\documentclass[11pt]{article}
\usepackage[margin=1in]{geometry}

\usepackage{amsthm}

\usepackage{setspace}
\usepackage{tabularx}
\usepackage{xcolor}
\usepackage{float,subcaption,placeins}
\usepackage{amsmath}
\usepackage{thm-restate}
\usepackage{float}
\usepackage{enumitem}
\usepackage{xr}
\usepackage{graphicx}
\usepackage{thmtools}
\usepackage{url}
\usepackage{bm}
\usepackage[normalem]{ulem}
\usepackage{balance}
\usepackage{lipsum}
\usepackage{etoc}
\usepackage{comment}
\usepackage{soul}
\usepackage{makecell}
\usepackage{hyperref}

\usepackage[capitalise,noabbrev]{cleveref}
\usepackage{multirow}

\usepackage{physics}

\newcommand{\Comments}{1}
\newcommand{\mynote}[2]{\ifnum\Comments=1\textcolor{#1}{#2}\fi}

\newcommand{\parbold}[1]{\vspace{.25em}\noindent\textbf{#1}}

\DeclareMathOperator{\thetaset}{\boldsymbol{\theta}}
\DeclareMathOperator{\full}{\textsc{full}}
\DeclareMathOperator{\subb}{\textsc{sub}}

\DeclareMathOperator{\prob}{\mathbb{P}}
\DeclareMathOperator{\expec}{\mathbb{E}}
\DeclareMathOperator{\indic}{\mathbf{1}}
\DeclareMathOperator{\N}{\mathcal{N}}

\newcommand\bbE{\ensuremath{\mathbb{E}}}
\newcommand\E{\ensuremath{\mathbb{E}}}

\newcommand*\diff{\mathop{}\!\mathrm{d}}

\renewcommand{\var}{\mathrm{Var}}

\usepackage[group-separator={,}]{siunitx}
\newcommand\blfootnote[1]{%
	\begingroup
	\renewcommand\thefootnote{}\footnote{#1}%
	\addtocounter{footnote}{-1}%
	\endgroup
}

\usepackage{pgfplots}

\pgfmathdeclarefunction{gauss}{2}{%
  \pgfmathparse{1/(#2*sqrt(2*pi))*exp(-((x-#1)^2)/(2*#2^2))}%
}

\usetikzlibrary{patterns}

\pgfplotsset{compat=1.10}
\usepgfplotslibrary{fillbetween}

\usepackage{xcolor}
\definecolor{babyblueeyes}{rgb}{0.63, 0.79, 0.95}
\definecolor{powderblue}{rgb}{0.69, 0.88, 0.9}
\definecolor{bluebell}{rgb}{0.64, 0.64, 0.82}

\usepackage{chngcntr}
\usepackage{apptools}
\AtAppendix{\counterwithin{lemma}{section}}

\usetikzlibrary{patterns}

\pgfplotsset{compat=1.10}
\usepgfplotslibrary{fillbetween}

\usepackage{xcolor}
\definecolor{babyblueeyes}{rgb}{0.63, 0.79, 0.95}
\definecolor{powderblue}{rgb}{0.69, 0.88, 0.9}

\usepackage{chngcntr}
\usepackage{apptools}
\AtAppendix{\counterwithin{lemma}{section}}

\usepackage[sort]{natbib}
\usepackage{amssymb}

\newtheorem{proposition}{Proposition}
\newtheorem{lemma}{Lemma}

\newtheorem{corollary}{Corollary}
\newtheorem{definition}{Definition}

\renewcommand{\cite}[1]{\citep{#1}}

\AtBeginDocument{%
  \providecommand\BibTeX{{%
    \normalfont B\kern-0.5em{\scshape i\kern-0.25em b}\kern-0.8em\TeX}}}

\begin{document}

\title{Dropping Standardized Testing for Admissions Trades Off Information and Access}

\author{
	Nikhil Garg\\
	Cornell Tech\\
	\texttt{ngarg@cornell.edu} \\
	\and
	Hannah Li\\
	Columbia University\\
	\texttt{hannah.li@columbia.edu} \\
	\and
	Faidra Monachou\\
	Yale University\\
	\texttt{faidra.monachou@yale.edu}\blfootnote{The journal version is available at \url{https://pubsonline.informs.org/doi/abs/10.1287/mnsc.2023.02573}.
The authors are extremely grateful to Itai Ashlagi for insightful discussions and suggestions throughout this project.
We also thank Jack Buckley, Sharad Goel, Josh Grossman, Ramesh Johari,  Wanyi Li, Irene Lo, Muriel Niederle, Collin Raymond, Alvin E. Roth, Philipp Strack, Sabina Tomkins, Johan Ugander, Gabriel Weintraub, and anonymous reviewers of Management Science, ACM FAccT 2021, and ACM EAAMO 2021 for helpful comments. 
This research uses public data from the Texas Higher Education Opportunity Project (THEOP) and acknowledges the following agencies that made THEOP data available through grants and support: Ford Foundation, The Andrew W. Mellon Foundation, The William and Flora Hewlett Foundation, The Spencer Foundation, National Science Foundation (NSF Grant SES-0350990), The National Institute of Child Health \& Human Development (NICHD Grant R24 H0047879) and The Office of Population Research at Princeton University.}}

\maketitle

\begin{abstract}
We study the role of information and access in capacity-constrained selection problems with fairness concerns. We develop a statistical discrimination framework, where each applicant has multiple features and is potentially strategic. The model formalizes the trade-off between the (potentially positive) informational role of a feature and its (negative) exclusionary nature when members of different social groups have unequal access to this feature. 
Our framework finds a natural application to policy debates on dropping standardized testing in admissions. Our primary takeaway is that the decision to drop a feature (such as test scores) cannot be made without the joint context of the information provided by other features and how the requirement affects the applicant pool composition. Dropping a feature may exacerbate disparities by decreasing the amount of information available for each applicant, especially those from non-traditional backgrounds. However, in the presence of access barriers to a feature, the interaction between the informational environment and the effect of access barriers on the applicant pool size becomes highly complex.  
Furthermore, we consider an extension with two schools and costly tests, where strategic students decide whether to take the test or not. Our theoretical results reveal that the students' test-taking behavior can be non-monotonic. We characterize the two-school policy equilibria and show that each school's optimal decision to drop the test critically depends on the other school's test policy.
Finally, using calibrated simulations, we demonstrate the presence of practical instances where the decision to eliminate standardized testing improves or worsens all metrics. 
\end{abstract}

\section{Introduction}

Recent debates on the use of standardized testing in college admissions have increasingly garnered national attention, initially during the COVID-19 pandemic as test centers shut down and schools were forced to reconsider their admissions practices \citep{covidtests2020}. Independently of the COVID-19 pandemic, in an attempt to increase equity and diversity in admissions, the University of California (UC) 
settled a lawsuit by eliminating all consideration of SAT and ACT scores for admissions and scholarships, 
following an earlier decision to suspend testing requirements and ultimately design its own test \citep{delrio_2021}.
{Most recently, in response to the United States Supreme Court ruling to end race-based affirmative action, more colleges are expected to drop those requirements permanently, ``responding to critics who say the tests favor students from wealthier families'' and at the same time, protecting schools from lawsuits~\citep{saul_2023}.} On the other hand, schools such as MIT 
reinstated standardized testing requirements that were dropped during the initial years of the pandemic~\citep{mit2022}.

These discussions primarily center on highly selective institutions and their efforts to shape the student body through the admissions process.\footnote{
	We analyze settings in which the capacity constraint for a school means that they must be selective about the students that they admit, i.e., that the school can accept less than half its applicants. While most students do not attend such colleges in the United States, selective college admissions are the subject of considerable academic study, partially due to their importance in improving downstream outcomes for the students who do attend them \citep{missingoneoffscollegeundermatching, bleemer2022affirmative, chetty2023diversifying, tomkins2023showing, grossman2024disparate}.
	Most schools are not selective, accepting most applicants. The admissions considerations of these schools differ substantially from those of more selective institutions~\citep{selingo2020whogetsin}. Modeling nonselective admissions for schools who accept the vast majority of their applicants would require specifying objective functions for how the school trades off class size with academic merit, diversity, and other desiderata of the admitted class; in practice, many schools accept all students who meet a minimum academic requirement such as high school graduation. More generally, given such an objective function trading off class size, results similar to ours may hold; for example, our model can be viewed as analyzing one extreme of such a model in which the admitted class size is an exact equality constraint.} These schools promise great opportunities to their students, but---due to perceived capacity constraints---limit their acceptances to students that they deem to have high potential 
in academics, athletics, creative endeavors, or leadership and service~\citep{espenshade2013no}. They typically attempt to identify these students through a combination of standardized tests, high school grades, letters of recommendation, personal essays, and extracurricular activities~\citep{zwick2002fair,espenshade2013no}. 

The question is whether each of these components, and the application as a whole, allows the schools to assess individuals from different backgrounds effectively and `fairly,' including students from different racial, ethnic, and socioeconomic groups. {Implicitly, the debate concerns how to \textit{design} an admission policy to aid fair and efficient decision-making, in terms of both deciding which information to collect from applicants and how to use this information.  
Our exposition focuses on the context of college admissions; however, our model and the questions we ask are more broadly applicable to other settings of information design and fair decision-making in capacity-constrained settings, 
such as labor markets, award committees, and {social welfare programs}.\footnote{{Our model applies to selection problems where there is a trade-off in the value of additional information and the fraction of applicants who can provide it. For example, in means-testing welfare programs, requiring long forms might help in better targeting benefits but  might also discourage eligible recipients from applying~\citep{welfare2021}.}}
In each of these cases, the decisions are being made on limited information but have far-reaching consequences for employment or education opportunities. Thus, it is important to analyze these policies and their potential disparate impact across different groups of applicants.}

\parbold{Background.} {A high-profile debate has surrounded the use of standardized testing for admissions, in which social scientists and education experts have highlighted specific fairness concerns.}
Test critics argue that tests exhibit racial gaps \citep{reardon2011widening} and reinforce inequality in higher education \citep{reeves2017race}. \citet{espenshade2013no} find that only 8\% of lower-income, compared to 78\% of high-income students, use a test preparation service. %
The testing process is expensive and time-consuming; \cite{hyman2016act} finds that
``for every ten poor students who score college-ready on the ACT or SAT, there are an additional five poor students who would score college-ready but who take neither exam'' and so cannot apply to colleges that require it.\footnote{{After UC Berkeley eliminated GRE requirements, 
``while overall graduate applications have increased 19 percent when compared to [the 2019-2020 cycle], the number of under-represented minority (URM) doctoral applicants increased by 42 percent and URM applicants to academic master’s programs increased by 82 percent'' \citep{berkeley2021}. }}  \cite{rothstein2004college}, correcting for selection biases, finds that the SAT correlates with high socioeconomic status, and that its orthogonal predictive power and thus weight given to it in an optimal predictor is low.

On the other hand, supporters of testing argue that it is ``a systematic means of collecting information,'' thereby contributing to decision-making when used appropriately~\citep{phelps2005defending}. Some supporters claim that tests actually \textit{help} schools evaluate under-represented minorities; in the absence of standardized testing, ``a capable student from a little-known school in the South Bronx may be more challenging to evaluate,'' further benefiting students from privileged---and historically familiar---backgrounds~\citep{bellafante_2020}. A report released by University of California explicitly uses the language of precision and predictive power of test scores compared to other features: ``The predictive power of the standardized test scores is higher for those student groups who are under-represented [\dots] Thus, consideration of test scores allows campuses to select those students from under-represented groups who are more likely to earn higher grades and to graduate on time [\dots] One implication is that consideration of test scores allows greater precision when selecting from [under-represented minority] populations''~\citep{berkeleyreport}. MIT in 2022 reinstated the SAT, highlighting that their 
    ``research shows standardized tests help us better assess the academic preparedness of all applicants, and also help us identify socioeconomically disadvantaged students who lack access to advanced coursework or other enrichment opportunities that would otherwise demonstrate their readiness'' \citep{mit2022}. {Other application components such as recommendation letters~\citep{dutt2016gender} and application essays~\citep{alvero_21} may also be unreliable.}\footnote{{For example, letter writers use different language to describe women and other under-represented groups, giving weaker recommendations~\citep{dutt2016gender}, and  application essays have a stronger correlation to reported household income than do SAT scores~\citep{alvero_21} (although they are not necessarily differentially scored).}} A school that does not consider test scores must rely more heavily on these components.

\parbold{Research questions.} The competing claims from critics and supporters largely center around two issues: \textit{access} and \textit{information}. 
We develop a model to capture these arguments in favor of and against dropping test scores and formalize the underlying trade-off. 
The model considers a Bayesian school that wishes to admit students based on their skill level, which we refer to as ``academic merit,'' and also values the ``diversity'' of the admitted class. Not every student applies to a school that requires testing---they may face group-dependent barriers or costs to applying. The school admits applicants to meet a capacity constraint and tries to maximize the average academic merit of the accepted cohort. However, it has imperfect knowledge of the students skills and instead must rely on 
noisy and potentially biased signals, one of which is the test score. The school decides whether to require the test score; the decision affects both who applies and how the school evaluates applicants.

We then provide a framework for evaluating potential trade-offs in these decisions. In particular, alongside the \textit{academic merit} objective, we analyze two fairness notions: \textit{diversity} and \textit{individual fairness}. The former captures group-level disparities. The latter quantifies disparities in individual opportunities, by measuring the difference in the admissions probability between two individuals of equal skill but different demographic groups. We focus on the trade-off between two effects:

\begin{description}
	\item[Differential informativeness.] Colleges often have better information---through, e.g., familiar letter writers and transcripts---on students from privileged backgrounds, and so can better estimate their true academic merit. Standardized testing reduces this measurement gap, and so especially helps identify well-qualified, non-traditional students. 

	\item[Applicant pool composition due to disparate access and strategic behavior.] Some students---especially those from disadvantaged backgrounds---either do not take standardized tests or do not report their scores,\footnote{{A University of California report on testing  states that under-represented students might be discouraged from applying based on their score, even if their score would be competitive~\citep{berkeleyreport}.}} due to cost and other exogenous access barriers. Without a test score, students cannot apply to a school with a test requirement, even if they are well-qualified. Dropping the requirement thus expands the applicant pool {but also alters its composition at different rates across groups.} 

    We further study when applicant composition results from \textit{strategic}  decisions by students, who choose whether to pay testing 
    costs, as a function of their other features.

\end{description}

\parbold{Contributions.} 
Given these effects, we study: \textit{Under what settings of informativeness and disparate access should standardized testing be dropped, if a college values both diversity and academic merit? Furthermore, what is the effect on these metrics when students can overcome disparate application costs, i.e., when students are strategic and schools may differ in their testing requirements?} 
To the best of our knowledge, our paper is the first theoretical study examining the impact of eliminating testing requirements in college admissions.

Modeling-wise, we introduce a Bayesian model that extends the classic statistical discrimination theory by \cite{phelps1972statistical} to include {multiple} application components, access assymetries to some feature and potentially strategic student behavior and multiple schools (see Section~\ref{sec:related} for a more detailed comparison). Our multi-feature model   allows us to study the \textit{design} of the information structure used in a selection process, and provide a testable framework for reasoning about how the new feature would interact with the current set of features, including when applicants can make strategic decisions. More broadly, we thus believe that our work provides a useful conceptual framework of independent interest, for studying emerging problems in fair decision-making and public policy. 

From a technical perspective, we formalize a trade-off between informativeness and access, two basic arguments in favor of and against the inclusion of a given feature, and show how the set of features required influences the admitted class's academic merit and diversity, through these competing effects.
{Our main technical insight shows that differences in the total \textit{variance} of features lead to information disparities across groups: even though the school manages to correct for the existing mean bias in the features of different groups, it is generally impossible to correct for variance---this variance effect is thus central when considering the set of features to use.} 
We characterize the settings where dropping test scores introduces a trade-off between diversity and academic merit and where it simultaneously improves or worsens all objectives. 

We further extend the model to consider the effect of students' strategic test-taking behavior and two schools simultaneously admitting students. Students can choose to pay (potentially heterogeneous) costs to take the test and apply to a school that requires it. At equilibrium,  students self-select to apply to a test-based school if their perceived probability of admission outweighs their relative cost-to-valuation ratio. We find that such strategic behavior disproportionately affects the applicant pool composition but not always at the expense of the group facing higher test costs. Additionally, in the case of two schools, where only the top school requires the test, we uncover an interesting discontinuity in the students' self-selecting behavior, which in turn leads to a potential mismatch between academic merit and the ranking of the school.
More broadly, we analyze schools' strategic incentives for whether to require the test, given the behavior of strategic students and the behavior of the other school---here, even a school that is more preferred by all applicants has best response strategies that differ based on the policies of the other school. 

Finally, we use our model to perform calibrated simulations, using student application and transcript data from the University of Texas at Austin. 
Our results establish that there exist practical settings both in which dropping testing concurrently worsens or improves all metrics, and that such effects especially depend on the strategic behavior of potential applicants. Thus, our primary takeaway for practice
is that the decision to drop testing cannot be made without jointly considering the interaction between the information provided by other features \textit{relative to test scores} and how dropping the test requirement affects the applicant pool composition. This interaction between information and access is complex.

\parbold{Organization.}  Section~\ref{sec:related} discusses the related literature. Section \ref{sec:baseline_model} introduces our baseline model. Section \ref{sec:estimatesintuition} provides intuition on the effect of informativeness and test access in our model. Section \ref{sec:single} formalizes a trade-off between informativeness and access when students may face access barriers to taking the test. Section \ref{sec:strategic} extends the model to include students' strategic test-taking behavior and two schools.  \Cref{sec:conclusion} concludes. Proofs, additional results, and simulations with UT Austin data are in the Electronic Companion.
\subsection{Related Work}
\label{sec:related}

Our work broadly relates to the study of discrimination and admissions in the economics and fair machine learning and operations communities.

\smallskip
\noindent\textbf{Economics of discrimination.} 
In economics, there are two lines of  related work: discrimination theories~\citep{becker1957economics}, especially statistical discrimination~\citep{ arrow1971theory, phelps1972statistical}, as well as theoretical models of affirmative action in student admissions~(e.g., \cite{ chan2003does, abdulkadirouglu2005college, avery2006cost, epple2006admission, fu2006theory, chade2014student, kamada2019fair,   fershtman2020soft}). There is also an important line of empirical work investigating the implications of affirmative action (e.g., \cite{arcidiacono2011does, backes2012affirmative,  bagde2016does,bleemer2022affirmative}) and race-neutral alternatives such as top percent plans and holistic reviews (e.g., \cite{long2004race,   kapor2020distributional,ellison2021efficiency, bleemer2023affirmative}).

From a conceptual viewpoint, our work is most closely related to the statistical discrimination theory of \cite{phelps1972statistical}, which---surprisingly---is rarely adopted in the admissions literature (except \citet{ kannan2019downstream, emelianov2020fair}). 
\citet{emelianov2020fair} 
 use Phelps' model
to study how differential variance of a single feature affects the admissions decisions of a school that greedily admits students with the highest test scores, without factoring in the differential variance. 

Both our work and \citet{emelianov2020fair} adopt the seminal theory of statistical discrimination  \citep{phelps1972statistical}. 
However, our work moves beyond \cite{emelianov2020fair} and \cite{phelps1972statistical}, as well as the standard matching-based approach of other theoretical models  (e.g., \cite{chade2014student, abdulkadirouglu2005college,karni2021fairness}),  in several ways. 
To our knowledge, 
 our paper is the first to extend Phelps' model to {multiple} features with non-identical distributions and access asymmetries to some feature. We further combine such statistical discrimination with a model of strategic student behavior. These modeling contributions allow us to study the complex interactions between the test and several other factors, including the remaining application components, access barriers and test costs (that induce student strategic behavior). Furthermore, our multi-feature model allows the decision-maker to potentially remove a feature, thus enabling us to reason about policy changes such as dropping standardized testing in a tractable manner. On the other hand, \citet{emelianov2020fair} include an effort component: in their framework, candidates have the ability to increase the mean of their single feature at a quadratic cost. Their finding that affirmative action can enhance both diversity and academic merit arises from the balancing of average efforts across groups in certain equilibria.

\parbold{Fairness in machine learning and operations.} Recent machine learning work applies fairness notions to admissions and related allocation problems, studying implicit bias~\citep{kleinberg2018selection, emelianov2020fair, faenza2020impact}, downstream effects~\citep{kannan2019downstream}, grade signaling~\citep{immorlica2019access}, greenlining~\citep{borgs2019algorithmic}, school choice~\citep{allman2022designing}, bus scheduling~\citep{banerjee2019incorporating}, and classification algorithms~\citep{hu2019disparate, liu2020disparate}. More broadly, our work contributes to the emerging literature on fairness in operational contexts (e.g., \cite{bertsimas2011price, monachou2019discrimination,  baek2021fair, kallus2021fairness, manshadi2021fair,  cohen2022price, sinclair2022sequential}), especially with respect to equity in education~\citep{smilowitz2020use}.

A line of literature specializes on different types of barriers for students, including implicit bias \citep{faenza2020impact}  and when only one group can take the test multiple times \citep{niu2022best}. 
These barriers affect the treatment of applicants, but do not prevent students from even applying, as is our focus in our baseline model. In relation to our strategic setting, note that \cite{faenza2020impact} do not consider strategic students. \cite{niu2022best} allow students to decide whether to take the test twice or not, but their model does not include costs and students have only binary skill levels.

A follow-up  paper~\citep{liu2021test} extends our model to provide (im)possibility results under test-optional policies (see also \cite{dessein2023test}). \cite{castera2024correlation} also build upon our work to study disparities due correlations across two schools in how they evaluate a student.   
Using data from the Education Longitudinal Study of 2002, \cite{borghesan2022heterogeneous}  finds that banning the SAT leads to a small increase in the population of low-income students but has a negligible effect on under-represented minority students.

\section{Model}
\label{sec:baseline_model}
{We develop a model where the school can design their admissions procedure and, in particular, choose the information that it requires the applicants to submit.}

We consider 
a continuum of students and a  single {school}. A unit mass 
of students is applying to college. Each student belongs to a group $g \in \{A, B\}$, and the mass of students in group $B$ is $\pi$.
Each student has a latent (unobserved) \textit{skill level} $q$, Normally distributed according to $\mathcal{N}(\mu, \sigma^2)$ identically for each group, 
as well as a set of observed \textit{features} $\thetaset = (\theta_1, \ldots, \theta_K)$. Each $\theta_{k}$ is a noisy  function of $q$, i.e.,  $\theta_k =  q + \epsilon_{k}$, $k = 1, \ldots, K$,  with Gaussian noise $\epsilon_{k}~\sim~N(\mu_{g k}, \sigma_{g k}^2)$. The distribution of noise $\epsilon_{k}$ is feature- and group-dependent, but each $\epsilon_{k}$ is drawn independently across features and students. Features represent application components like recommendation letters, grades, and test scores.

Students differ in their \textit{access} to the features. {When a student does not have access to feature $K$, then they cannot apply to a school that requires it.}  In our primary model, only a fraction $\gamma_g$ of group $g \in \{A, B\}$ has access to the \textit{full} set of features $\full = \{1, \dots, K\}$, i.e., $\thetaset = (\theta_1, \dots, \theta_K)$; the remainder only has access to the \textit{subset} $\subb = \{1, \dots, K-1\}$. 
Whether a student has access to all features is independent of their skill $q$ and conditionally independent of the feature values given  group membership. 
In \Cref{sec:strategic} we consider a setting where students are \textit{strategic} about whether to take the test.

\smallskip
\noindent\textbf{Admissions policy.} 
{We now turn to the question of interest: the design of the admissions policy.} The school admits a mass $C<1$  to fill its capacity. 
The admissions procedure consists of a feature requirement policy, skill estimation, and selection given estimates. 

The feature requirement policy choice is whether to require the {full} set of features or the {subset}. If it requires the full set, then students without full access cannot apply.
If it only requires the subset, then it observes only  that subset for each student. 
Then, given a student's features $\thetaset$, the school estimates a \textit{perceived skill} $\tilde q$ of their true skill $q$. 
The school is Bayesian, knows the distribution of $q$ and the (group-dependent) distributions of $\epsilon_k$, and is \textit{group-aware}: it can use the student's group membership in constructing its estimate. 
The resulting Bayesian estimate is the `best' one can do, given the available information:
$$\tilde q(\thetaset, g) \triangleq \bbE[ q \mid \thetaset, g].$$

After estimating the skill level of each applicant, the school selects  the mass $C$ of students with the highest skill estimates $\tilde q$. This selection process induces a \textit{threshold} $\tilde q^*_S$ such that applicants with perceived skill above the threshold are admitted.

Holding the estimation and selection policies fixed, the \textit{admissions policy} $P_S$ is determined by the required feature set $S$.

\smallskip
\noindent\textbf{Academic merit and fairness metrics.}
 We evaluate a policy $P$ using three metrics on the admitted class. Let $Y\in\{0,1\}$ denote the admission decision for a given student; $Y=1$ means that the student is admitted.

\textbf{Academic merit} $\E[q \mid  Y=1, P] $, the expected skill level of accepted students. We also use group-specific measures, $\E[q \mid  Y=1, g, P] $.

\textbf{Diversity level} $\tau (P)$, the fraction of students admitted that are of group $B$. Policy $P$ satisfies \textit{group fairness} if and only if the  
fraction matches the population: $\tau (P) = \pi$.

\textbf{Individual fairness gap} $I(q; P)$, the difference in admissions probability between two students of identical true skill $q$, one belonging to group $A$ and the other to group $B$:
\begin{equation*}
	I(q ; P) \triangleq  \prob\left(Y=1\mid q, A, P\right) - \prob\left(Y=1 \mid q, B, P\right).
\end{equation*}
Policy $P$ satisfies \textit{individual fairness} if and only if the gap is $0$ for all skill levels $q$.

\smallskip
\noindent\textbf{College admissions and relationship to practice.} 
While our model and results are more general, our exposition primarily considers college admissions in the United States and the debate to drop standardized testing {as our main example}. We focus on how policies differentially affect {privileged} (group $A$) versus {disadvantaged} (group $B$) students.

We refer to the potentially inaccessible last feature $\theta_K$ as the \textit{test score} of a student in a common standardized exam like the SAT or ACT, and assume that more privileged students have access to testing; 
as~\citet{hyman2016act} notes, many well-qualified disadvantaged students do not have access to standardized tests and so cannot apply to schools that require them.
On the other hand, as the~\citet{berkeleyreport}, \citet{bellafante_2020}, and \citet{mit2022} posit, without testing it may be  especially difficult to evaluate students from non-traditional backgrounds, as colleges instead rely on transcripts and recommendations from familiar (privileged) high schools. This aspect could be captured---as we do for our simulations---by considering the first $K-1$ features as substantially more informative for group $A$ ($\sigma_{Ak} < \sigma_{Bk}$), with a smaller informativeness discrepancy for the test score. 

The model's focus differs from feature bias as traditionally understood, if a feature systematically under-values one group; e.g., weaker letters of recommendation for under-represented students. In our model, the school fully corrects for such bias (cancelling out $\mu_{gk})$; in practice, schools interpret signals in context, for example, by benchmarking how many AP courses are offered by a student's school.  
In contrast, differential informativeness (a function of $\sigma_{gk})$ and disparate access ($\gamma_g$) are harder to correct at admissions time. The former represents an information-theoretic limit to identifying the most qualified students, and the latter prevents some students from even applying. These effects cannot even be completely mitigated using affirmative action, 
which is  insufficient in identifying qualified disadvantaged students. We study affirmative action in Electronic Companion~\ref{sec:affirmative_action}.

Without loss of generality, we assume that the features are less informative for group $B$ than they are for group $A$. Specifically, under policy $P_S$ let \textit{unequal precisions} between groups mean $\sum_{k\in S} \sigma_{Ak}^{-2} > \sum_{k \in S} \sigma_{Bk}^{-2}$, and  \textit{equal precision} mean  $\sum_{k\in
 S} \sigma_{Ak}^{-2} = \sum_{k\in S} \sigma_{Bk}^{-2}$. In settings with barriers, we assume that group $A$ also has more access to the test, i.e., $\gamma_A \geq \gamma_B$.\footnote{We further assume that, even in the presence of barriers, the market is \textit{over-demanded}, 
 i.e., $C< (1-\pi)\gamma_A +\pi \gamma_B$.} Finally, the school is \textit{selective} with capacity $C<1/2$. These assumptions are for exposition; our model's tractability allows us to solve analogously for the omitted cases.

\Cref{sec:strategic} extends our model to one in which students make a strategic decision to take the test as a function of their admissions probability and the test cost that differs across groups, in both single- and two-school settings. 

\section{Intuition: The role of differential informativeness} 
\label{sec:estimatesintuition}

\begin{figure}
	\centering
	\resizebox{0.7\columnwidth}{!}{
	\begin{tikzpicture}
		\begin{axis}[
			no markers, domain=0:10, samples=100,
			axis lines*=left, xlabel=$\tilde{q}$, ylabel=\empty,
			axis y line=none,
			every axis x label/.style={at=(current axis.right of origin),anchor=west},
			height=5cm, width=12cm,
			xtick=\empty, ytick=\empty,
			enlargelimits=false, clip=false, axis on top,
			grid = major,
			legend style={at={(0.75,0.8)},anchor=west, font=\small},
			legend cell align={left}
			]

			\addlegendimage{line width=0.3mm, dashed, color=black}
			\addlegendentry{$q$}
			\addlegendimage{line width=0.3mm,  color=green!50!black}
			\addlegendentry{$\tilde q \mid A, P_S$}
			\addlegendimage{line width=0.3mm,  color=magenta!80!black}
			\addlegendentry{$\tilde q \mid B, P_S$}

			\addplot [fill=powderblue!70, draw=none,domain=6.5:10, name path =D] {gauss(5,1)} \closedcycle;
			\addplot [very thick,black, dashed, name path=C] {gauss(5,1.5)};
			\addplot [very thick,green!50!black, name path=A] {gauss(5,1)};
			\addplot [very thick,magenta!80!black, name path=B] {gauss(5,0.7)};

			\node[below] at (axis cs:6.5, 0)  {$\tilde{q}^*_{S}$};
			\node[below] at (axis cs:5, 0)  {$\mu$};

		\end{axis}
	\end{tikzpicture}
	}
	\caption{The distribution of skill estimates $\tilde q$ at an aggregate level for each group, as it depends on the \textit{informativeness} of the features. When the application components are more precise for one group (group $A$, in green), the variance in the skill estimates of their group is \textit{higher}---there is more signal for individuals to demonstrate that their skill is different than the mean. Then, more group $A$ students have high skill estimates above threshold $\tilde q^*_S$, and thus more are admitted. This effect occurs even though the true skill $q$ distribution is identical across groups. If dropping the test causes such differential informativeness, then doing so may worsen both fairness and academic merit (estimated skill of admitted students). Figure~\ref{fig:2destimates} illustrates how the differential informativeness interacts with \textit{disparate access}, due to which dropping test scores may improve all objectives. 
}
	\label{fig:normals_intuition}
\end{figure}

We begin our analysis in Section~\ref{sec:skillestimation} by deriving how a Bayesian-optimal school estimates the students' skill level.
Then, 
we preview our main results, illustrating how the relationship between skill estimates and true skills of the applicant pool depends on the informativeness  of features and the access barriers, with implications for how admissions differ by group.

\subsection{School's optimal Bayesian estimation procedure}

\label{sec:skillestimation}

Our Bayesian school---with knowledge of the model's feature noise means and variances---observes each student's features and group membership and estimates their expected skill level, using properties of Normal distributions. Repeating this process for all applicants induces the following distribution of skill level estimates for each group. 

\begin{restatable}[Estimated skill]{lemma}{lemperceivedskilldisttribution}
	\label{lemma:perceived_skill_distribution}
	Consider a school that uses feature set $S \subseteq \{1, \dots, K\}$ for each applicant. 
 {Then, the perceived skill of an applicant in group $g\in\{A,B\}$ with feature values $\thetaset = (\theta_k)_{k\in S}$ is:} 
 \begin{equation}
 \label{eq:skill_estimate}
 	\tilde q(\thetaset, g) = 
 	\frac{\mu \sigma^{-2} + \sum_{{k\in S}}(\theta_{k} - \mu_{gk})\sigma_{gk}^{-2}} {\sigma^{-2} +\sum_{{k\in S}} \sigma_{gk}^{-2}}.
 \end{equation}
Further, the skill level estimates for students in group $g$ are Normally distributed:
\begin{equation}
\label{eq:estimate_distribution_group}
	{\tilde q \mid  g, P_{S}} \sim \mathcal{N}  \left( \mu, \sigma^{2}\left[\frac{\sum_{{k\in S}} \sigma_{gk}^{-2}  }{\sigma^{-2} +\sum_{{k\in S}} \sigma_{gk}^{-2}}\right]\right).
\end{equation}
\end{restatable}

As \Cref{eq:skill_estimate} shows,\footnote{Note that \Cref{eq:skill_estimate} is a direct generalization of \cite{phelps1972statistical} from a single to $K$ features.} when the school estimates the skill level $\tilde q (\thetaset, g)$ of an individual and knows the skill and feature noise distributions, it perfectly cancels out the mean bias terms $\mu_{gk}$ such that they do not affect estimation.\footnote{\cite{berkeleyreport}: ``test scores are considered in the context of comprehensive review, which in effect re-scales the scores to help mitigate between-group differences.''} {The school also re-weights each feature $\theta_k$ proportionately to the relative informativeness of this feature for group $g$:  the less informative a feature is for a group (smaller \textit{precision} $\sigma_{gk}^{-2}$), the less it contributes to estimates. Thus, due to differences in $\sigma_{gk}^{-2}$ across groups, two students from different social groups with the same features $\thetaset$ are evaluated differently. However, even in this idealized scenario, the school cannot fully correct for the variance terms $\sigma^2_{gk}$;
two students with same skill $q$ but in different groups have different skill estimates in expectation.} 

These individual estimation effects accumulate at the group level (\Cref{eq:estimate_distribution_group}) {and drive our results on disparities}. The school {knows} that $q \sim \mathcal{N}(\mu,\sigma^2)$ is identically distributed across social groups. However, as illustrated in Figure~\ref{fig:normals_intuition}, the distribution of its skill estimates ${\tilde q \mid  g, P_{S}}$  
can differ across groups. For each group, the skill estimates are regularized toward the mean skill level $\mu$. The regularization strength depends on the total precision $\sum_{{k\in S}} \sigma_{gk}^{-2}$: the larger the total precision for a group is (or the more informative its features are), the higher the variance in the estimated skills for that group is. 
In Figure~\ref{fig:normals_intuition}, group $A$ has larger total precision and for any value $\bar q>\mu$, there is a larger mass of students from group $A$ than $B$ with estimated skill higher than $\bar q$. Thus a school with capacity $C<\frac{1}{2}$ admits more students from group $A$.

\subsection{Intuition for the impact of admissions policy}
\label{sec:intuition}
\begin{figure*}[t]
	\begin{subfigure}[b]{0.325\textwidth}
		\centering
		\includegraphics[width=1\linewidth]{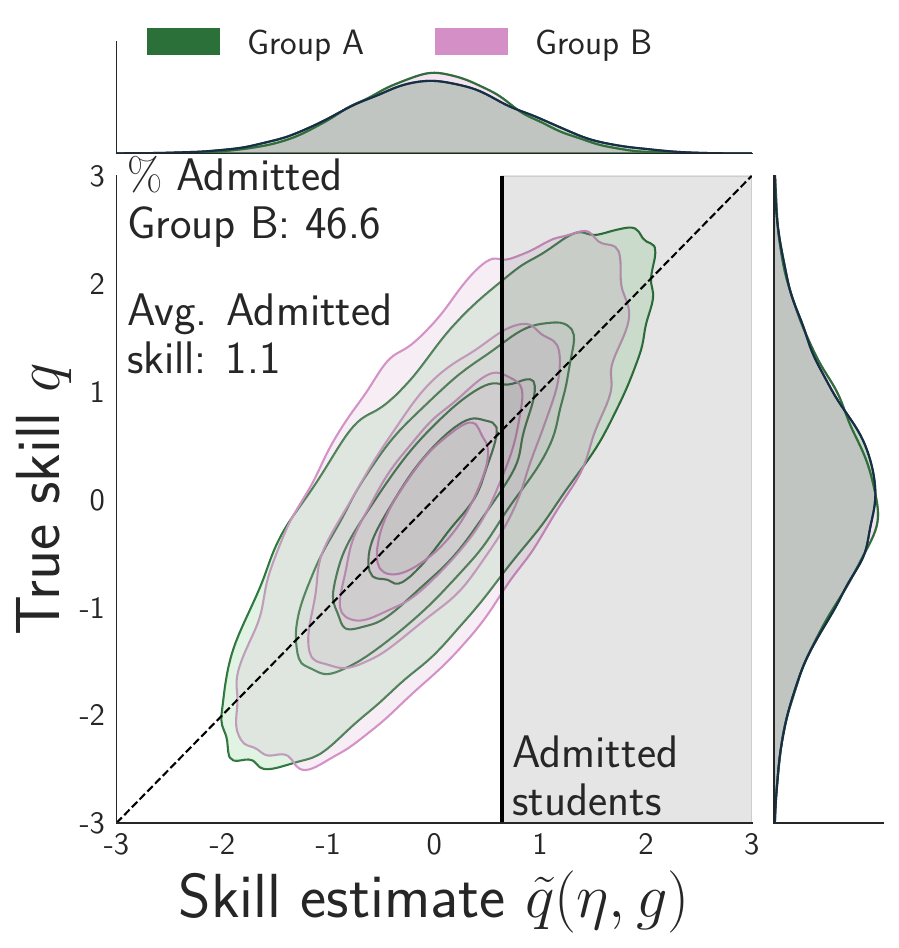}
		\caption{\small With test, no barriers \normalsize}
		\label{fig:2dcontour_ideal}
	\end{subfigure}
	\hfill
	\begin{subfigure}[b]{0.325\textwidth}
		\centering
		\includegraphics[width=1\linewidth]{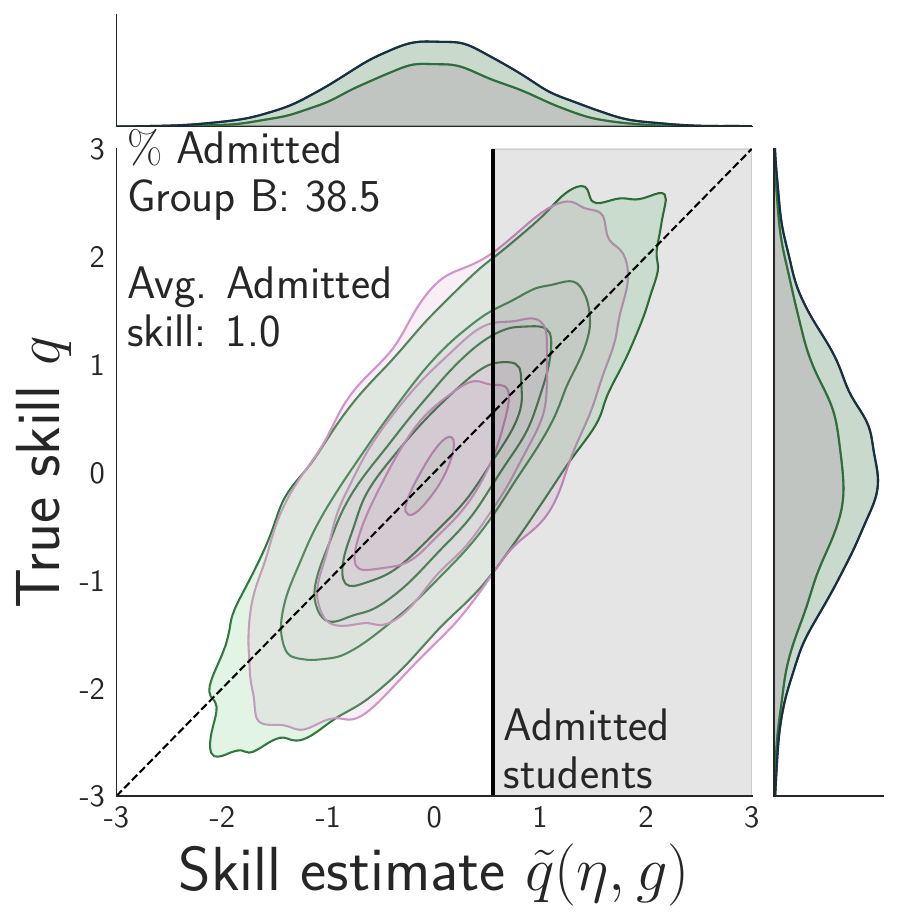}
		\caption{\small With test, inaccessible to some \normalsize}
		\label{fig:2dcontour_barriertest}
	\end{subfigure}
	\hfill
	\begin{subfigure}[b]{0.325\textwidth}
		\centering
		\includegraphics[width=1\linewidth]{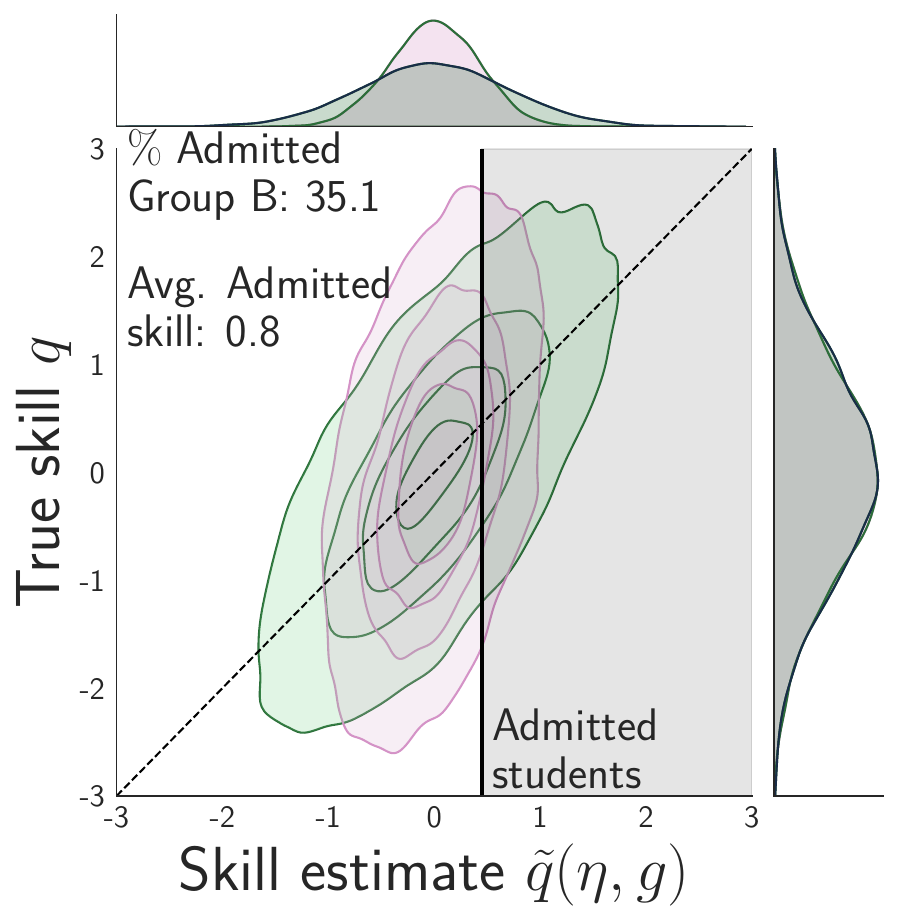}
		\caption{\small Without test \normalsize}
		\label{fig:2dcontour_barrierfree}
	\end{subfigure}
	
	\caption{Skill vs estimate joint distribution for each group.  Above and to the right of each joint distribution we plot the corresponding marginal distributions---e.g., the plot on the right of each joint distribution corresponds to the true skill distribution, which is equal across groups. The dashed diagonal lines correspond to perfect estimation. Figure~\ref{fig:2dcontour_ideal} represents a world without access barriers and when the features are approximately equally informative across groups. Figure~\ref{fig:2dcontour_barriertest} illustrates the consequences of requiring the test when group B (in pink) has access barriers: fewer can apply and so can be admitted. Figure~\ref{fig:2dcontour_barrierfree} illustrates potential consequences of dropping the test: the school may be unable to distinguish among group B applicants, leading to worse estimates (rotated away from diagonal) and fewer admitted from that group. Full parameters  
    are in Electronic Companion~\ref{appsec:simparams}.
    }

	\label{fig:2destimates}
\end{figure*}

Before proceeding to our main results, we first illustrate our primary insight regarding the trade-off between informativeness and the applicant pool size.
In Figure~\ref{fig:2destimates}, each sub-figure shows, for one scenario, the joint distribution between true skill $q$ and the corresponding skill estimates $\tilde q$ for each group, along with the respective marginal distributions.
Since both groups have identical true skill distributions, the joint distributions would  \textit{ideally} be identical for the two groups (and perfectly aligned along the diagonal) and group $B$ would comprise a proportion $\pi$ of the admitted class.

Consider the case where the potentially dropped feature (the ``test score") is equally informative for both groups, whereas the remaining features are more informative for group $A$.
Figure~\ref{fig:2dcontour_ideal} illustrates the scenario when there are no access barriers to the test.  Due to the differential informativeness induced by the other features, (slightly) more group $A$ students are admitted: the college can better estimate their true skill, as illustrated by the group $A$ joint distribution being closer to the diagonal.
 Figure~\ref{fig:2dcontour_barriertest} illustrates the consequences of requiring test scores in the presence of unequal access levels ($\gamma_A=1$ and $\gamma_B=\frac{2}{3}$). \textit{Among those who apply}, the college can estimate their true skill as well as it could in Figure~\ref{fig:2dcontour_ideal}. However, fewer group $B$ students can apply, as indicated by the smaller marginal count histogram, and so fewer are admitted.
Figure~\ref{fig:2dcontour_barrierfree} illustrates a scenario where the school removes the test score. Estimates for both groups are worse, as reflected in the joint distributions being further from the perfect estimation diagonal. However, skill estimates for group $B$ students are especially degraded as their other features may be less informative, and so they make up a smaller proportion of the admitted class. Whether the effect in Figure~\ref{fig:2dcontour_barriertest} or \ref{fig:2dcontour_barrierfree} dominates depends on the parameter context.

\section{Analysis of the baseline model}%
\label{sec:single}

We now apply the insights from Section~\ref{sec:estimatesintuition} to our baseline admissions model. We show that differences in informativeness alone (Section~\ref{sec:baseline_policy}) generate disparities in academic merit, diversity, and individual fairness. We then compare admissions \textit{with} and \textit{without} a given feature (Section~\ref{sec:dropping_tests}), showing that under full and equal access to testing, removing information may further decrease both fairness and academic merit. Under unequal access, however, a trade-off emerges between test-imposed barriers and the value of the information the test provides. We characterize the school’s optimal testing policy accordingly.

\subsection{Informational effects of fixed testing policies}
\label{sec:baseline_policy}

In general, our fairness notions are not achievable, \textit{even though both groups have the same true skill distribution}. We  study how differential informativeness affects our three metrics.

\begin{restatable}[Metrics with a fixed policy]{proposition}{propgroupawarenoAA}
\label{prop:group_aware_noAA}
\label{prop:appendix_baseline}
Suppose that a selective school uses admissions policy $P_{S}$.
\emph{Group fairness} and \emph{individual fairness} fail except for equal precision, even   
in the absence of barriers. 
Given \textit{unequal precisions}:

\begin{itemize}
	\item [(i)] \emph{Diversity level}:  Group $B$ students are under-represented, i.e.,  $\tau(P_{S}) <\pi$. 
    Furthermore, a larger informativeness gap leads to decreased diversity:
    Fix group $B$ precision, $\sum_{k\in S} \sigma_{Bk}^{-2}$. Then, as group $A$ precision increases, the diversity level $\tau(P_{S})$ decreases.

	\item[(ii)] \emph{Individual fairness}: High-skilled group $B$ students are hard to target,  i.e., $I(q; P_{S}) >0$,
	iff $q  >  \tilde q^*_{S}+ \frac{\sigma^{-2}(\tilde q^*_{S} -\mu)}{\sqrt{\sum_{k\in S} \sigma_{Bk}^{-2}}  \sqrt{\sum_{k\in S} \sigma_{Ak}^{-2}}}$.

Increasing the informativeness gap  increases the individual fairness gap for high-skilled students: fix group $B$ precision, $\sum_{k\in S} \sigma_{Bk}^{-2}$; then as group $A$ precision increases, $I(q; P_{S})$ increases for ${q>\mu + \sigma \Phi^{-1}(1-C)}$, where $\Phi$ denotes the CDF of 
$\N(0,1)$.

	\item[(iii)] \emph{Academic merit}: 
    Admitted group $B$ students have lower academic merit than group $A$.

\end{itemize}
\end{restatable}

Intuitively, although the school's Bayesian-optimal decision-making process can eliminate bias from skill estimates in terms of mean differences (see Section \ref{sec:estimatesintuition}),
the informativeness gap---as quantified via the difference in the total precision across groups---induces disparities in admission outcomes even for
ex-ante identical student groups. 
As Figure~\ref{fig:variancevarying}  illustrates, and as we prove in Electronic Companion~\ref{app.A3}, with overall equal precision (the vertical line) both groups are admitted according to their population fractions (here, $1-\pi = \pi = 0.5$); however, all fairness metrics degrade as the gap in informativeness between the two groups increases.
Access barriers (even if limited to one group) would have a similarly negative effect, albeit for a different reason: high-skilled students who otherwise would be admitted cannot even apply as they have not taken the test, cf.~\citet{hyman2016act}.

The errors in estimation due to unequal precision affect not only the diversity of the class but also the academic merit of each admitted group.
As parts (i) and (iii) establish, under unequal precisions (and no other disparities), students from group $A$ admitted to selective colleges are not only admitted at a higher rate, but---contrary to existing theoretical results \citep{faenza2020impact}---are also of higher true skill, on average, than the admitted students from group $B$.
{This discrepancy arises because the school fails to identify high-skilled students
from group $B$.  Part (ii) shows that high-skilled students in group 
$B$ are less likely to be admitted than they would be in group $A$. Although the individual fairness gap is positive for all sufficiently high-skilled students, its magnitude varies: for students in the far right tail, the gap eventually decreases because---despite the noise---their estimated skills remain high enough for admission. We prove this in the lemma below.

\begin{restatable}{lemma}{IFgapdecrease}
\label{lemma:IFgap_decrease}
Consider policy $P_{S}$, and assume unequal precision. The individual fairness gap $I(q; P_{S})$ is decreasing in $q$ for $q>q_e$, where
\begin{equation*}
    q_e \triangleq \tilde q^*_S + \sqrt{\frac{ \sigma^{-4}(\mu - \tilde q^*_S)^2 }{{\sum_{k \in S}\sigma^{-2}_{Ak}}  {\sum_{k \in S}\sigma^{-2}_{Bk}}  } +  \frac{\ln\left({{\sum_{k \in S}\sigma^{-2}_{Ak}}}\right) - \ln\left({{\sum_{k \in S}\sigma^{-2}_{Bk}}}\right)}{{\sum_{k \in S}\sigma^{-2}_{Ak}} - {\sum_{k \in S}\sigma^{-2}_{Bk}}}}.
\end{equation*}
Furthermore,
$\lim_{q \rightarrow \infty} I(q; P_{S}) = 0$.
\end{restatable}

\begin{figure*}[tb]
		\begin{subfigure}[b]{0.32\textwidth}
		\centering
		\includegraphics[width=\linewidth]{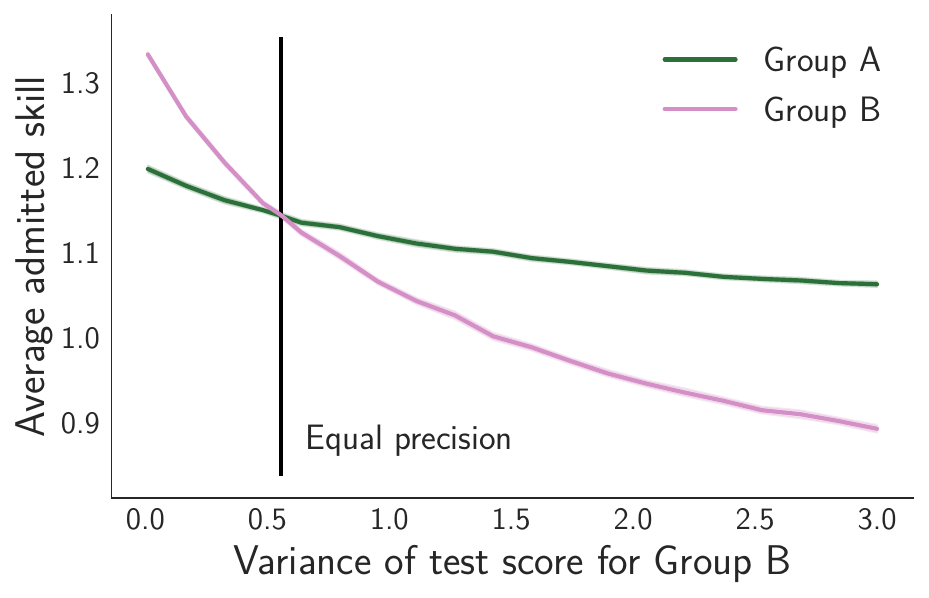}
		\caption{\small Academic merit \normalsize}
		\label{fig:varianceadmittedskill}
	\end{subfigure}
	\hfill
 \begin{subfigure}[b]{0.32\textwidth}
	\centering
\includegraphics[width=\linewidth]{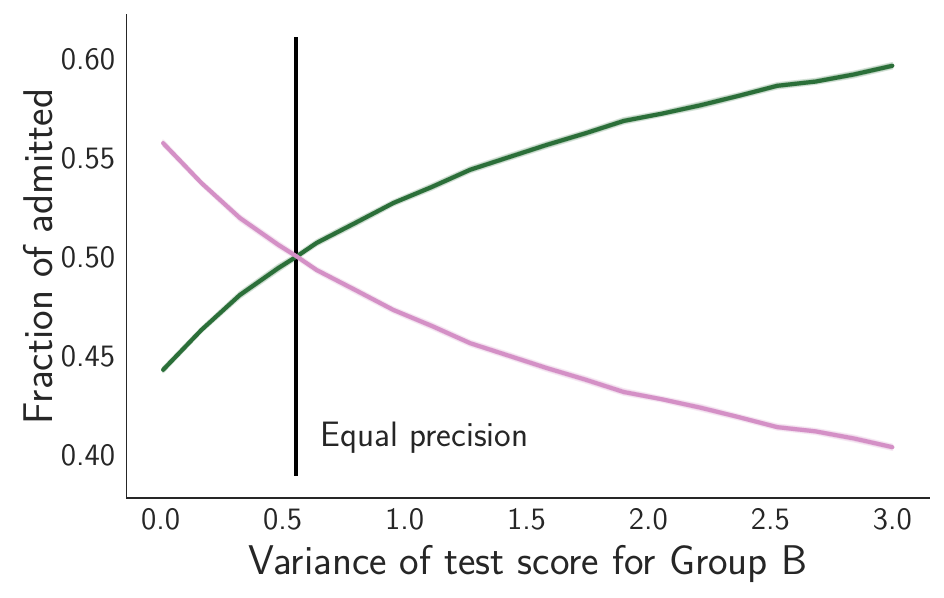}
		\caption{\small Fraction of admitted students \normalsize}
		\label{fig:variancefractionadmitted}
	\end{subfigure}
	\hfill
\begin{subfigure}[b]{0.32\textwidth}
	\centering
	\includegraphics[width=\linewidth]{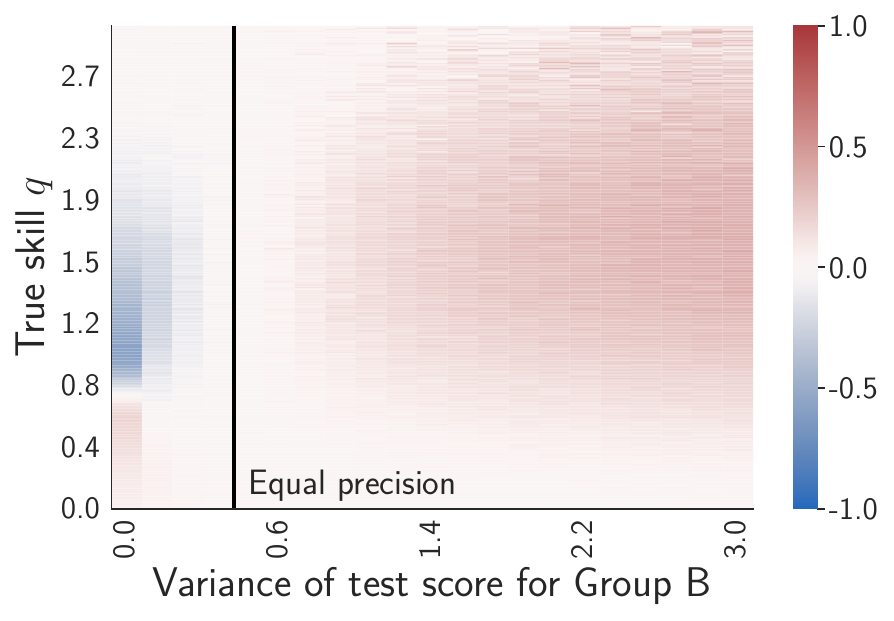}
	\caption{\small Individual fairness gap \normalsize}
	\label{fig:varianceIF}
\end{subfigure}

	\begin{subfigure}[b]{0.32\textwidth}
		\centering
		\includegraphics[width=\linewidth]{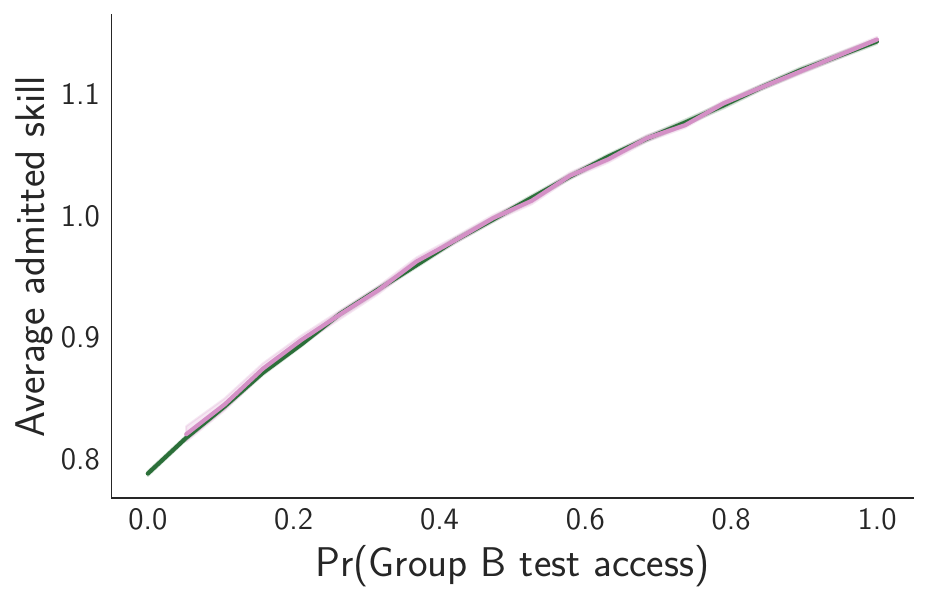}
		\caption{\small Academic merit \normalsize}
		\label{fig:barrieradmittedskill}
	\end{subfigure}
	\hfill
	\begin{subfigure}[b]{0.32\textwidth}
		\centering
		\includegraphics[width=\linewidth]{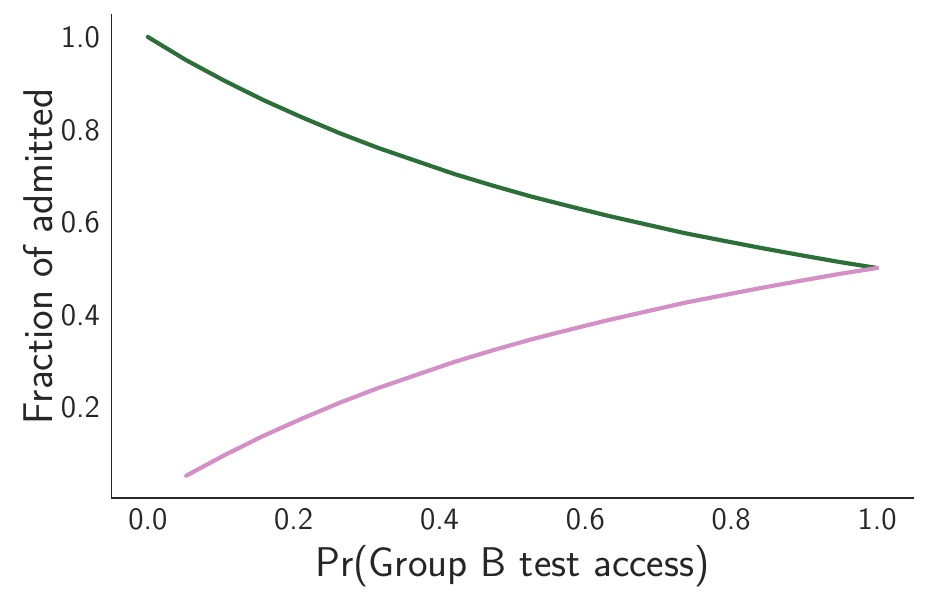}
		\caption{\small Fraction of admitted students \normalsize}
		\label{fig:barrierfractionadmitted}
	\end{subfigure}
	\hfill
	\begin{subfigure}[b]{0.32\textwidth}
		\centering
		\includegraphics[width=\linewidth]{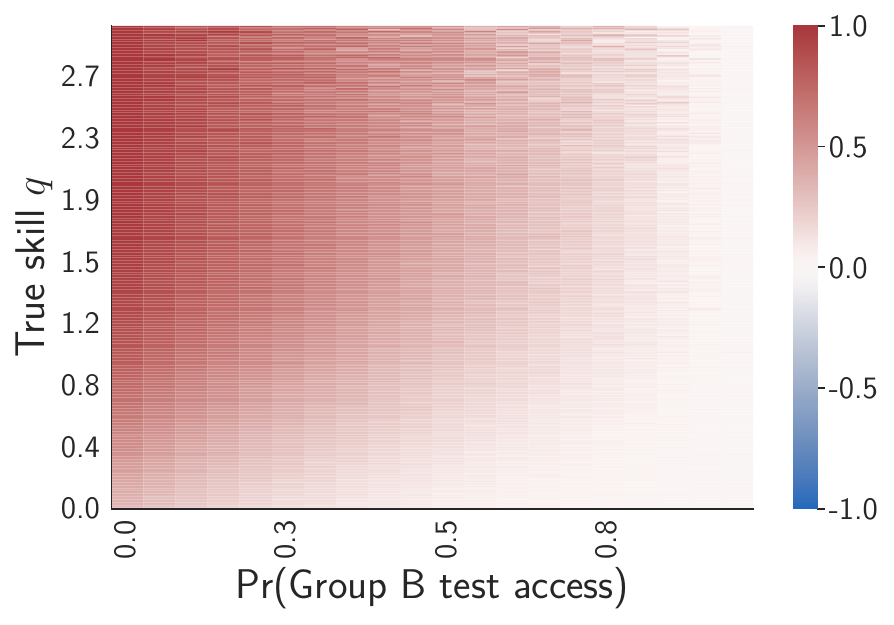}
		\caption{\small Individual fairness gap \normalsize}
		\label{fig:barrierIF}
	\end{subfigure}
	
 \caption{How the admitted students' academic merit, fraction of each group, and individual fairness gap change with group $B$ test score variance and test access, respectively.
 Figures (a)-(c) fix test access $\gamma_A=\gamma_B=1$ and vary the test score variance. Figures (d)-(f) fix the test variance to be equal for both groups and vary test access $\gamma_B$. 
 With equal precision and no barriers, groups are treated equitably. As the feature variance or barriers increase for group $B$, both academic merit of admitted $B$ students and fairness metrics worsen. We considered $\pi = 1-\pi = 0.5$; the full parameter set can be found in Electronic Companion~\ref{appsec:simparams}. 
}
\label{fig:variancevarying}
\end{figure*}

These results hint at the difficulty in deciding whether to drop standardized testing. Removing the test increases estimation variance (perhaps differentially, as~\citet{bellafante_2020} and~\citet{berkeleyreport} posit) which worsens all metrics, but it also reduces access barriers, which improves them. The interaction of these two forces determines the overall effect. Our next section formalizes this interaction.

\subsection{Dropping test scores with and without barriers}
\label{sec:dropping_tests}

{We now ask: \textit{Under what conditions would dropping a feature benefit the school and the applicants?} We study this question by comparing the test-free policy $P_{\subb}$ to the test-based policy $P_{\full}$, both \textit{with} (Theorem~\ref{thm:threshold_char_without_aa}) and \textit{without} barriers (Theorem~\ref{prop.test_group_aware_noAA}).}

\begin{restatable}[Dropping tests with barriers]{theorem}{propthresholdnoAA}
\label{thm:threshold_char_without_aa}
Consider policies $P_{\full}$ and $P_{\subb}$ and assume unequal precisions under $P_{\full}$.
With barriers, the following is true:

\begin{itemize}
    \item[(i)] \emph{Diversity level}:  
    Holding other parameters fixed, there exists a threshold $\bar{\gamma}$ such that diversity improves under $P_{\subb}$ iff the fraction of group B students with access $\gamma_B < \bar{\gamma}$.

    \item[(ii)] \emph{Academic merit}: For each group $g$, holding other parameters fixed, there exists a threshold $\bar{\bar\gamma}_{g}$ such that academic merit of group $g$  increases under $P_{\subb}$ iff $\gamma_g < \bar{\bar\gamma}_g$. 
    
\end{itemize}
\end{restatable}

{Perhaps surprisingly,  Theorem~\ref{thm:threshold_char_without_aa}  establishes that the academic merit of the admitted class may improve after dropping the test score. Similarly, diversity may deteriorate after dropping test scores.  
More specifically, Theorem~\ref{thm:threshold_char_without_aa} offers a threshold characterization, where the thresholds $\bar{\bar\gamma}_g$ and $\bar{\gamma}$ are functions of both the access levels of the two groups as well as the variance parameters, with and without the test. 
We provide the full characterization 
and additional illustrations  in Electronic Companion~\ref{app.A4} and~\ref{app.suppfigs.nonstrategic}, respectively.  

{At a high level, Theorem~\ref{thm:threshold_char_without_aa} implies that the decision to drop the test requirement is not just a matter of increasing access for the disadvantaged group. Rather, it depends on the complex interaction between the informational environment and the access levels of both groups. 
First, dropping test scores increases the applicant pool size but also affects its composition at different rates for each group. Second, 
 the information loss incurred by dropping the test may not necessarily benefit students in group $B$. In particular, it is possible that the informational disadvantage faced by group $B$ students may be exacerbated by the absence of test score information even if test scores are more noisy for group $B$ than group $A$; especially when the testing barriers are relatively small, the negative informational effect may not be counterbalanced sufficiently by the increase in the group's pool size. }

{Beyond the equivocal impact of dropping test scores on diversity and group-level academic merit, doing so introduces additional trade-offs. As part (ii) of Theorem~\ref{thm:threshold_char_without_aa} and Figure~\ref{fig:policycompare_barriers2dplot} show, one admitted group’s  academic merit may decrease even when overall academic merit rises. Depending on the parameters, this may be an inevitable consequence of dropping the test, raising important fairness trade-offs for policymakers.}

{Our next result studies the role of information loss in more depth, focusing on just the effect of the variance parameters in a setting without access barriers.}

\begin{restatable}[Dropping tests without barriers]{theorem}{proptestgroupawarenoAA}
\label{prop.test_group_aware_noAA}
Consider policies  $P_{\full}$ and $P_{\subb}$, and assume unequal precisions under $P_{\full}$. 
\begin{itemize}

    \item[(i)] \emph{Diversity level}:
    Diversity level improves after dropping the test,  $\tau(P_{\subb})>\tau(P_{\full})$, iff
    \begin{equation}
    \label{eq:condition_diversity_dropping_test}
 \frac{  \sum_{k\in \subb} \sigma_{Ak}^{-2} \left(\sigma^{-2}+\sum_{k\in \full} \sigma_{Ak}^{-2}   \right)}{\sum_{k\in \subb} \sigma_{Bk}^{-2} \left(\sigma^{-2} +\sum_{k\in \full} \sigma_{Bk}^{-2} \right)}  < \frac{\sigma^{-2}_{AK}}{\sigma_{BK}^{-2}}.
      \end{equation}

\item[(ii)] 
\emph{Individual fairness}:  For each group $g$, there  exist thresholds $q_{g}$ such that the admission probability for students of skill $q$ in group $g$ decreases under $P_{\subb}$ iff $q>q_{g}$. 
Further, there exists a threshold $\hat{q} \geq \max\{q_A, q_B\}$ such that the individual fairness gap increases for all $q> \hat{q}$, but may decrease otherwise.

    \item[(iii)]  \emph{Academic merit}: Academic merit decreases for both groups $g\in\{A,B\}$, that is,
    ${\E [q \mid Y=1, g, P_{\full}] > \E [q \mid Y=1, g, P_{\subb}].}$
\end{itemize}
\end{restatable}

Without barriers, the effect on the  diversity level and individual fairness gap of dropping a feature depends on relative informativeness. However, it always worsens academic merit for both groups as the school has access less information and skill estimates are noisier.

The exact effect on diversity depends on both the total precision of the remaining ${K-1}$ features and how much the test precisions $\sigma^{-2}_{A,K}$, $\sigma^{-2}_{B,K}$ differ.  
\Cref{eq:condition_diversity_dropping_test} is equivalent to: 
\begin{equation}
\label{eqn:dropariancescondition_2}
   \frac{\left[\frac{\sum_{k \in \subb} \sigma_{Ak}^{-2}}{\sigma^{-2} + \sum_{k\in \subb} \sigma_{Ak}^{-2}}\right]}{\left[\frac{\sum_{k\in \subb} \sigma_{Bk}^{-2}}{\sigma^{-2} + \sum_{k\in \subb} \sigma_{Bk}^{-2}}\right]} < \frac{\left[\frac{\sum_{k \in \full} \sigma_{Ak}^{-2}}{\sigma^{-2} + \sum_{k\in \full} \sigma_{Ak}^{-2}}\right]}{\left[\frac{\sum_{k\in \full} \sigma_{Bk}^{-2}}{\sigma^{-2} + \sum_{k\in \full} \sigma_{Bk}^{-2}}\right]},
\end{equation}
which intuitively encodes how informativeness for each group changes after dropping the test. 
If \Cref{eq:condition_diversity_dropping_test} holds, then the diversity level improves since dropping the test narrows the relative informativeness gap. However, if \cref{eq:condition_diversity_dropping_test} does not hold (as~\citet{berkeleyreport} attests),  removing test scores exacerbates the informational disadvantage of of group $B$; then, dropping the test decreases diversity.

Similarly, dropping the test may worsen individual fairness. As part (ii) shows, the admission probability of students with sufficiently high true skill, for either group, decreases after removing the test.
Furthermore,  
for sufficiently high-skilled students, the individual fairness gap increases after dropping test scores. 
This implication is separate from of the effect on overall diversity; although the school may manage to improve diversity by dropping the test, the targeting of high-skilled students in both groups becomes less effective, leaving high-skilled students in group $B$ disproportionately affected.  

Even without access barriers, the result establishes the importance of understanding features \textit{other} than the test score---not just their biases $\mu_{gk}$, which are canceled out given full knowledge---but also their informativeness. More broadly, our theoretical results show that even in a simple model, the debate over dropping standardized testing cannot be held without the particulars of the context: whether one cares about overall academic merit of the admitted class or  fairness criteria, the effects hinge on the interplay between access barriers, the test informativeness, and the informativeness of other application components.

\begin{figure}[t]
    \centering
    \includegraphics[width=0.45\linewidth]{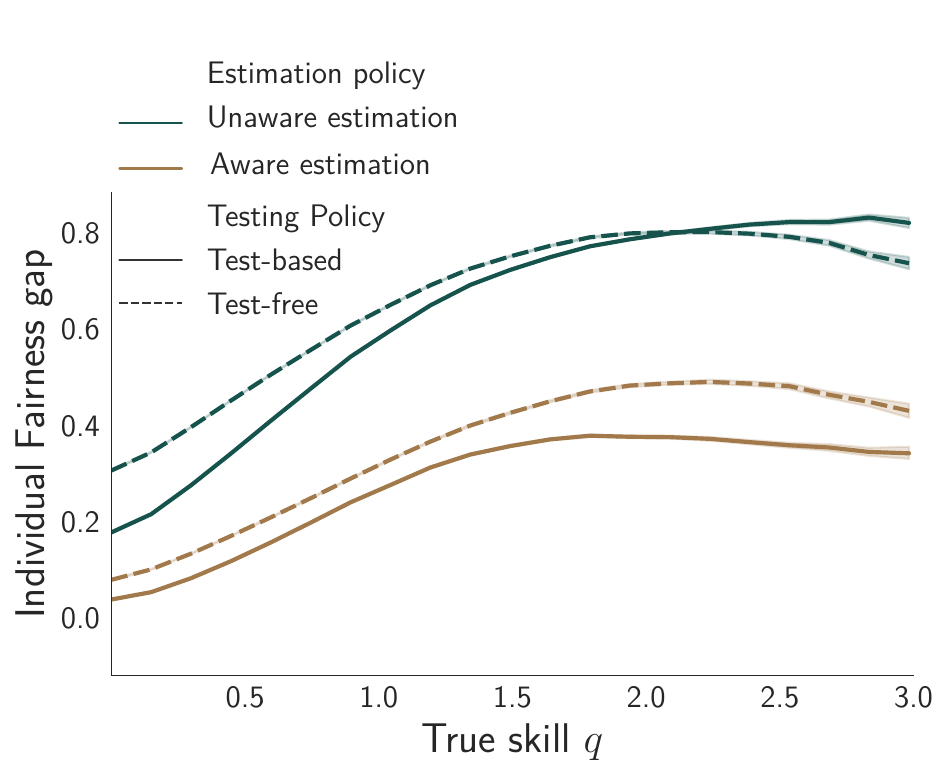}

    \caption{Individual fairness gap of various policies, in simulations of a setting where features are more informative for group $A$, and with testing barriers for group $B$.  Group-unaware policies, analyzed in \Cref{app.group_unaware}, are those in which the school does not use group information in its Bayesian optimal estimation. In this setting, dropping the test worsens the individual fairness gap across all skill levels under the group-aware policy, whereas it improves the individual fairness gap for large true skill values under the group-unaware policy. The group-aware policy has lower individual fairness gap at all true skill levels, compared to the group-unaware policy.  The full parameter set can be found in Electronic Companion \ref{appsec:simparams}.
 }
	\label{fig:policycompare_barriers_withunaware_no_AA}
\end{figure}

\smallskip
\noindent\textbf{Comparing the policies in simulation.} Figure~\ref{fig:policycompare_barriers_withunaware_no_AA} compares, for one parameter setting, our policies with and without testing. Dropping tests worsens the individual fairness gap for the main (group-aware) estimation policy, across all true skill levels $q$. {Figure~\ref{fig:policycompare_barriers_withunaware} in the Electronic Companion also includes \textit{group-unaware estimation} policies, that ignore the social group that a student belongs to  (see Electronic Companion~\ref{app.group_unaware}).  
Ignoring group attributes is an oft-proposed but often problematic policy proposal to combat bias~\citep{corbett2018measure}.
Perhaps unsurprisingly, group-unaware estimation policies perform poorly and have higher individual fairness gaps across all true skill levels, compared to group-aware estimation policies.} Finally, Figure \ref{fig:policycompare_barriers_withunaware}   also compares these policies with and without \textit{affirmative action}.

\section{Extensions: Strategic students and two schools}
\label{sec:strategic}

We extend our analysis to incorporate student incentives and school competition. Students strategically decide whether to take the test based on their expected admission probabilities and testing costs. Schools can also strategically decide which testing policy to use, since the outcomes of one school may depend on the policy of another. We first develop the model with incentives for schools and students, then characterize the student test-taking behavior illustrating that it may exhibit a non-monotonic pattern, and finally characterize equilibrium testing policies when schools compete and decide their test policy strategically.

\subsection{Extended model}

We extend the baseline model from Section~\ref{sec:baseline_model} to incorporate strategic student behavior and potential competition between schools.

\smallskip
\noindent\textbf{Schools.}
We consider two settings with strategic students: with one and two schools, respectively; the latter introduces competition between schools.  
Each school $J_i$, $i=1,2$, admits up to capacity $C_i$, and we assume an over-demanded market such that $C_1 + C_2 <1$.
Let $P^i$ denote the test policy of school $J_i$. For example, with two schools, if $J_1$ is test-based ($P^1_{\full}$) and $J_2$ is test-free ($P^2_{\subb}$), we write $\mathbf P=(P^1_{\full},P^2_{\subb})$. The testing policy determines the subset of features $S_i$ that school $i$ sees.  When the context is clear, we drop notation $i$.

\smallskip
\noindent\textbf{Students.}
Each student in group $g$ incurs a constant cost $c_g$ to take the test. A student does not know their own true skill $q$, but does know their other features $\thetaset_{\subb}$ and group membership $g$. They use this information to assess the probability of admission when deciding whether to take the test and apply. 
Admission to school $i$ is of valuation $v_i$. With two schools, we assume that students in both groups strictly prefer $J_1$ to $J_2$, i.e., $v_1>v_2$.\footnote{Our theoretical results do not depend on value $v$ directly but rather the ratio $\frac{c_g}{v}$ (or $\frac{c_g}{v_1 - v_2}$ in the two-school case).} 
To rule out trivial equilibria where no students take the test due to high testing costs, we assume that the \textit{cost-to-valuation ratio} $\frac{c_g}{v_i} < 1$ (and $\frac{c_g}{v_1 - v_2}<1$ in the two-school case).

Students are \textit{strategic} in the sense that they decide whether to apply to test-requiring schools; $\alpha\in\{0,1\}$ denotes the action of the student, where $\alpha =0$ corresponds to not taking the test and thus not applying to test-requiring schools and $\alpha = 1$ corresponds to taking the test and applying to test-requiring schools.  
If a school $J_i$ uses a test-free policy, then all students apply to it without taking the test.  

\smallskip
\noindent\textbf{School's selection policy, conditional on test policy.}
As in the base model, the school maximizes the academic merit of the admitted class, 
an objective that is strictly increasing in the skill estimate $\tilde q_{S_i}$. We define $Y_i(\thetaset_{S_i}, g; \mathbf{P})$ to be each school's selection policy. We show that, in both the single-school and the two-school setting, each school $i$ optimal selection policy is a threshold policy in which they determine a lower bound on the skill estimate $\tilde{q}^*_{S_i}$ and admit all students with estimated skill above the threshold, i.e., $Y_i(\thetaset_{S_i}, g; \mathbf{P}) = \indic\{\tilde q (\thetaset_{S_i}, g) \geq \tilde q^*_{S_i}\}$.  (See \Cref{eq:singleschoolproblem_main} for the single-school selection policy and Electronic Companion~\ref{app:strategic_two_schools}, Lemma~\ref{lemma:twoschoools_thresholdpolicies} for the two-school selection policy.)

\smallskip
\noindent\textbf{Student decisions.}
First consider a single-school setting. With testing policy $P_{\full}$, each student's expected utility depends on the valuation $v$ from getting admitted, the probability of being admitted, and the testing cost $c_g$. A student who does not apply to the school or is not admitted receives an outside option with valuation 0.  If their perceived probability of admission is sufficiently high to outweigh the cost of the test, they decide to take the test and apply to the school. 
Thus each student solves the following optimization problem:
\begin{equation}
\label{eq.student.decision.cost.theta}
    \alpha(\thetaset_{\subb}, g; P_{\full})=\arg \max_{\alpha \in \{0,1\}} \alpha \left( v \prob_{\theta_K} (Y = 1 \mid \thetaset_{\subb}, g, P_{\full}) -  c_g \right).
\end{equation}
If a school does not require the test ($P_{\subb}$), then the student always applies, i.e., $\alpha(\thetaset_{\subb}, g; P_{\subb}) =1$ for all $\thetaset_{\subb} \in \mathbb{R}^{K-1}$, $g \in \{A,B\}$.

Now consider the setting with two schools. With policies $\mathbf{P} = (P^1, P^2)$, students decide whether to apply to each school. Because applying to test-free schools is always beneficial, the key decision is whether to also apply to test-requiring schools. 
For example, if $\mathbf{P} = (P^1_{\full}, P^2_{\subb})$, then the decision to take the test (and thus apply to school $J_1$) is given by:
\small
\begin{equation*}
    \alpha(\thetaset_{\subb}, g; \mathbf{P}) = \arg \max_{\alpha \in \{0,1\}} \alpha \left( v_1 \prob (Y_1 = 1 \mid \thetaset_{\subb}, g, P^1_{\full})   - c_g  \right) + v_2 \prob (Y_1 = 0 \cap Y_2 =1 \mid \thetaset_{\subb}, g, {P}^2_{\subb}).
\end{equation*}
\normalsize
If the student applies to both schools and is admitted to both, they enroll in $J_1$. (Student decisions for other policies are given in Electronic Companion~\ref{app:strategic_two_schools}.)
We will show that this optimization problem induces more complex application behavior than with one school.

\smallskip
\noindent\textbf{Summary of school and student decisions.}
Both schools simultaneously choose whether to require the test. After observing this pair of policies, students decide whether to take the test and where to apply, taking into account testing costs and their preferences over schools. Each school $J_i$ then observes its applicant pool and admits its top $C_i$ applicants according to its admission rule. 
Finally, students enroll in their most preferred school among those that admit them.

\smallskip\noindent\textbf{Equilibria.}
In the setting with a single school, we restrict attention to the interesting setting where $P_{\full}$, i.e., the school requires the test. We say that a pair $(\alpha^*, Y^*)$
constitutes an \textit{equilibrium} if: (i) for all  $\thetaset_{\subb} \in \mathbb{R}^{K-1}$ and $g \in \{A,B\}$, $$\alpha^*(\thetaset_{\subb}, g; P_{\full}) = \arg \max_{\alpha \in \{0,1\}} \alpha \left( v \prob (Y^* = 1 \mid \thetaset_{\subb}, g; P_{\full}) - c_g \right);$$ and  (ii) for all  $\thetaset_{\full} \in \mathbb{R}^{K}$ and $g \in \{A,B\}$, $Y^*(\thetaset_{\full}, g; P_{\full}) = \indic\{\tilde q (\thetaset_{\full}, g) \geq \tilde q^*_{\full}\}$, where the optimal admission threshold  
$\tilde q^*_{\full}$  is the corresponding solution to 
\begin{equation}
    \begin{split}
    \label{eq:singleschoolproblem_main}
    \tilde q^*_{\full} =  \min \left\{z \in \mathbb{R}: 
    \sum_g \pi_g \expec_{\thetaset_{\full}}[\alpha^*(\thetaset_{\subb},g; P_{\full})  \mid \tilde q (\thetaset_{\full}, g)  \geq z, g, P_{\full}  ] \leq C\right\}.
\end{split}
\end{equation}
 In words, (i) requires that each student makes their optimal test-taking decision given the school's equilibrium decisions, and (ii) requires that the school's decisions correspond to their optimal threshold policy given their capacity and skill estimates.\footnote{Note that the above equilibrium definition depends on the full vector of $K-1$ feature values $\thetaset_{\subb}$. 
However, as formalized in Lemma \ref{lemma:reduced_equilibria_single} in Electronic Companion~\ref{app:strategic_single}, we can equivalently focus solely on the set of student actions that are dependent on $\tilde q_{\subb}$ instead of the entire vector of $K-1$ feature values $\thetaset_{\subb}$, where $\tilde q_{S} \triangleq \tilde q (\thetaset_{S}, g)$  is the skill estimate using feature set $S$ and the group $g$ of the given student.
We thus work directly with the \textit{reduced-form equilibrium} $(\alpha^*(\tilde q_{S_i}, g; P_{S_i}), Y^*(\tilde q_{S_i}; P_{S_i}))$. In this case, the actions still also depend on the precisions of this estimate and the test score, which may be group specific. The same also applies to the two-school setting (see Lemma~\ref{lemma:twoschoools_thresholdpolicies}).}
With two schools, we say that a triple $(\mathbf{\alpha}^*, \mathbf{Y}^*, \mathbf{P})$ constitutes an \textit{equilibrium} if  (i) students' application decisions are optimal conditional on the schools' selection policies and (ii) each school's testing policy and selection policy maximizes its academic merit given student application decisions and the other school's policy. The formal definition is in Electronic Companion~\ref{app:strategic_two_schools}. 
In the next sections, we will analyze the set of policy pairs $(P^1, P^2)$ and characterize the equilibrium behavior of schools.

\subsection{Optimal student test-taking behavior}
\label{sec.testtakingbehavior}

In order to analyze the equilibrium behavior described above and the schools' optimal testing policy, we first analyze the optimal student test-taking behavior, as a response to a fixed testing policy. We begin with the single-school case where the school requires the test, which provides intuition for the key factors in students' strategic behavior. In this setting, students follow a simple threshold strategy: those with sufficiently high skill estimates from their non-test features take the test, while those below the threshold opt out. This threshold structure results in positive assortative behavior: higher-skilled students are more likely to apply and, conditional on applying, more likely to be admitted.

In contrast, introducing a second school fundamentally changes this behavior. In a two-school setting where the more preferred 
school $J_1$ uses a test-based policy while the less preferred   
school $J_2$ uses a test-free policy, we show that the single-threshold structure breaks down. Instead, students may exhibit a non-monotonic test-taking behavior: high-skilled students with high admission probability and 
some low-skilled students with lower outside option take the test and apply, but middle-skilled students may opt out, preferring a guaranteed admission to the less selective school. This breakdown in monotonicity can lead to cases where $J_1$ may attain lower academic merit than the less preferred school $J_2$, if $J_1$ requires the test and $J_2$ does not.

\subsubsection{Single school: Threshold-based student strategies}
\label{sec:singleschoolstratetic}
Consider a setting with a single school that uses the testing policy $P_{\full}$. Students can observe their non-test features $\thetaset_{\subb}$ before assessing their probability of admission, which depends on the other students' behavior, and deciding whether to take the test. 
Despite the complexity introduced, the following result shows that the student's optimal decision follows a threshold-strategy that depends on their skill estimate $\tilde q_{\subb}$ from the non-test features $\thetaset_{\subb}$.

\begin{restatable}{lemma}{lemmathresstrategy}
\label{lemma.thres.strategy}
Suppose that the school uses a test-based policy $P_{\full}$. There exists a unique equilibrium $(\alpha^*,Y^*)$, with the following property:  there is a threshold $\underline{q}^ g$ such that  students in group $g$  take the test ($a=1$) if and only if $ \tilde q_{\subb}\geq \underline{q}^ g$,
where 
\begin{equation}
    \label{eq.decision_threshold_taking_test}
    \underline{q}^g = \tilde q_{\full}^* - \Phi^{-1} \left(1 - \frac{c_g}{v}\right) \left ( \frac{\sigma_{gK}^{-2}}{\sigma^{-2} + \sum_{k \in \full} \sigma_{gk}^{-2}}\right)
\sqrt{\sigma^{2}_{gK} + \frac{1}{\sigma^{-2} + \sum_{k \in \subb} \sigma_{gk}^{-2}} },
\end{equation}
and  $\tilde q^*_{\full}$ is the solution to \Cref{eq:singleschoolproblem_main} so that $Y^*(\tilde q_{\full}; P_{\full})= \indic\{ \tilde q_{\full} \geq \tilde q^*_{\full}\}$.
\end{restatable}

The above result implies positive assortativity in test-taking decisions.
In contrast to the non-strategic setup with access barriers of Theorem~\ref{thm:threshold_char_without_aa}, where all students had  the same a priori probability of being eligible to apply, the decision to apply now correlates with each individual student's true skill level through the other features (see Figure \ref{fig:simulations_prob_apply_by_skill} in Electronic Companion~\ref{app:synthetic_simulations_strategic_single_school}). These selection effects change the composition of the applicant pool. 
We note that applying does not guarantee admission: conditional on applying, higher true skill (and thus higher skill estimate $\tilde q_{\subb}$) 
increases the admission probability, yet some students will still pay the cost $c_g$, apply, and ultimately be rejected. 
For further characterization of the equilibrium outcomes in this single-school, strategic student setting, see Electronic Companion~\ref{app.singleschool.effectsoftestcosts}.

\subsubsection{Two schools: Non-monotonic student test-taking behavior}
\label{sec:twoschoolstrategicbehavior}

The introduction of a second school fundamentally changes student strategic behavior. Consider an asymmetric case where the preferred school $J_1$ is test-based ($P^1_{\full}$) and $J_2$ is test-free ($P^2_{\subb}$); we write $\mathbf P=(P^1_{\full},P^2_{\subb})$. 
This is equivalent to adding a less preferred school to the single-school setting above, and mirrors a common real-world practice in which a more selective school requires a test and a less selective school does not.
Here, competition---induced by a less-preferred school---breaks monotonic testing behavior of the single-school setting. 

\begin{restatable}{proposition}{thmtwoschools}
\label{thm.twoschools}
  Consider the setting with two schools defined above. Then, there exists a unique equilibrium $(\alpha^*, \mathbf{Y}^*)$ with the following properties:
  \begin{itemize}
      \item [(i)] School $J_i$'s selection policy $Y^*_i$ takes a threshold form: $Y^*_i (\tilde q_{S_i}; \mathbf{P})= \indic\{\tilde q_{S_i} \geq \tilde q^*_{i, S_i}\}$, where $\tilde q^*_{i, S_i}$ is  the admission threshold of school $i$ using feature set $S_i$.
      
      \item[(ii)] Students in group $g$ with preliminary skill estimate $\tilde q_{\subb}$ take the test and apply to school $J_1$,  if and only if one of the following conditions holds:
      \begin{itemize}
          \item [1)]  either  $\tilde q^*_{2, \subb} > \tilde q_{\subb} \geq \underline{q}_{l}^g $ where
            \begin{equation}
            \label{eq.q_l_thres}
    \underline{q}_{l}^g = \tilde q_{1, \full}^* - \Phi^{-1} \left(1 - \frac{c_g}{v_1}\right) \left ( \frac{\sigma_{gK}^{-2}}{\sigma^{-2} + \sum_{k \in \full} \sigma_{gk}^{-2}}\right)
\sqrt{\sigma^{2}_{gK} + \frac{1}{\sigma^{-2} + \sum_{k \in \subb} \sigma_{gk}^{-2}} }; 
\end{equation}

\item[2)] or $\tilde q_{\subb} \geq \max\{\underline{q}_{h}^g , \tilde q^*_{2, \subb}\}$, where
\begin{equation}
\label{eq.q_h_thres}
    \underline{q}_{h}^g = \tilde q_{1, \full}^* - \Phi^{-1} \left(1 - \frac{c_g}{v_1-v_2}\right) \left ( \frac{\sigma_{gK}^{-2}}{\sigma^{-2} + \sum_{k \in \full} \sigma_{gk}^{-2}}\right)
\sqrt{\sigma^{2}_{gK} + \frac{1}{\sigma^{-2} + \sum_{k \in \subb} \sigma_{gk}^{-2}} }.
\end{equation}
      \end{itemize}
     
     Furthermore, $\underline{q}_{l}^g < \underline{q}_{h}^g $ for both groups $g \in \{A,B\}$.

\normalsize  
  \end{itemize}
\end{restatable}

\begin{figure}
\centering
\begin{tikzpicture}
  \begin{axis}[
    width=8cm,
    height=6cm,
    axis lines=middle,
    xmin=-4, xmax=4,
    ymin=0, ymax=0.45,
    xlabel={$\tilde q_{\subb}$},
    xtick={-0.3, 0.2, 0.9},
   xticklabels={$\underline{q}^g_l$, $~~~\tilde q^*_{2, \subb}$, $~~~~~\underline{q}^g_h$},
    ytick=\empty,
    domain=-4:4,
    samples=100,
    axis y line=none,
    smooth,
    ticklabel style = {font=\small}
    ]  
    \addplot [fill=yellow!30, draw=none, domain=-0.3:10, area legend] {gauss(0,1)} \closedcycle;
    
    \addplot [fill=bluebell!50, draw=none, domain=-0.7:10, area legend] {min(gauss(0,1),gauss(-0.9,0.97))} \closedcycle;

  \addplot [fill=bluebell!50, draw=none, domain=0.2:0.9, area legend] {gauss(0,1)} \closedcycle;

    \addplot [fill=white, draw=none, domain=-0.75:0.2, area legend] {gauss(-0.9,0.97))} \closedcycle;
    
    \draw [dashed] (axis cs: 0.2,0) -- (axis cs: 0.2, 0.39);

    \draw [dashed] (axis cs: -0.3,0) -- (axis cs: -0.3, 0.38);

    \draw [dashed] (axis cs: 0.9,0) -- (axis cs: 0.9, 0.26);
      \addplot [black, very thick] {gauss(0,1)};
  \end{axis}
\end{tikzpicture}
~
\begin{tikzpicture}
  \begin{axis}[
    width=8cm,
    height=6cm,
    axis lines=middle,
    xmin=-4, xmax=4,
    ymin=0, ymax=0.45,
    xlabel={$\tilde q_{\subb}$},
    xtick={-0.3, 0.15, 0.7},
   xticklabels={$\underline{q}^g_l$,  $\underline{q}^g_h$, $~~~~~\tilde q^*_{2, \subb}$},
    ytick=\empty,
    domain=-4:4,
    samples=100,
    axis y line=none,
    smooth,
    ticklabel style = {font=\small}
    ]

    \addplot [fill=yellow!30, draw=none, domain=-0.3:10, area legend] {gauss(0,1)} \closedcycle;
     \addlegendentry{$J_1$}

    \addplot [fill=bluebell!50, draw=none, domain=0.7:10, area legend] {min(gauss(0,1),gauss(-0.9,0.97))} \closedcycle;
     \addlegendentry{$J_2$}

    \addplot [fill=white, draw=none, domain=-0.75:0.7, area legend] {gauss(-0.9,0.97))} \closedcycle;

    \draw [dashed] (axis cs: 0.15,0) -- (axis cs: 0.15, 0.40);

    \draw [dashed] (axis cs: -0.3,0) -- (axis cs: -0.3, 0.38);

    \draw [dashed] (axis cs: 0.7,0) -- (axis cs: 0.7, 0.31);
      \addplot [black, very thick] {gauss(0,1)};
  \end{axis}
\end{tikzpicture}
\caption{Students in group $g$ enrolled at schools $J_1$ and $J_2$. The left and right panel of the figure correspond to $\tilde q^*_{2, \textrm{SUB}} < \underline{q}_h^g$ and $\tilde q^*_{2,\textrm{SUB}} > \underline{q}_h^g$, respectively. 
In both cases, all students apply to $J_2$ but only the students with skill estimates above the cutoff $\tilde q^*_{2, \textrm{SUB}}$ are admitted. Among those students only the mass of students in the purple-shaded area will accept $J_2$'s offer, since a fraction of the students admitted to $J_2$ also applied and got admitted to $J_1$ (yellow-shaded area). However, as the left panel illustrates for $\tilde q^*_{2, \textrm{SUB}} < \underline{q}_h^g$, not all students who apply to $J_2$ necessarily apply to $J_1$ as well; indeed, the admitted students with $\tilde q^*_{2, \textrm{SUB}} \leq \tilde q_{\textrm{SUB}} <\underline{q}^g_h$ have applied only to $J_2$. 
Furthermore, observe that the students who successfully apply to $J_1$ (yellow-shaded area) are not always characterized by higher skill estimates than the admitted students at $J_2$ (purple-shaded area). Finally, in both panels, all students with $\underline{q}_l^g < \tilde q^*_{2, \textrm{SUB}}$ applied to $J_1$, however only the yellow-shaded mass is admitted.}
\label{fig.merit.twoschools.strategic}
\end{figure}

The above theorem establishes several interesting equilibrium properties. Even this two-school setting induces complex student strategic behavior.
\Cref{fig.merit.twoschools.strategic} illustrates the student test-taking behavior, described in Part (ii) above.  In particular, student test-taking behavior is not necessarily characterized by a single threshold in their skill estimates $\tilde q_{\subb}$. Some students with lower skill estimates $\tilde q_{\subb}$ (who, after observing the first $K-1$ features, know they will not be admitted to the test-free school $J_2$) take the test to reattain a chance of admission to school $J_1$, {while other students with higher skill estimates $\tilde q_{\subb}$ may choose not to take the test, preferring the safer option of $J_2$.}

This non-monotonicity in application behavior can generate a mismatch between student skill and school ranking. Although schools continue to use admission thresholds that are increasing in their ranking (Part (i)), the non-monotonic student responses (Part (ii)) break the positive assortativeness that matching models typically exhibit~\citep{chade2017sorting}. 
In Proposition~\ref{thm.twoschools.full} (Parts (iii)-(iv)), we show that, as a result, $J_1$ may, in some cases, attain lower academic merit than the less-preferred school $J_2$. Such a reversal cannot occur when both schools adopt the same testing policy---then, each student either applies to both schools or neither, eliminating the possibility of differential self-selection. We further show that heterogeneous testing policies across schools may also lead to differing levels of diversity, although our analysis does not preclude both schools experiencing low diversity.

\subsection{Equilibrium testing policy outcomes}
In light of the complex dynamics in the student test-taking behavior, we consider whether \textit{schools} themselves are incentivized to require the test.
We return to the general setting where the two schools $J_1$ and $J_2$ can each choose their testing policy and optimize for the academic merit of the student body.\footnote{{For a partial characterization of school behavior with respect to  diversity, see Proposition~\ref{prop.strategic.twoschool.diversity}.}} Now, a school's optimal testing policy may depend on its competitor's requirements. For example, a school may be incentivized to drop the test if most other schools have already done so, as fewer students would be taking the test.
Conversely, if all other schools require the test, then many students may take it regardless, potentially incentivizing the school to require the test as well. Section \ref{ssec:two_schools_strategic_theory} 
characterizes the schools' equilibrium policies. Section 
\ref{ssec:simulation_strategic_schools} highlights, through simulations, settings in which one school's policy depends on the other school's decision.

\subsubsection{Characterization of equilibrium policies.}
\label{ssec:two_schools_strategic_theory}

{Theorem \ref{thm.twoschool.strategic.equilibria} analyzes the set of policy pairs $(P^1, P^2)$ and identifies the conditions under which $(P_{\subb}, P_{\full})$, $(P_{\full}, P_{\subb})$, and $(P_{\full}, P_{\full})$ are equilibrium policies.} Recall that, at equilibrium, $\tilde{q}^*_{i, S_i}$ denotes the optimal admission threshold of school $i$ using feature set $S_i$.

\begin{restatable}[Academic merit and two-school equilibria with strategic students]{theorem}{thmtwoschoolsstrategicequlibria}
\label{thm.twoschool.strategic.equilibria}
Suppose that each school chooses a policy to maximize academic merit. 
 Under school policies $\mathbf{P}$, let $M_g (\mathbf{P})$ denote the mass of test-taking students from group $g$ and let $L^i_g (\mathbf{P})$ denote the academic merit of admitted students from group $g$ to school $J_i$ (see Electronic Companion~\Cref{app:strategic_two_schools} for definitions). 
 Define also 
 $$K_g (q)  \triangleq 1- \Phi\left(\frac{q - \mu}{\sigma \sqrt{\frac{\sum_{k\in \subb} \sigma_{gk}^{-2}}{\sigma^{-2} + \sum_{k \in \subb} \sigma_{gk}^{-2}} }}\right)$$ to be the mass in group $g$ with skill estimates $\tilde q_{\subb} \geq q$.
 Then, the following hold:
\begin{itemize}
    \item[(i)]  Policy $(P_{\subb}, P_{\full})$ is an equilibrium if and only if 
     \begin{equation}
    \label{eq.condition.drop.twoschools.strategic.J1}
       c_g \geq \hat{c}_g \mathrm{~and~} M_g ({P}_{\full}, P_{\full}) < 
       K_g (\tilde q^*_{1, \subb})
        \end{equation} 
        for exactly one group $g \in \{A, B\}$ and some threshold $\hat{c}_g >0$, while for $g' \neq g$, it holds
            \begin{equation}
                \label{eq.sub.full.J2.cond1}
                M_{g'} (P_{\subb}, P_{\full}) \geq C_2
            \end{equation}
\begin{equation}
\label{eq.sub.full.J2.cond2}
\mathrm{~~~\textit{and}~~~~}
\pi_{g'}L^2_{g'} (P_{\subb}, P_{\full}) >  \pi_A L^2_{A} (P_{\subb}, P_{\subb}) + \pi_B L^2_{B} (P_{\subb}, P_{\subb}).
\end{equation}

\item[(ii)] Policy $(P_{\full}, P_{\subb})$ is an equilibrium if and only if 
\begin{equation}
\label{eq.condition.drop.twoschools.strategic.J2.1} \begin{split}
       & c_g \geq  \hat{c}'_g,  ~~~
              M_g ({P}_{\full}, P_{\full}) < K_g (\tilde q^*_{2, \subb}) 
                    \end{split}
        \end{equation}
\begin{equation}
\label{eq.condition.drop.twoschools.strategic.J2.2} \begin{split}
             & \mathrm{~~~\textit{and either}~}  M_g ({P}_{\full}, P_{\subb})  > K_g (\tilde q^*_{1, \subb})
                \mathrm{~~\textit{or}~~} c_g \leq\hat{c}_g'',
        \end{split}
        \end{equation}
        for exactly one group $g \in \{A, B\}$ and  thresholds $ \hat{c}_g'' > \hat{c}_g'>0$.

    \item[(iii)] There exist functions $\underline{c}_A, \underline{c}_B: \mathbb{R}_+ \rightarrow \mathbb{R}_+$  
    such that policy $(P_{\full}, P_{\full})$ is an equilibrium if and only if $c_A \leq \underline{c}_A (c_B)$ and $c_B \leq \underline{c}_B (c_A)$.
\end{itemize}
\end{restatable}

{Intuitively, both schools want to maintain the testing requirement when test costs are sufficiently low for both groups (Part (iii)). However, there exist regimes where one school wants to drop the test. Part (i) shows that the more preferred school, $J_1$, wants to drop the test when test costs of one group $g$ are sufficiently high and an insufficient number of group $g$ students take the test. Dropping the test helps $J_1$ attract more students from group $g$ in a way that increases the average skill of its applicant pool, thanks to group $g$. At the same time, there are extreme scenarios school $J_2$ might prefer to keep the test, even if this implies admitting zero students from group $g$; however, our simulations suggest that this is a rare case and school $J_2$ is better off following $J_1$'s lead and dropping the test. Part (ii) shows that  $J_1$ might wish to maintain its test requirement if enough students from both groups take the test. As the more preferred school, $J_1$ has the advantage to keep the test, whereas the less preferred school, $J_2$, must drop the test to attract students from the group with higher costs. Typically, this equilibrium arises if one group faces moderately high test costs.} Note that the theorem does not rule out $(P_{\subb}, P_{\subb})$ being an equilibrium---the conditions are not mutually exclusive, and our analysis does not imply that the equilibrium is unique. Further, in simulations, we find settings in which $(P_{\subb}, P_{\subb})$ is the equilibrium (Figure \ref{fig:twoschool_strat_cost2.0}).

\subsubsection{Simulation results}
\label{ssec:simulation_strategic_schools}
We simulate the resulting admission outcomes decisions when there are two schools vying for the same student population and highlight two different settings where one school's optimal testing decision depends on the other school's decision. 
The simulation closely resembles that of the single school, strategic student setting outlined in Appendix \ref{app:synthetic_simulations_strategic_single_school}; for details, see Appendix \ref{app:synthetic_simulations_strategic_two_schools}.

\begin{figure*}[tb]
\centering
\begin{subfigure}[b]{0.49\textwidth}
    \centering
    \includegraphics[width=\linewidth]{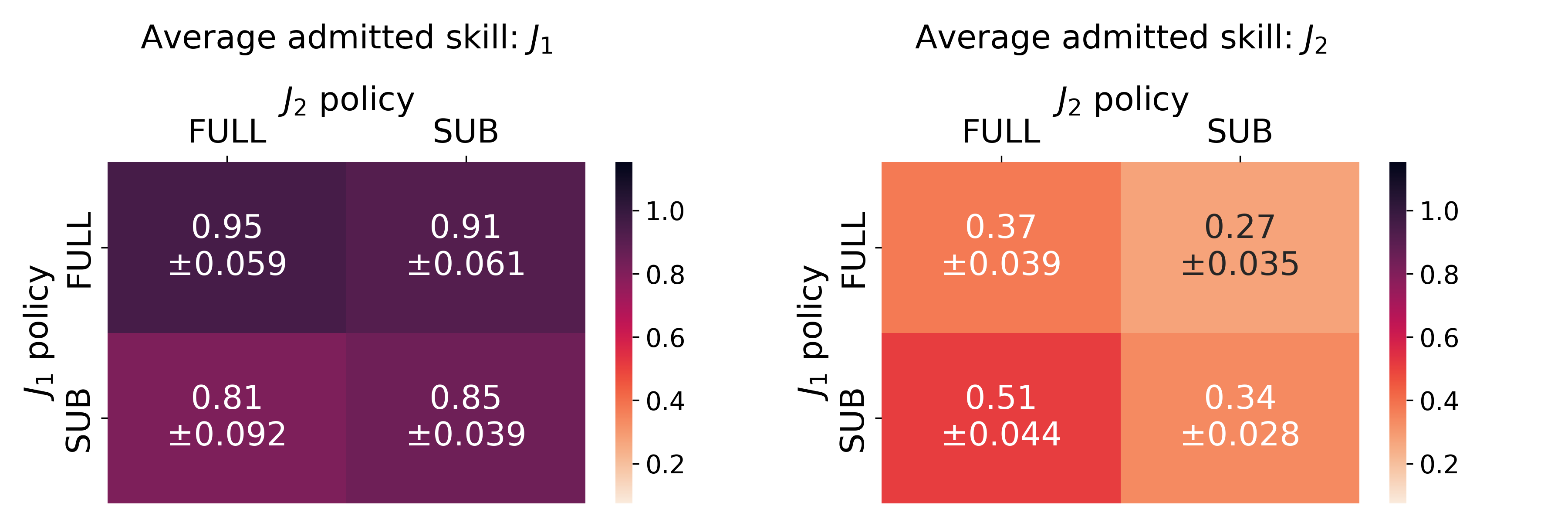}
    \caption{\small \centering Test cost $c_A = c_B = 0.5$. Unique equilibrium is $(P^1_{\full}, P^2_{\full})$.
    }
    \label{fig:twoschool_strat_cost0.5}
\end{subfigure}
\hfill
\begin{subfigure}[b]{0.49\textwidth}
    \centering
    \includegraphics[width=\linewidth]{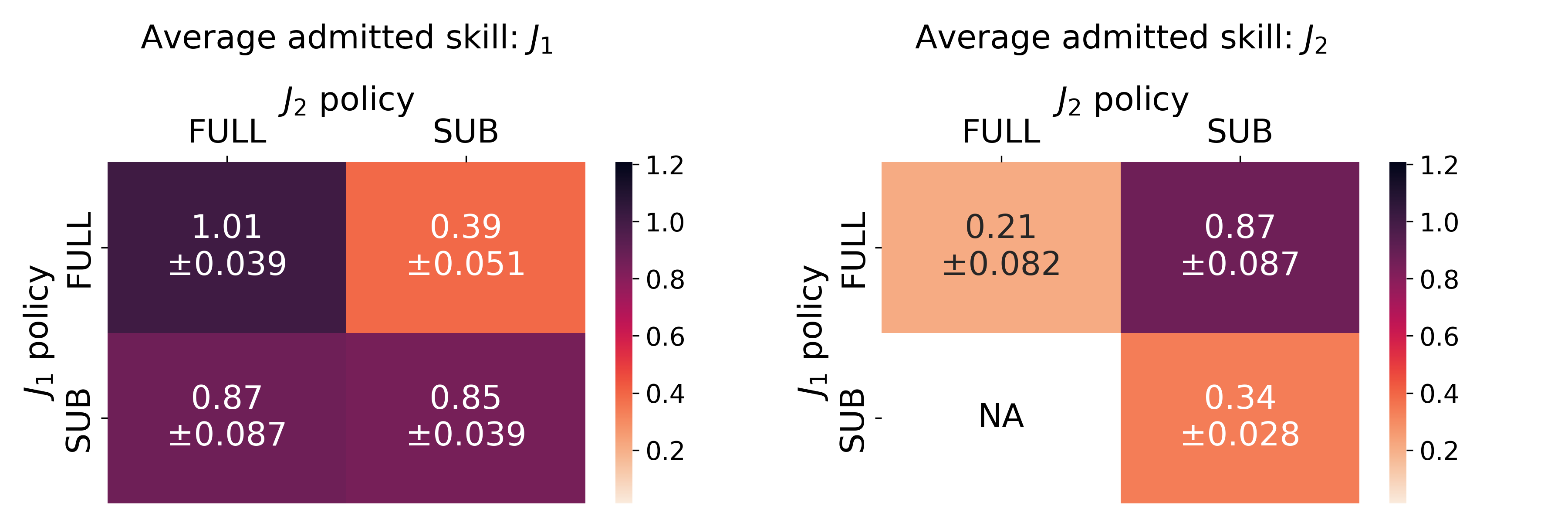}
    \caption{\small \centering Test cost $c_A = c_B = 2.0$. Unique equilibrium is $(P^1_{\subb}, P^2_{\subb})$.
    }
    \label{fig:twoschool_strat_cost2.0}
\end{subfigure}

\caption{Average admitted skill of the resulting student body, in $J_1$ (preferred school) and $J_2$ (less preferred school), as a function of $J_1$ and $J_2$ testing policies. {Error bars are shown for 2 standard errors.} In both settings, students in both groups have utilities $v_1=3$ and $v_2=2$. Figure \ref{fig:twoschool_strat_cost0.5} corresponds to low test cost ($c_A=c_B=0.5$) and Figure \ref{fig:twoschool_strat_cost2.0} corresponds to high test cost ($c_A=c_B=2.0$). See Appendix~\ref{app:synthetic_simulations_strategic_two_schools} for full parameter details. {In Figure \ref{fig:twoschool_strat_cost2.0}, no students apply to $J_2$ under policy pair $(P^1_{\subb}, P^2_{\full})$.}}
\label{fig:two_school_strategic_heatmap}
\end{figure*}

Figure \ref{fig:two_school_strategic_heatmap} shows the academic merit for schools $J_1$ and $J_2$ under various policy pairs $(P^1, P^2)$. 
{Figure \ref{fig:twoschool_strat_cost0.5} demonstrates a low test cost setting where the optimal policy of both the preferred school, $J_1$, and the less preferred school, $J_2$, is to require the test, regardless of the other school's policy. 
By contrast, \Cref{fig:twoschool_strat_cost2.0} demonstrates a high test cost setting where $J_1$'s optimal policy depends on $J_2$'s policy. In this setting, $J_1$'s optimal policy is to require the test when $J_2$ requires the test and to drop the test otherwise. $J_2$'s optimal policy is to drop the test, regardless of $J_1$'s policy.} 
Notably, in both settings, {when the preferred school $J_1$ requires the test, it benefits when the other school $J_2$ also requires it.}

{The results in the two settings in Figure \ref{fig:two_school_strategic_heatmap} parallel the results in the non-strategic, single school setting, albeit with more complex behavior. In the low test cost setting of Figure \ref{fig:twoschool_strat_cost0.5},
if both schools aim to maximize academic merit, then the equilibrium strategy is $(P^1_{\full}, P^2_{\full})$. This policy pair also ensures the optimal academic merit for each school, out of the four possible policy pairs. This setting parallels the non-strategic setting where students have a high level of access $\gamma_g$, and thus dropping the test does not significantly increase the number of students who apply. Conversely, in the high test cost setting of Figure \ref{fig:twoschool_strat_cost2.0}, the unique equilibrium strategy is $(P^1_{\subb}, P^2_{\subb})$, which parallels the non-strategic setting where students have low level of test access $\gamma_g$ and so dropping the test increases the number of students who apply. However, the strategic components induce more complex behavior. The preferred school, $J_1$, receives the highest average admitted skill under the policy pair $(P^1_{\full}, P^2_{\full})$, but the less preferred school, $J_2$, is incentivized to deviate and drop the test. Similarly, $J_2$ receives the highest average admitted skill under the policy pair $(P^1_{\full}, P^2_{\subb})$, when few students apply to $J_1$ and so $J_2$ has a larger pool of students who will enroll in $J_2$ if admitted, but $J_1$ is incentivized to deviate and drop the test. Also note that, if $J_1$ drops the test, then $J_2$ under the $P^2_{\full}$ policy has no applicants and so must also drop the test.}

The high cost setting of Figure \ref{fig:twoschool_strat_cost2.0}
depicts a setting where the preferred school $J_1$'s best response can depend on the less preferred school's policy. 
Here, the equilibrium strategy is $(P^1_{\subb}, P^2_{\subb})$. While $J_2$'s best response is to choose $P^2_{\subb}$ regardless of $J_1$'s policy, $J_1$'s best response varies: if $J_2$ chooses $P^2_{\full}$, then $J_1$'s best response is $P^1_{\full}$ and if $J_2$ chooses $P^2_{\subb}$, $J_1$'s best response is $P^1_{\subb}$. 
Intuitively, this occurs because the cost of the test is high enough relative to the students' valuation for $J_1$; thus, many students do not take the test solely for $J_1$ if they have the option of applying to $J_2$ without incurring test costs, but would take the test if both schools required it.

For the settings in which one school's optimal testing decision is independent of their competitor's strategy, the change in academic merit due to the other school's testing policy is not large enough to change the optimal testing behavior. 

Together, these simulations demonstrate that multiple schools environments can significantly influence each school's optimal testing policy. Schools cannot determine its testing strategy in isolation without considering other schools, even in the simplified setting in which all students prefer one school over the other.

\vspace{-0.3em}
\section{Conclusion}
\label{sec:conclusion}

We formalize the trade-off between information access and barriers in a testable framework, an important aspect of the decision for colleges to keep or drop standardized testing. As we show, there are reasonable parameter settings in which dropping testing improves or worsens both academic merit and diversity goals.

From a conceptual and modeling perspective, our work contributes  to the growing literature of fairness in decision-making systems. 
Our multi-feature version of the seminal model by \citet{phelps1972statistical} naturally 
 provides a general framework for analyzing how decision-makers use imperfect and differentially informative information in education and beyond. Our results further underscore that the \textit{design and choice of input features}---not just the decision rule itself---constitute a key lever for shaping fairness and efficiency.

Practically, our work suggests that schools must further invest in better signals and in expanding their applicant pools. In settings where test scores are found to be highly effective for skill estimation \textit{but} also impose large barriers, our analysis further suggests the value of another option for increasing fairness in admission: decreasing the access barriers. For example, several states have implemented policies to make the SAT and/or ACT mandatory for all public school students, while also reducing both financial and logistical barriers by paying the financial costs of test registration and offering the tests at more convenient times \citep{hyman2016act}.

We derived our theoretical results in a highly stylized setting where the school is Bayesian-optimal and knows the parameters of the model. While such a scenario is, in practice, unattainable, this work can be viewed as an information-theoretic limit to how well schools can identify the most qualified students. Even if a school had full knowledge of each group's feature distributions (i.e., were able to perfectly evaluate students' skills in context),  the school could not completely mitigate inequalities in admissions.

We remark that many of our results may extend to more general models. While we work with the standard Gaussian framework by~\citet{phelps1972statistical}, our insights hold for broader distributions where group $A$’s skill estimates are a mean-preserving spread of group $B$’s \citep{blackwell1953equivalent}. Similarly, relaxing the assumptions of feature independence and additive, uncorrelated noise would preserve our main conclusions, albeit without closed-form analytic solutions. Further details 
appear in Electronic Companion~\ref{app.general.distributions}.

Finally, our analysis considered the impact that access barriers to testing and strategic incentives might have  on the applicants' behavior and the schools' policy decisions. However, several other factors such as differential access to test preparation services \citep{park2015benefits} and family support~\citep{ espenshade2013no}, 
 may also constitute significant barriers for certain groups of students.  Many of these factors may additionally introduce compounding effects that contribute to students' future success, beyond those that we consider in our model.

\bibliographystyle{ACM-Reference-Format}
\bibliography{references,rebuttal_bib}


\newpage

%
%
%

 \appendix
\onecolumn

\section{Calibrated simulations with UT Austin data}
\label{sec:simulations_revision}

{We calibrate our model to empirical data to assess the effects of dropping test requirements under our model.} 
Our results establish that (a) there are reasonable parameter ranges both in which dropping the test can be beneficial and harmful for the desiderata, and (b)  when tests are required, outcomes can depend on whether the model allows students to self-select to take the test. 
Real admissions decisions are much more complex than our model \citep{selingo2020whogetsin}, and a key challenge in empirical admissions settings is selective data availability \citep{rothstein2004college}, since we typically only observe college outcome data for those admitted, which partially depends on test scores and other admission features. Given these limitations, our calibrated simulation exercise should be viewed as suggestive examples that dropping the test can either improve or worsen the desiderata, as opposed to establishing optimal policy for any particular setting.

\smallskip

\noindent\textbf{Data.} Our data is from the Texas Higher Education Opportunity Project (THEOP), a semipublic dataset of applications and transcripts for universities in Texas~\citep{tienda2011texas}. We focus on data from the University of Texas at Austin, for students who enrolled there in 1992-1997 and completed at least 24 credit hours.\footnote{This period represents admissions from before the time Texas adopted the Top Ten Percent rule, in which all students at the top of their Texas public high school class were accepted regardless of other application components.} For each student, we observe their high school \textit{class rank} (rounded to nearest decile), standardized \textit{test score} (SAT, or ACT score translated to equivalent SAT score); we also observe characteristics of their high school (including relative \textit{economic privilege} rounded to nearest quartile, which is a measure of the socioeconomic status of the students the high school serves). Since we consider enrolled students, we observe their GPA and number of credit hours for each enrolled semester, that we use to calculate overall GPA in their first year and afterwards. 

\smallskip

\noindent\textbf{Calibration and simulation setup.} We conduct a calibrated simulation exercise for a hypothetical admissions setting in which the applicant population looks distributionally similar to students who in reality {enrolled} to UT Austin.\footnote{This could reflect, for example, admissions at a college more selective than UT Austin in the time period considered.} For each individual, we use their cumulative \textit{college GPA}---not counting their first year---to represent their true skill. Then, as features, we use (in various simulations) their 
high school \textit{class rank}, 
standardized \textit{test score} and/or college \textit{first-year GPA}. To form the two groups, we take the upper (group $A$) and lower (group $B$) halves of the high schools' \textit{economic privilege}\footnote{Column by the data provider, defined as ``Publicly available data from the Texas Education Agency (TEA) is used to stratify regular, Texas public high schools according to the socioeconomic status of the students they serve. The 25\% of high schools with the lowest percent of students ever economically disadvantaged are designated as Upper quartile. The 25\% of high schools having the highest percent of students ever economically disadvantaged are designated as Lower quartile. Because the statewide share of economically disadvantaged students rose over time, quartile cut points are calculated separately for each year.'' We then binarize the quartiles.} index.

We calibrate our model parameters to the empirical data. We calibrate the true skill mean $\mu$ and variance $\sigma^2$ to the empirical mean and variance of the cumulative college GPA, excluding the first year. We then calibrate the conditional feature distributions for each group, which in our model are distributed as $\theta_k \sim N(q + \mu_{gk}, \sigma^2_{gk})$; i.e., for each group $g$ and feature $k$ pair, we need estimates of $\mu_{gk}$ and $\sigma^2_{gk}$, the \textit{conditional} mean and variance of the feature given the student's true skill. We estimate these values  by running an ordinary least squares regression $\theta_{k} \sim q$, where $q$ is the observed college GPA. Let the fitted regression model be $\hat{\theta}_k = \hat{\beta}_0 + \hat{\beta}_1 q$, so that $q = (\hat{\theta}_k - \hat{\beta}_0) / \hat{\beta}_1 $. To normalize the features so that a one unit increase in the feature corresponds to a one unit increase in skill level (so that the feature has mean $q + \mu_{gk}$), we center and scale each observed feature to obtain $\theta_k' = (\theta_k - \hat{\beta}_0) / \hat{\beta}_1$, and likewise for the predicted features $\hat{\theta}_k$ to obtain $\hat{\theta}_k'$. Now, we calibrate the model to the distribution of $\theta_k'$. We set $\mu_{gk}$ to be the sample mean of $\theta_k'$ and $\sigma^2_{gk}$ to be the sample variance of the residuals $\hat{\theta}_k' - \theta_k'$. 
{The calibrated standard deviation parameters $\sigma_{gk}$ are in \Cref{tab:theop_calibrated_feature_var}.} 
\begin{table}[tbh]
\small
	\centering
 \begin{tabular}{c |c|c | c}
\textbf{Group} & \textbf{HS class rank} & \textbf{College GPA, 1st year} & \textbf{Test score}\\ \hline
A (high economic privilege)   &  2.00  &  0.98 & 3.30 \\
B (low economic privilege)   &  2.65  &  0.91 & 3.11 \\
 \end{tabular}

\caption{Calibrated feature standard deviations $\sigma_{gk}$ for each group $g$ and feature $\theta_k$. This calibration suggests that class rank is relatively more predictive of cumulative college GPA for the high economic privilege group while the test score is more predictive for the low economic privilege group -- consistent with \citet{berkeleyreport} and \citet{mit2022}. Most predictive for each is  the first-year college GPA.}
\label{tab:theop_calibrated_feature_var}
\end{table}

Using these calibrated mean and variance parameters, we then simulate our model, with the students' applications and the school's Bayesian updating as described earlier.  
We simulate the admission outcomes in both the setting with strategic students and the setting with non-strategic students. In both settings, we fix group $A$ to have full access to the test ($\gamma_A=1$ and $c_A=0$ in the non-strategic and strategic settings, respectively) and vary the level of access for group $B$ students. We set the student utility for the school to be $v=5$. We fix an equal proportion of students from each group in the candidate pool ($\pi=0.5$). We simulate a setting with 10,000 applicants and a capacity of 1,000. For each parameter set, we run 100 simulations and report the mean and 95\% confidence intervals across simulation runs.

We simulate two informational cases, which correspond to the school having access to different features when making its decision. 
\begin{description}
	\item[Low informativeness:] \textit{Class rank} and (potentially) \textit{Test score}. Simulates, for example, a setting in which the application pool and information available is incoming first-year students at UT Austin.
	\item[High informativeness:] \textit{First-year GPA} and (potentially) \textit{Test score}. Simulates, for example, a setting in which the application pool and information available is students at the end of their first year at UT Austin. 
\end{description}

To make the non-strategic and strategic settings comparable, we define the notion of \textit{test access level} for group $B$ as the proportion of group $B$ students taking the test. In the non-strategic setting, this is $\gamma_B$ by definition. In the strategic setting, each cost level $c_B$ induces a test access level which can be found through simulation. We note that while the overall number of group $B$ students taking the test is the same for a fixed test access level, in the strategic setting this group of students are disproportionately high-skilled (see Lemma~\ref{lemma.thres.strategy} and Figure \ref{fig:simulations_prob_apply_by_skill}).

\begin{table}[tb]
\small
\centering

 \begin{tabular}{c|c|cc|cc}
	&  & \multicolumn{2}{c|}{\textbf{Academic merit}} & \multicolumn{2}{c}{\textbf{Diversity Level}} \\ \hline
  
 \textbf{Informational Case}
 & \textbf{Student behavior} & With test           & Without test           & With test           & Without test           \\ \hline
Low      
& Strategic    & 3.42               & 3.33                   & 40.8\%              & 35.7\%              \\
&     & $\pm$ .005 & $\pm$ .0005 & $\pm$ .3\% & $\pm$ .03 \% \\
   & Non-strategic      & 3.38            & 3.33                 & 23.7\%              & 35.7\%     \\
   &    & $\pm$ .005 & $\pm$ .0005 & $\pm$ .3\% & $\pm$ .03\% \\
\hline
 High   
 & Strategic    &   3.76           &     3.74           & 52.5\%             &  52.4\%  
 \\
 &  & $\pm$ .005 & $\pm$ .0004 & $\pm$ .3\% & $\pm$ .03\% \\
    & Non-strategic     &       3.66   &   3.74     & 29.4\%            & 52.4\%  \\
    &   &  $\pm$ .004 & $\pm$ .0004 & $\pm$.3\% & $\pm$ .03 \% \\
\end{tabular}

\caption{How academic merit and diversity level of admitted students change with and without requiring a test score, for two informational cases (how informative the non-test feature is) and for the strategic and non-strategic settings. Academic merit (GPA) ranges from 1.0-4.0. Diversity level is shown as a percentage of admitted students. Shown with 95\% confidence intervals. 
This table assumes a 40\% test access for group $B$; for outcomes for the full range of group $B$ test access, see Figures \ref{fig:calibrated_sims_theop_test_and_firstyeargpa} and  \ref{fig:calibrated_sims_theop_test_and_hsclassrank}. 
Values are averaged across 100 simulation runs with 95\% confidence intervals shown. All differences are statistically significant.
}
\label{tab:empiricsresults_with_confints}
\end{table}

\begin{figure*}[tb]
    \begin{subfigure}[b]{0.31\textwidth}
    \centering
    \includegraphics[width=\linewidth]{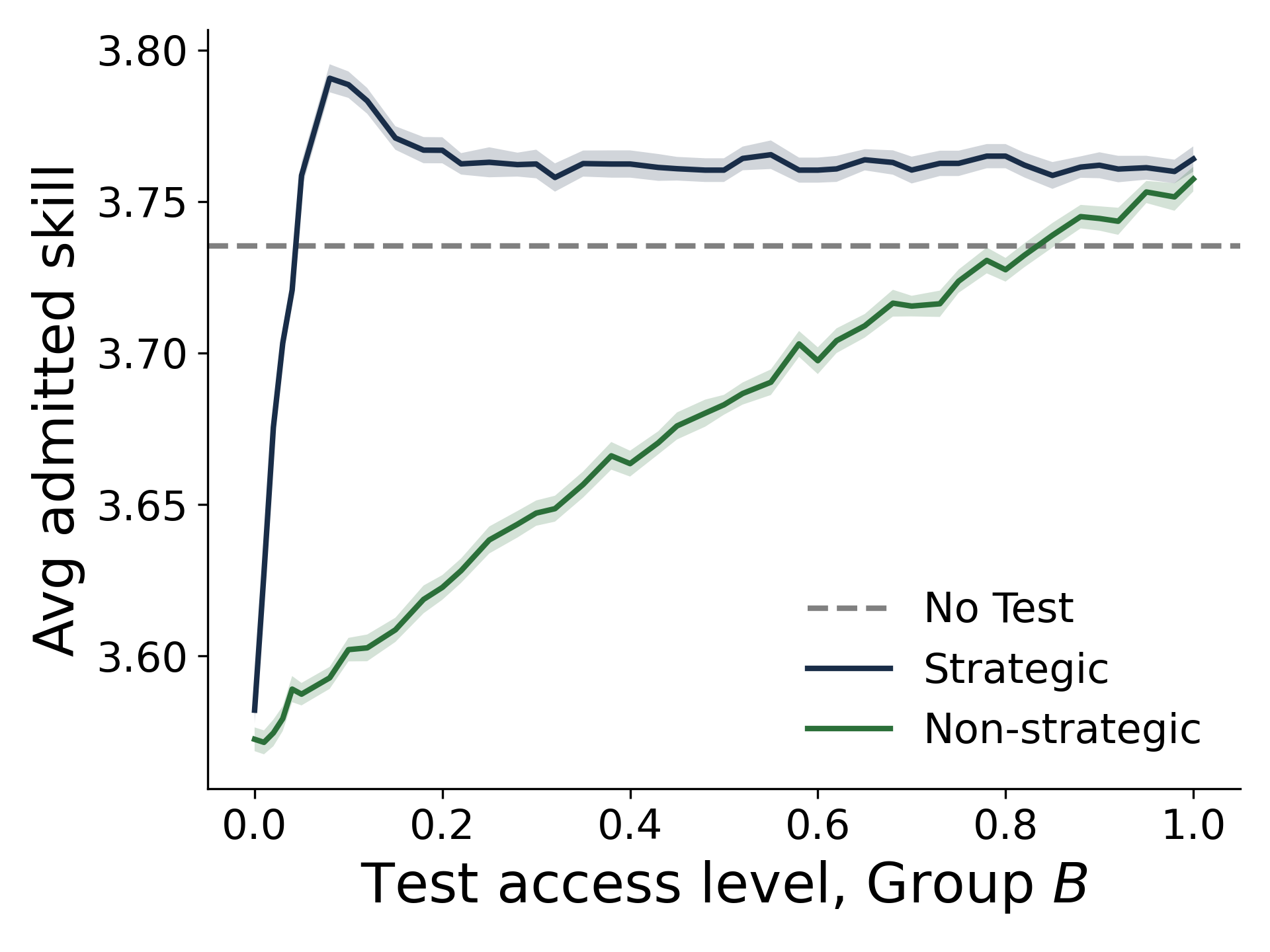}
	\caption{\small Academic merit\normalsize}
 \label{fig:calibrated_with_cgpa_academic_merit}
    \end{subfigure}
    \hfill
    \begin{subfigure}[b]{0.31\textwidth}
    \centering
    \includegraphics[width=\linewidth]{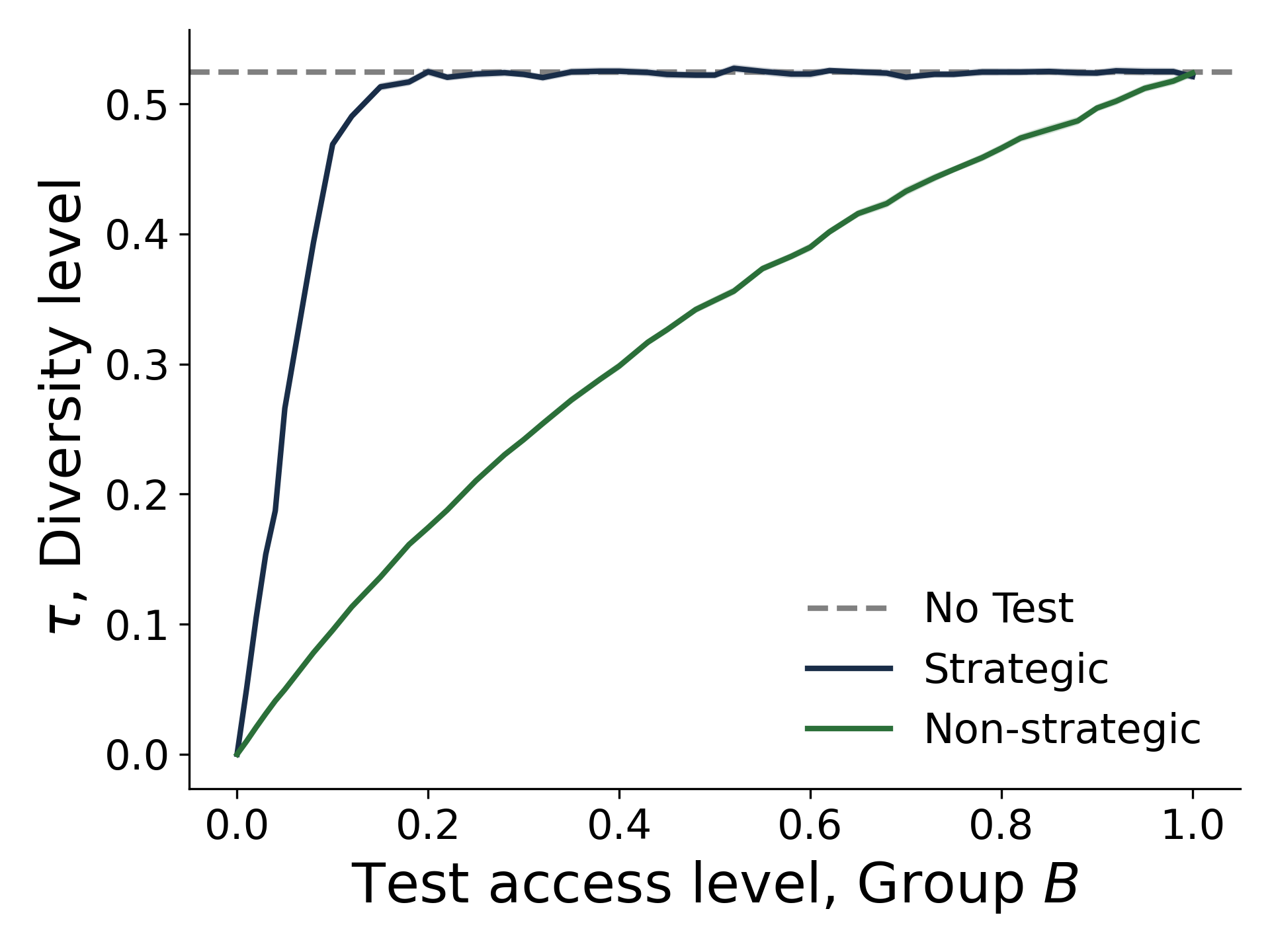}
	\caption{\small Diversity level \normalsize}
	\label{fig:calibrated_with_cgpa_diversity}
\end{subfigure}
\hfill
\begin{subfigure}[b]{0.31\textwidth}
    \centering
    \includegraphics[width=\linewidth]{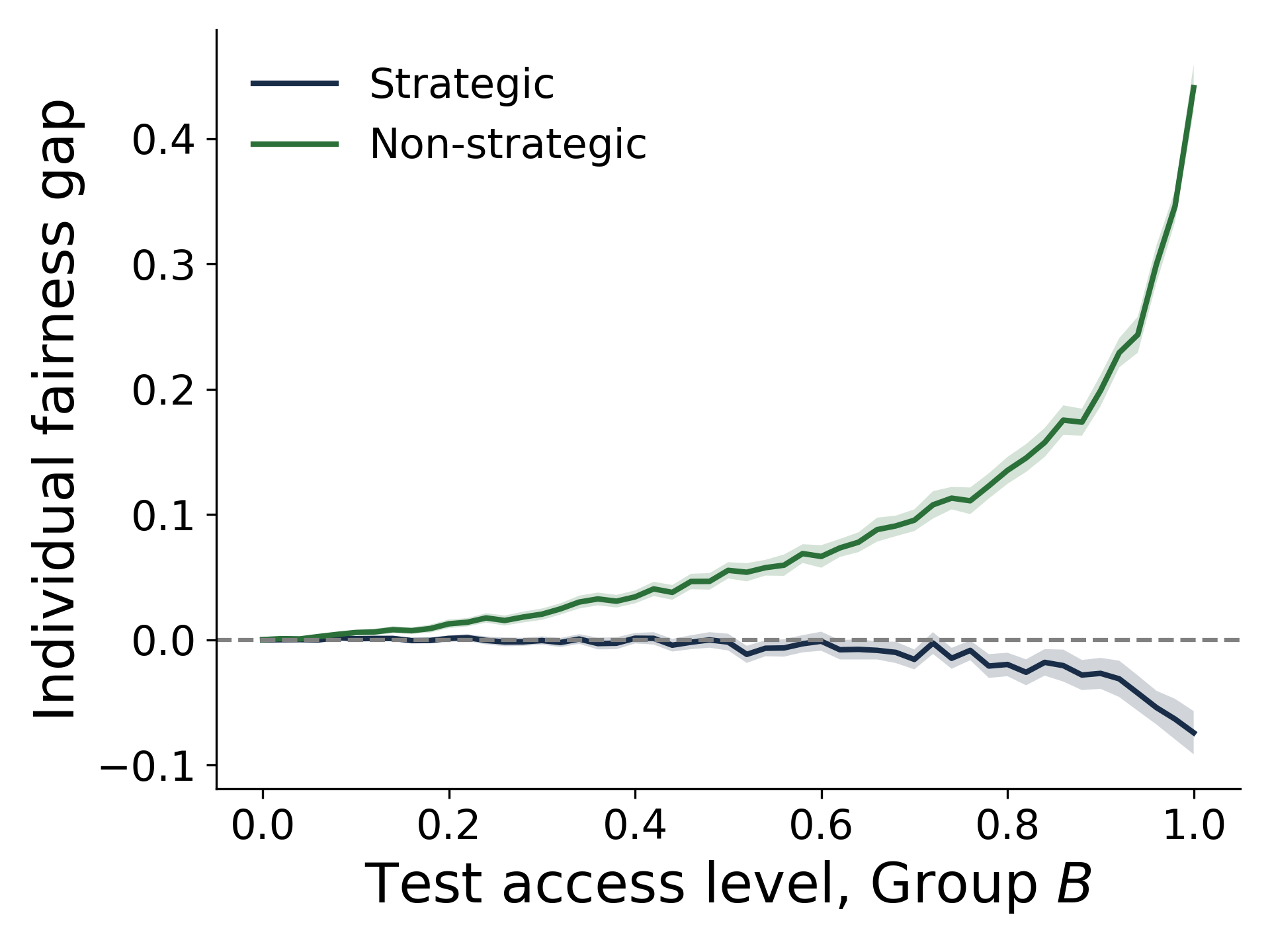}
    \caption{\small Individual fairness gap \normalsize}
    \label{fig:calibrated_with_cgpa_ifgap}
\end{subfigure}
\hfill
     \caption{Calibrated simulations when full set of features is $\{$First year GPA, test score$\}$ -- high informativeness case (see Table \ref{tab:theop_calibrated_feature_var} for informativeness of features). Figures (a) and (b) vary test the access level for group $B$. Figure (c) shows the individual fairness gap at a fixed group $B$ test access level of 40\%. Figures show average value across 100 simulation runs and 95\% confidence intervals.
     }
 \label{fig:calibrated_sims_theop_test_and_firstyeargpa}
\end{figure*}

\begin{figure*}[h]
    \begin{subfigure}[b]{0.31\textwidth}
    \centering
    \includegraphics[width=\linewidth]{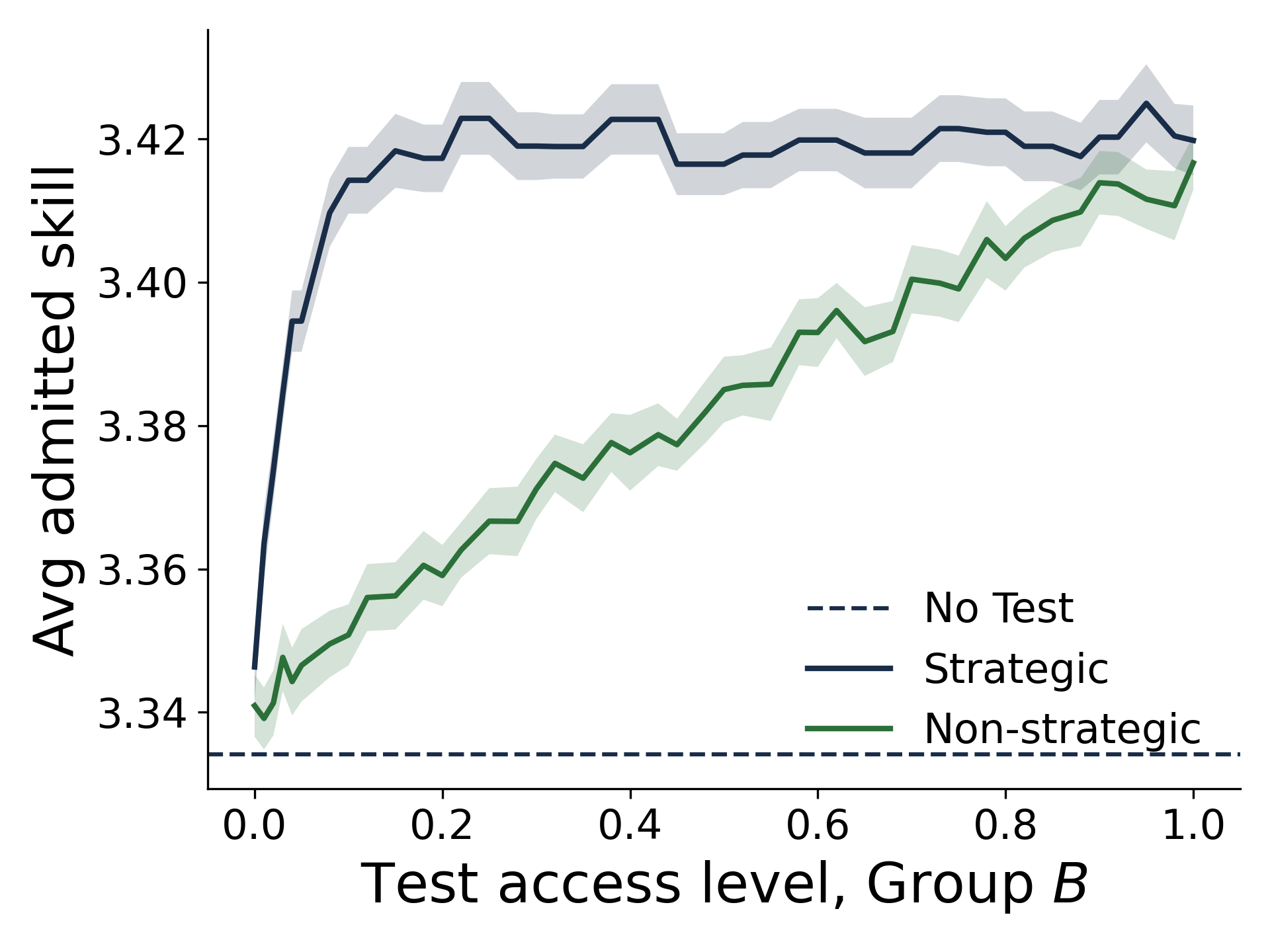}
	\caption{\small Academic merit \normalsize}
	\label{fig:calibrated_with_hsclassrank_academic_merit_group_a}
    \end{subfigure}
    \hfill
    \begin{subfigure}[b]{0.31\textwidth}
    \centering
    \includegraphics[width=\linewidth]{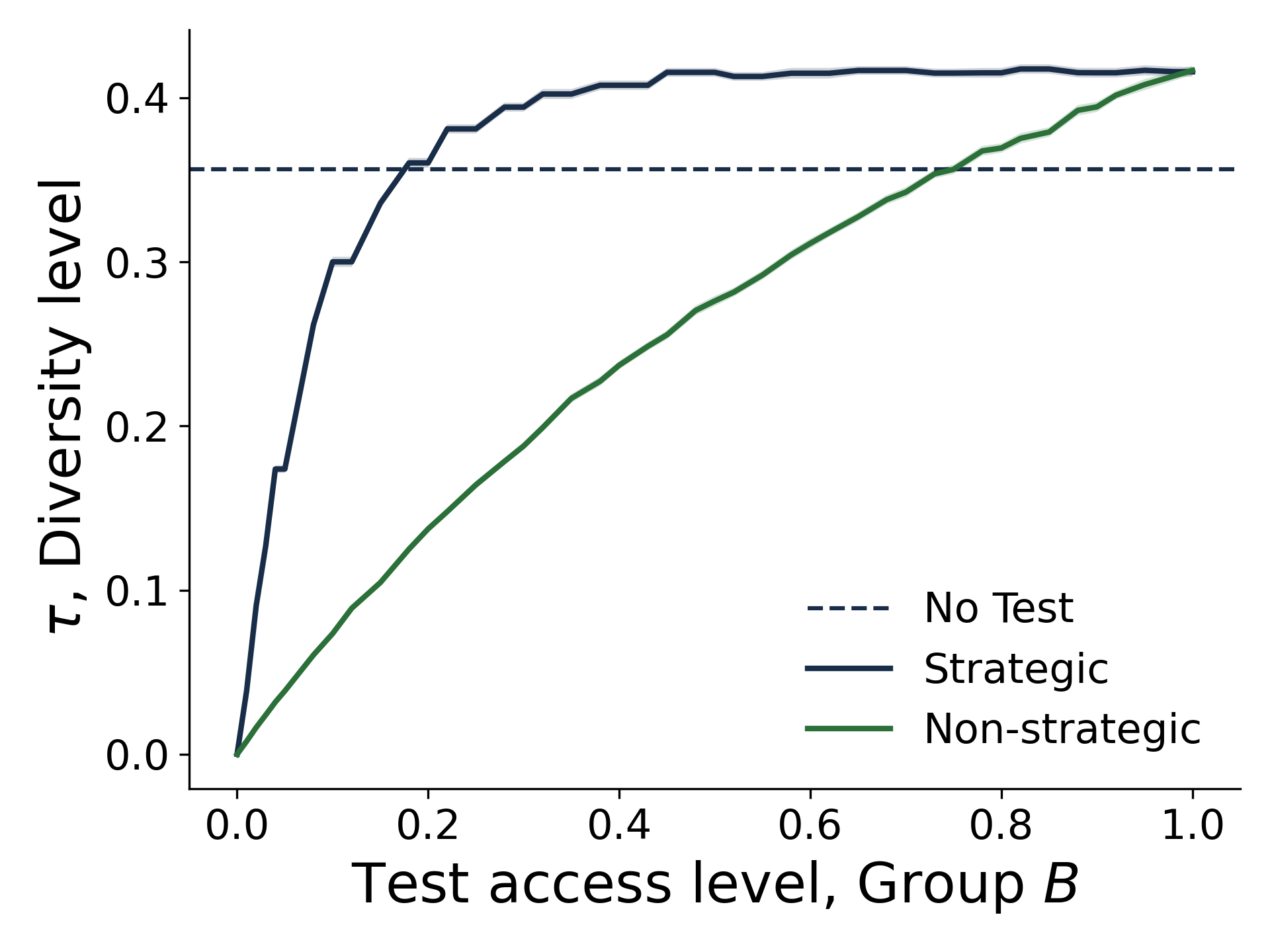}
	\caption{\small Diversity level \normalsize}
	\label{fig:calibrated_with_hsclassrank_diversity}
\end{subfigure}
\hfill
\begin{subfigure}[b]{0.31\textwidth}
    \centering
    \includegraphics[width=\linewidth]{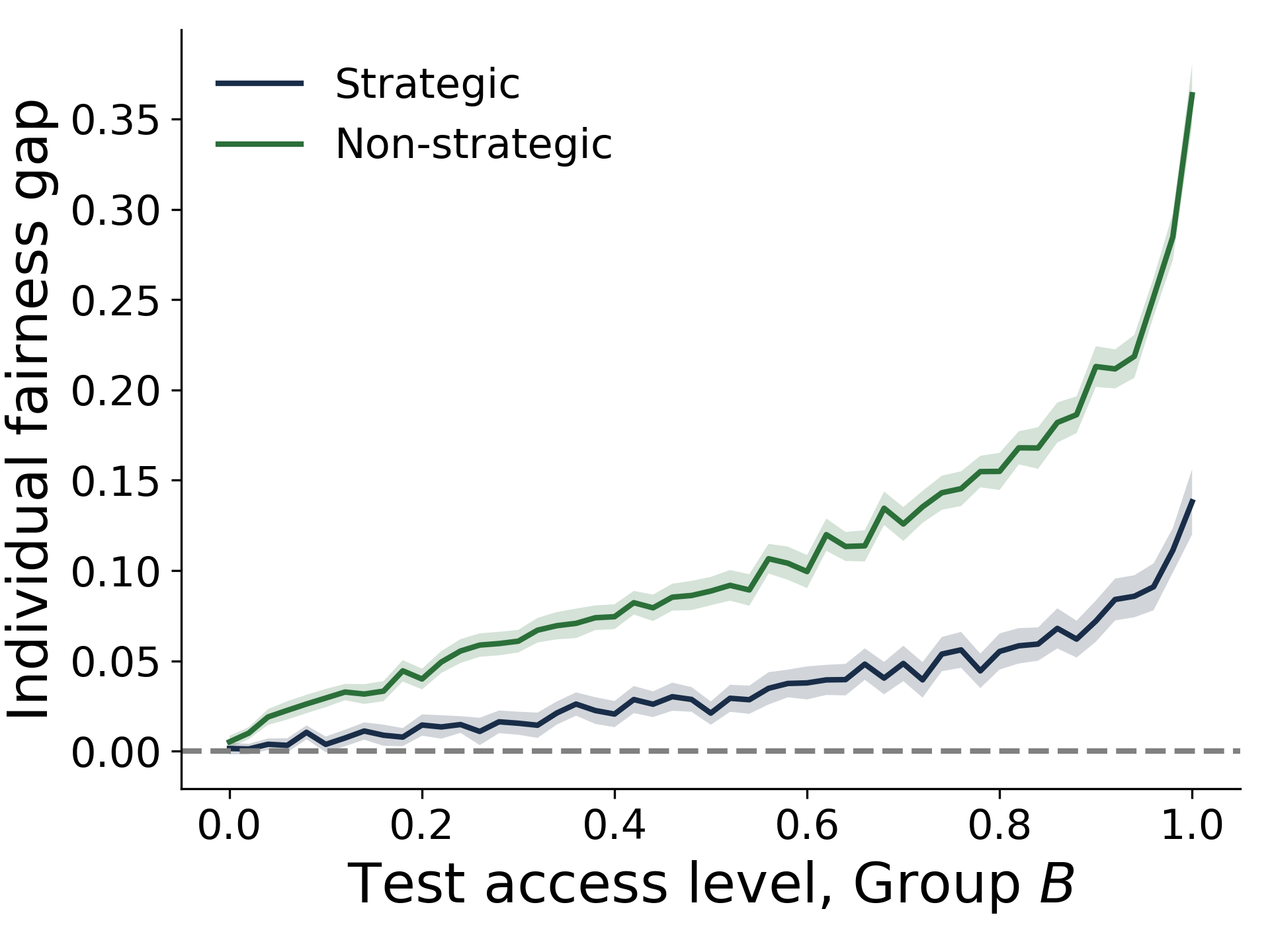}
    \caption{\small Individual fairness gap \normalsize}
    \label{fig:calibrated_with_hsclassrank_ifgap}
\end{subfigure}
\hfill
     \caption{Calibrated simulations when full set of features is $\{$High school class rank, test score$\}$. Low informativeness case. See Table \ref{tab:theop_calibrated_feature_var} for informativeness of features.  
     }
 \label{fig:calibrated_sims_theop_test_and_hsclassrank}
\end{figure*}

\smallskip

\noindent\textbf{Simulation results.} Table \ref{tab:empiricsresults_with_confints} summarizes the admission outcomes with and without the test, for a fixed level of group $B$ students having access (40\%), while all group $A$ students have access.\footnote{Using the \cite{2022California} California SAT Suite of Assessments Annual Report, we calculate that a student from the bottom two quintiles of family income are 38\% as likely to take the test as a student from the top two quintiles. Thus we focus on an access levels of 100\% and 40\% for groups $A$ and $B$, respectively.} 
For outcomes for the full range of group $B$ test access, see Figures \ref{fig:calibrated_sims_theop_test_and_firstyeargpa} and \ref{fig:calibrated_sims_theop_test_and_hsclassrank}, for the high and low informational environment, respectively. 

In this setting, exactly half of the students are in group $B$ ($\pi=0.5$). For any diversity level below 50\% (i.e., students in group $B$ make up less than half of the admitted student body), we consider group $B$ to be under-represented.

Overall, the results show that the effects of dropping the test requirement depend crucially on both the informational environment and whether students are strategic. At a test access level of 40\% for group $B$, dropping the test worsens both academic merit and diversity level when students are strategic in both informational cases, although only slightly in the high information case. However, when students are non-strategic, dropping the test improves both metrics when the remaining feature has high informativeness, whereas dropping the test has mixed results when the other feature has low informativeness.

\textit{Comparing effect of test access in strategic and non-strategic settings.} The results show that in both informational settings, academic merit, diversity, and individual fairness all worsen when fewer group $B$ students have access to the test. 
However, for a given level of test access, the outcomes for all three metrics are better when students are strategic, compared to when they are non-strategic. 
In the strategic setting, the students with higher skill levels are more likely to take the test (see Lemma~\ref{lemma.thres.strategy} and Figure \ref{fig:simulations_prob_apply_by_skill}), as opposed to the non-strategic setting where all students in group $B$ have the same probability $\gamma_B$ of taking the test. Thus, as we see in Figures  \ref{fig:calibrated_sims_theop_test_and_firstyeargpa} and \ref{fig:calibrated_sims_theop_test_and_hsclassrank}, even when the test access levels are as low as 30 percent, the admission outcomes of academic merit, diversity, and individual fairness are comparable to when group $B$ has full test access. This observation, of course, relies on the students appropriately assessing their likelihood of admission upon taking the test, which we assume in our model. 
We also note that academic merit in particular is not monotonic in the test access level (Figure \ref{fig:calibrated_sims_theop_test_and_firstyeargpa}). As the access level for group $B$ approaches 0, the average skill level for admitted students increases for group $B$ but decreases for group $A$, leading to non-monotonicity in the overall academic merit. See Figure \ref{fig:academic_merit_by_cost} for an illustration of average skill level of admitted students, by group.

\textit{Effect of the informational environment.} 
When the college has access to a high quality signal on all students---\textit{first-year GPA}---dropping test scores increases both academic merit  and diversity when costs are high enough; it allows more students to apply, without incurring a substantial informational loss. 
In contrast, in the low informativeness case, without test scores the school must rely on students' high school ranks, which are especially uninformative for group $B$, thus leading to worse admissions outcomes.

These findings underscore our theoretical results: the consequences of dropping test scores depend crucially on the information content of other signals, the level of strategic behavior by applicants, and the levels of access to the test. Decisions to require the test should not (and cannot) be made in a context-independent manner.

\smallskip

\noindent\textbf{Discussion.} 
There are several ways in which our simulation setup differs from reality, for example: (1) We use college GPA as a measure of student true skill; in reality, GPA is a function of many other aspects as well, such as college major and barriers faced during college \citep{engle2008moving}. (2) Because of our choice to use college GPA as a true skill measure, we cannot simulate our model for all students who \textit{apply} to UT Austin, as data is censored\footnote{This is a  common barrier to measuring the predictive power of standardized testing in admissions \citep{rothstein2004college,weissman2020gre}.}---we do not observe their college GPA unless they enrolled. Thus, we must simulate a hypothetical admissions setting for which the enrolled population at UT Austin is a reasonable application pool. (3) To closely simulate our model, we fit Normal distributions to the data, while the respective distributions may not be Normally distributed (e.g., many of the features are truncated). (4) We do not have estimates of the barriers or costs to testing, and in fact {almost all applicants in the data (over 99.9\%)} have test scores due to school policies at the time; 
thus, we have to artificially simulate some students as not having access. For these reasons, our simulations should not be interpreted as making statements about the UT Austin context or any particular admissions setting.

\section{Simulations with synthetic data}

\subsection{Supplementary simulations for the non-strategic setting}

\subsubsection{Supplemental simulation figures for the non-strategic setting}
\label{app.suppfigs.nonstrategic}
\label{app.B4}
\begin{figure}[h]
	\begin{subfigure}[b]{0.48\textwidth}
		\centering
		\includegraphics[width=\linewidth]{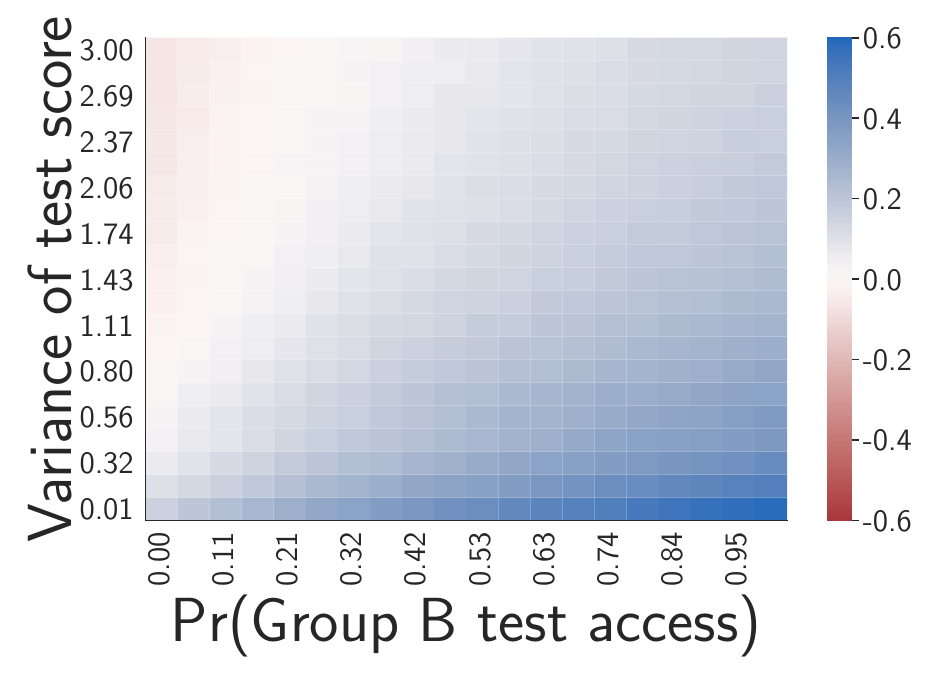}
		\caption{Difference in Average admitted skill}
		\label{fig:2dplot_avgskill}
	\end{subfigure}
	\hfill
	\begin{subfigure}[b]{0.48\textwidth}
		\centering
		\includegraphics[width=\linewidth]{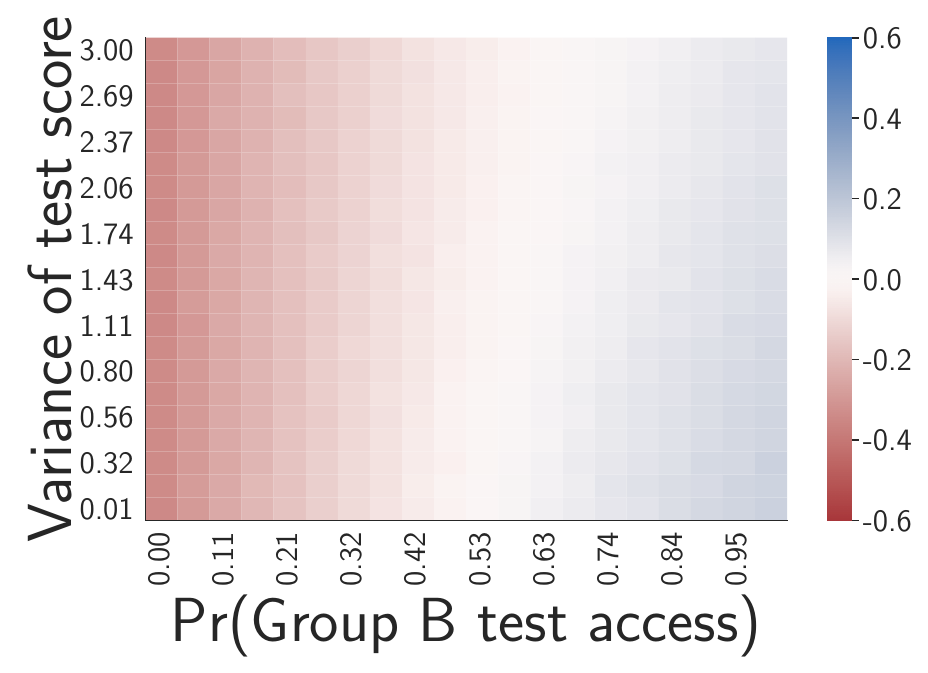}
		\caption{Difference in Diversity level}
		\label{fig:2dplot_diversity}
	\end{subfigure}

	\begin{subfigure}[b]{0.48\textwidth}
		\centering
		\includegraphics[width=\linewidth]{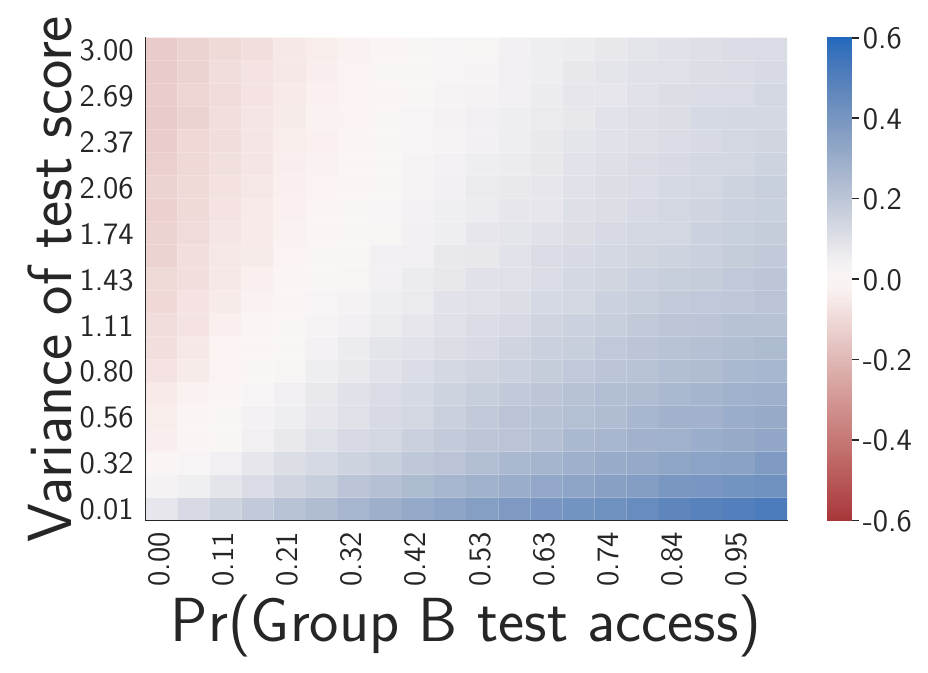}
		\caption{\small Difference in Average admitted skill, Group A \normalsize}
		\label{fig:2dplot_avgskilla}
	\end{subfigure}
	\hfill
	\begin{subfigure}[b]{0.48\textwidth}
		\centering
		\includegraphics[width=\linewidth]{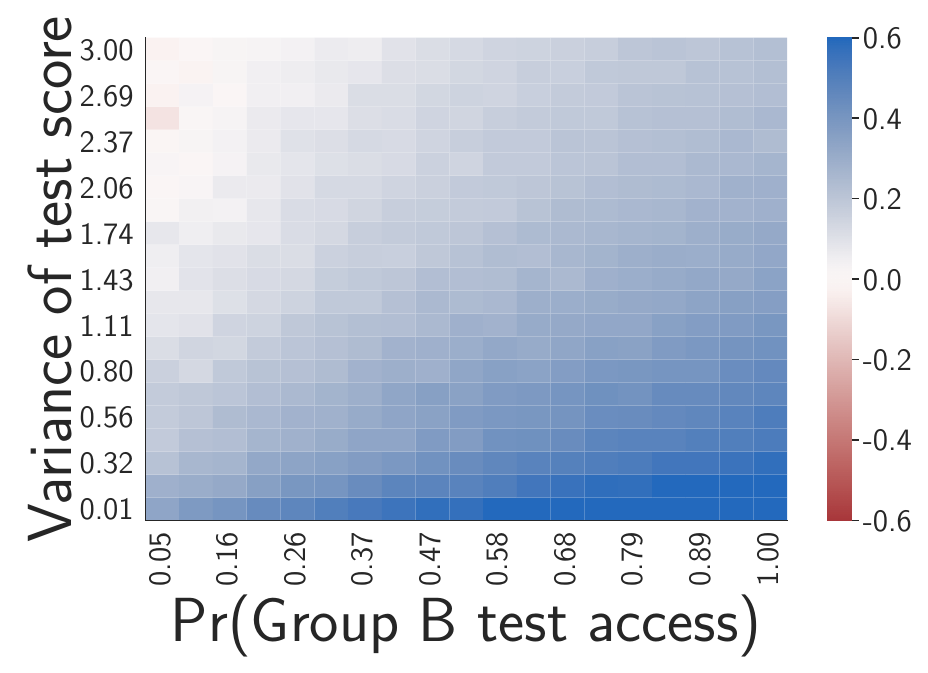}
		\caption{\small Difference in Average admitted skill, Group B \normalsize}
		\label{fig:2dplot_avgskillb}
	\end{subfigure}
	\caption{Difference between test-based and test-free policies with respect to various objective functions. The more negative (red) the difference, the more that dropping the test \textit{improves} that metric compared to test-based policies. Simulation is with budgets case, using parameters as given in Electronic Companion~\ref{appsec:simparams}. The plot reads as follows: in Figure~\ref{fig:2dplot_avgskill}, a difference of $0.6$ means that the average academic merit with a test-based policy is $0.6$ higher than that with a test-free policy.}

	\label{fig:policycompare_barriers2dplot}
\end{figure}

Figure~\ref{fig:policycompare_barriers2dplot} supplements the results in Theorem~\ref{thm:threshold_char_without_aa} and Proposition~\ref{prop:diversitythresholdnoAA}, regarding the thresholds at which academic merit and diversity improve after dropping the test. In particular, they illustrate that for high enough test score variance or high enough barriers, dropping the test score improves the objectives.

\subsubsection{Simulation parameters}
\label{appsec:simparams}

We report the parameters for the simulations with non-stretegic students. 

\noindent\textbf{Figure~\ref{fig:2destimates}.}
$    C = 0.2, 
    \pi = 0.5, 
     q, 
     \theta_{A0}, 
    \theta_{A1} \sim N(0,1),
     \theta_{B0} \sim N(-4,5), 
    \theta_{B1} \sim N(-4,1), 
    \gamma_A = 1, 
    \gamma_B = \frac{2}{3}.$

\smallskip
\noindent\textbf{Figure~\ref{fig:variancevarying}.} Same as Figure~\ref{fig:2destimates}, except with  $\theta_{B1} \sim N(-4,\sigma_{B1}^2)$, where   $\sigma_{B1}^2 \in (0,5)$. For subfigures (\ref{fig:varianceadmittedskill}) - (\ref{fig:varianceIF}) we fix test access $\gamma_A = \gamma_B=1$. For subfigures (\ref{fig:barrieradmittedskill}) - (\ref{fig:barrierIF}) fix the test score variance of group B to be equal to that of group A, so that $\sigma^2_{B1} = \sigma^2_{A1}=1$ and we vary $\gamma_B$.

\smallskip
\noindent\textbf{Figure~\ref{fig:policycompare_barriers2dplot}.} Same as Figure~\ref{fig:2destimates}, except with test score precision varying together for both groups $\sigma_{A1}^2 = \sigma_{B1}^2 \in (0, 3)$, and group $B$ test access varying, $\gamma_B \in (0, 1)$.

\smallskip
\noindent\textbf{Figure~\ref{fig:policycompare_barriers_withunaware}.}
Same as Figure~\ref{fig:2destimates}.

\subsection{Additional simulations for the strategic setting}

In this appendix, we report results for the setting with strategic students. Section~\ref{app:synthetic_simulations_strategic_single_school} focuses on a single school and Section~\ref{app:synthetic_simulations_strategic_two_schools} studies two competing schools.

\subsubsection{Simulations with synthetic data and a single school}
\label{app:synthetic_simulations_strategic_single_school}

We first describe the general simulation setup for a single school and then report the parameters and other details for each figure.

\noindent\textbf{Simulation setup for strategic students and a single school.}
Fix a single school that uses policy $P_{\full}$. 
All students are initialized with realizations of their true skill $q$ and non-test features $\thetaset_{\subb}$. 
Students follow the behavior outlined in the proof of Lemma \ref{lemma.thres.strategy} (see Appendix \ref{app:strategic_single}).
Note that this result shows the existence of an equilibrium that is characterized by thresholds $\underline{q}^g_{\subb}$, where students take the test if and only if $\tilde{q}_{\subb} \geq \underline{q}^g_{\subb}$ but does not directly give the value of these thresholds. 
We simulate admissions process under candidate equilibrium thresholds $\underline{q}^g_{\subb}$, each of which results in a certain number of students being admitted. We find the thresholds $\underline{q}^g_{\subb}$ such that the school's capacity constraint is respected and size of the admitted student body is the closest to the threshold. 

The school's admission decisions and students' test taking decisions follow the proof of Lemma \ref{lemma.thres.strategy}. For the policy $P_{\full}$, fix a candidate admission threshold $\tilde{q}'$ for the estimated skill $\tilde{q}_{\full}|\thetaset_{\full},g$.  
First, consider the students' observations. 
Each student observes their non-test features $\theta_{\subb}$ and estimates the distribution of their estimated skill if they were to take the test $\tilde{q}_{\full} \mid \tilde{q}_{\subb}, g, P_{\full}$, given in Equation \eqref{eq:q_tilde_cond_thetasubb_q}.  
The student then calculate $\prob (Y = 1 \mid \thetaset_{\subb}, g, P_{\full})= \prob( \tilde q (\thetaset_{\full}, g) \geq \tilde q^*_{\full} \mid \thetaset_{\subb}, g)$ and solves for their optimal test taking decision $\arg \max_{\alpha \in \{0,1\}} \alpha ( v \prob (Y = 1 \mid \thetaset_{\subb}, g, P_{\full}) - c_g  ),$ as seen in Lemma \ref{lemma:reduced_equilibria_single}. In other words, the student takes the test when $v \prob (Y = 1 \mid \thetaset_{\subb}, g, P_{\full})\geq c_g$. The students who take the test then apply to the school. Now, the school admits all students with estimated skill $\tilde{q}_{\full}|\thetaset_{\full},g \geq \tilde{q}'$. Note, however, that this may result in a smaller or larger admitted class than the school's capacity. We then search across candidate thresholds $\tilde{q}'$ and set $\tilde{q}^*_{\full}$ to be the $q'$ that attains the largest admitted class size, while still respecting the capacity constraint. 

\begin{figure}
    \centering
    \includegraphics[scale=0.4]{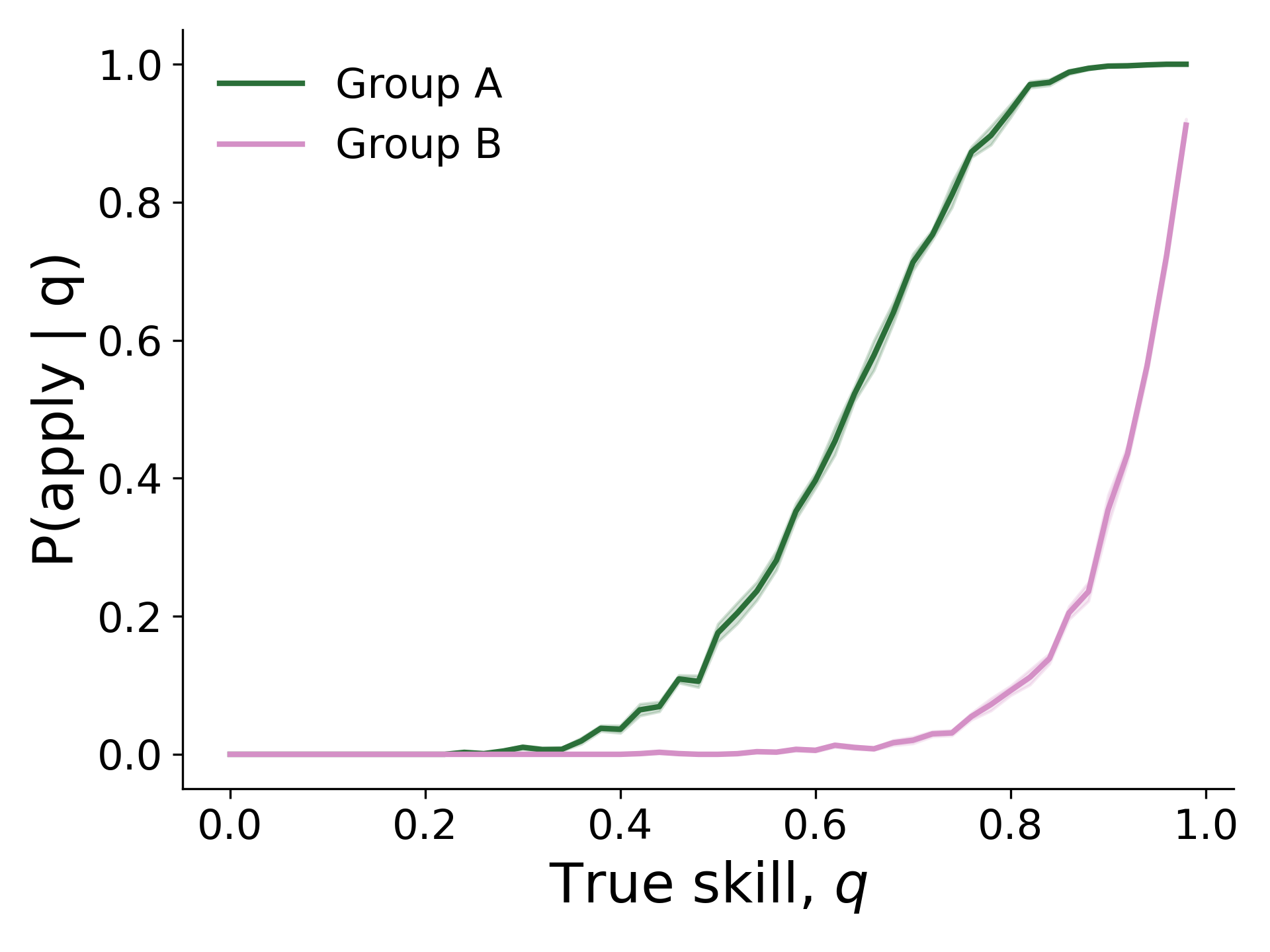}
    \caption{Results of a simulation calculating the student equilibrium decisions (characterized by \cref{eq.student.decision.cost.theta}) on whether to take the test and apply to a school that requires the test, as a function of their true skill and group. When students are strategic, high skilled students are more likely to take the test and apply. We consider a case where Group $A$ has higher precision and lower test costs: Parameters are $\mu_{gk}=0$ for all $g,k$, $\sigma^2_{A1}=\sigma^2_{A2}=\sigma^2_{B1}=1$ and $\sigma^2_{B2}=2$, where $k=2$ denotes the test feature. Costs are $c_A=0.5$ and $c_B=3$. In this case, Group $A$ students are more likely to take the test than Group $B$ students of the same true skill value. Full simulation parameter set can be
found in Electronic Companion \ref{app:synthetic_simulations_strategic_single_school}. }
\label{fig:simulations_prob_apply_by_skill}
\end{figure}

\noindent\textbf{Parameters for Figure \ref{fig:simulations_prob_apply_by_skill}.}
Figure \ref{fig:simulations_prob_apply_by_skill} shows simulation results illustrating the student equilibrium decisions
(characterized by \Cref{eq.student.decision.cost.theta}) on whether to take the test and apply to a school that requires the test, as a function of their true skill and group. 

There are two features, where the where the non-test feature is equally informative for both groups, but the test score is more informative for group $A$ than group $B$. The true skill distribution for both groups is Normally distributed with mean $\mu=0$ and variance $\sigma^2=1$. The features for the two groups are Normally distributed with mean $\mu_{gk}=0$ for all $g,k$ and $\sigma^2_{A1}=\sigma^2_{A2}=\sigma^2_{B1}=1$ and $\sigma^2_{B2}=2$, where $k=2$ denotes the test feature. Students of both groups have valuation $v=5$ for the school. Test costs are $c_A=0.5$ and $c_B=3$. 
There are $N=10000$ students and the school has capacity $0.1$. To find the equilibrium $\tilde{q}^*_{\full}$, we search over a grid of 250 threshold values. The mean over 20 simulation runs is presented, along with 95\% confidence intervals.

\begin{figure*}[tb]
    \begin{subfigure}[b]{0.32\textwidth}
    \centering
    \includegraphics[width=\linewidth]{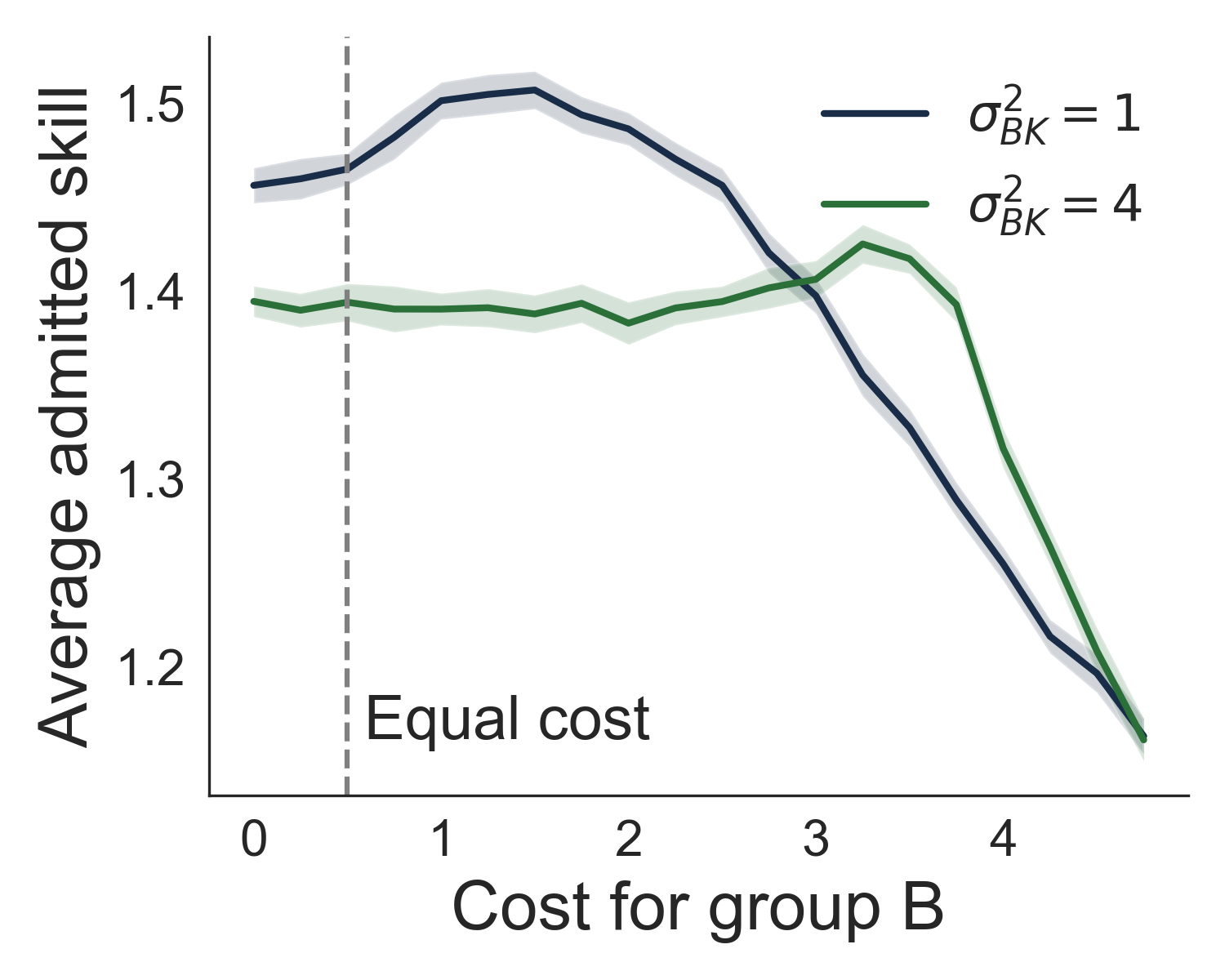}
	\caption{\small Academic merit \normalsize}
	\label{fig:academic_merit_by_cost_var}
\end{subfigure}
\hfill
    \begin{subfigure}[b]{0.32\textwidth}
    \centering
    \includegraphics[width=\linewidth]{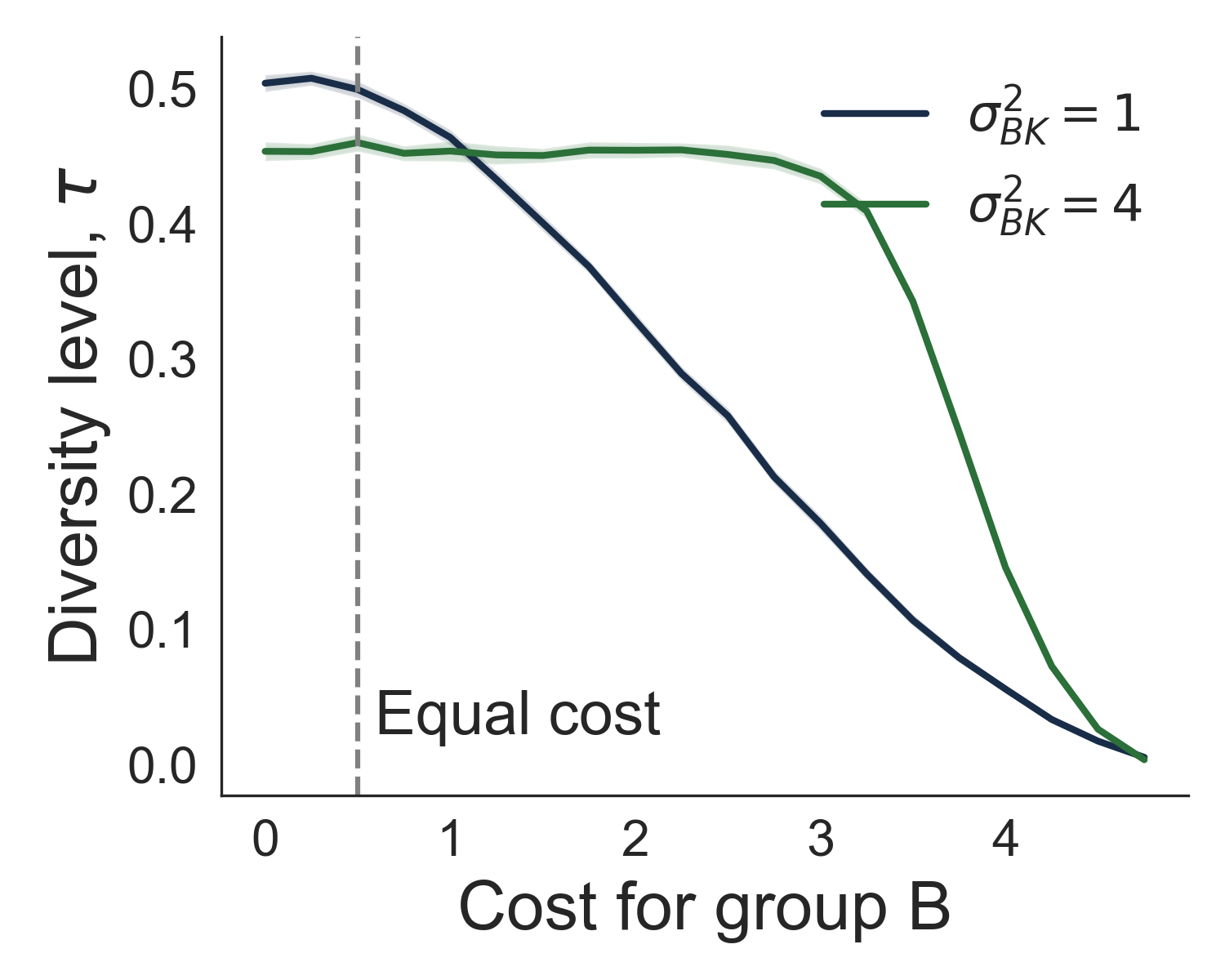}
	\caption{\small Diversity level \normalsize}
	\label{fig:diversity_level_by_cost_var}
\end{subfigure}
\hfill
\begin{subfigure}[b]{0.32\textwidth}
    \centering
    \includegraphics[width=\linewidth]{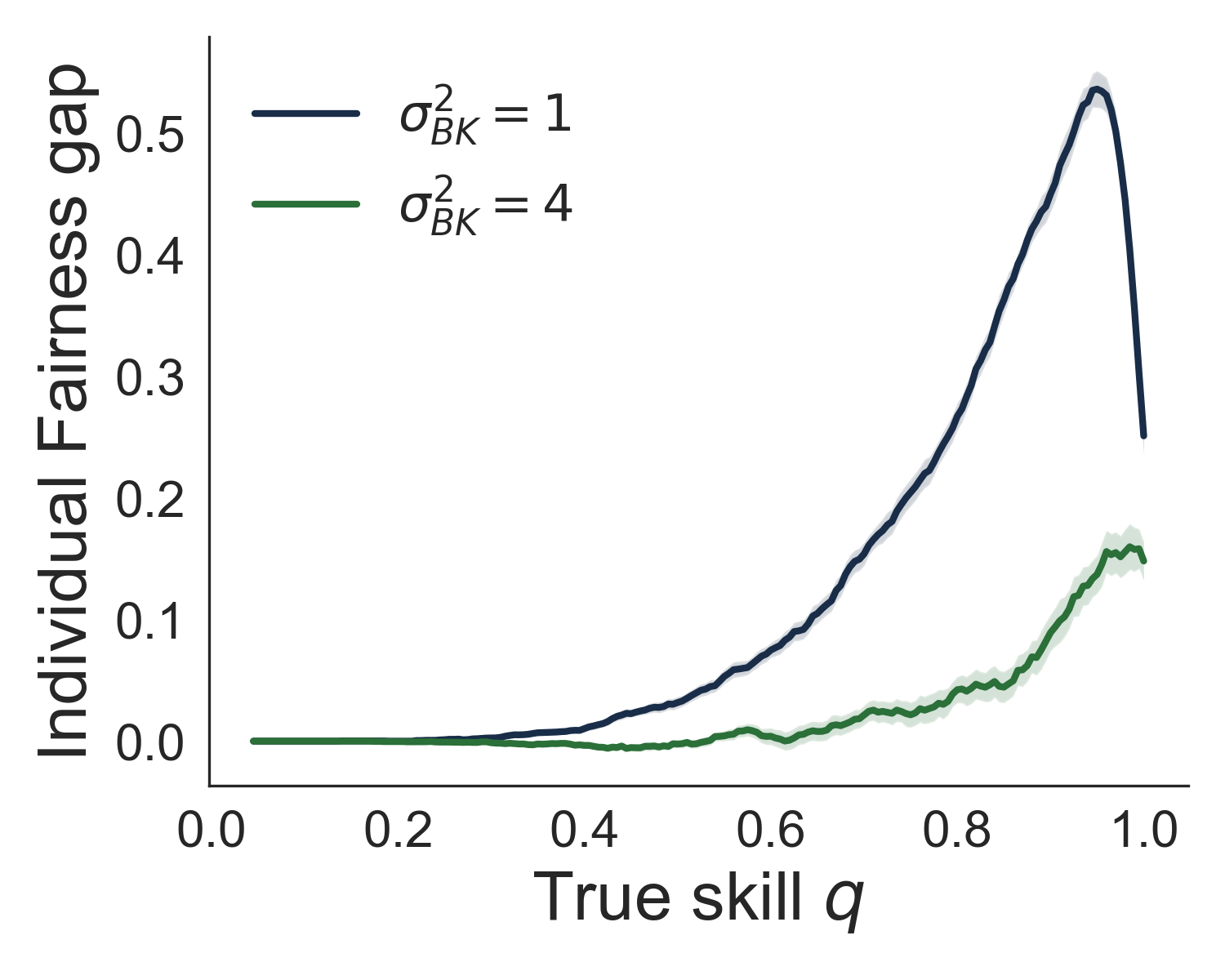}
    \caption{\small Individual fairness gap \normalsize}
    \label{fig:ifgap_by_skill}
\end{subfigure}
\hfill
\caption{Simulation results illustrating \Cref{prop.div_merit_strategic}, characterizing admission outcomes in a setting in which students make strategic decisions on whether to take the test. The figures show how the admitted students' (a) academic merit, (b) diversity level, and (c) individual fairness gap depend on test informativeness $\sigma^2_{BK}$ for Group $B$, as a function of either the test cost $c_B$ or the student skill level $q$. Figures (a) and (b)
fix cost $c_A=0.5$ and vary $c_B$. Academic merit and diversity are particularly harmed when the test is costly \textit{and} informative for Group $B$. Figure (c) 
considers a fixed cost $c_B=3$ and shows that individual fairness is worse when the test is more informative for Group $B$. The full parameter set can be
found in Appendix \ref{app:synthetic_simulations_strategic_single_school}. Overall, when the test is informative (low feature variance), higher costs for group $B$ correspond to worse outcomes across all metrics. }
\label{fig:strategic_varying_costs_and_feature_var}
\end{figure*}

\noindent\textbf{Parameters for Figure \ref{fig:strategic_varying_costs_and_feature_var}.}
Figure \ref{fig:strategic_varying_costs_and_feature_var} illustrates \Cref{prop.div_merit_strategic}, which characterizes admission outcomes in a setting in which students make strategic decisions on whether to take the test. The figures show how the admitted students' (a) academic merit, (b) diversity level, and (c) individual fairness gap depend on test informativeness $\sigma^2_K$ for Group $B$, as a function of either the test cost $c_B$ or the student skill level $q$.

There are two features, where the where the non-test feature is equally informative for both groups, and compares instance a) where the test score is equally informative for both groups ($\sigma^2_{AK}=\sigma^2_{BK}=1$) and b) where the test score is more informative for group $A$ than $B$ ($\sigma^2_{AK}1, \sigma^2_{BK}=4$). Figures \ref{fig:academic_merit_by_cost_var} and \ref{fig:diversity_level_by_cost_var} fix cost $c_A=0.5$ and vary $c_B \in [0,5)$. Figure \ref{fig:ifgap_by_skill} considers $c_A=0.5$ and a fixed cost $c_B=3$. 
The remainder of the parameters are the same as Figure \ref{fig:simulations_prob_apply_by_skill}. The true skill distribution for both groups is Normally distributed with mean $\mu=0$ and variance $\sigma^2=1$. The features for the two groups are Normally distributed with mean $\mu_{gk}=0$ for all $g,k$ and $\sigma^2_{A1}=\sigma^2_{A2}=\sigma^2_{B1}=1$. Students of both groups have valuation $v=5$ for the school. Test costs are $c_A=0.5$ and $c_B=3$. 
There are $N=10000$ students and the school has capacity $0.1$. To find the equilibrium $\tilde{q}^*_{\full}$, we search over a grid of 250 threshold values. 
The mean over 20 simulation runs is presented, along with 95\% confidence intervals.

\begin{figure*}[tbh]
    \begin{subfigure}[b]{0.31\textwidth}
		\centering
		\includegraphics[width=\linewidth]{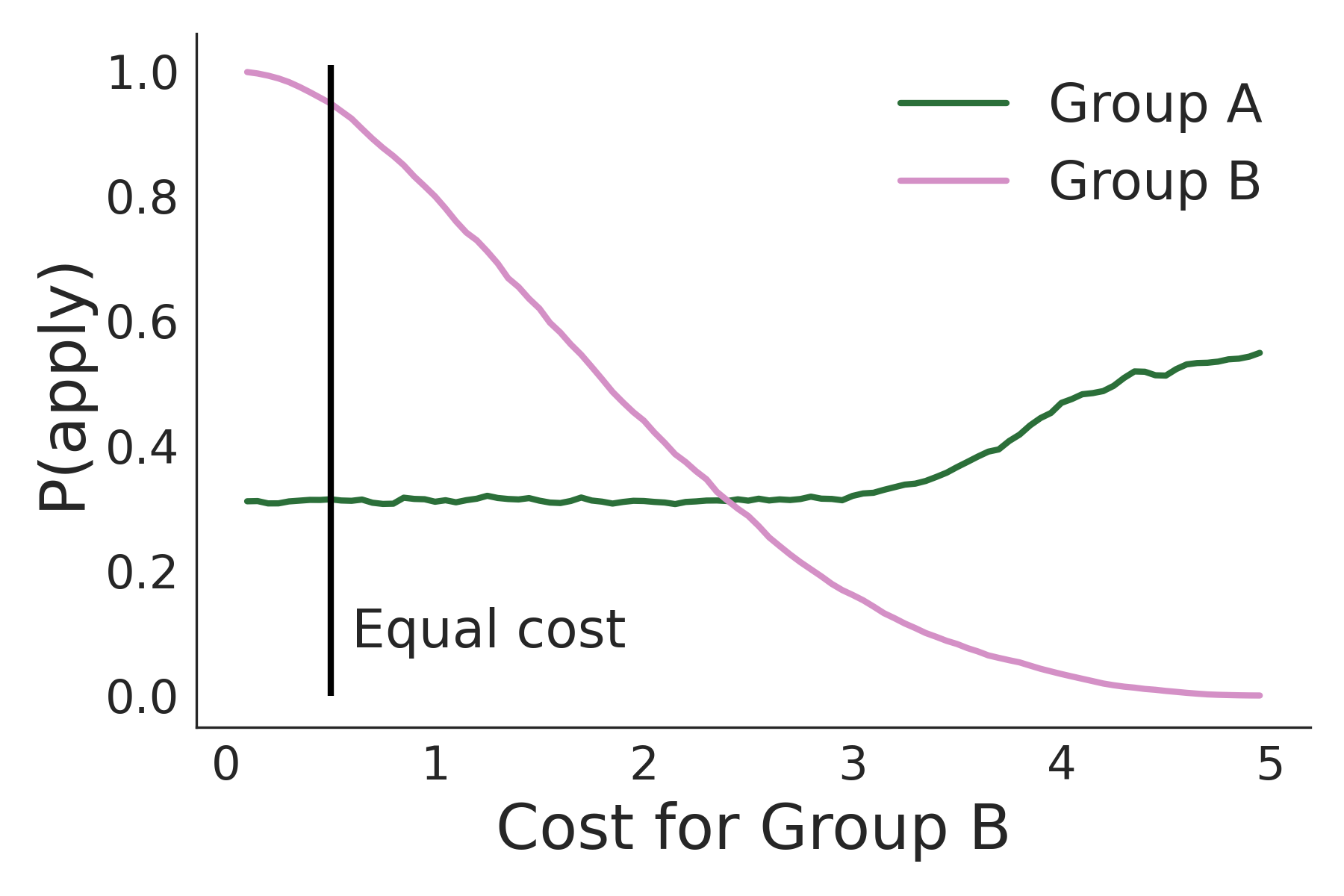}
    \caption{\small Probability of applying \normalsize}
    \label{fig:p_apply_by_cost}
    \end{subfigure}
    \hfill
    \begin{subfigure}[b]{0.31\textwidth}
    \centering
    \includegraphics[width=\linewidth]{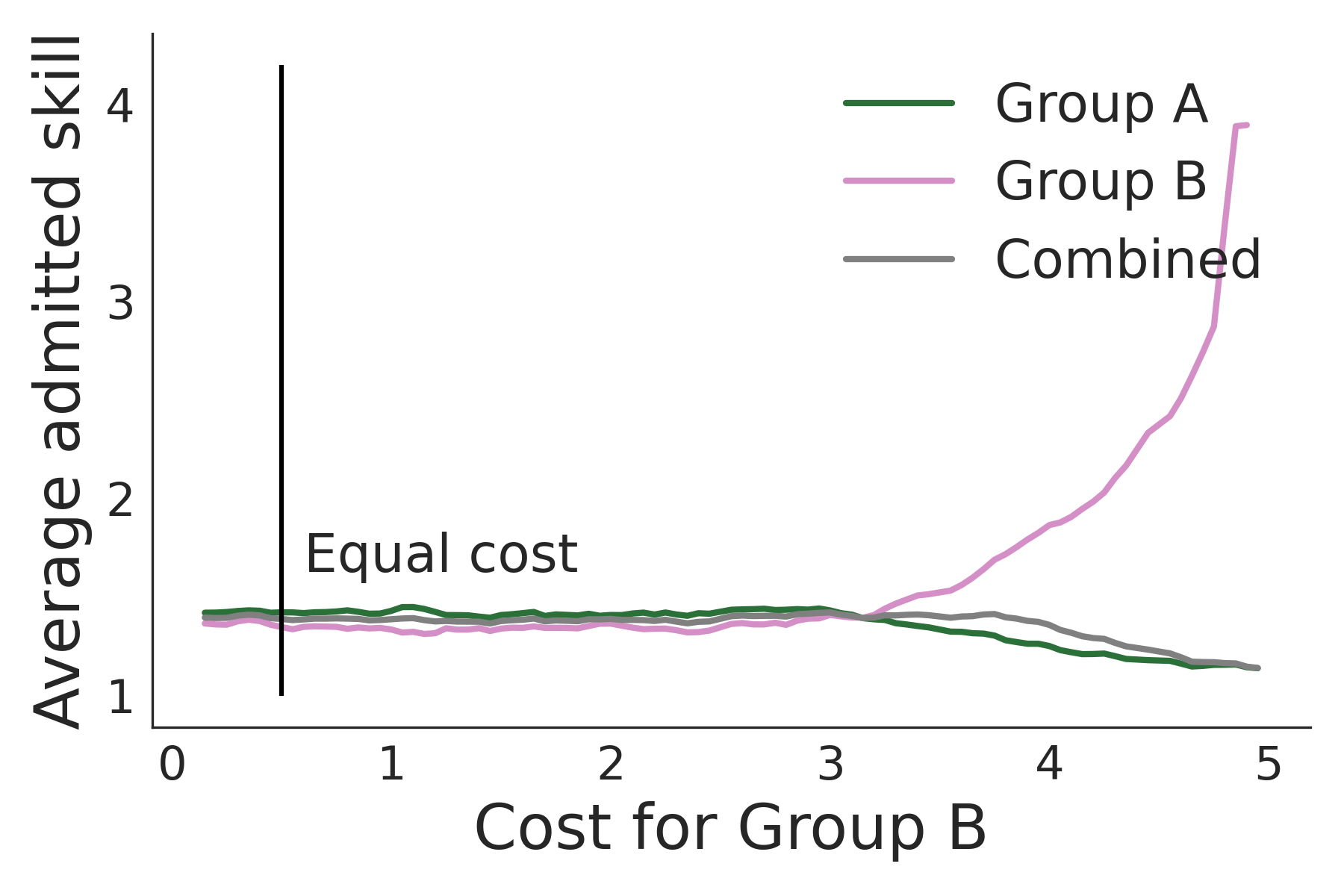}
    \caption{\small Academic merit \normalsize}
    \label{fig:academic_merit_by_cost}
    \end{subfigure}
    \hfill
    \begin{subfigure}[b]{0.31\textwidth}
    \centering
    \includegraphics[width=\linewidth]{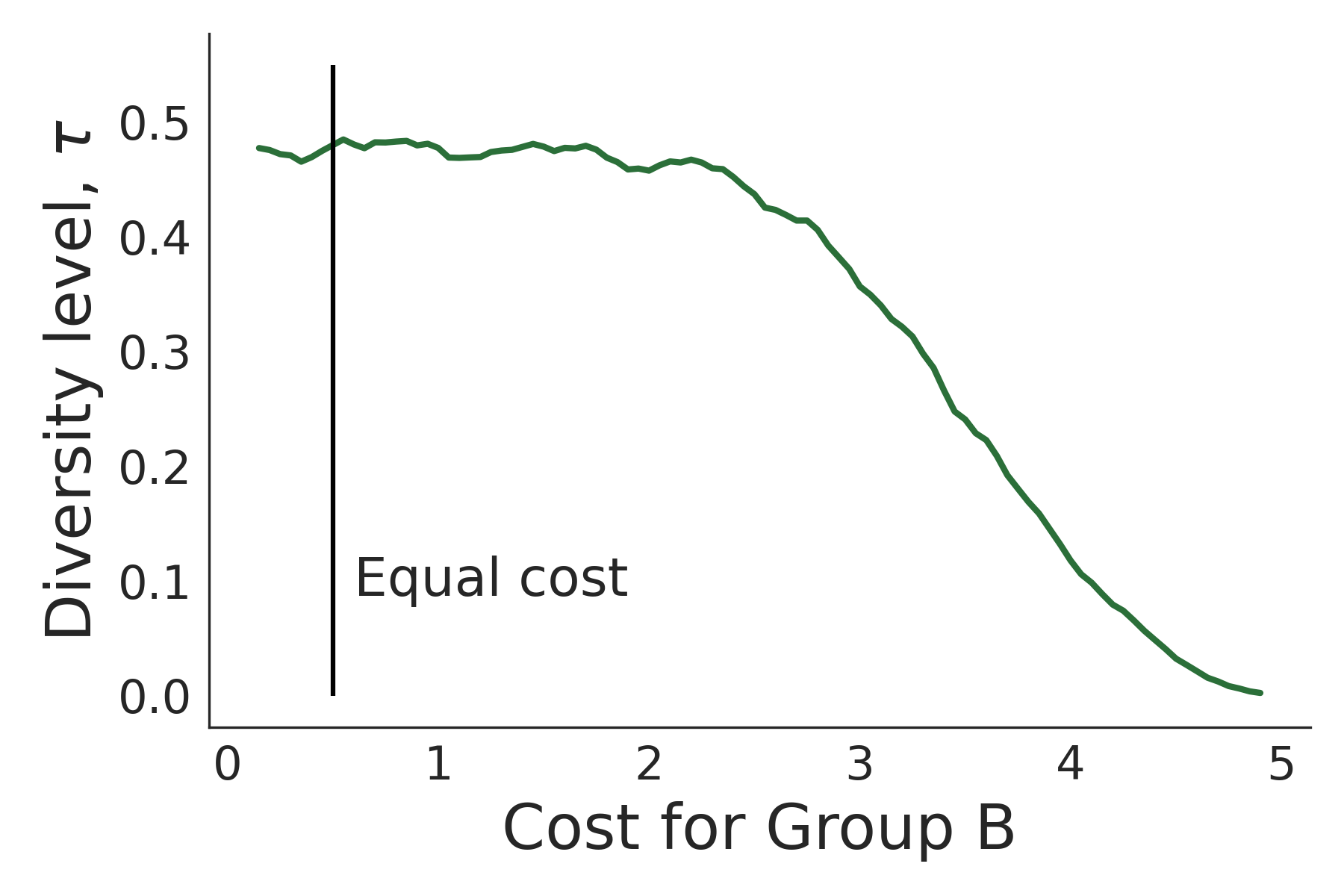}
    \caption{\small Diversity level \normalsize}
    \label{fig:diversity_by_cost}
    \end{subfigure}
    \hfill
\caption{Strategic students setting. How the probability of applying and the admitted students' academic merit change when Group $A$ has cost $c_A=0.5$ and the cost for Group $B$ varies. As the cost for Group $B$ increases, fewer Group $B$ students apply and more Group $A$ students apply (since the threshold decreases). Academic merit of admitted Group $B$ students increases while that of Group $A$ decreases.
We consider a setting where the variances of non-test features are equal for both groups, but Group $B$ has higher variance for test; the full parameter set can be
found in Electronic Companion \ref{app:synthetic_simulations_strategic_single_school}. 
}
\label{fig:strategic_papply_and_am_by_group}
\end{figure*}

\medskip

\noindent\textbf{Dropping the test score.} Figure \ref{fig:strategic_drop_test_diversity_and_academic_merit} shows the change in the diversity level and average skill level of the admitted students, after dropping the test. In this scenario, since the variance of the non-test feature $\sigma^2_{A0}=\sigma^2_{B0}=1$ are equal for both groups, a test-free policy will have a diversity level of $\tau=0.5$. 
\begin{figure}[tbh]
    \centering
    \includegraphics[scale=0.4]{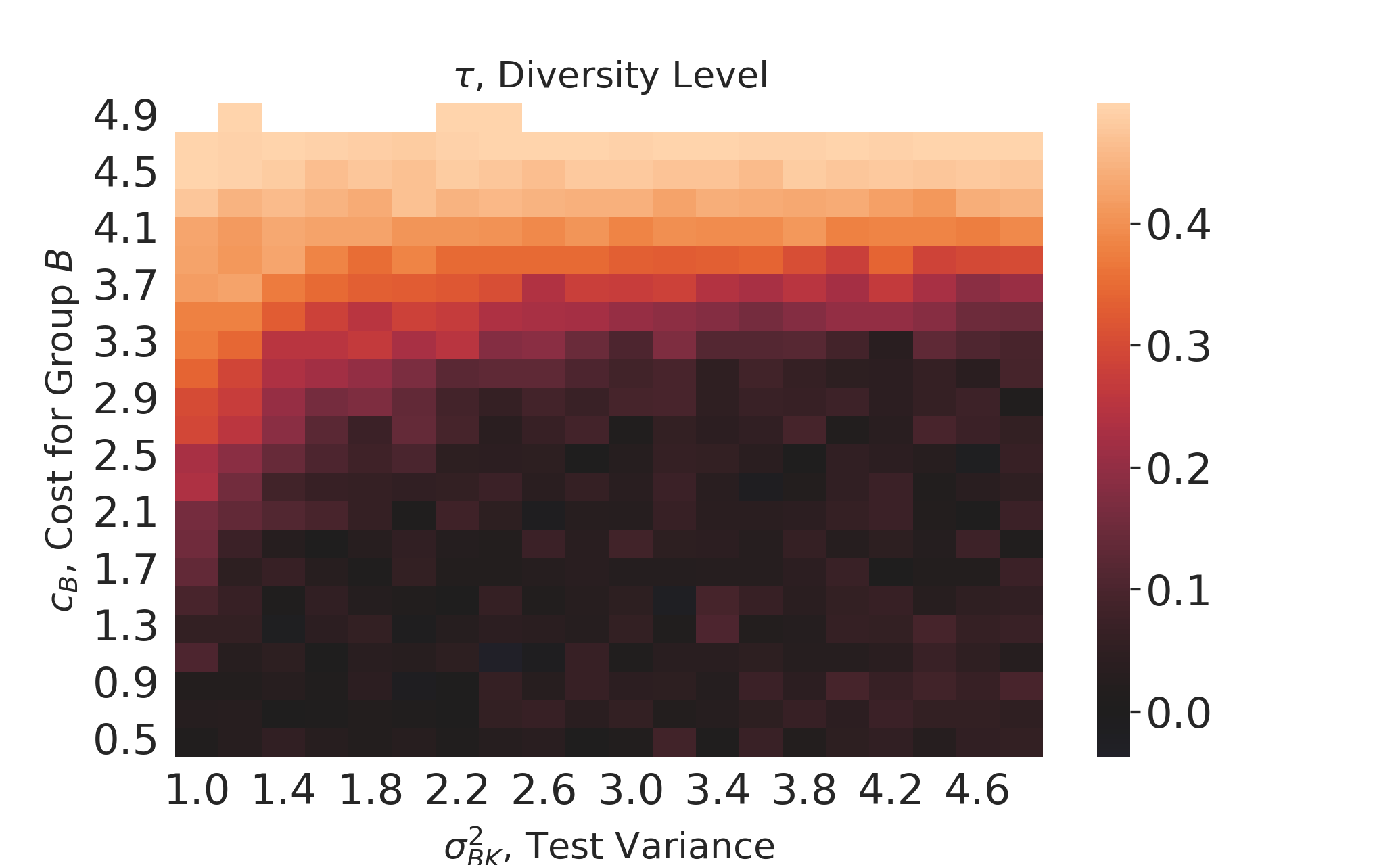}
    \includegraphics[scale=0.4]{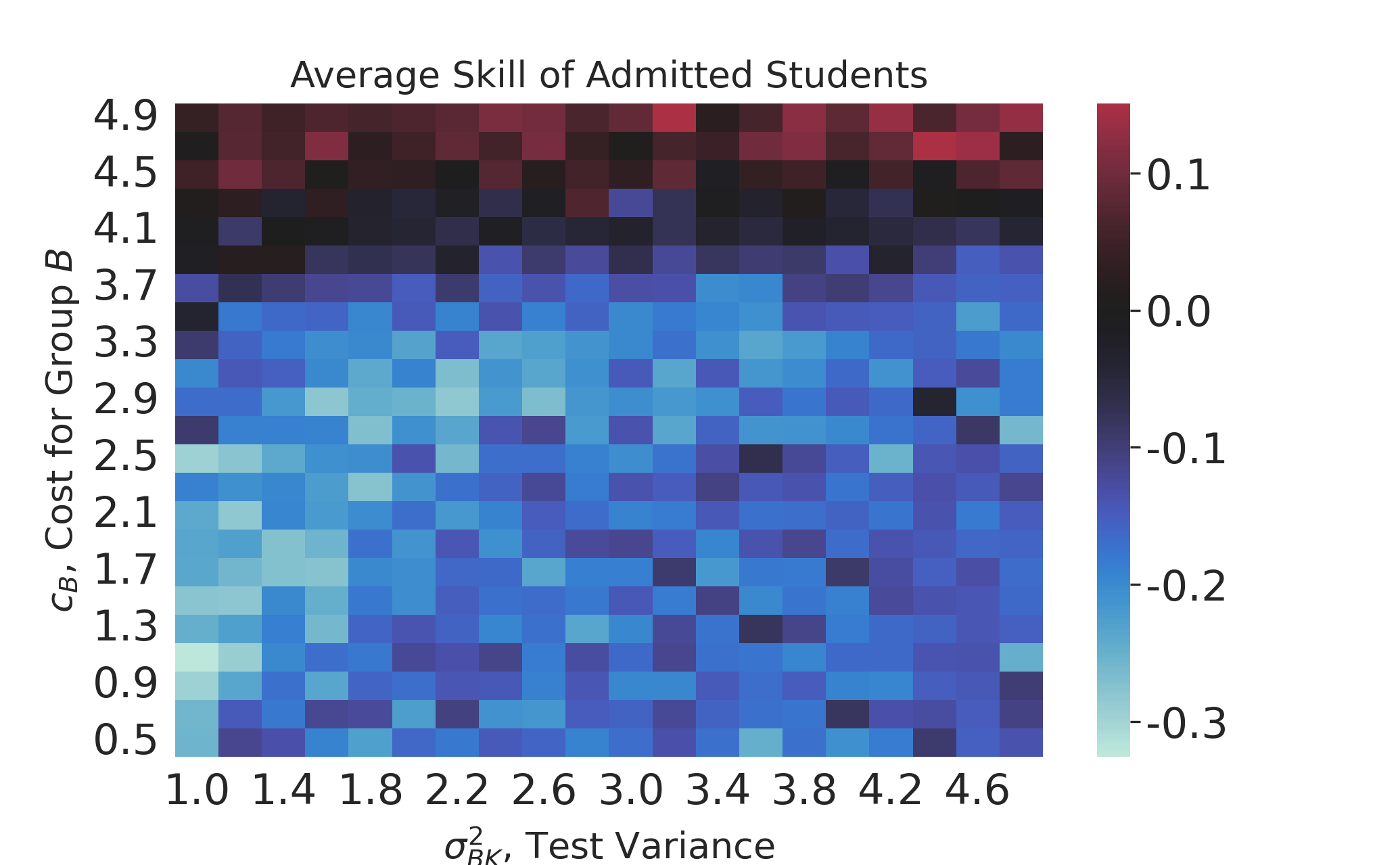}
    \caption{Change in diversity and average skill when the school drops the test requirement. When the test cost is large enough for Group $B$, dropping the test requirement increases the average academic merit and the diversity of the admitted student body. The full parameter set can be
found in Electronic Companion \ref{app:synthetic_simulations_strategic_single_school}.
    }
    \label{fig:strategic_drop_test_diversity_and_academic_merit}
\end{figure}

\subsection{Simulations with synthetic data and two schools}
\label{app:synthetic_simulations_strategic_two_schools}

The simulation setup closely resembles that of the single school, strategic student setting outlined in \ref{app:synthetic_simulations_strategic_single_school}.
In the same way as the single school, strategic student setting, each student is initialized with their true skill $q$ and non-test features $\thetaset_{\subb}$. The student observes $\thetaset_{\subb}$ and makes their test decision by calculating their expected reward for the situation in which they take the test and the situation in which they do not.

The simulation setup with two schools $J_1$ and $J_2$ differs in that we now have two testing policies ($P^1, P^2$) and two admission thresholds ($\tilde{q}^*_1, \tilde{q}^*_2$). For each policy pair --($P^1_{\full}, P^2_{\full}$), ($P^1_{\full}, P^2_{\subb}$), ($P^1_{\subb}, P^2_{\full}$), and ($P^1_{\subb}, P^2_{\subb}$)---we simulate the resulting admission outcomes in equilibrium.

Given a fixed policy pair ($P^1, P^2$) and admission threshold pair $(\tilde{q}^*_1, \tilde{q}^*_2)$, students observe their non-test features $\thetaset_{\subb}$ and calculates the distribution of their estimated skill if they were to take the test $\tilde{q}_{\full} \mid \tilde{q}_{\subb}, g, P_{\full}$, in the same way that they do in the single school, strategic setting. Then, the student solves for their optimal test decision. 
For example, if the policy pair is ($P^1_{\full}, P^2_{\subb})$, then the student solves the following optimization problem:
\begin{equation*}
    \alpha(\tilde q(\thetaset_{\subb}, g), g; \mathbf{P}) = \arg \max_{\alpha \in \{0,1\}} \alpha \left( v_1 \prob (Y_1 = 1 \mid \thetaset_{\subb}, g, P^1_{\full})   - c_g  \right) + v_2 \prob (Y_1 = 0 \cap Y_2 =1 \mid \thetaset_{\subb}, g, {P}^2_{\subb}).
\end{equation*}
The student then applies to all schools that do not require the test and applies to a test-required school if they choose $\alpha=1$. 

For each policy pair ($P^1, P^2$), finding an equilibrium amounts to finding an admission threshold pair $(\tilde{q}^*_1, \tilde{q}^*_2)$ such that each school admits the largest number of students while respecting their capacity constraints. We do a grid search to find the equilibrium admission threshold pair $(\tilde{q}^*_1, \tilde{q}^*_2)$.

\subsubsection{Parameters for the figures in the main text}
\label{app:ssec_two_school_strategic_parameter_setup}

Figure \ref{fig:two_school_strategic_heatmap} shows the average admitted skill (academic merit) of the resulting student body in $J_1$ (the preferred school) and $J_2$ (the less preferred school), as a function of $J_1$ and $J_2$ testing policies. 

In both Figure \ref{fig:twoschool_strat_cost0.5} and Figure \ref{fig:twoschool_strat_cost2.0}, there are $N=1000$ students, with half in group $A$ and half in group $B$. The true skill is Normally distributed with mean $\mu=0$ and variance $\sigma^2=1$. Students have two features. Group $A$ has feature distributions with mean $\mu_{A1}=\mu_{A2}=1$ and variance $\sigma^2_{A1}=\sigma^2_{A2}=1$. Group $B$ has feature distributions with mean $\mu_{B1}=\mu_{B2}=-1$ and variance $\sigma^2_{B1}=3, \sigma^2_{B2}=5$. 
Students have valuation $v_1=3$ for $J_1$ and $v_2=2$ for $J_2$. We simulate $N=1000$ students. $J_1$ and $J_2$ each have capacity $0.2$. In Figure \ref{fig:twoschool_strat_cost0.5}, the test costs are $c_A = c_B = 0.5$. 
In Figure \ref{fig:twoschool_strat_cost2.0}, the test costs are $c_A = c_B = 2.0$. To find the equilibrium threshold pair $(\tilde{q}^*_1, \tilde{q}^*_2)$ we do a grid search over 100 values of $\tilde{q}^*_1$ and 100 values of $\tilde{q}^*_2$.

\begin{figure}[tbh]
\centering
    \includegraphics[width=0.8\linewidth]{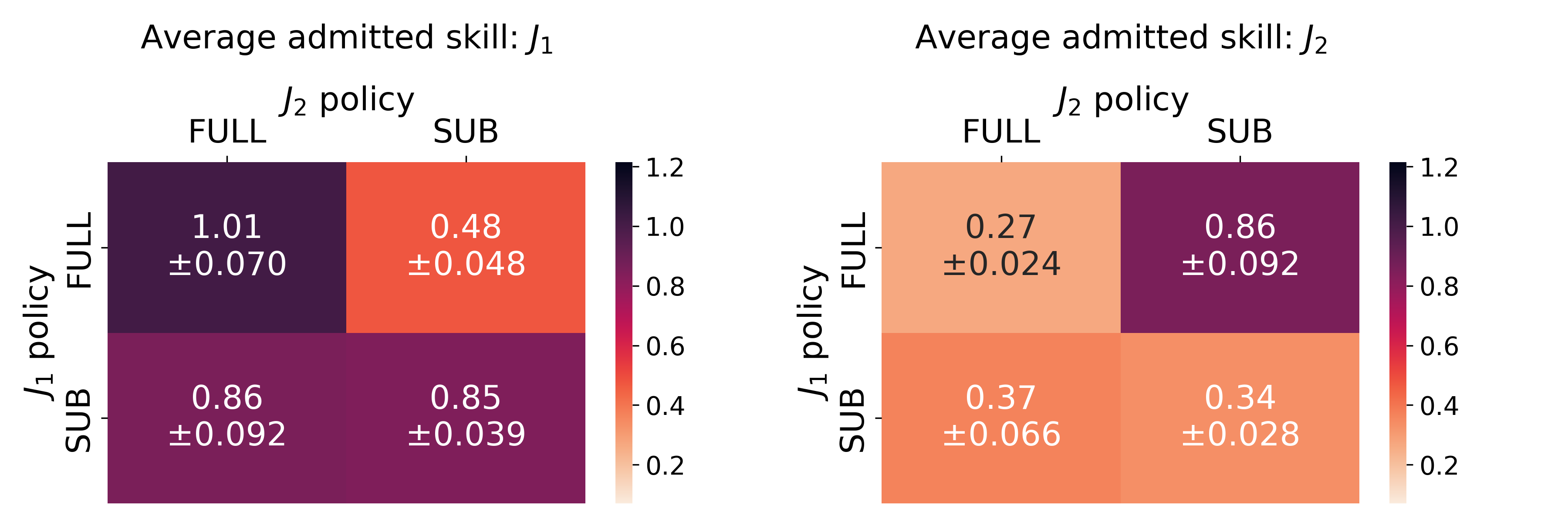}

 \caption{Average admitted skill of the resulting student body, in $J_1$ (preferred school) and $J_2$ (less preferred school), as a function of $J_1$ and $J_2$ testing policies. Students in both groups have valuations $v_1=3$ and $v_2=2$ and students have test cost $c_A=c_B=1.5$.  See \ref{app:ssec_two_school_strategic_parameter_setup} for full parameter details. 
}
\label{fig:two_school_strategic_heatmap_appendix}
\end{figure}

Figure \ref{fig:simulations_prob_apply_by_skill} shows simulation results illustrating the student equilibrium decisions
(characterized by \Cref{eq.student.decision.cost.theta}) on whether to take the test and apply to a school that requires the test, as a function of their true skill and group. 

There are two features, where the where the non-test feature is equally informative for both groups, but the test score is more informative for group $A$ than group $B$. The true skill distribution for both groups is Normally distributed with mean $\mu=0$ and variance $\sigma^2=1$. The features for the two groups are Normally distributed with mean $\mu_{gk}=0$ for all $g,k$ and $\sigma^2_{A1}=\sigma^2_{A2}=\sigma^2_{B1}=1$ and $\sigma^2_{B2}=2$, where $k=2$ denotes the test feature. Students of both groups have valuation $v=5$ for the school. Test costs are $c_A=0.5$ and $c_B=3$. 
There are $N=10000$ students and the school has capacity $0.1$. To find the equilibrium,  
we search over a grid of 250 threshold values.

\subsubsection{Additional simulations for synthetic data and two schools.}
\label{app:ssec_additional_two_school_strategic}

Figure \ref{fig:two_school_strategic_heatmap_appendix} shows additional simulations for the two school, strategic student setting, under different utility and test costs parameters.
In this setting, the optimal policies of both schools can depend on the policy of their competitor. $J_1$'s optimal policy is to require the test when $J_2$ requires the test, but drop the test when $J_2$ drops the test. $J_2$'s optimal policy is to drop the test when $J_1$ requires the test, but when $J_1$ drops the test, $J_2$ receives a quite similar average admitted skill when dropping or requiring the test.

\section{Supplementary information and discussion}
\FloatBarrier
Here, we provide additional information to support the main text analysis and writing. \Cref{tab:notation} provides a table of key notation. We also include further discussion on various modeling points.

{
\begingroup
\tiny
\begin{table}
\footnotesize
\centering
\caption{Key mathematical notation}
\begin{tabular}{l|p{8cm}|p{5cm}}
\hline
\textbf{Symbol} & \textbf{Meaning} & \textbf{Section} \\
\hline\hline
$q$ & Student's latent (unobserved) skill level & Section 2 (Base Model) \\
$\theta_k$ & Feature $k$ (e.g., test score, grades, etc.) & Section 2 (Base Model) \\
$\theta = (\theta_1,...,\theta_K)$ & Vector of all features & Section 2 (Base Model) \\
$\epsilon_k$ & Gaussian noise term for feature $k$ & Section 2 (Base Model) \\
$g \in \{A,B\}$ & Student group (A or B) & Section 2 (Base Model) \\
$\pi$ & Mass/proportion of students in group B & Section 2 (Base Model) \\
$\mu$ & Mean of skill distribution & Section 2 (Base Model) \\
$\sigma^2$ & Variance of skill distribution & Section 2 (Base Model) \\
$\mu_{gk}$ & Mean of noise distribution for feature $k$ and group $g$ & Section 2 (Base Model) \\
$\sigma^2_{gk}$ & Variance of noise distribution for feature $k$ and group $g$ & Section 2 (Base Model) \\
$\gamma_g$ & Fraction of group $g$ with access to full set of features & Section 2 (Base Model) \\
$\tilde{q}(\theta,g)$ & Perceived skill estimate given features $\theta$ and group $g$ & Section 3.1 (Bayesian Estimation) \\
$\tilde{q}^*_S$ & Admission threshold under policy $P_S$ & Section 3.1 (Bayesian Estimation) \\
$Y \in \{0,1\}$ & Admission decision (1 = admitted) & Section 2 (Base Model) \\
$\tau(P)$ & Diversity level (fraction of admitted students from group B) under policy $P$ & Section 4 (Analysis) \\
$I(q;P)$ & Individual fairness gap at skill level $q$ under policy $P$ & Section 4 (Analysis) \\
$c_g$ & Cost for group $g$ to take test (in strategic model) & Section 5 (Strategic Model) \\
$v$ & Value/utility of admission (in strategic model) & Section 5 (Strategic Model) \\
$\alpha \in \{0,1\}$ & Student's action (1 = apply/take test) & Section 5 (Strategic Model) \\
$P_{{\full}}$ & Policy requiring full set of features & Section 4.2 (Policy Analysis) \\
$P_{{\subb}}$ & Policy requiring only subset of features & Section 4.2 (Policy Analysis) \\
$\underline{q}^g$ & Threshold for taking test in strategic model & Section 5.1 (Single School) \\
\hline
\end{tabular}
\label{tab:notation}
\end{table}
\endgroup
}

\paragraph{Strategic student behavior in practice.} Our model of rational student behavior requires that students know school cutoffs in equilibrium. In practice, there is substantial uncertainty about school admission policies across application settings \citep{tomkins2023showing,idoux2023integrating,kapor_heterogeneous_2020,ajayi2020school}, and students may behave suboptimally given knowledge of historical college decisions \citep{tomkins2023showing}. It may be possible to incorporate such behavioral and informational effects into the model, using ideas from application search models under imperfect information \citep{ajayi2020school,calsamiglia2020structural,agarwal2018demand,agarwal2020revealed,idoux2023integrating}. Our results provide qualitative, directional insight regarding student behavioral effects. For example, we expect the results to continue to hold in settings where student beliefs regarding their admission chances is monotonic increasing in their test scores and their knowledge of their other features; however, showing such a result would require moving beyond our distributional assumptions and specifying a specific search model or belief structure for students.  

\paragraph{Model with general competition across many schools and arbitrary student preferences.} Our multi-school analysis is restricted to studying two schools, where all applicants prefer one program over another. Studying competition more generally would be of interest, such as when student preferences are heterogeneous (students differ in which schools they relatively prefer). We note that our results suggest that the homogeneous setting already induces competition: the best response policy of the more-preferred school can depend on the policy of the less-preferred school; intuitively, students may not choose to take the test if they can be admitted to the less-preferred school without the test. The policies of the schools jointly affect student strategic incentives, and in turn school optimal policies. We foresee that analyzing more general competition effects would use substantially different technical tools (and likely start from a different model than our base model).  

\FloatBarrier
\section{Proofs of statements}

In this appendix, 
we provide and prove the full statement of each result appearing in the main text.

\subsection{Auxiliary lemmas}
\label{app.A1}

Let $\Phi$ denote the CDF of $\mathcal{N}(0,1)$ and $\mathrm{HR}(x) = \frac{\phi(x)}{ 1-\Phi(x)}$ the \textit{Hazard Rate} 
of $X\sim\mathcal{N}(0,1)$.

\begin{lemma}
	\label{lemma:prior_marginal_normal}
	Let $X \mid M \sim \mathcal{N}(M, \sigma^2)$ and
	$M \sim \mathcal{N}(\mu_0, \sigma_0^2)$.
	Then, $X \sim \mathcal{N}(\mu_0, \sigma^2+ \sigma^2_0)$.
\end{lemma}

\begin{lemma}
	\label{lemma:prior_marginal_normal_2}
	Let $X \mid M \sim \mathcal{N}(M, \sigma^2)$ and
	$M \sim \mathcal{N}(\mu_0, \sigma_0^2)$. Then,
	$$ M \mid X \sim \mathcal{N}\left(\frac{\sigma_0^2}{\sigma^2+ \sigma_0^2} X + \frac{\sigma^2}{\sigma^2+ \sigma_0^2} \mu_0,
	\frac{1}{\sigma^{-2}+\sigma_0^{-2}}\right).$$
\end{lemma}

\begin{lemma}
\label{lemma:conditional_expectation_normal}
Let $X \sim \mathcal{N}(\mu, \sigma^2)$. Then, for any $a \in \mathbb{R}$,
$\expec[X \mid X>a] = \mu + \sigma \frac{\phi(t)}{1-\Phi(t)},$
where $t = \frac{a - \mu}{\sigma}$.
\end{lemma}

\begin{lemma}
\label{lemma:HR}
The hazard rate $\mathrm{HR}(x) = \frac{\phi(x)}{1-\Phi(x)}$, $x\in \mathbb{R}$ has the following properties:
\begin{itemize}
    \item[(i)] Its derivative equals $\frac{\diff \mathrm{HR}(x)}{ \diff x}  = \mathrm{HR}(x) (\mathrm{HR}(x)-x)$;
    \item[(ii)] It holds that $\mathrm{HR}(x) > x$ for all $x >0$;
\end{itemize}
\end{lemma}

\begin{lemma}
\label{lemma:HR_h}
Let $a>0$. The function $h(x) = \frac{x}{a}  \mathrm{HR}( \frac{a}{x})$ is increasing in $x>0$.
\end{lemma}
\begin{proof}{\textit{Proof.}}
Let $y = a/x$. We study the monotonicity of $\hat{h}(y) = \mathrm{HR}(y) / y $. The derivative of $\hat{h}(y)$ equals
$$\frac{\diff \hat{h}(y)}{ \diff y}  = \frac{\frac{\diff \mathrm{HR}(y)}{ \diff y} y - \mathrm{HR}(y)}{y^2}.$$
For any $y>0$, it holds that $\frac{\diff \hat{h}(y)}{ \diff y}<0$ if and only if
$\frac{\diff \mathrm{HR}(y)}{ \diff y} y - \mathrm{HR}(y) < 0.$
Using Part (i) in Lemma~\ref{lemma:HR}, we get that
$$\frac{\diff \mathrm{HR}(y)}{ \diff y} y - \mathrm{HR}(y) = \mathrm{HR}(y) \left (
        \mathrm{HR}(y) \, y -y^2 -1 \right),$$
which is negative for $y>0$ if and only if $\mathrm{HR}(y) \, y -y^2 -1 <0$ for all $y>0$.

By Theorem 2.3 in \cite{baricz2008mills}, we know that
$\mathrm{HR}(y) < \frac{y}{2} + \frac{\sqrt{y^2 +4}}{2}.$
Thus, using this inequality, we can bound the quantity $\mathrm{HR}(y) \, y -y^2 -1 $ as follows:
$$\mathrm{HR}(y) \, y -y^2 -1 < \frac{y^2}{2} + y\,\frac{\sqrt{y^2 +4}}{2} -y^2 -1 = \frac{y}{2} ( -y + \sqrt{y^2 +4}) - 1,$$
which  is negative for any $y \in \mathbb{R}$. Therefore,  $ \frac{\diff \hat{h}(y)}{ \diff y} <0$ for all $y >0$. 
Finally, since $\hat{h}(y)$ is decreasing in $y>0$ and $y= \frac{a}{x}$, $a >0$, is decreasing in $x>0$, it follows that $h(x) = \hat{h}\left( \frac{a}{x}\right)$ is increasing in $x>0$.
\end{proof}

\subsection{Group-aware estimation (Proofs from Section~\ref{sec:estimatesintuition})}
\label{app.A2}

\noindent\textbf{Gaussian social learning with feature set $S \subseteq \{1, \ldots, K\}$.} Given that $q \sim \N(\mu, \sigma^2)$, $\epsilon_{kg} \sim \N(\mu_{gk}, \sigma_{gk}^2)$ and the noise is drawn independently, each feature $k\in S$ is also Normally distributed conditional on $q$, i.e.,
$\theta_k \mid q, g \sim \N (q + \mu_{gk}, \sigma_{gk}^2).$
Then, we inductively find that
$q \mid \thetaset, g \sim \N \left(\tilde q(\thetaset, g), \tilde \sigma^2(\thetaset, g) \right),$
where
\begin{equation}
\label{eq.app.perceived_skill_cond_features}
\tilde q(\thetaset, g) = \frac{\mu \sigma^{-2} + \sum_{k\in S}(\theta_{k} - \mu_{gk})\sigma_{gk}^{-2}} {\sigma^{-2} +\sum_{k\in S} \sigma_{gk}^{-2}}, \, \, \, \, 
\tilde \sigma^2(\thetaset, g)= \frac{1} {\sigma^{-2} +\sum_{k\in S} \sigma_{gk}^{-2}}.
\end{equation}

\noindent\textbf{Perceived skill conditional on true skill.}  \Cref{eq.app.perceived_skill_cond_features} gives us the skill estimate $\tilde q$ of a student conditional on features $\thetaset$. Another useful distribution is  
$\tilde q \mid q, g, P_S$, which is also Gaussian. Indeed, observe that $\tilde q(\thetaset, g)$ in \Cref{eq.app.perceived_skill_cond_features} is a linear combination of independent (conditional on $q$) Gaussian variables $\theta_k= q + \epsilon_{kg}$, $k \in S$. Thus,
\begin{equation}
\label{eq:tilde_q_cond_q_distribution}
    \tilde q  \mid q, g, P_S \sim \N\left(
    \frac{\mu\sigma^{-2} + q\sum_{k\in S} \sigma_{gk}^{-2}}{\sigma^{-2} + \sum_{k\in S} \sigma^{-2}_{gk}},
    \frac{\sum_{k\in S} \sigma^{-2}_{gk}}{\left(\sigma^{-2}+\sum_{k\in S} \sigma^{-2}_{gk}\right)^2}
    \right).
\end{equation}

\begin{lemma}
\label{lemma:comp_tildeq_cond_q}
For group-aware estimation policies, the following properties hold:
\begin{itemize}
    \item[(i)] $\expec[\tilde q  \mid q, A, P_{S}] > \expec[\tilde q \mid q, B, P_{S}]$ if and only if $(q-\mu)\left(
    \sum_{k\in S}\sigma_{Ak}^{-2} -
    \sum_{k\in S}\sigma_{Bk}^{-2}
    \right)>0.$
    \item[(ii)] $\textrm{Var}[\tilde q \mid q, A, P_{S} ] >\textrm{Var}[\tilde q \mid q, B, P_{S} ]$ if and only if $${\left( \sigma^{-4} - \sum_{k\in S} \sigma_{Ak}^{-2} \sum_{k\in S}\sigma_{Bk}^{-2} \right)
    \left(\sum_{k\in S}\sigma_{Ak}^{-2} - \sum_{k\in S}\sigma_{Bk}^{-2} \right)>0.}$$
\end{itemize}
\end{lemma}
\begin{proof}{\textit{Proof.}} The proof follows immediately from simple algebra thus it is ommitted.
\end{proof}

\smallskip\noindent\textbf{Distribution of skill estimates per group.} 
We find the distribution  
$\tilde q \mid g, P_{S}$, that we denote by 
$F_{\tilde q \mid g, P_{S}}$.

\begin{lemma}[Lemma~\ref{lemma:perceived_skill_distribution}]
	Consider a school that uses feature set $S \subseteq \{1, \dots, K\}$ for each applicant. 
For $g \in \{A, B\}$, the skill level estimates for students in group $g$ are  Normally distributed:
\begin{equation*}
	{\tilde q \mid  g, P_{S}} \sim \mathcal{N}  \left( \mu, \sigma^{2}\left[\frac{\sum_{{k\in S}} \sigma_{gk}^{-2}  }{\sigma^{-2} +\sum_{{k\in S}} \sigma_{gk}^{-2}}\right]\right).
\end{equation*}
\end{lemma}

\begin{proof}{\textit{Proof.}} 
An application of Lemma \ref{lemma:prior_marginal_normal} for $X= \tilde q$ and $M = \frac{\mu \sigma^{-2} + q  \sum_{k\in \full} \sigma_{gk}^{-2}} {\sigma^{-2} +\sum_{k\in \full} \sigma_{gk}^{-2}}$ gives us the result. 
Analytically, the parameters of this distribution can be computed as follows:
\begin{equation*}
 \begin{split}
    \expec[ \tilde q \mid g, P_{S} ] =& \expec_q[\expec[ \tilde q \mid q, g, P_{S}]] =\frac{\mu \sigma^{-2} + \mu \sum_{k\in S} \sigma_{gk}^{-2}} {\sigma^{-2} +\sum_{k\in S} \sigma_{gk}^{-2}}
 =\mu,
    \\
  \textrm{Var}[ \tilde q \mid g, P_{S} ] =& \expec[\tilde q^2 \mid g,P_{S} ] - \mu^2
  = \expec_q [\expec[\tilde q^2\mid q, g, P_{S}] ] -\mu^2\\
 =& \expec_q \left[
  \textrm{Var} [\tilde q \mid q, g, P_{S}] + \left(\frac{\mu \sigma^{-2} + q \sum_{k\in S} \sigma_{gk}^{-2}} {\sigma^{-2} +\sum_{k\in S} \sigma_{gk}^{-2}}\right)^2\right] - \mu^2
  \\
  =&\expec_q \left[
  \textrm{Var} [\tilde q \mid q, g, P_{S}]\right]+  \textrm{Var}\left[\frac{\mu \sigma^{-2} + q \sum_{k\in S} \sigma_{gk}^{-2}} {\sigma^{-2} +\sum_{k\in S} \sigma_{gk}^{-2}}\right] \\
  =&  \frac{\sum_{k\in S} \sigma_{gk}^{-2}   }{(\sigma^{-2} +\sum_{k\in S} \sigma_{gk}^{-2})^2} + \sigma^2 \left(\frac{\sum_{k\in S} \sigma_{gk}^{-2}   }{\sigma^{-2} +\sum_{k\in S} \sigma_{gk}^{-2}} \right)^2\\
  = &\sigma^{2}\frac{\sum_{k\in S} \sigma_{gk}^{-2}    }{\sigma^{-2} +\sum_{k\in S} \sigma_{gk}^{-2}}.
  \end{split}
 \end{equation*}
\end{proof}

\begin{corollary}
\label{cor:var_tilde_q_comp}
$\textrm{Var}[\tilde q \mid A, P_{S}] > \textrm{Var}[\tilde q \mid B, P_{S}]$ if and only if $\sum_{k\in S} \sigma_{Ak}^{-2}>\sum_{k\in S} \sigma_{Bk}^{-2}$.
\end{corollary}

\begin{corollary}[Second-order stochastic dominance]
\label{cor:ssd}
    If $\sum_{k\in S} \sigma_{Ak}^{-2}>\sum_{k\in S} \sigma_{Bk}^{-2}$, then ${(\tilde q \mid B,  P_{S}) \succ_{\textrm{SSD}} (\tilde q \mid A,  P_{S})}$
 and $\tilde q \mid A,   P_{S}$ is a mean-preserving spread of $\tilde q \mid B,   P_{S}$.
\end{corollary}

\noindent\textbf{Distribution of true skill conditional on skill estimate.}
To answer questions about the academic merit of the admitted student body, we need to be able to compute the expected value of $q$ conditional on acceptance and the social group $g$ of a student, i.e., $\expec[q \mid  Y=1, g, P_S]$.
Thus, we first the conditional distribution $q\mid \tilde q, g, P_S$ in the following lemma.

\begin{lemma}
\label{lemma.true_skill_conditional_perceived}
Suppose that the school uses policy $P_{S}$. Then, the true skill level $q$ of students in group $g\in \{A,B\}$ conditional on the estimated skill level $\tilde q$ is Normally distributed as follows
\begin{equation}
\label{eq.true_skill_conditional_perceived}
    q \mid \tilde{q}, g, P_{S} \sim \N \left( \tilde q , \frac{1}{\sigma^{-2} + \sum_{k\in S} \sigma_{gk}^{-2}}\right).
  \end{equation}
\end{lemma}

  \begin{proof}{\textit{Proof.}}
  We apply Lemma \ref{lemma:prior_marginal_normal_2} by using the transformation $M = \frac{\mu\sigma^{-2} + q \sum_{k\in S} \sigma_{gk}^{-2}}{\sigma^{-2} + \sum_{k\in S} \sigma_k \sigma_{gk}^{-2}}$ and $X=\tilde q$.
More specifically, let
$$X \mid M \sim \N \left(M, \frac{\sum_{k\in S} \sigma_{gk}^{-2}}{(\sigma^{-2} + \sum_{k\in S} \sigma_{gk}^{-2})^2}\right), \,\,\, \,M \sim \mathcal{N}\left( \mu, \sigma^2 \left(  \frac{\sum_{k\in S} \sigma_{gk}^{-2}}{\sigma^{-2} + \sum_{k\in S} \sigma_{gk}^{-2}} \right)^2\right).$$
Then, by Lemma \ref{lemma:prior_marginal_normal_2}, we get that
\begin{equation*}
\begin{split}
    \expec[M \mid \tilde q, g, P_{S}] &=
    \frac{\sigma^2\frac{ (\sum_{k\in S} \sigma_{gk}^{-2})^2}{(\sigma^{-2} + \sum_{k\in S} \sigma_{gk}^{-2})^2} \tilde q + \mu \frac{\sum_{k\in S} \sigma_{gk}^{-2}}{(\sigma^{-2} + \sum_{k\in S} \sigma_{gk}^{-2})^2}}{\sigma^2 \frac{(\sum_{k\in S} \sigma_{gk}^{-2})^2}{(\sigma^{-2} + \sum_{k\in S} \sigma_{gk}^{-2})^2}  + \frac{\sum_{k\in S} \sigma_{gk}^{-2}}{(\sigma^{-2} + \sum_{k\in S} \sigma_{gk}^{-2})^2}}\\
   &= \frac{\sum_{k\in S} \sigma_{gk}^{-2} \tilde q + \mu \sigma^{-2} }{\sum_{k\in S} \sigma_{gk}^{-2}  +  \sigma^{-2}}
    \\
    \textrm{Var}[M \mid \tilde q,  g, P_{S}] & = \left ( \left(\frac{\sum_{k\in S} \sigma_{gk}^{-2}}{(\sigma^{-2} + \sum_{k\in S} \sigma_{gk}^{-2})^2} \right)^{-1} +  \sigma^{-2} \left(  \frac{\sum_{k\in S} \sigma_{gk}^{-2}}{\sigma^{-2} + \sum_{k\in S} \sigma_{gk}^{-2}} \right)^{-2} \right)^{-1} \\
       &= \frac{\left(\sum_{k\in S}\sigma_{gk}^{-2} \right)^2}{\left(\sigma^{-2} + \sum_{k\in S} \sigma_{gk}^{-2}\right)^3}.
\end{split}
\end{equation*}
Therefore,
   $M \mid  \tilde q, g, P_{S}  \sim \N \left(\frac{\sum_{k\in S} \sigma_{gk}^{-2} \tilde q + \mu \sigma^{-2} }{\sum_{k\in S} \sigma_{gk}^{-2}  +  \sigma^{-2}}, \frac{\left(\sum_{k\in S}\sigma_{gk}^{-2} \right)^2}{\left(\sigma^{-2} + \sum_{k\in S} \sigma_{gk}^{-2}\right)^3}\right).$
Finally, using the linear transformation $$q = \frac{M \left(\sigma^{-2} + \sum_{k\in S}  \sigma_{gk}^{-2}\right) - \mu \sigma^{-2}}{\sum_{k\in S} \sigma_{gk}^{-2}},$$ 
we get that
$q \mid \tilde q, g, P_{S} \sim \N \left(\tilde q , \frac{1}{\sigma^{-2} + \sum_{k\in S} \sigma_{gk}^{-2}}\right).$
\end{proof}

\subsection{Baseline policy in the absence of barriers (Proofs from Sections~\ref{sec:baseline_policy})}
\label{app.A3}

\noindent
Let $\tilde q^*_S$ denote the optimal decision threshold used by the school under policy $P_S$. Using the distribution $F_{\tilde{q} \mid g, P_{S}}$, it follows that threshold $\tilde q^*_S$ is the solution to the equation
\begin{equation}
\label{def:threshold_groupaware_noAA}
   (1-\pi) F_{\tilde{q} \mid A, P_{S}} (\tilde q^*_S)  + \pi F_{\tilde{q} \mid B, P_{S}} (\tilde q^*_S) = 1-C.
\end{equation}

By Lemma~\ref{lemma:perceived_skill_distribution}, the \textit{Gaussian mixture} of $F_{\tilde{q} \mid A, P_{S}}$, $F_{\tilde{q} \mid B, P_{S}}$ with weights $1-\pi$, $\pi$ has mean $\mu$ and variance
$$(1-\pi) \sigma^{2}\left[\frac{\sum_{k\in S} \sigma_{Ak}^{-2}  }{\sigma^{-2} +\sum_{k\in S} \sigma_{Ak}^{-2}}\right] +  \pi \sigma^{2}\left[\frac{\sum_{k\in S} \sigma_{Bk}^{-2}  }{\sigma^{-2} +\sum_{k\in S} \sigma_{Bk}^{-2}}\right].$$

Recall that for a Gaussian random variable $X\sim N(\mu_0, \sigma_0^2)$, it holds that $\frac{X-\mu_0}{\sigma_0} \sim N(0,1)$.
Thus, \Cref{def:threshold_groupaware_noAA} can be equivalently written as
\begin{equation}
\label{eq:threshold_groupaware_noAA_Phi}
    \Phi \left ( \left(\tilde q^*_S - \mu\right) \left( (1-\pi) \sigma^{2}\left[\frac{\sum_{k\in S} \sigma_{Ak}^{-2}  }{\sigma^{-2} +\sum_{k\in S} \sigma_{Ak}^{-2}}\right] +  \pi \sigma^{2}\left[\frac{\sum_{k\in S} \sigma_{Bk}^{-2}  }{\sigma^{-2} +\sum_{k\in S} \sigma_{Bk}^{-2}}\right]\right)^{-1/2} \right) = 1-C.
\end{equation}

We also introduce some additional definitions.
Given any fixed value of $\sum_{k\in S} \sigma_{Bk}^{-2}$, the \textit{informativeness gap} $\Delta$ is defined as $\Delta = \sum_{k\in S} \sigma_{Ak}^{-2} - \sum_{k\in S} \sigma_{Bk}^{-2}$.
Given all parameters, except $\sum_{k\in S} \sigma_{Ak}^{-2}$ fixed, let $F_{\tilde{q} \mid g, P_{S}} (q ; \Delta)$ denote the CDF $F_{\tilde{q} \mid g, P_{S}}$ parameterized by $\Delta \geq 0$ and $\tilde{q}^*_S (\Delta)$ and $\tau(P_{S};\Delta)$ denote the corresponding admission threshold and diversity level, respectively, for any $\Delta \geq 0$ under  baseline policy $P_{S}$.

We now provide the proof to Proposition~\ref{prop:group_aware_noAA}.
Note that the result below considers a general feature set $S$ where the assumption on unequal precisions holds.

\propgroupawarenoAA*

\begin{proof}{\textit{Proof of Part (i)}.} We break the proof into two steps.

\noindent \textit{Step 1: We show that group fairness fails except for equal precision. Given unequal precisions, we further show that $\tau(P_S) < \pi$.}
If $\sum_{k\in S} \sigma_{Ak}^{-2}=\sum_{k\in S} \sigma_{Bk}^{-2}$, then the two distributions $F_{\tilde q \mid A, P_{S}}$, $F_{\tilde q \mid B, P_{S}}$ are identical so it trivially holds that $F_{\tilde q \mid A, P_{S}}(\tilde q_S^*) = F_{\tilde q \mid B, P_{S}} (\tilde q_S^*)=1-C$. Consequently, group fairness is achieved.

Next, assume that $\sum_{k\in S} \sigma_{Ak}^{-2}>\sum_{k\in S} \sigma_{Bk}^{-2}$. Then, by Lemma \ref{lemma:perceived_skill_distribution} and Corollary \ref{cor:ssd}, ${(\tilde q \mid B, P_{S}) \succ_{SSD} (\tilde q \mid A, P_{S})}$ and $\tilde q \mid A, P_{S}$ is a mean-preserving spread of $\tilde q \mid B, P_{S}$.
Thus, the CDFs $F_{\tilde q \mid A, P_{S}}$ and  $F_{\tilde q \mid B, P_{S}}$ cross once at $\tilde q = \mu$. Furthermore, $F_{\tilde q \mid A, P_{S}} (\tilde q) <F_{\tilde q \mid B, P_{S}} (\tilde q)$, for $\tilde q > \mu$ and  $F_{\tilde q \mid A, P_{S}} (\tilde q) > F_{\tilde q \mid B, P_{S}} (\tilde q)$, for $\tilde q < \mu$.

Since $C<0.5 = F_{\tilde q \mid A, P_{S}} (\mu) = F_{\tilde q \mid B, P_{S}} (\mu)$, then $\tilde q_S^* > \mu$. Therefore, $F_{\tilde q \mid A, P_{S}}(\tilde q_S^*) < F_{\tilde q \mid B, P_{S}}(\tilde q_S^*)$, which due to \Cref{def:threshold_groupaware_noAA} implies that $1- F_{\tilde q \mid B, P_{S}}(\tilde q_S^*)<C$ thus
$$\tau(P_{S}) = \frac{\pi(1-F_{\tilde q \mid B, P_{S}}(\tilde q^*_S))}{C} < \pi.$$

\smallskip
\noindent\textit{Step 2: We show that the marginal effect of $\Delta$ on $\tau(P_{S})$ is negative.}
 Consider $0\leq \Delta < \Delta '$.
Since $F_{\tilde{q} \mid B, P_{S}} (q ; \Delta)$ 
depends only on $\sum_{k\in S} \sigma_{Bk}^{-2}$, it remains unchanged under both $\Delta, \Delta'$.

Recall that the admission threshold is the solution to  \Cref{eq:threshold_groupaware_noAA_Phi}.
Solving for $\tilde q^*_S (\Delta)$ gives us
\begin{equation}
\label{eq:threshold_tilde_q_Delta_fun}
    \tilde q^*_S (\Delta)  = \mu+\Phi^{-1}(1-C) \cdot \left( (1-\pi) \sigma^{2}\left[\frac{\sum_{k\in S} \sigma_{Bk}^{-2} +\Delta  }{\sigma^{-2} +\sum_{k\in S} \sigma_{Bk}^{-2}+\Delta}\right] +  \pi \sigma^{2}\left[\frac{\sum_{k\in S} \sigma_{Bk}^{-2}  }{\sigma^{-2} +\sum_{k\in S} \sigma_{Bk}^{-2}}\right]\right)^{1/2},
\end{equation}
which is an increasing function of $\Delta$.
Thus, $\tilde{q}^*_S(\Delta') > \tilde{q}^*_S(\Delta)$.

Therefore, given that the capacity remains constant at $C$, the diversity level decreases as $\Delta$ increases since
\small
\begin{equation*}
\begin{split}
    \tau(P_{S};\Delta') = \frac{\pi (1-F_{\tilde{q} \mid B, P_{S}} (\tilde{q}^*_{S}(\Delta') ; \Delta')) }{C} 
      =\frac{\pi (1-F_{\tilde{q} \mid B, P_{S}} (\tilde{q}^*_{S}(\Delta');\Delta))}{C} 
    < \frac{\pi (1-F_{\tilde{q} \mid B, P_{S}} (\tilde{q}^*_{S}(\Delta); \Delta))}{C} = \tau(P_{S};\Delta).
\end{split}
\end{equation*}
\normalsize
 \end{proof}

 \begin{proof}{\textit{Proof of Part (ii)}.} We prove each claim in different steps.
 
\noindent\textit{Step 1: We show that $I(q; P_S)>0$ if and only if}
 $$ q  > \tilde q^*_{S}+ \frac{\sigma^{-2}(\tilde q^*_{S} -\mu)}{\sqrt{\sum_{k\in S} \sigma_{Bk}^{-2}}  \sqrt{\sum_{k\in S} \sigma_{Ak}^{-2}}}.$$
 
 Recall that for a Gaussian variable $X\sim N(\mu_0, \sigma_0^2)$, it holds that $\frac{X-\mu_0}{\sigma_0} \sim N(0,1)$.
Thus, given policy $P_{S}$, the probability of admission for a student in group $g$ equals
\begin{equation}
\label{eq:admission_prob_cond_q}
    \prob[Y=1 \mid q, g, P_{S}] = 1- F_{\tilde q \mid q, g, P_{S}} (\tilde q^*_S) = 1- \Phi\left( \frac{\sigma^{-2}+\sum_{k\in S} \sigma^{-2}_{gk}}{\sqrt{\sum_{k\in S} \sigma^{-2}_{gk}}} \left(\tilde q^*_S  - \frac{\mu \sigma^{-2} + q\sum_{k\in S}\sigma_{gk}^{-2}} {\sigma^{-2} +\sum_{k\in S} \sigma_{gk}^{-2}}\right)\right),
\end{equation}
where 
$$\E[\tilde q \mid q, g, P_{S}] = \frac{\mu \sigma^{-2} + q\sum_{k\in S}\sigma_{gk}^{-2}} {\sigma^{-2} +\sum_{k\in S} \sigma_{gk}^{-2}}, \, \, \, \textrm{Var}[\tilde q \mid q, g, P_{S}]= \frac{\sum_{k\in S} \sigma^{-2}_{gk}}{\left(\sigma^{-2}+\sum_{k\in S} \sigma^{-2}_{gk}\right)^2}.$$

Consequently, due to the monotonicity of $\Phi$, it holds that $I(q ; P_{S}) >0$ if and only if
    \begin{equation}
        \label{eq:condition_IF_noAA}
        \begin{split}
        & \frac{\sigma^{-2}+\sum_{k\in S} \sigma^{-2}_{Ak}}{\sqrt{\sum_{k\in S} \sigma^{-2}_{Ak}}} \left(\tilde q^*_S  - \frac{\mu \sigma^{-2} + q\sum_{k\in S}\sigma_{Ak}^{-2}} {\sigma^{-2} +\sum_{k\in S} \sigma_{Ak}^{-2}}\right)  < \frac{\sigma^{-2}+\sum_{k\in S} \sigma^{-2}_{Bk}}{\sqrt{\sum_{k\in S} \sigma^{-2}_{Bk}}} \left(\tilde q^*_S  - \frac{\mu \sigma^{-2} + q\sum_{k\in S}\sigma_{Bk}^{-2}} {\sigma^{-2} +\sum_{k\in S} \sigma_{Bk}^{-2}}\right)\\
           & \Longleftrightarrow \frac{ \tilde q^*_S\sigma^{-2} + \tilde q^*_S\sum_{k\in S} \sigma_{Ak}^{-2}  - \mu \sigma^{-2} - q \sum_{k\in S} \sigma_{Ak}^{-2}}{ \sqrt{\sum_{k\in S} \sigma_{Ak}^{-2}}} < \frac{ \tilde q^*_S\sigma^{-2} +\tilde q^*_S \sum_{k\in S} \sigma_{Bk}^{-2}  - \mu \sigma^{-2} - q \sum_{k\in S} \sigma_{Bk}^{-2}}{ \sqrt{\sum_{k\in S} \sigma_{Bk}^{-2}}}\\
        & \Longleftrightarrow 
        \left (\sqrt{\sum_{k\in S} \sigma_{Bk}^{-2}} - \sqrt{\sum_{k\in S} \sigma_{Ak}^{-2}} \right) \left( \sigma^{-2}(\tilde q^*_S -\mu) + (q -\tilde q^*_S) \sqrt{\sum_{k\in S} \sigma_{Bk}^{-2}}  \sqrt{\sum_{k\in S} \sigma_{Ak}^{-2}}
        \right) < 0.
        \end{split}
    \end{equation}
    
Due to our assumption on unequal precisions, the last inequality further translates to
$$(\tilde q^*_S-q) \sqrt{\sum_{k\in S} \sigma_{Bk}^{-2}}  \sqrt{\sum_{k\in S} \sigma_{Ak}^{-2}} < \sigma^{-2}(\tilde q^*_S -\mu),$$
where the RHS is always positive due to  school selectivity which implies that $\tilde q^*_S >\mu$.
Thus, we conclude that $I(q; P_S) >0$ if and only if
$$ q  > \tilde q^*_{S}+ \frac{\sigma^{-2}(\tilde q^*_{S} -\mu)}{\sqrt{\sum_{k\in S} \sigma_{Bk}^{-2}}  \sqrt{\sum_{k\in S} \sigma_{Ak}^{-2}}}.$$

\smallskip
\noindent\textit{Step 2: We show that individual fairness fails except for equal precisions.}
As an immediate corollary of the previous analysis in Step 1, observe that individual fairness fails unless the LHS in \Cref{eq:condition_IF_noAA} equals 0 for all $q$; equivalently, individual fairness fails except for equal precision, i.e., $\sqrt{\sum_{k\in S} \sigma_{Bk}^{-2}} - \sqrt{\sum_{k\in S} \sigma_{Ak}^{-2}}=0$.

\medskip\noindent\textit{Step 3: Finally, we show that for $q>\mu + \sigma \Phi^{-1}(1-C)$, $I(q; P_{S})$ increases as the informativeness gap increases.}
We begin with group $B$.
By \Cref{eq:tilde_q_cond_q_distribution}, it follows that
\small
\begin{equation*}
\begin{split}
    \prob[Y=1 \mid q, B, P_{S}, \Delta] = 1- F_{\tilde{q} \mid B, P_{S}} ( \tilde q^*_{S}(\Delta) ; \Delta) = 1-\Phi\left( \frac{\sigma^{-2}+\sum_{k\in S} \sigma^{-2}_{Bk}}{\sqrt{\sum_{k\in S} \sigma^{-2}_{Bk}}}
    \left (\tilde q^*_S(\Delta) - \frac{\mu\sigma^{-2} + q\sum_{k\in S} \sigma_{Bk}^{-2}}{\sigma^{-2} + \sum_{k\in S} \sigma^{-2}_{Bk}}\right) \right).
    \end{split}
\end{equation*}
\normalsize
By   \Cref{eq:threshold_tilde_q_Delta_fun}, it further follows that $\tilde q^*_S (\Delta)$ is increasing in $\Delta$. Consequently, the above probability is  decreasing in $\Delta$ since $\Phi$ is an increasing function and all terms except for $\tilde q^*_S(\Delta)$ do not depend on $\Delta$. Therefore, we conclude that the admission probability of group $B$ students decreases for any $q$ as $\Delta$ increases.

Next, for group $A$, note that students with $q>\mu+\sigma \Phi^{-1}(1-C)$ are exactly those students in group $A$ who---given perfectly observable skills $q$---would be admitted to the class; 
due to imperfect information, a group $A$ student of true skill $q>\mu+\sigma \Phi^{-1}(1-C)$ has a non-zero probability to get rejected.
Next, observe that as $\Delta$ increases, the total precision $\sum_{k\in S}\sigma^{-2}_{Ak}$ of group $A$ must increase. Consequently, the variance $\textrm{Var}[\tilde q \mid q, A, P_S]$ decreases thus the estimates $\tilde q \mid q, A, P_S$ of all group $A$ students (including those with true skill $q>\mu+\sigma \Phi^{-1}(1-C)$) become more precise.
Combining this observation with the facts that the capacity $C$ remains constant and the admission probability of group $B$ students decreases, it follows that the probability that the top-skilled group $A$ students with $q>\mu+\sigma \Phi^{-1}(1-C)$ are rejected (either in favor of lower-skilled students in $A$ or students in $B$) decreases as $\Delta$ increases. 
Equivalently,  their admission probability  $\prob[Y=1 \mid q, A, P_{S}, \Delta]$ increases as $\Delta$ grows.

Putting everything together, we conclude that, given $q>\mu+\sigma \Phi^{-1}(1-C)$, the individual fairness gap $I(q; P_S)$ increases as the informativeness gap $\Delta$ increases.
 \end{proof}

\smallskip
\begin{proof}{\textit{Proof of Part (iii).}}
We break the proof into the following steps.

\noindent\textit{Step 1}: \textit{We compute the expected value $\expec[\tilde q \mid Y=1, g, P_{S}]$ and show that $\expec[\tilde q \mid Y=1, A, P_{S}] \geq \expec[\tilde q \mid Y=1, B, P_{S}]$.} Applying Lemma~\ref{lemma:conditional_expectation_normal}, we get that
\begin{equation}
\label{eq:exp_estimated_skill_given_admit}
\begin{split}
    \expec[\tilde q \mid Y=1, g, P_{S}]= \expec[\tilde q \mid \tilde q \geq \tilde q_S^*, g, P_{S}] &= \expec[\tilde q \mid g] + \sqrt{\textrm{Var}[\tilde{q} \mid g, P_{S}]} \frac{\phi(t_g)}{1-\Phi(t_g)} \\&= \mu +  \sigma \sqrt{\frac{\sum_{k\in S} \sigma_{gk}^{-2}  }{\sigma^{-2} +\sum_{k\in S} \sigma_{gk}^{-2}}} \cdot \frac{\phi(t_g)}{1-\Phi(t_g)},
\end{split}
\end{equation}
where $t_g = \frac{\tilde q_S^* - \expec[\tilde q \mid g, P_{S}]}{\sqrt{\textrm{Var}[\tilde{q} \mid g, P_{S}]}}$.
Due to school selectivity, we have $\tilde{q}_S^* >\mu$. 
By Lemma~\ref{lemma:HR_h}, 
the function
$$ h(x) = x \frac{\phi\left(\frac{\tilde q^*_S - \mu}{x}\right)}{1-\Phi\left(\frac{\tilde q^*_S - \mu}{x}\right)}=  x  \textrm{HR}\left(\frac{\tilde q^*_S - \mu}{x}\right)$$
is increasing in $x>0$ 
for $\tilde{q}_S^* >\mu$.
Thus, by  Corollary \ref{cor:var_tilde_q_comp}, we get that
that $$\expec[\tilde q \mid \tilde q \geq \tilde q_S^*, A, P_{S}] \geq \expec[\tilde q \mid \tilde q \geq \tilde q_S^*, B, P_{S}].$$

\medskip\noindent\textit{Step 2}:\textit{ We compute the expected value $\expec[ q \mid \tilde q \geq \tilde q_K^*, g, P_{S}]$.} Specifically,
\begin{equation}
\label{eq:expected_academic_merit}
\begin{split}
     \expec[ q \mid  Y=1, g, P_{S}]
     =\expec[ q \mid \tilde q \geq \tilde q_K^*, g, P_{S}]
     =\expec_{\tilde q}[ \expec_q[q \mid \tilde q, g, P_{S}] \mid \tilde q \geq \tilde q_K^*, g, P_{S}] =  \expec[\tilde q \mid \tilde q \geq \tilde q_K^*, g, P_{S}]  ,
\end{split}
\end{equation}
where the last equality follows from Lemma \ref{lemma.true_skill_conditional_perceived}.

\medskip\noindent\textit{Step 3:} \textit{We show that
$\expec[ q \mid Y=1, A, P_{S}]  > \expec[ q \mid Y=1, B, P_{S}]$.}
Given our assumptions on unequal precisions and school selectivity, 
the proof follows  from Steps 1 and 2.
I.e., if $\sum_{k\in S} \sigma_{Ak}^{-2}>\sum_{k\in S} \sigma_{Bk}^{-2}$ and $C<0.5$, then
$\expec[ q \mid Y=1, A, P_{S}]  > \expec[ q \mid Y=1, B, P_{S}].$
\end{proof}

\medskip
\noindent\textbf{Explaining why the individual fairness gap decreases for high-skilled students.}
Although the individual fairness gap is positive for sufficiently high-skilled students, the magnitude of this gap varies.
For students at the end of the right tail of the true skill distribution, the individual fairness gap starts to decrease.
This property can be graphically observed in Figure~\ref{fig:IFsimulations_barriers_withunaware}. 

\IFgapdecrease*

\begin{proof}{\textit{Proof.}}
By \Cref{eq:tilde_q_cond_q_distribution}, the individual fairness gap equals
\footnotesize
\begin{equation*}
    I(q; P_{S}) = \left(1- 
    \Phi\left( \frac{\sigma^{-2}+\sum_{k \in S}\sigma^{-2}_{Ak}}{\sqrt{\sum_{k \in S}\sigma^{-2}_{Ak}}}
    \left (\tilde q^*_S - \frac{\mu\sigma^{-2} + q\sum_{k \in S}\sigma_{Ak}^{-2}}{\sigma^{-2} + \sum_{k \in S}\sigma^{-2}_{Ak}}\right) \right)\right)
    - \left(1-\Phi\left( \frac{\sigma^{-2}+\sum_{k \in S}\sigma^{-2}_{Bk}}{\sqrt{\sum_{k \in S}\sigma^{-2}_{Bk}}}
    \left (\tilde q^*_S - \frac{\mu\sigma^{-2} + q\sum_{k \in S}\sigma_{Bk}^{-2}}{\sigma^{-2} + \sum_{k \in S}\sigma^{-2}_{Bk}}\right) \right) \right).
\end{equation*}
\normalsize
Taking the derivative of $I(q; P_S)$ with respect to $q$, we find that
\begin{equation*}
\begin{split}
    \frac{\diff I(q; P_{S})}{ \diff q} =&\phi\left( \frac{\sigma^{-2}+\sum_{k \in S}\sigma^{-2}_{Ak}}{\sqrt{\sum_{k \in S}\sigma^{-2}_{Ak}}}
    \left (\tilde q^*_S - \frac{\mu\sigma^{-2} + q\sum_{k \in S}\sigma_{Ak}^{-2}}{\sigma^{-2} + \sum_{k \in S}\sigma^{-2}_{Ak}}\right) \right) \sqrt{\sum_{k \in S}\sigma^{-2}_{Ak}}\\
    &-\phi
    \left( \frac{\sigma^{-2}+\sum_{k \in S}\sigma^{-2}_{Bk}}{\sqrt{\sum_{k \in S}\sigma^{-2}_{Bk}}}
    \left (\tilde q^*_S - \frac{\mu\sigma^{-2} + q\sum_{k \in S}\sigma_{Bk}^{-2}}{\sigma^{-2} + \sum_{k \in S}\sigma^{-2}_{Bk}}\right) \right) \sqrt{\sum_{k \in S}\sigma^{-2}_{Bk}}.
\end{split}
\end{equation*}

Thus, to prove that $ \frac{\diff I(q; P_{S})}{ \diff q} <0$, it suffices to show that 
\begin{equation*}
\begin{split}
    &\ln \left(\phi\left( \frac{\sigma^{-2}+\sum_{k \in S}\sigma^{-2}_{Ak}}{\sqrt{\sum_{k \in S}\sigma^{-2}_{Ak}}} 
    \left (\tilde q^*_S - \frac{\mu\sigma^{-2} + q\sum_{k \in S}\sigma_{Ak}^{-2}}{\sigma^{-2} + \sum_{k \in S}\sigma^{-2}_{Ak}}\right) \right) \sqrt{\sum_{k \in S}\sigma^{-2}_{Ak}}\right)\\
    <& \ln \left(\phi
    \left( \frac{\sigma^{-2}+\sum_{k \in S}\sigma^{-2}_{Bk}}{\sqrt{\sum_{k \in S}\sigma^{-2}_{Bk}}}
    \left (\tilde q^*_S - \frac{\mu\sigma^{-2} + q\sum_{k \in S}\sigma_{Bk}^{-2}}{\sigma^{-2} + \sum_{k \in S}\sigma^{-2}_{Bk}}\right) \right) \sqrt{\sum_{k \in S}\sigma^{-2}_{Bk}}\right).
\end{split}
\end{equation*}
The above condition is equivalent to 
\small
\begin{equation*}
\begin{split}
 & - \frac{ \left((\tilde q^*_S - \mu)\sigma^{-2} + (\tilde q^*_S - q) \sum_{k \in S}\sigma^{-2}_{Ak} \right)^2}{\sum_{k \in S}\sigma^{-2}_{Ak}}  +\ln \left ( \sum_{k \in S}\sigma^{-2}_{Ak}\right) < 
  -\frac{ \left((\tilde q^*_S - \mu)\sigma^{-2} + (\tilde q^*_S - q) \sum_{k \in S}\sigma^{-2}_{Bk} \right)^2}{\sum_{k \in S}\sigma^{-2}_{Bk}}  +\ln \left ( \sum_{k \in S}\sigma^{-2}_{Bk}\right)  \\
 \Longleftrightarrow &    \left({\sum_{k \in S}\sigma^{-2}_{Ak}} - {\sum_{k \in S}\sigma^{-2}_{Bk}} \right) \left(  \sigma^{-4}(\mu - \tilde q^*_S)^2 - {\sum_{k \in S}\sigma^{-2}_{Ak}} {\sum_{k \in S}\sigma^{-2}_{Bk}} (q - \tilde q^*_S)^2\right) 
    + {\sum_{k \in S}\sigma^{-2}_{Ak}} {\sum_{k \in S}\sigma^{-2}_{Bk}} \ln\left( \frac{{\sum_{k \in S}\sigma^{-2}_{Ak}}}{{\sum_{k \in S}\sigma^{-2}_{Bk}}}\right)<0.
\end{split}
\end{equation*}
\normalsize
Given our assumption on unequal precision, i.e.,
${\sum_{k \in S}\sigma^{-2}_{Bk}} < {\sum_{k \in S}\sigma^{-2}_{Ak}},$
we further get that this condition is satisfied for
$$q> q_e \triangleq \tilde q^*_S + \sqrt{\frac{ \sigma^{-4}(\mu - \tilde q^*_S)^2 }{{\sum_{k \in S}\sigma^{-2}_{Ak}}  {\sum_{k \in S}\sigma^{-2}_{Bk}}  } +  \frac{\ln\left({{\sum_{k \in S}\sigma^{-2}_{Ak}}}\right) - \ln\left({{ \sum_{k \in S}\sigma^{-2}_{Bk}}}\right)}{{\sum_{k \in S}\sigma^{-2}_{Ak}} - {\sum_{k \in S}\sigma^{-2}_{Bk}}}}.$$
Therefore, the individual fairness gap $I(q; P_{S})$ is decreasing in $q$ for $q>q_e$ as desired.

Furthermore, by the definition of $I(q; P_{S})$ and the fact that $\lim_{q' \rightarrow \infty} \Phi(q') =1$, we immediately  get that
$\lim_{q \rightarrow \infty} I(q; P_{S}) = 0.$
\end{proof}

\subsection{Dropping a feature with and without barriers (Proofs from Section~\ref{sec:dropping_tests})}
\label{app.A4}

\noindent \textbf{Dropping a feature  in the absence of barriers.}
We are interested in comparing the group-aware policies $P_{\full}$ and $P_{\subb}$.
By our previous result in Lemma \ref{lemma:perceived_skill_distribution}, we get that
\begin{equation*}
    \tilde q \mid g, P_{\subb}  \sim  \N \left( \mu, \sigma^{2}\frac{\sum_{k\in \subb} \sigma_{gk}^{-2}  }{\sigma^{-2} +\sum_{k\in \subb} \sigma_{gk}^{-2}}\right),\,\,\,\,
    \tilde q \mid g, P_{\full}  \sim  \N \left( \mu, \sigma^{2}\frac{\sum_{k\in \full} \sigma_{gk}^{-2}  }{\sigma^{-2} +\sum_{k\in \full} \sigma_{gk}^{-2}}\right).
\end{equation*}

\begin{lemma}
\label{lemma:mean_var_dropping_tests}
 The variance of $\tilde q \mid g, P_{\subb}$ is lower than that of $\tilde q \mid g, P_{\full}$ but their means are both equal to $\mu$.
\end{lemma}
\begin{proof}{\textit{Proof.}}
The proof follows trivially from the fact that the function $h(x)=\frac{x}{\sigma^{-2} +x }$ is increasing in $x>0$ and $$\sum_{k\in \full} \sigma_{gk}^{-2} = \sum_{k=1}^{K} \sigma_{gk}^{-2}> \sum_{k=1}^{K-1} \sigma_{gk}^{-2} = \sum_{k\in \subb} \sigma_{gk}^{-2}$$
for any $g$.
\end{proof}

Let $\tilde q^*_{ \subb}$ be the decision threshold of a school considering only features $k=1$ to $K-1$. 
By \Cref{def:threshold_groupaware_noAA}, $\tilde q^*_{\subb}$ is the solution to the following equation
$$(1-\pi) F_{\tilde q \mid A, P_{\subb}}(\tilde{q}^*_{\subb}) + \pi F_{\tilde q \mid A, P_{\subb}}(\tilde{q}^*_{\subb}) = 1-C,$$
whereas $\tilde q^*_{\full}$ is the solution to 
$$(1-\pi) F_{\tilde q \mid A, P_{\full}}(\tilde{q}^*_{\full}) + \pi F_{\tilde q \mid A, P_{\full}}(\tilde{q}^*_{\full}) = 1-C.$$

\begin{lemma}
\label{cor.threshold_comp_group_aware}
The admission threshold decreases after dropping feature $k=K$, i.e., $\tilde{q}^*_{\subb} < \tilde{q}^*_{\full}$.
\end{lemma}
\begin{proof}{\textit{Proof.}}
The proof follows  from  the definitions of  $\tilde{q}^*_{\subb}, \tilde{q}^*_{\full}$, and Lemma \ref{lemma:mean_var_dropping_tests}.
\end{proof}

\proptestgroupawarenoAA*

\begin{proof}{\textit{Proof of Part (i).}}
Diversity improves if and only if 
$$\tau(P_{\subb}) = 1- F_{\tilde q \mid  A, P_{\subb}}(\tilde q^*_{\subb}) > 1- F_{\tilde q \mid A, P_{\full}}(\tilde q_{a}) = \tau(P_{\full}).$$
By the definition of diversity level and Lemma~\ref{lemma:perceived_skill_distribution}, this is equivalent to the following condition:
\begin{equation*}
    1 - \Phi \left( \frac{\tilde q^*_{\subb} -\mu}{ \sigma\sqrt{\frac{\sum_{k\in \subb} \sigma^{-2}_{Bk} }{\sum_{k\in \subb} \sigma^{-2}_{Bk}+\sigma^{-2}}}}\right) >
     1 - \Phi \left( \frac{\tilde q^*_{\full} -\mu}{ \sigma\sqrt{\frac{\sum_{k\in \full} \sigma^{-2}_{Bk} }{\sum_{k\in \full} \sigma^{-2}_{Bk}+\sigma^{-2}}}}\right).
\end{equation*}

Replacing $\tilde q^*_{\full}, \tilde q^*_{\subb}$ with their definitions as in \Cref{eq:threshold_groupaware_noAA_Phi}, the above inequality becomes
\begin{equation*}
    \Phi \left( \Phi^{-1}(1-C) \sqrt{(1-\pi) \frac{\frac{\sum_{k \in \subb} \sigma_{Ak}^{-2}}{\sigma^{-2} + \sum_{k\in \subb} \sigma_{Ak}^{-2}}}{\frac{\sum_{k\in \subb} \sigma_{Bk}^{-2}}{\sigma^{-2} + \sum_{k\in \subb} \sigma_{Bk}^{-2}}} + \pi } \right) < 
    \Phi \left( \Phi^{-1}(1-C) \sqrt{(1-\pi) \frac{\frac{\sum_{k \in \full} \sigma_{Ak}^{-2}}{\sigma^{-2} + \sum_{k\in \full} \sigma_{Ak}^{-2}}}{\frac{\sum_{k\in \full} \sigma_{Bk}^{-2}}{\sigma^{-2} + \sum_{k\in \full} \sigma_{Bk}^{-2}}} + \pi } \right),
\end{equation*}
which---due to the monotonicity of $\Phi$---holds if and only if
\begin{equation*}
    \frac{\frac{\sum_{k \in \subb} \sigma_{Ak}^{-2}}{\sigma^{-2} + \sum_{k\in \subb} \sigma_{Ak}^{-2}}}{\frac{\sum_{k\in \subb} \sigma_{Bk}^{-2}}{\sigma^{-2} + \sum_{k\in \subb} \sigma_{Bk}^{-2}}} < \frac{\frac{\sum_{k \in \full} \sigma_{Ak}^{-2}}{\sigma^{-2} + \sum_{k\in \full} \sigma_{Ak}^{-2}}}{\frac{\sum_{k\in \full} \sigma_{Bk}^{-2}}{\sigma^{-2} + \sum_{k\in \full} \sigma_{Bk}^{-2}}}.
\end{equation*}
Using the substitution $ \sum_{k \in \full} \sigma_{gk}^{-2} = \sum_{k \in \subb} \sigma_{gk}^{-2} + \sigma_{gK}$,  the last relation equivalently simplifies to \Cref{eq:condition_diversity_dropping_test}.
\end{proof}
\begin{proof}{\textit{Proof of Part (ii)}.} We prove each claim at a separate step.

\noindent\textit{Step 1: We show that, for group $B$, $\prob(Y=1 \mid q, B, P_{\full}) < \prob(Y=1 \mid q, B, P_{\subb})$ if and only if
\begin{equation*}
\begin{split}
    q< q_{B} \triangleq \mu+ & \frac{ \sigma \Phi^{-1}(1-C) }{\sqrt{\sum_{k\in \full} \sigma_{Bk}^{-2}} - \sqrt{\sum_{k\in \subb} \sigma_{Bk}^{-2}} }
    \left ( \sqrt{\sigma^{-2} + \sum_{k\in \full} \sigma_{Bk}^{-2}} \sqrt{(1-\pi) \frac{ \frac{\sum_{k\in \full} \sigma_{Ak}^{-2}}{\sigma^{-2} + \sum_{k\in \full} \sigma_{Ak}^{-2}} }{\frac{\sum_{k\in \full} \sigma_{Bk}^{-2}}{\sigma^{-2} + \sum_{k\in \full} \sigma_{Bk}^{-2}} } + \pi} \right.\\
    &- \left. \sqrt{\sigma^{-2} + \sum_{k\in \subb} \sigma_{Bk}^{-2}} \sqrt{(1-\pi) \frac{ \frac{\sum_{k\in \subb} \sigma_{Ak}^{-2}}{\sigma^{-2} + \sum_{k\in \subb} \sigma_{Ak}^{-2}} }{\frac{\sum_{k\in \subb} \sigma_{Bk}^{-2}}{\sigma^{-2} + \sum_{k\in \subb} \sigma_{Bk}^{-2}} } + \pi} \right),
\end{split}
\end{equation*}
Similarly, for group $A$, it holds that $\prob(Y=1 \mid q,A, P_{\full}) > \prob(Y=1 \mid q, A, P_{\subb})$ if and only if}
\begin{equation*}
    \begin{split}
q<  q_{A} \triangleq \mu+ & \frac{ \sigma \Phi^{-1}(1-C) }{\sqrt{\sum_{k\in \full} \sigma_{Ak}^{-2}} - \sqrt{\sum_{k\in \subb} \sigma_{Ak}^{-2}} }
    \left (\sqrt{\sigma^{-2} + \sum_{k\in \full} \sigma_{Ak}^{-2}} \sqrt{(1-\pi) + \pi \frac{ \frac{\sum_{k\in \full} \sigma_{Bk}^{-2}}{\sigma^{-2} + \sum_{k\in S} \sigma_{Ak}^{-2}} }{\frac{\sum_{k\in \full} \sigma_{Ak}^{-2}}{\sigma^{-2} + \sum_{k\in \full} \sigma_{Bk}^{-2}} } }\right. \\
    & -\left.    \sqrt{\sigma^{-2} + \sum_{k\in \subb} \sigma_{Ak}^{-2}} \sqrt{(1-\pi) + \pi\frac{ \frac{\sum_{k\in \subb} \sigma_{Bk}^{-2}}{\sigma^{-2} + \sum_{k\in \subb} \sigma_{Bk}^{-2}} }{\frac{\sum_{k\in \subb} \sigma_{Ak}^{-2}}{\sigma^{-2} + \sum_{k\in \subb} 
    \sigma_{Ak}^{-2}} } } \right),
   \end{split}
\end{equation*}
Assume $g=B$; the proof for group $A$ is analogous.
Replacing $\tilde q^*_S$ from \Cref{eq:threshold_groupaware_noAA_Phi} in \Cref{eq:admission_prob_cond_q}, we find that for  policy $P_{S}$, the admissions probability (conditional on true skill $q$ and group $g$) equals
\footnotesize
\begin{equation*}
    \prob(Y=1 \mid q, B, P_{S}) = 1- \Phi \left( (\mu-q) \sqrt{\sum_{k\in S} \sigma_{Bk}^{-2}} + \sigma \Phi^{-1}(1-C) \sqrt{\sigma^{-2} + \sum_{k\in S} \sigma_{Bk}^{-2}}  \sqrt{\frac{(1-\pi) \frac{\sum_{k\in S} \sigma_{Ak}^{-2}}{\sigma^{-2} + \sum_{k\in S} \sigma_{Bk}^{-2}} + \pi \frac{\sum_{k\in S} \sigma_{Bk}^{-2}}{\sigma^{-2} + \sum_{k\in S} \sigma_{Bk}^{-2}}}{\frac{\sum_{k\in S} \sigma_{Bk}^{-2}}{\sigma^{-2} + \sum_{k\in S} \sigma_{Bk}^{-2}}}}\right).
\end{equation*}
\normalsize

Thus, the admission probability increases after dropping test scores,  if and only if
\footnotesize
\begin{equation}
\label{eq:cond_prob_admission_drop_tests}
\begin{split}
    (\mu-q) &\left(\sqrt{\sum_{k\in \full} \sigma_{Bk}^{-2}} - \sqrt{\sum_{k\in \subb} \sigma_{Bk}^{-2}} \right)  + 
    \sigma \Phi^{-1}(1-C) \sqrt{\sigma^{-2} + \sum_{k\in \full} \sigma_{Bk}^{-2}}  \sqrt{(1-\pi) \frac{ \frac{\sum_{k\in \full} \sigma_{Ak}^{-2}}{\sigma^{-2} + \sum_{k\in \full} \sigma_{Ak}^{-2}} }{\frac{\sum_{k\in \full} \sigma_{Bk}^{-2}}{\sigma^{-2} + \sum_{k\in \full} \sigma_{Bk}^{-2}} } + \pi} \\ 
    &- \sigma \Phi^{-1}(1-C)\sqrt{\sigma^{-2} + \sum_{k\in \subb} \sigma_{Bk}^{-2}} \sqrt{(1-\pi) \frac{ \frac{\sum_{k\in \subb} \sigma_{Ak}^{-2}}{\sigma^{-2} + \sum_{k\in \subb} \sigma_{Ak}^{-2}} }{\frac{\sum_{k\in \subb} \sigma_{Bk}^{-2}}{\sigma^{-2} + \sum_{k\in \subb} \sigma_{Bk}^{-2}} } + \pi}   >0.
 \end{split}
\end{equation}
\normalsize
This is equivalent to  $q <q_{ B}$, i.e., \small
\begin{equation*}
\begin{split}
    q< \mu+ & \frac{ \sigma \Phi^{-1}(1-C) }{\sqrt{\sum_{k\in \full} \sigma_{Bk}^{-2}} - \sqrt{\sum_{k\in \subb} \sigma_{Bk}^{-2}} }
    \left ( \sqrt{\sigma^{-2} + \sum_{k\in \full} \sigma_{Bk}^{-2}} \sqrt{(1-\pi) \frac{ \frac{\sum_{k\in \full} \sigma_{Ak}^{-2}}{\sigma^{-2} + \sum_{k\in \full} \sigma_{Ak}^{-2}} }{\frac{\sum_{k\in \full} \sigma_{Bk}^{-2}}{\sigma^{-2} + \sum_{k\in \full} \sigma_{Bk}^{-2}} } + \pi} \right.\\
    &- \left. \sqrt{\sigma^{-2} + \sum_{k\in \subb} \sigma_{Bk}^{-2}} \sqrt{(1-\pi) \frac{ \frac{\sum_{k\in \subb} \sigma_{Ak}^{-2}}{\sigma^{-2} + \sum_{k\in \subb} \sigma_{Ak}^{-2}} }{\frac{\sum_{k\in \subb} \sigma_{Bk}^{-2}}{\sigma^{-2} + \sum_{k\in \subb} \sigma_{Bk}^{-2}} } + \pi} \right).
\end{split}
\end{equation*}
\normalsize

\smallskip
\noindent\textit{Step 2: We show that there exists a threshold $\hat{q} \geq \max\{q_A, q_B\}$ such that the individual fairness gap increases for all $q> \hat{q}$. Otherwise, it may decrease.} 
Let 
\small
\begin{equation*}
   \underline{q} \triangleq \arg \min_{q \in \mathbb{R}} \left \{ (\mu-q) \sqrt{\sum_{k\in S} \sigma_{gk}^{-2}} + \sigma \Phi^{-1}(1-C) \sqrt{\sigma^{-2} + \sum_{k\in S} \sigma_{gk}^{-2}}  \sqrt{\frac{(1-\pi) \frac{\sum_{k\in S} \sigma_{Ak}^{-2}}{\sigma^{-2} + \sum_{k\in S} \sigma_{Ak}^{-2}} + \pi \frac{\sum_{k\in S} \sigma_{Bk}^{-2}}{\sigma^{-2} + \sum_{k\in S} \sigma_{Bk}^{-2}}}{\frac{\sum_{k\in S} \sigma_{gk}^{-2}}{\sigma^{-2} + \sum_{k\in S} \sigma_{gk}^{-2}}}} \leq 0, \forall g, S \right\}.
\end{equation*}
\normalsize

Next, consider only $ q> \max \{\underline{q}, q_A, q_B\}.$ 
Since $\Phi$ is monotone and convex  in $(-\infty, 0]$ 
and by Step 1 for any group $g$, it also holds that $\prob(Y=1 \mid q, g, P_{\full}) > \prob(Y=1 \mid q, g, P_{\subb})$ for all $q>q_g$, a sufficient condition for $I(q; P_{\full}) > I(q; P_{\subb})$ 
to hold  is 
\begin{equation*}
\begin{split}
    &\frac{\sigma^{-2} + \sum_{k\in \subb} \sigma^{-2}_{Ak} }{ \sqrt{\sum_{k\in \subb} \sigma^{-2}_{Ak}}} \left (\tilde q^*_{\subb} - \frac{\mu\sigma^{-2} + q \sum_{k\in \subb} \sigma^{-2}_{Ak} }{\sigma^{-2} + \sum_{k\in \subb} \sigma^{-2}_{Ak}}\right)
    - \frac{\sigma^{-2} + \sum_{k\in \full} \sigma^{-2}_{Ak} }{ \sqrt{\sum_{k\in \full} \sigma^{-2}_{Ak}}} \left (\tilde q^*_{\full} - \frac{\mu\sigma^{-2} + q \sum_{k\in \full} \sigma^{-2}_{Ak} }{\sigma^{-2} + \sum_{k\in \full} \sigma^{-2}_{Ak}}\right)
    \\
    <&  
    \frac{\sigma^{-2} + \sum_{k\in\subb} \sigma^{-2}_{Bk} }{ \sqrt{\sum_{k\in \subb} \sigma^{-2}_{Bk}}} \left (\tilde q^*_{\subb} - \frac{\mu\sigma^{-2} + q \sum_{k\in\subb} \sigma^{-2}_{Bk} }{\sigma^{-2} + \sum_{k\in\subb} \sigma^{-2}_{Bk}}\right)
    - \frac{\sigma^{-2} + \sum_{k\in \full} \sigma^{-2}_{Bk} }{ \sqrt{\sum_{k\in \full} \sigma^{-2}_{Bk}}} \left (\tilde q^*_{\full} - \frac{\mu\sigma^{-2} + q \sum_{k\in \full} \sigma^{-2}_{Bk} }{\sigma^{-2} + \sum_{k\in \full} \sigma^{-2}_{Bk}}\right).
\end{split}
\end{equation*}
Let \small
\begin{equation*}
\begin{split}
   \underline{\underline{q}} \triangleq \arg \min_{q \in \mathbb{R}} &\left\{  \frac{\sigma^{-2} + \sum_{k\in \subb} \sigma^{-2}_{Ak} }{ \sqrt{\sum_{k\in \subb} \sigma^{-2}_{Ak}}} \left (\tilde q^*_{\subb} - \frac{\mu\sigma^{-2} + q \sum_{k\in \subb} \sigma^{-2}_{Ak} }{\sigma^{-2} + \sum_{k\in \subb} \sigma^{-2}_{Ak}}\right)
    - \frac{\sigma^{-2} + \sum_{k\in \full} \sigma^{-2}_{Ak} }{ \sqrt{\sum_{k\in \full} \sigma^{-2}_{Ak}}} \left (\tilde q^*_{\full} - \frac{\mu\sigma^{-2} + q \sum_{k\in \full} \sigma^{-2}_{Ak} }{\sigma^{-2} + \sum_{k\in \full} \sigma^{-2}_{Ak}}\right)  \right.
     \\
    <& \left. 
    \frac{\sigma^{-2} + \sum_{k\in\subb} \sigma^{-2}_{Bk} }{ \sqrt{\sum_{k\in \subb} \sigma^{-2}_{Bk}}} \left (\tilde q^*_{\subb} - \frac{\mu\sigma^{-2} + q \sum_{k\in\subb} \sigma^{-2}_{Bk} }{\sigma^{-2} + \sum_{k\in\subb} \sigma^{-2}_{Bk}}\right)
    - \frac{\sigma^{-2} + \sum_{k\in \full} \sigma^{-2}_{Bk} }{ \sqrt{\sum_{k\in \full} \sigma^{-2}_{Bk}}} \left (\tilde q^*_{\full} - \frac{\mu\sigma^{-2} + q \sum_{k\in \full} \sigma^{-2}_{Bk} }{\sigma^{-2} + \sum_{k\in \full} \sigma^{-2}_{Bk}}\right)
    \right\}.
\end{split}
\end{equation*}
\normalsize

Define $\hat{q} \triangleq \max\{\underline{q}, \underline{\underline{q}}, q_A, q_B\}$. Then, by the previous conditions, we have $I(q; P_{\full}) > I(q; P_{\subb})$ for all $q>\hat{q}$, thus the individual fairness gap decreases.
Furthermore, $\hat{q} \geq \max\{q_A, q_B\}$ as required.

Finally, if $q_A < q_B$, then for all $q_A < q < q_B$, ${\prob(Y=1 \mid q, A, P_{\full}) > \prob(Y=1 \mid A, g, P_{\subb})}$  but ${\prob(Y=1 \mid q, B, P_{\full}) < \prob(Y=1 \mid B, g, P_{\subb})}$ (by Step 1). Thus, $I(q; P_{\full}) > I(q; P_{\subb})$. 
\end{proof}

\begin{proof}[\textit{Proof of Part (iii).}]
Since $\textrm{Var}[\tilde q \mid g, P_{\subb}] < \textrm{Var}[\tilde q \mid g, P_{\full}] $ and, by Corollary \ref{cor.threshold_comp_group_aware}, $\tilde q^*_{\subb} < \tilde q^*_{\full}$, the expected estimated skill of each admitted group decreases, that is, $$\expec[\tilde q \mid \tilde q \geq \tilde q^*_{\subb}, g, P_{\subb} ] < \expec[\tilde q \mid \tilde q\geq \tilde q^*_{\full}, g, P_{\full} ].$$ \Cref{eq:expected_academic_merit} further implies that
$\expec[q \mid Y=1, g, P_{\subb} ] < \expec[q \mid Y=1, g, P_{\full} ].$
\end{proof}

\bigskip
\noindent \textbf{Admissions with barriers to testing.}
 In a setting with barriers to testing and policy $P_{\full}$, let $\tilde w^*_{\full}$ the decision threshold of the school with policy $P_{\full}$. Then, observe that $\tilde w^*_{\full} < \tilde q^*_{\full}$, where
\begin{equation}
  (1-\pi)\gamma_A  (1-F_{\tilde q \mid A, P_{\full}} (\tilde w^*_{\full})) +\pi\gamma_B (1- F_{\tilde q \mid B, P_{\full}} (\tilde w^*_{\full})) = C.
\end{equation}
 
We now study the trade-off between  barriers and informativeness.
For brevity, we use $\pi_A=1- \pi$, $\pi_B=\pi$.

\begin{restatable}[
Theorem~\ref{thm:threshold_char_without_aa}]{theorem}{propthresholdnoAA}
Consider policies  $P_{\full}$ and $P_{\subb}$ and assume unequal precisions under $P_{\full}$. 
\begin{itemize}

    \item[(i)] For each group $g$ there exists a constant $\Delta_g (\xi_g, \rho^g_{\subb})$  
such that the academic merit of group $g$ increases if and only if
\begin{equation}
\label{eq:gamma_A_gamma_B_tension}
    \beta_g ( \gamma_A, \gamma_B, \rho^g_{\full}) \leq \Delta_g (\xi_g, \rho^g_{\subb}),
\end{equation}
where
\begin{equation*}
	\rho^A_{S}  = \frac{1}{\rho^B_{S}} \triangleq \frac{\sum_{k \in S } \sigma_{Bk}^{-2}}{\sigma^{-2}+\sum_{k \in S} \sigma_{Bk}^{-2}} \left ( \frac{\sum_{k \in S} \sigma_{Ak}^{-2}}{\sigma^{-2}+\sum_{k \in S} \sigma_{Ak}^{-2}} \right)^{-1},
\end{equation*}
\begin{equation*}
    \xi_g =
\frac{\sum_{k \in \subb } \sigma_{gk}^{-2} }{\sigma^{-2}+\sum_{k \in \subb } \sigma_{gk}^{-2} } \left(\frac{\sum_{k \in \full } \sigma_{gk}^{-2}}{\sigma^{-2}+\sum_{k \in \full } \sigma_{gk}^{-2}}\right)^{-1}, 
\end{equation*}
\begin{equation*}
	\beta_g ( \gamma_g, \gamma_{g'}, \rho^g_{S}) \triangleq \Phi^{-1} \left(1-\frac{C}{\pi_g \gamma_g + \pi_{g'} \gamma_{g'}}\right) \sqrt{ \frac{\frac{\pi_{g'}\gamma_{g'}}{\pi_{g}\gamma_{g}} \rho^g_{S} +1}{1+ \frac{\pi_{g'}\gamma_{g'}}{\pi_{g}\gamma_{g}} }},
\end{equation*}
\begin{equation*}
    \Delta_g (\xi_g, \rho^g_{\subb})= \textrm{HR}^{-1}\left({\xi_g}
    \textrm{HR}\left(
    \Phi^{-1}(1-C)
     \sqrt{\pi_{g}
    + \pi_{g'} \rho^g_{\subb} }
    \right)\right).
\end{equation*}
As barriers to group $g$ increase ($\gamma_g$ decreases), then $\beta_g ( \gamma_A, \gamma_B, \rho^g_{\subb})$ decreases. 
{Thus, given any group $g$ and $\gamma_{g'} \in (0,1]$, $g'\neq g$, there exists threshold $\bar{\bar \gamma}_g \in (0,1]$, such that academic merit of group $g$ improves by dropping feature $K$ if and only if $\gamma_g < \bar{\bar\gamma}_{g}$.}

    \item[(ii)]	Diversity strictly improves after dropping test scores if and only if
		$\eta(1, 1, \rho^B_{\subb}) > \eta(\gamma_A, \gamma_B, \rho^B_{\full}),$ where
	\begin{equation*}
		\eta(\gamma_A, \gamma_B, \rho^B_S) \triangleq \frac{(1-\pi)\gamma_B}{C} \left( 1- \Phi\left(
		\Phi^{-1}\left(1-\frac{C}{(1-\pi)\gamma_A  + \pi \gamma_B}\right) \sqrt{\frac{(1-\pi) \gamma_A \rho^B_{S}  +  {\pi\gamma_B} }{(1-\pi)\gamma_A  + {\pi\gamma_B}}}\right)\right).
	\end{equation*}
	Given any $\gamma_A \in (0,1]$, there exists a threshold $\bar{\gamma} \in (0,1]$, such that diversity strictly improves after dropping test scorers if and only if $\gamma_B < \bar{\gamma}$.
	\label{prop:diversitythresholdnoAA}
\end{itemize}

\end{restatable} 

\begin{proof}{\textit{Proof of Part (i).}}
We break the proof into the following parts.

\noindent\textit{Step 1: We show that the academic merit of group $g$ increases if and only if \Cref{eq:gamma_A_gamma_B_tension} holds.}
We adopt an argument similar to the proof of  Proposition~\ref{thm.budget_tests_vs_poolsize}. We prove the statement for $g=A$. The argument for group $B$ is similar.

First, similarly to \Cref{eq:threshold_groupaware_noAA_Phi}, we derive that
\begin{equation}
\label{eq:w_threshold}
\begin{split}
    \tilde w_{\full}^* &= \mu + \Phi^{-1}\left(1-\frac{C}{(1-\pi)\gamma_A + \pi \gamma_B}\right)
    \sigma  \sqrt{
    \frac{(1-\pi) \gamma_A  \frac{\sum_{k\in \full} \sigma_{Ak}^{-2}}{\sigma^{-2}+\sum_{k\in \full} \sigma_{Ak}^{-2}}
    +
    \pi \gamma_B  \frac{\sum_{k\in \full} \sigma_{Bk}^{-2}}{\sigma^{-2}+\sum_{k\in \full} \sigma_{Bk}^{-2}}
    }
    {(1-\pi)\gamma_A +\pi\gamma_B}}.
\end{split}
\end{equation}

Second, requiring that $\expec[q \mid Y=1, A, P_{\full}] \leq \expec[q \mid Y=1, A, P_{\subb}]$ and adapting Lemma~\ref{lemma:conditional_expectation_normal} to our setting with barriers gives us
\small
\begin{equation*}
\begin{split}
    &\sqrt{\frac{\sum_{ k \in \full} \sigma_{Ak}^{-2}}{\sigma^{-2}+\sum_{ k \in \full} \sigma_{Ak}^{-2}}}
    \textrm{HR} \left(
    \frac{\Phi^{-1}\left(1-\frac{C}{(1-\pi)\gamma_A + \pi \gamma_B}\right)
     \sqrt{
    \frac{ \frac{(1-\pi)\gamma_A}{\pi\gamma_B}  \frac{\sum_{ k \in \full} \sigma_{Ak}^{-2}}{\sigma^{-2}+\sum_{ k \in \full} \sigma_{Ak}^{-2}}
    +    \frac{\sum_{ k \in \full} \sigma_{Bk}^{-2}}{\sigma^{-2}+\sum_{ k \in \full} \sigma_{Bk}^{-2}}
    }
    {1+\frac{(1-\pi)\gamma_A}{\pi\gamma_B}}}}{\sqrt{\frac{\sum_{ k \in \full} \sigma_{Ak}^{-2}}{\sigma^{-2}+\sum_{ k \in \full} \sigma_{Ak}^{-2}}}}
    \right)\\
    \leq&
    \sqrt{\frac{\sum_{ k \in \subb} \sigma_{Ak}^{-2}}{\sigma^{-2}+\sum_{ k \in \subb} \sigma_{Ak}^{-2}}}
    \textrm{HR}\left(
    \frac{\Phi^{-1}(1-C)
     \sqrt{(1-\pi)
    \frac{  \sum_{ k \in \subb} \sigma_{Ak}^{-2}}{\sigma^{-2}+\sum_{ k \in \subb} \sigma_{Ak}^{-2}}
    + \pi   \frac{\sum_{ k \in \subb} \sigma_{Bk}^{-2}}{\sigma^{-2}+\sum_{ k \in \subb} \sigma_{Bk}^{-2}}
    }}{\sqrt{\frac{\sum_{ k \in \subb} \sigma_{Ak}^{-2}}{\sigma^{-2}+\sum_{ k \in \subb} \sigma_{Ak}^{-2}}}}
    \right).
\end{split}
\end{equation*}
\normalsize
 
Replacing with the definitions of $\Delta_A, \rho^A_{\full}$, we finally obtain that
\begin{equation*}
    \Phi^{-1} \left(1-\frac{C}{(1-\pi)\gamma_A + \pi \gamma_B}\right) \sqrt{\frac{\frac{(1-\pi)\gamma_A}{\pi\gamma_B} +\rho^A_{\full}}{1+\frac{(1-\pi)\gamma_A}{\pi\gamma_B} }}
    \leq \Delta_A(\xi_A, \rho^A_{\subb}).
\end{equation*}
Equivalently, using the definition of $\beta_g$, we finally get that academic merit in group $A$ improves after dropping feature $K$ if and only if
$\beta_A(\gamma_A, \gamma_B, \rho^A_{\full}) \leq \Delta_A(\xi_A, \rho^A_{\subb}).$

\medskip\noindent\textit{Step 2: We show that, for each group $g\in \{A,B\}$, $\beta_g(\gamma_g, \gamma_{g'}, \rho^g_{\full})$ is increasing in $\gamma_g$.}
Given some group $g$, fix all parameters except $\gamma_g$.
Then, the function
$\Phi^{-1} \left(1-\frac{C}{(1-\pi)\gamma_A + \pi \gamma_B}\right)$
is increasing in $\gamma_g$ since $\Phi^{-1}$ is increasing in its argument and $1-\frac{C}{(1-\pi)\gamma_A + \pi \gamma_B}$ is an increasing function of both $\gamma_A, \gamma_B$.

Now consider the expression in the second term of $\beta$:
\begin{equation}
\label{eq:beta_second_term}
   \sqrt{ \frac{\frac{\pi_{g'}\gamma_{g'}}{\pi_{g}\gamma_{g}} \rho^g_{S} +1}{1+ \frac{\pi_{g'}\gamma_{g'}}{\pi_{g}\gamma_{g}} }}.
\end{equation}
We show that this function is increasing in $\gamma_g$, for both $g=A$ and $g=B$.
More specifically, for group $g=A$, the derivative of  \Cref{eq:beta_second_term}  with respect to $\gamma_A$ equals
\begin{equation*}
    \begin{split}
        \frac{\partial}{\partial \gamma_A}\left(\sqrt{\frac{\frac{(1-\pi)\gamma_A}{\pi\gamma_B} +\rho^A_{\full}}{1+\frac{(1-\pi)\gamma_A}{\pi\gamma_B} }}\right) = 
\frac{ (1-\pi) \pi \gamma_B (1-\rho^A_{\full})}{ 2((1-\pi)\gamma_A + {\pi\gamma_B})^2}
\left(\sqrt{\frac{{(1-\pi)\gamma_A} + {\pi\gamma_B} \rho^A_{\full}}{{(1-\pi)\gamma_A}+ {\pi\gamma_B} }}\right)^{-1},
    \end{split}
\end{equation*}
and is positive since $\rho^A_{\full}<1$. 
A similar argument applies for group $g=B$ since $\rho^B_{\full} >1$.

\smallskip
\noindent\textit{Step 3: We show that for any given group $g$ and $\gamma_{g'}\in(0,1]$, $g'\neq q$, there exists threshold $\bar{\bar\gamma}_g\in(0,1]$ such that academic merit of group $g$ improves if and only if $\gamma_g < \bar{\bar\gamma}_g$.}
Fix group $A$; the proof is analogous for group $B$.
It suffices to show that (a) $\bar{\bar\gamma}_A$ is the unique solution to $\beta_A(\bar{\bar\gamma}_A, \gamma_B, \rho^A_{\full}) = \Delta_A(\xi_A, \rho^A_{\subb})$
and (b) $\bar{\bar\gamma}_A \in (0,1]$.

Conditional on the existence of $\bar{\bar{\gamma}}_A$, uniqueness in (a) follows immediately from the monotonicity of $\beta_A$ shown in Step 2. Existence in turn can be shown as follows.  In the absence of barriers, Part (iii) in Theorem~\ref{prop.test_group_aware_noAA} guarantees that the academic merit of group $g$ decreases after dropping test scores, thus 
$\beta_A (1, \gamma_B,\rho^A_{\full}) > \Delta_A (\xi_A, \rho^A_{\full})$.
Furthermore, observe that for $\gamma_A =0$, academic merit trivially improves from $\beta_A (0,\gamma_B, \rho^A_{\full}) =0$ to a positive value $\Delta_A(\xi_A, \rho^A_{\subb})>0$ after dropping test scores.
Thus, by the continuity of $\beta_A(\gamma_A, {\gamma}_B, \rho^A_{\full})$, such a $\bar{\bar\gamma}_A$ exists.
For Part (b), continuity of $\beta_A$ further implies that there must exist an interval $[0, \epsilon)$, $\epsilon)>0$, such that  $ \beta_A(\gamma_A, \gamma_B, \rho^B_{\full}) < \Delta_A(\xi_A, \rho^A_{\subb})$ for all $\gamma_A \in [0, \epsilon)$. Consequently, $\bar{\bar\gamma}_A \geq \epsilon>0$.
\end{proof}

\medskip
\begin{proof}[\textit{Proof of Part (ii).}] 
Plugging  \Cref{eq:w_threshold} 
into the definition of diversity with and without test scores, respectively, it immediately follows that diversity  improves if and only if
$ \eta(1, 1, \rho^B_{\subb}) > \eta(\gamma_A, \gamma_B, \rho^B_{\full}).$

\noindent\textit{Step 1: Fix all parameters (including $\gamma_A \in (0,1]$) except for $\gamma_B \in(0,1]$. We  show that diversity strictly increases as barriers decrease ($\gamma_B$ increases), i.e., $\eta(\gamma_A, \gamma_B', \rho^B_{\full}) > \eta(\gamma_A, \gamma_B, \rho^B_{\full})$ for $\gamma_B'>\gamma_B$.}

By \Cref{eq:w_threshold}, the admission threshold increases as $\gamma_B$ increases. 
Indeed, $\tilde q^*_{\subb}$ is the solution to
$ {(1-\pi) \gamma_A (1-F_{\tilde q \mid A, P_{\subb}}( \tilde q^*_{\subb}))  + \pi \gamma_B (1-F_{\tilde q \mid B, P_{\subb}}( \tilde q^*_{\subb}))  = 1-C.}$
Thus, as $\gamma_B$ increases, the solution $\tilde q^*_{\subb}$ must decrease since each $F_{\tilde q \mid g, P_{\subb}}$ is increasing in its argument.

Then, since the admission threshold $\tilde q^*_{\subb}$ increases but the capacity $C$, barriers $\gamma_A$ (thus the mass of students in group $A$ who are eligible to apply), and the perceived skill distributions for both groups remain constant, it follows that a lower mass of students are admitted from group $A$.
As a result, the remaining capacity is filled with more students from group $B$, which in turn implies that diversity increases.

\smallskip\noindent\textit{Step 2: We show that, given all other parameters fixed including $\gamma_A$, there exists a threshold $\bar{\gamma}_B(\gamma_A)$ such that diversity increases after dropping the test if and only if $\gamma_B < \bar{\gamma}$.}
It suffices to show that (a) $\bar{\gamma}$ is the unique solution to $ \eta(1, 1, \rho^B_{\subb}) = \eta(\gamma_A, \bar{\gamma}, \rho^B_{\full})$
and (b) $\bar{\gamma} \in (0,1]$.
The proof follows as in Step 3 in Part (i).
\end{proof}

\subsection{Strategic students: Single school (Proofs from Section~\ref{sec:singleschoolstratetic})}
\label{app:strategic_single}

\begin{lemma}
\label{lemma:reduced_equilibria_single}
Fix testing policy $P_{\full}$. Let $\alpha(\tilde q_{\subb}, g; P_{\full}): \mathbb{R} \times \{A, B\} \rightarrow \{0,1\}$ denote the function that describes the action of students in group $g$ with skill estimate $\tilde q_{\subb}$,
i.e., 
\begin{equation}
\label{eq.student.decision.cost}
    \alpha(\tilde q_{\subb}, g; P_{\full}) \triangleq \arg \max_{\alpha \in \{0,1\}} \alpha \left( v \prob (Y = 1 \mid \tilde q_{\subb} , g, P_{\full}) - c_g \right).
\end{equation}
At equilibrium, for any $\thetaset_{\subb} \in \mathbb{R}^{K-1}$ and $g\in \{A,B\}$, it holds that
$$\alpha(\tilde q(\thetaset_{\subb}, g), g; P_{\full}) = \arg \max_{\alpha \in \{0,1\}} \alpha \left( v \prob (Y = 1 \mid \thetaset_{\subb}, g, P_{\full}) - c_g  \right) $$
\end{lemma}

\begin{proof}{\textit{Proof.}}
Recall that $\tilde q^*_{\full}$ denotes the admission threshold of the school with policy $P_{\full}$ at a given equilibrium. To solve \Cref{eq.student.decision.cost.theta}, the student computes the following probability:
\begin{equation*}
    \begin{split}
        \prob (Y = 1 \mid \thetaset_{\subb}, g, P_{\full})  & = \prob( \tilde q (\thetaset_{\full}, g) \geq \tilde q^*_{\full} \mid \thetaset_{\subb}, g) \\ & = \prob_{\theta_K} \left(\frac{\tilde q (\thetaset_{\subb},g) (\sigma^{-2} + \sum_{k\in \subb} \sigma_{gk}^{-2}) + (\theta_{K} - \mu_{gK})\sigma_{gK}^{-2}}{\sigma^{-2} +\sum_{k\in \full} \sigma_{gk}^{-2}} \geq \tilde q^*_{\full} \mid \thetaset_{\subb}, g\right)\\
        & = \prob_{\theta_K}\left (\theta_K \geq \mu_{gK} + \tilde q^*_{\full} + \sigma_{gK}^{2}(\tilde q^*_{\full}  - \tilde q (\thetaset_{\subb},g)) (\sigma^{-2} + \sum_{k\in \subb} \sigma_{gk}^{-2}) \mid \thetaset_{\subb},g \right) \\
        &= \prob_{\theta_K}\left (\theta_K \geq \mu_{gK} + \tilde q^*_{\full} + \sigma_{gK}^{2}(\tilde q^*_{\full}  - \tilde q (\thetaset_{\subb},g)) (\sigma^{-2} + \sum_{k\in \subb} \sigma_{gk}^{-2}) \mid \tilde q(\thetaset_{\subb},g), g \right)\\
        & = \prob (Y = 1 \mid \tilde q(\thetaset_{\subb}, g), g, P_{\full}),
    \end{split}
\end{equation*}
where in the second line we used \Cref{eq.app.perceived_skill_cond_features} for $\thetaset= \thetaset_{\full},\thetaset_{\subb}$ to rewrite $\tilde q (\thetaset_{\full},g)$ in terms of  $\tilde q_{\subb} (\thetaset_{\subb})$ and $\theta_K$, i.e.,
\begin{equation}
\label{eq.qfull_qsub_thetaK}
\begin{split}
  \tilde q (\thetaset_{\full},g) &= \frac{\mu \sigma^{-2} + \sum_{k\in \full}(\theta_{k} - \mu_{gk})\sigma_{gk}^{-2}} {\sigma^{-2} +\sum_{k\in \full} \sigma_{gk}^{-2}} \\ &= 
  \frac{\tilde q (\thetaset_{\subb},g) (\sigma^{-2} + \sum_{k\in \subb} \sigma_{gk}^{-2}) + (\theta_{K} - \mu_{gK})\sigma_{gK}^{-2}}{\sigma^{-2} +\sum_{k\in \full} \sigma_{gk}^{-2}}.
\end{split}
\end{equation}

This equality immediately implies that for any $\alpha \in \{0,1\}$:
$$\alpha (v\prob (Y = 1 \mid \thetaset_{\subb}, g, P_{\full})  -c_g)= \alpha(v \prob (Y = 1 \mid \tilde q(\thetaset_{\subb}, g), g, P_{\full}) -c_g).$$
Consequently, 
\begin{equation*}
  \begin{split}
      \arg \max_{\alpha \in \{0,1\}}\alpha (v\prob (Y = 1 \mid \thetaset_{\subb}, g, P_{\full})  -c_g)&= \arg \max_{\alpha \in \{0,1\}} \alpha(v \prob (Y = 1 \mid \tilde q(\thetaset_{\subb}, g), g, P_{\full}) -c_g) \\&= \alpha (\tilde q(\thetaset_{\subb}, g),g; P_{\full}),
  \end{split}  
\end{equation*}
which concludes the proof of the lemma.
\end{proof}

\lemmathresstrategy*

\begin{proof}{\textit{Proof.}}
Without loss of generality, we fix group $g$ throughout the proof as the arguments are analogous for both groups of students.

\medskip
 \textit{Step 1: We derive the distribution of $\tilde q_{\full}  \mid  \tilde q_{\subb}, g, P_{\full}$}. The student uses this distribution to solve \Cref{eq.student.decision.cost}. 
 
 Fix test-free skill estimate $\tilde q_{\subb}$. 
 By Lemma~\ref{lemma.true_skill_conditional_perceived}, we have that 
 \begin{equation}
 \label{eq:expected_skill_given_qtilde_sub}
    q \mid \tilde{q}_{\subb}, g \sim \N \left( \tilde{q}_{\subb}, \frac{1}{\sigma^{-2} + \sum_{k\in \subb} \sigma_{gk}^{-2}}\right).
  \end{equation}
  Furthermore, conditional on her true skill $q$, the student's test score $\theta_K$ is drawn from a  distribution  $\theta_K \mid q, g \sim \N(q + \mu_{gK}, \sigma_{gK}^2)$. 
  Applying Lemma~\ref{lemma:prior_marginal_normal}, we get that
  \begin{equation}
  \label{eq:expected_test_given_qtilde_sub}
      \theta_K \mid \tilde q_{\subb}, g \sim \N \left ( \tilde q_{\subb} +\mu_{gK}, \sigma^{2}_{gK} + \frac{1}{\sigma^{-2} + \sum_{k \in \subb} \sigma_{gk}^{-2}} \right).
  \end{equation}
 
By applying Lemma~\ref{lemma:prior_marginal_normal}, combined with \Cref{eq.qfull_qsub_thetaK} and the above distribution,
the student then finds that her projected skill estimate $\tilde q_{\full} \mid \tilde q_{\subb}, g, P_{\full}$, after they take the test and submit the score $\theta_K$ to the school, will follow a Normal distribution:
\begin{equation}
\label{eq:q_tilde_cond_thetasubb_q}
\tilde q_{\full} \mid \tilde q_{\subb}, g, P_{\full} \sim \N \left(\tilde q_{\subb}, 
\left ( \frac{\sigma_{gK}^{-2}}{\sigma^{-2} + \sum_{k \in \full} \sigma_{gk}^{-2}}\right)^2
\left(\sigma^{2}_{gK} + \frac{1}{\sigma^{-2} + \sum_{k \in \subb} \sigma_{gk}^{-2}} \right)  \right).
\end{equation}

\textit{Step 2:} \textit{Neither $\alpha^*(\tilde q_{\subb}, g; P_{\full}) = 1$, $\forall ~ \tilde q_{\subb}$, or $\alpha^*(\tilde q_{\subb}, g; P_{\full}) =0$,  $\forall ~ \tilde  q_{\subb}$, constitute an equilibrium.} For the sake of contradiction, assume that $\alpha^*(\tilde q_{\subb}, g; P_{\full}) = 1$, $\forall ~\tilde q_{\subb}$, is an equilibrium. Then, all students take the test and apply to the school as in the main setup without barriers. 

The student has probability $\prob (Y=1 \mid \tilde q_{\subb}, g, P_{\full}) = \prob (\tilde q_{\full} \geq \tilde q^*_{\full} \mid \tilde q_{\subb}, g, P_{\full}) $ to be accepted by the school. Keeping $\tilde q^*_{\full}$ fixed, by \Cref{eq:q_tilde_cond_thetasubb_q}, there exists a small enough $\underline{q}$ such that for all $\tilde q_{\subb} < \underline{q}$,  $v\prob (Y=1 \mid \tilde q_{\subb}, g, P_{\full}) - c_g <0$. Thus, students with $\tilde q_{\subb} < \underline{q}$ have incentive not to apply, implying that $\alpha^* (\tilde q_{\subb}, g; P_{\full}) =0$ for a positive mass of students, which contradicts our assumption. 

A similar argument also shows that $\alpha^* (\tilde q _{\subb}, g; P_{\full}) =0$ cannot be an equilibrium, since  students with  $\tilde q_{\subb} > \overline{q}$ for some threshold $\overline{q}$ will have the incentive to deviate and take the test.

\medskip
\textit{Step 3: If $\alpha^*(\tilde q_{\subb}, g; P_{\full})$ is an equilibrium student strategy, then  it must be non-decreasing in $\tilde q_{\subb}$.}
We prove this claim by contradiction. Suppose that there exist $\tilde q_{\subb}', \tilde q_{\subb}''$, with $\tilde q_{\subb}'< \tilde q_{\subb}''$, such that $1=\alpha^*(\tilde q_{\subb}', g; P_{\full}) > \alpha^*(\tilde q_{\subb}'', g; P_{\full})=0$. 

We show that this cannot hold true. Indeed, since the mean of \Cref{eq:q_tilde_cond_thetasubb_q} is increasing in $\tilde q_{\subb}$ and the variance does not depend on $\tilde q_{\subb}$, it follows that $$\prob ( Y =1 \mid \tilde q_{\subb}', g, P_{\full}) \leq \prob ( Y =1 \mid \tilde q_{\subb}'', g, P_{\full}),$$ 
therefore 
$0 \leq v \prob ( Y =1 \mid \tilde q_{\subb}', g, P_{\full}) - c_g \leq v \prob ( Y =1 \mid \tilde q_{\subb}'', g, P_{\full}) - c_g,$ where the first inequality follows from the fact that $\alpha^*(\tilde q_{\subb}', g; P_{\full})=1$.  Consequently, the student with $\tilde q_{\subb}''$ also has the incentive to apply, i.e., $\alpha^*(\tilde q_{\subb}'', g; P_{\full}) =1$ which is a contradiction. Thus, $\alpha^*(\tilde q_{\subb}, g; P_{\full})$ must be non-decreasing in $\tilde q_{\subb}$.

\medskip
\textit{Step 4: If an equilibrium exists, it takes a threshold form: $\alpha^*(\tilde q_{\subb}, g; P_{\full}) = \indic\{\tilde q_{\subb} \geq \underline{q}^ g\}$.}
An immediate corollary of Steps 2 and 3 is that if an equilibrium $\alpha^*(\tilde q_{\subb}, g; P_{\full})$  exists, it must take a threshold form, i.e.,  
there must exist a threshold $\underline{q}^ g$ such that $\alpha^*(\tilde q_{\subb}, g; P_{\full}) = \indic\{ \tilde q_{\subb} \geq \underline{q}^ g\}$. In other words, $\underline{q}^ g$ corresponds to the unique skill level that characterizes students who are indifferent between taking and not taking the test.

\medskip
\textit{Step 5: An equilibrium $(\alpha^*,Y^*)$ exists and is unique.}
As explained in the main text, the selection policy $Y^*$ of the school remains the same as in the baseline setting without test costs: among the students who apply, the school sets a threshold $\tilde q^*_{\full}$ to accept the top mass $C$ of applicants thus $Y^*(\tilde q^*_{\full}; P_{\full}) = \indic\{\tilde q_{\full} \geq \tilde q^*_{\full}\}$ where $\tilde q^*_{\full}$ is the unique solution to: 
\begin{equation}
\begin{split}
    \label{eq:singleschoolproblem_main_appendix}
    \tilde q^*_{S} =  \min \left\{z \in \mathbb{R}:  
    \sum_g \pi_g \expec_{\thetaset_S}[\alpha^*(\thetaset_{\subb},g; P_S)  \mid \tilde q (\thetaset_S, g)  \geq z, g, P_S  ] \leq C\right\}.
\end{split}
\end{equation}
Regarding $\alpha^*$, we will prove the slightly more general statement: given any threshold $\tilde q_{\full}^*$, there exists a unique equilibrium with $\alpha^*(\tilde q_{\subb}, g; P_{\full}) = \indic\{\tilde q_{\subb} \geq \underline{q}^ g\}$ where $\underline{q}^ g$ is the solution to
\begin{equation}
    \label{eq:indiff_student_equilibrium}
\prob (\tilde q_{\full} \geq \tilde q_{\full}^* \mid \underline{q}^ g, \thetaset_{\subb},g)= \frac{c_g}{v}.
\end{equation}

Indeed, given the admission cutoff $\tilde q_{\full}^*$ and using \Cref{eq:q_tilde_cond_thetasubb_q}, the student computes her admission probability:
$$\prob (\tilde q_{\full} \geq \tilde q_{\full}^* \mid \underline{q}^ g,g, P_{\full}) = 1-\Phi\left(\frac{\tilde q_{\full}^* -\underline{q}^ g }{\var(\tilde q_{\full} \mid \underline{q}^ g, g, P_{\full} )}\right).$$
Given that the CDF $\Phi$ is a continuous, strictly increasing function in $\tilde q_{\subb}$   and $c_g/v<1$, it follows that \Cref{eq:indiff_student_equilibrium} has a unique solution $\underline{q}^ g$. 
Then, $\alpha^*(\tilde q_{\subb}, g; P_{\full}) = \indic\{ \tilde q_{\subb} \geq \underline{q}^ g\}$ is an equilibrium:  all students with $\tilde q_{\subb} \geq\underline{q}^ g$ receive weakly positive expected utility if they apply, whereas all students with $\tilde q_{\subb} < \underline{q}^ g$ get strictly negative expected utility therefore they choose not to apply.  By the uniqueness of the solution $\underline{q}^ g$ to \Cref{eq:indiff_student_equilibrium}, it follows that no other equilibrium of a threshold form can exist. Due to Step 4, this further implies that $\alpha^*(\tilde q_{\subb}, g; P_{\full}) = \indic\{\tilde q_{\subb} \geq \underline{q}^ g\}$ must be unique.
Extending the arguments to students of any  $g$ concludes the proof. 
\end{proof}

\subsubsection{Effect of test cost and informativeness on admissions}  
\label{app.singleschool.effectsoftestcosts}
In the non-strategic setting of \Cref{prop:group_aware_noAA},  the sign of the informativeness gap $\sum_{k \in S} \sigma^{-2}_{Ak} - \sum_{k \in S} \sigma^{-2}_{Bk}$ determined the diversity level and academic merit in a straightforward manner: if group $B$ has lower total precision than group $A$, then it is under-represented and has lower academic merit among the admitted class. The same holds even in the presence of barriers as long as $\gamma_A \geq \gamma_B$.  With costly testing, however, Proposition~\ref{prop.div_merit_strategic} below shows that the relationship between informativeness and fairness becomes more complex, depending on the costs and informativeness of the features with and without the test score.
Recall that $\Phi_2(x,y; \rho)$ denotes the CDF of the  standard bivariate Normal distribution with correlation $\rho$.

\begin{proposition}
\label{prop.div_merit_strategic.full}
\label{prop.div_merit_strategic}
Consider the equilibrium under policy $P_{\full}$.  
\begin{itemize}
    \item [(i)] \emph{Diversity level}: Group $B$ students are under-represented, i.e., $\tau(P_{\full}) < \pi$, if and only if  
    \begin{equation}
    \label{eq.diversity_condition_barriers}
        \frac{\Phi\left(\frac{a_A+b_A\mu}{\sqrt{1 + \tilde \sigma_A ^2 b_A^2}}\right) - \Phi_2 \left( \frac{a_A+b_A\mu}{\sqrt{1 + \tilde \sigma_A ^2 b_A^2}}, \frac{\tilde q^*_{\full} -\mu}{\tilde \sigma_A}; -\frac{\tilde \sigma_A b_A}{\sqrt{1 + \tilde \sigma_A ^2 b_A^2}} \right)}
        {\Phi\left(\frac{a_B+b_B\mu}{\sqrt{1 + \tilde \sigma_B ^2 b_B^2}}\right) -  \Phi_2 \left( \frac{a_B+b_B\mu}{\sqrt{1 + \tilde \sigma_B ^2 b_B^2}}, \frac{\tilde q^*_{\full} -\mu}{\tilde \sigma_B}; -\frac{\tilde \sigma_B b_B}{\sqrt{1 + \tilde \sigma_B ^2 b_B^2}} \right)}
        > \frac{ \tilde \sigma_B}{ \tilde \sigma_A}, 
    \end{equation}
    where  
    $$\tilde \sigma_g = \sigma \sqrt{
    \frac{\sum_{k\in \full} \sigma^{-2}_{gk}}{\sigma^{-2}+ \sum_{k\in \full} \sigma^{-2}_{gk}}},$$ 
    $$a_g \triangleq a_g(\tilde q_{\full}^*)= \frac{  \frac{\sigma_{gK}^{-2}\Phi^{-1} \left(1 - \frac{c_g}{v}\right)}{\sigma^{-2}  + \sum_{k \in \full} \sigma_{gk}^{-2} }   
\sqrt{\sigma^{2}_{gK} + \frac{1}{\sigma^{-2} + \sum_{k \in \subb} \sigma_{gk}^{-2}} } + \mu \frac{ \sigma^{-2} }{\sigma^{-2} + \sum_{k \in \subb} \sigma_{gk}^{-2}}- \tilde q_{\full}^*}{\frac{\sqrt{\sum_{k \in \subb} \sigma_{gk}^{-2}}}{\sigma^{-2} + \sum_{k \in \subb} \sigma_{gk}^{-2}} \sqrt{1 + \frac{ \sum_{k \in \subb} \sigma_{gk}^{-2}}{\sigma^{-2} + \sum_{k \in \full} \sigma_{gk}^{-2}}}},$$
   $$b_g = \sqrt{\frac{\sum_{k\in \subb} \sigma^{-2}_{gk} (\sigma^{-2} + \sum_{k\in \full} \sigma^{-2}_{gk})}{\sigma^{-2}+ \sum_{k\in \subb} \sigma^{-2}_{gk} + \sum_{k\in \full} \sigma^{-2}_{gk}}}.$$
 \item[(ii)] \emph{Academic merit}: Policy $P_{\full}$ achieves worse academic merit for group $B$ than group $A$ if and only if 
    $ \lambda (a_A, b_A, \tilde \sigma_A, \tau_A) >  \lambda (a_B, b_B, \tilde \sigma_B, \tau_B),$
    where
    \begin{equation*}
        \begin{split}
       \lambda (a_g, b_g, \tilde \sigma_g, \tau_g)\triangleq \mu  
       &- 
     \frac{\tilde \sigma_g^2 b_g}{\tau_g \sqrt{1 +\tilde \sigma_g^2 b_g^2}} \phi \left(\frac{a_g + b_g \mu}{\sqrt{1 +\tilde \sigma_g^2 b_g^2}}\right) 
     \Phi\left(\tilde q^*_{\full} \sqrt{1 +\tilde \sigma_g^2 b_g^2}  + \frac{ b_g\tilde \sigma_g (a_g+b_g \mu)} {\sqrt{1 +\tilde \sigma_g^2 b_g^2}} \right) \\
     &+  \Phi(a_g + b_g \mu + b_g \tilde \sigma_g \tilde q^*_{\full}) \frac{\tilde \sigma_g^2  \phi(\tilde q^*_{\full})}{\tau_g}. 
            \end{split}
    \end{equation*}
    \normalsize
\end{itemize} 
\end{proposition}

\begin{proof}[\textit{Proof of Part (i).}]
Fix group $g$. We break the proof into steps. 

\medskip\textit{Step 1: We derive the distribution of $\tilde q_{\subb} \mid \tilde q_{\full}, g$.}
By Lemma~\ref{lemma.true_skill_conditional_perceived}, we have that
\begin{equation*}
    q \mid \tilde q_{\full}, g \sim \N \left(\tilde q, \frac{1}{\sigma^{-2} + \sum_{k\in \full} \sigma_{gk}^{-2}} \right),  
\end{equation*}
while by \Cref{eq:tilde_q_cond_q_distribution},
\begin{equation*}
    \tilde q_{\subb}  \mid q, g \sim \N\left(
    \frac{\mu\sigma^{-2} + q\sum_{k\in \subb} \sigma_{gk}^{-2}}{\sigma^{-2} + \sum_{k\in \subb} \sigma^{-2}_{gk}},
    \frac{\sum_{k\in \subb} \sigma^{-2}_{gk}}{\left(\sigma^{-2}+\sum_{k\in \subb} \sigma^{-2}_{gk}\right)^2}
    \right).
\end{equation*}
Applying Lemma~\ref{lemma:prior_marginal_normal} gives us
\begin{equation*}
    \tilde q_{\subb} \mid \tilde q_{\full}, g \sim \N\left(
    \frac{\mu \sigma^{-2} + \tilde q_{\full} \sum_{k\in \subb} \sigma^{-2}_{gk} }{\sigma^{-2} + \sum_{k\in \subb} \sigma^{-2}_{gk}}, 
    \frac{\sum_{k\in \subb} \sigma^{-2}_{gk}}{(\sigma^{-2}+ \sum_{k\in \subb} \sigma^{-2}_{gk})^2} \left( 
    1 + \frac{\sum_{k\in \subb} \sigma^{-2}_{gk}}{\sigma^{-2} + \sum_{k\in \full} \sigma^{-2}_{gk}}
    \right)
    \right).
\end{equation*}

\medskip
\textit{Step 2: We show that} 
$$\tau_g = \frac{\pi_g \tilde \sigma_g }{C} \Phi\left(\frac{a_g+b_g\mu}{\sqrt{1 + \tilde \sigma_g ^2 b_g^2}}\right) - \frac{\pi_g \tilde \sigma_g }{C}  \Phi_2 \left( \frac{a_g+b_g\mu}{\sqrt{1 + \tilde \sigma_g ^2 b_g^2}}, \frac{\tilde q^*_{\full} -\mu}{\tilde \sigma_g}; -\frac{\tilde \sigma_g b_g}{\sqrt{1 + \tilde \sigma_g ^2 b_g^2}} \right),$$
\textit{where $\tilde \sigma_g, a_g, b_g$ are defined as above.}

 Given that the school's admission threshold is $\tilde q^*_{\full}$, only students with $\tilde q \geq \tilde q^*_{\full} $ get admitted. If no costs existed, the fraction of students who would get admitted under a fixed threshold $\tilde q^*_{\full}$ would be 
 $$\int_{\tilde q^*_{\full}}^{\infty} \phi\left(\frac{\tilde q - \mu}{\tilde \sigma_g}\right) \diff \tilde q,$$
 by Lemma~\ref{lemma:perceived_skill_distribution}.
 However, in the presence of costs, by Lemma~\ref{lemma.thres.strategy}, among all students who in our continuum model could have $\tilde q_{\full} > \tilde q^*_{\full}$,  only students with $\tilde q_{\subb} \geq \underline{q}^g_{\subb}$ apply. Conditional on having the same $\tilde q_{\full}$, Step 1 implies that the fraction of applying students from group $g$ equals 
 $$\Phi\left(\frac{   \frac{\mu \sigma^{-2} + \tilde q_{\full} \sum_{k\in \subb} \sigma^{-2}_{gk}}{\sigma^{-2} + \sum_{k\in \subb} \sigma^{-2}_{gk}} - \underline{q}^g_{\subb}}{\frac{\sqrt{\sum_{k\in \subb} \sigma^{-2}_{gk}}}{\sigma^{-2}+\sum_{k\in \subb} \sigma^{-2}_{gk}}\sqrt{
    1 + \frac{\sum_{k\in \subb} \sigma^{-2}_{gk}}{\sigma^{-2} + \sum_{k\in \full} \sigma^{-2}_{gk}}} }\right).$$
Consequently, putting everything together, we get that
\begin{equation}
\label{eq.diversity_barriers}
\tau_g = \frac{\pi_g }{C} \int_{\tilde q^*_{\full}}^{\infty}
\Phi\left(\frac{  \frac{\mu \sigma^{-2} + \tilde q \sum_{k\in \subb} \sigma^{-2}_{gk} }{\sigma^{-2} + \sum_{k\in \subb} \sigma^{-2}_{gk}} - \underline{q}^g_{\subb} }{\frac{\sqrt{\sum_{k\in \subb} \sigma^{-2}_{gk}}}{\sigma^{-2}+\sum_{k\in \subb} \sigma^{-2}_{gk}} \sqrt{ 
    1 + \frac{\sum_{k\in \subb} \sigma^{-2}_{gk}}{\sigma^{-2} +\sum_{k\in \full} \sigma^{-2}_{gk}}}}
    \right)
\phi\left(\frac{\tilde q - \mu}{\tilde \sigma_g}\right) \diff \tilde q.
\end{equation}

By Equation (10,010.1) in \cite{owen1980table}, we have that
\begin{equation*}
    \int_{-\infty}^{u} \Phi(a   +  bx) \phi\left(\frac{x-\mu}{\rho}\right) \diff x = \rho \Phi_2 \left(
    \frac{a + b \mu}{\sqrt{1 + \rho^2 b^2}}, \frac{u -  \mu}{\rho}; -\frac{\rho b}{\sqrt{1 + \rho^2 b^2}}\right).
\end{equation*}
Furthermore, by Equation (10,010.1) in \cite{owen1980table}, 
\begin{equation*}
    \int_{-\infty}^{\infty} \Phi(a   +  bx) \phi\left(\frac{x-\mu}{\rho}\right) \diff x = \rho \Phi\left(\frac{a+b\mu}{\sqrt{1 + \rho^2 b^2}}\right).
\end{equation*}
Substituting the definition of $\underline{q}^g_{\subb}$ from Lemma~\ref{lemma.thres.strategy}, plugging the definitions of $a_g$, $b_g$ and $\tilde\sigma_g$ into \Cref{eq.diversity_barriers} and using the two Owen's formulae above completes the current step.

\medskip\textit{Step 3: An immediate corollary  is that group $B$ is under-represented if and only if $\tau_B < \tau_A$, which by Step 2 is equivalent to \Cref{eq.diversity_condition_barriers}.}
\end{proof}

\medskip
\begin{proof}[\textit{Proof of Part (ii).}]
 The academic merit of admitted students from group $g$ equals
\begin{equation*}
\begin{split}
    \expec[q \mid Y =1, g, P_{\full}] & = \expec[ \tilde q_{\full} \mid \tilde q_{\full} \geq \tilde q^*_{\full}, \tilde q_{\subb} \geq \underline{q}^g_{\subb},  g, P_{\full}] \\
    &=     \frac{\pi_g}{\tau_g C} \int_{\tilde q^*_{\full}}^{\infty} \tilde q
\Phi\left(\frac{ \frac{\mu \sigma^{-2} + \tilde q \sum_{k\in \subb} \sigma^{-2}_{gk} }{\sigma^{-2} + \sum_{k\in \subb} \sigma^{-2}_{gk}} -\underline{q}^g_{\subb}  }{\frac{\sqrt{\sum_{k\in \subb} \sigma^{-2}_{gk}}}{\sigma^{-2}+\sum_{k\in \subb} \sigma^{-2}_{gk}} \sqrt{
    1 + \frac{\sum_{k\in \subb} \sigma^{-2}_{gk}}{\sigma^{-2} +\sum_{k\in \full} \sigma^{-2}_{gk}}
    }}\right)
\phi\left(\frac{\tilde q - \mu}{\tilde \sigma_g}\right) \diff \tilde q. 
\end{split}
\end{equation*}

By Equation (10,011.1) in \cite{owen1980table}, we get that
\begin{equation*}
\begin{split}
    \int x \Phi(a +bx) \phi\left(\frac{x-\mu}{\rho}\right) \diff x = &
    \frac{\rho^2 b}{\sqrt{1 +\rho^2 b^2}} \phi \left(\frac{a + b \mu}{\sqrt{1 +\rho^2 b^2}}\right) \Phi\left(x \sqrt{1 +\rho^2 b^2}  +  \frac{ b\rho (a+b \mu)} {\sqrt{1 +\rho^2 b^2}} \right) 
    \\
    &- \rho^2 \Phi(a + b \mu + b \rho x) \phi(x)
    + \mu \rho \int \Phi(a +bx) \phi\left(\frac{x-\mu}{\rho}\right) \diff x.
\end{split}
\end{equation*}
Observe that the last integral simplifies because
\begin{equation*}
    \frac{\mu \pi_g \tilde \sigma_g}{\tau_g  C} \int_{\tilde q^*_{\full}}^\infty \Phi(a_g +b_g \tilde q) \phi\left(\frac{\tilde q-\mu}{\tilde \sigma_g}\right) \diff \tilde q = \frac{\mu \tau_g }{\tau_g} = \mu.
\end{equation*}
For the first and second term, we find that
\begin{equation*}
\begin{split}
     -\frac{\tilde \sigma_g^2 b_g}{\sqrt{1 +\tilde \sigma_g^2 b_g^2}} \phi \left(\frac{a_g + b_g \mu}{\sqrt{1 +\tilde \sigma_g^2 b_g^2}}\right)
     \Phi\left(\tilde q^*_{\full} \sqrt{1 +\tilde \sigma_g^2 b_g^2}  +\frac{ b_g\tilde \sigma_g (a_g+b_g \mu)} {\sqrt{1 +\tilde \sigma_g^2 b_g^2}} \right) +  \tilde \sigma_g ^2 \Phi(a_g + b_g \mu + b_g \tilde \sigma_g \tilde q^*_{\full}) \phi(\tilde q^*_{\full}).
\end{split}
\end{equation*}
Putting everything together, we get that
\begin{equation*}
\begin{split}
    \expec[q \mid Y =1, g, P_{\full}]  = & \mu -   
     \frac{\tilde \sigma_g^2 b_g}{\tau_g \sqrt{1 +\tilde \sigma_g^2 b_g^2}} \phi \left(\frac{a_g + b_g \mu}{\sqrt{1 +\tilde \sigma_g^2 b_g^2}}\right) 
     \Phi\left(\tilde q^*_{\full} \sqrt{1 +\tilde \sigma_g^2 b_g^2}  + \frac{ b_g\tilde \sigma_g (a_g+b_g \mu)} {\sqrt{1 +\tilde \sigma_g^2 b_g^2}} \right) \\
     &+  \Phi(a_g + b_g \mu + b_g \tilde \sigma_g \tilde q^*_{\full}) \frac{ \tilde \sigma_g^2  \phi(\tilde q^*_{\full})}{\tau_g}\\
     = & \lambda(
     a_g, b_g, \tilde \sigma_g, \tau_g).
\end{split}
\end{equation*}
Requiring that $\expec[q \mid Y =1, A, P_{\full}]  > \expec[q \mid Y =1, B, P_{\full}] $ concludes the proof.
\end{proof}

Because testing is costly, admissions outcomes reflect both the informativeness of the test and other $K-1$ features and the cost-to-valuation ratio $c_g/v$. Low diversity can occur either because group $B$ students self-select out of the test at higher rates (due to higher costs), or are admitted at lower rates even if they apply similarly (due to low feature informativeness). On the other hand, unlike exogenous barriers, student incentives can improve outcomes: higher-skilled students in both groups are more likely to take the test and apply (see \Cref{fig:calibrated_sims_theop_test_and_hsclassrank}). %

Overall, Figure~\ref{fig:strategic_varying_costs_and_feature_var} shows how test informativeness and test costs interact to determine academic merit, diversity, and individual fairness. When test costs are high for group  $B$, both academic merit and diversity decline—an effect that is amplified when the test is \emph{more informative} (lower conditional variance), up to a point. In such cases, more group $B$ students self-select out of testing, exacerbating these outcomes. Intuitively, when feature informativeness is increased, group $B$ students near the previous decision boundary have a lower admissions probability, because they can no longer can ``get lucky'' with a higher test score.

\subsection{Two schools (Proofs from Section~\ref{sec:twoschoolstrategicbehavior} and \ref{ssec:two_schools_strategic_theory})}
\label{app:strategic_two_schools}

\noindent\textbf{Student decisions.} In a two-school setting with policies $\mathbf{P} = (P^1, P^2)$, students' decisions to take the test and thus apply to test-requiring schools are determined per case as follows:
\begin{equation}
\label{eq:alpha_twoschools_definition}
    \alpha(\thetaset_{\subb}, g; \mathbf{P})
=\arg\max_{\alpha \in \{0,1\}} h (\alpha, \thetaset_{\subb}, g; \mathbf{P}),
\end{equation}
where
\begin{small}
\begin{equation*}
h (\alpha, \thetaset_{\subb}, g; \mathbf{P})
= 
\begin{cases}
\displaystyle 
\alpha\big( v_1 \Pr(Y_1=1 \mid \thetaset_{\subb}, g, P^1_{\full}) - c_g \big)
+ v_2 \Pr(Y_1=0 \cap Y_2=1 \mid \thetaset_{\subb}, g, P^2_{\subb}),
& \mathbf{P}=(P^1_{\full}, P^2_{\subb}),
\\[1.1em]
\displaystyle
\alpha\big( 
v_1 \Pr(Y_1=1 \mid \thetaset_{\subb}, g, P^1_{\full})
+ v_2 \Pr(Y_1=0 \cap  Y_2=1 \mid \thetaset_{\subb}, g, P^2_{\full})
- c_g
\big),
& \mathbf{P}=(P^1_{\full}, P^2_{\full}),
\\[1.1em]
\displaystyle
v_1 \Pr(Y_1=1 \mid \thetaset_{\subb}, g, P^1_{\subb})
+ \alpha \big( 
v_2 \Pr(Y_1=0 \cap  Y_2=1 \mid \thetaset_{\subb}, g, P^2_{\full})
- c_g
\big),
& \mathbf{P}=(P^1_{\subb}, P^2_{\full}),
\\[0.9em]
0, 
& \mathbf{P}=(P^1_{\subb}, P^2_{\subb}).
\end{cases}
\end{equation*}
\end{small}

\smallskip
 \noindent\textbf{Schools' selection policies.} Recall that $Y_i (\tilde q_{S_i}; \mathbf{P})$ denotes the selection policy of school $J_i$. For brevity, we also define the indicator function $$\chi_i(\thetaset_{\subb}, g, \mathbf{P}) =  1+  (\alpha^*(\thetaset_{\subb},g; \mathbf{P}) -1) \cdot \indic\{P^i=P_{\full}\} ), $$
which takes value 1 in two cases: either when school $J_i$ does not require the test ($P^i_{\subb}$) or school $J_i$ requires the test and a student in group $g$ with features $\thetaset_{\subb}$ takes the test ($P^i_{\full}$). 

At equilibrium, given the student preference for $J_1$ over $J_2$, the more preferred school, $J_1$, picks students first. 
In particular, school $J_1$ optimizes  
the academic merit of its admitted class as follows:
\begin{equation}
\begin{split}
    \label{eq:twoschoolproblem_1}
        &\max_{Y_1} \sum_g \pi_g \expec_{\thetaset_{S_1}}[ \tilde q (\theta_{S_1}, g)  \cdot \chi_1(\thetaset_{\subb}, g, \mathbf{P}) \cdot   Y_1(\tilde q (\thetaset_{S_1}, g); \mathbf{P}) \mid g, P^1] \\
    &\mathrm{~s.t.~}  \sum_g \pi_g\expec_{\thetaset_{S_1}}[ \chi_1(\thetaset_{\subb}, g, \mathbf{P})   \cdot Y_1(\tilde q (\thetaset_{S_1}, g); \mathbf{P}) \mid g, P^1 ] \leq C_1.
\end{split}
\end{equation}
Similarly, $J_2$ optimizes academic merit by selecting among the students who either did not apply to $J_1$ at all (if $J_1$ requires the test) or applied but did not get admitted, i.e.,
\small
\begin{equation}
\begin{split}
    \label{eq:twoschoolproblem_2}
        &\max_{Y_2} \sum_g \pi_g \expec_{\thetaset_{S_2}}[ \tilde q (\theta_{S_2}, g) \cdot \chi_2(\thetaset_{\subb}, g, \mathbf{P}) \cdot Y_2(\tilde q (\thetaset_{\subb}, g); \mathbf{P})  \mid \chi_1(\thetaset_{\subb}, g, \mathbf{P})    \cdot Y_1(\tilde q (\thetaset_{S_i}, g); \mathbf{P})=0, g, P^2] \\
    &\mathrm{~s.t.~}  \sum_g \pi_g \expec_{\thetaset_{S_2}}[\chi_2(\thetaset_{\subb}, g, \mathbf{P}) \cdot Y_2(\tilde q (\thetaset_{\subb}, g); \mathbf{P})  \mid \chi_1(\thetaset_{\subb}, g, \mathbf{P})    \cdot Y_1(\tilde q (\thetaset_{S_i}, g); \mathbf{P})=0, g, P^2] \leq C_2.
\end{split}
\end{equation}
\normalsize

\smallskip\noindent\textbf{Two-school equilibria.} Given testing policies $\mathbf{P}$ and capacities $C_i$, we say that a triple $(\alpha^*, \mathbf{Y}^*, \mathbf{P})$ constitutes an \textit{equilibrium} if: 
(i) for all  $\thetaset_{\subb} \in \mathbb{R}^{K-1}$ and $g \in \{A,B\}$, 
$ \alpha^*(\thetaset_{\subb}, g; \mathbf{P})
=\arg\max_{\alpha \in \{0,1\}} h (\alpha, \thetaset_{\subb}, g; \mathbf{P});$
and  (ii) for all  $\thetaset_{\full} \in \mathbb{R}^{K}$ and $g \in \{A,B\}$, $Y_i^*(\thetaset_{S_i}, g; \mathbf{P}) = \indic\{\tilde q (\thetaset_{S_i}, g) \geq \tilde q^*_{i, S_i}\}$, where  
$\tilde q^*_{i, S_i}$  is the corresponding solutions to \Cref{eq:twoschoolproblem_1} and \Cref{eq:twoschoolproblem_2}.

Note that, as in the single school case, we can focus on admission strategies of the form $Y^*_i (\tilde q_{S_i}; \mathbf{P})$.
In Lemma~\ref{lemma:twoschoools_thresholdpolicies} below we  formalize that each such $Y_i$ preserves its threshold-based form.

\normalsize

\begin{lemma}
\label{lemma:twoschoools_thresholdpolicies}
At an equilibrium $(\alpha^*, \mathbf{Y}^*)$, each school $J_i$'s selection  policy $Y^*_i (\tilde q_{S_i}; \mathbf{P})$, $i \in \{1,2\}$, takes a threshold form, i.e., there exists a threshold $\tilde q^*_{i, S_i}$ such that $Y^*_i (\tilde q_{S_i}; \mathbf{P}) = \indic\{\tilde q_S \geq \tilde q^*_{i, {S_i}}\}$ where  $\tilde q^*_{1, S_1}$, $\tilde q^*_{2,S_2}$, are the solutions to \Cref{eq:twoschoolproblem_1} and \Cref{eq:twoschoolproblem_2}, respectively.

\end{lemma}

\begin{proof}{\textit{Proof.}}
We provide the proof for $\mathbf{P}= (P^1_{\full}, P^2_{\subb})$; the remaining cases are analogous. 
     Since $v_1>v_2$ and school $J_2$ uses $P^2_{\subb}$, all students who apply to $J_1$ also apply to $J_2$ but not vice versa. All students have incentive to apply to $J_2$. 

    We begin with school $J_1$. Note that every student with $\alpha^* (\tilde q_{\subb}, g) =1$ admitted to $J_1$ will accept the offer since $v_1 > v_2$. 
    Therefore $J_1$ can pick any student as long as the student has applied to $J_1$. Let $G_1$ denote the CDF of all students with skill estimate  $\tilde q_{\full}$ who apply to school $J_1$ at equilibrium.  
    
    We show that $J_1$ admits the top mass $C_1^{-}\triangleq \min\{C_1, \int_{- \infty}^{\infty} \sum_g \pi_g a^*(\tilde q_{\subb},g; \mathbf{P}) \diff \tilde q_{\subb} \}$  with the highest skill estimates $\tilde q_{\full}$. I.e., $Y_1^*(\tilde q_{\full}; \mathbf{P}) = \indic\{\tilde q_{\full} \geq \tilde q^*_{1,\full}\}$, where $\tilde q^*_{1, \full}$ satisfies 
    $$1-G_1(\tilde q^*_{1,\full}) = C_1^{-}.$$

    First note that any other threshold-based policy is infeasible or suboptimal. This is because $Y^*_1$ either admits all applicants (in the case where $C_1^{-} < C_1$) or the capacity constraint in \Cref{eq:twoschoolproblem_1} binds.
    Next, consider any feasible selection policy $Y_1$ and 
    observe that under any $Y_1$ the academic merit objective in \Cref{eq:twoschoolproblem_1} can be written as
    $$\int_{-\infty}^{\infty} \tilde q_{\full} Y_1(\tilde q_{\full}; \mathbf{P}) \diff G_1(\tilde q_{\full}) = \int_{0}^{1} G_1^{-1} (s)  Y_1(G_1^{-1} (s); \mathbf{P}) \diff s, $$
    which is trivially convex in $Y_1$ and supermodular.
    By the threshold form of $Y^*_1$, $Y^*_1$ weakly majorizes any other feasible selection policy $Y_1$.
    Thus, by the Fan-Lorentz inequality \citep{fan1954integral}, it follows that
    \begin{equation*}
    \begin{split}
        \int_{-\infty}^{\infty} \tilde q_{\full} Y_1(\tilde q_{\full}; \mathbf{P}) \diff G_1(\tilde q_{\full}) &= \int_{0}^{1} G_1^{-1} (s)  Y_1(G_1^{-1} (s); \mathbf{P}) \diff s \\&\leq
        \int_{0}^{1} G_1^{-1} (s)  Y^*_1(G_1^{-1} (s); \mathbf{P}) \diff s \\& =  \int_{-\infty}^{\infty} \tilde q_{\full} Y^*_1(\tilde q_{\full}; \mathbf{P}) \diff G_1(\tilde q_{\full}),
      \end{split}
    \end{equation*}
    thus $Y_1^*$ is optimal.
\end{proof}

\begin{proposition}[Proposition~\ref{thm.twoschools}]
\label{thm.twoschools.full}
  Consider the setting with two schools defined above with testing policies $\mathbf{P}=(P^1_{\full}, P^2_{\subb})$. Then, there exists a unique equilibrium $(\alpha^*, \mathbf{Y}^*)$ with the following properties:
  \begin{itemize}
      \item [(i)] School $J_i$'s selection policy $Y^*_i$ takes a threshold form: $Y^*_i (\tilde q_{S_i}; \mathbf{P})= \indic\{\tilde q_{S_i} \geq \tilde q^*_{i, S_i}\}$ where $\tilde q^*_{1, \full}$, $\tilde q^*_{2, \subb}$, are the solutions to \eqref{eq:twoschoolproblem_1}, \eqref{eq:twoschoolproblem_2}, respectively.
      \item[(ii)] Students in group $g$ take the test and apply to school $J_1$,  if and only if one of the following conditions holds:
      \begin{itemize}
          \item [1)]  either  $\tilde q^*_{2 ,\subb} > \tilde q_{\subb} \geq \underline{q}_{l}^g $ where
            \begin{equation}
            \label{eq.q_l_thres}
    \underline{q}_{l}^g = \tilde q_{1,\full}^* - \Phi^{-1} \left(1 - \frac{c_g}{v_1}\right) \left ( \frac{\sigma_{gK}^{-2}}{\sigma^{-2} + \sum_{k \in \full} \sigma_{gk}^{-2}}\right)
\sqrt{\sigma^{2}_{gK} + \frac{1}{\sigma^{-2} + \sum_{k \in \subb} \sigma_{gk}^{-2}} }; 
\end{equation}

\item[2)] or $\tilde q_{\subb} \geq \max\{\underline{q}_{h}^g , \tilde q^*_{2,\subb}\}$, where
\begin{equation}
\label{eq.q_h_thres}
    \underline{q}_{h}^g = \tilde q_{1, \full}^* - \Phi^{-1} \left(1 - \frac{c_g}{v_1-v_2}\right) \left ( \frac{\sigma_{gK}^{-2}}{\sigma^{-2} + \sum_{k \in \full} \sigma_{gk}^{-2}}\right)
\sqrt{\sigma^{2}_{gK} + \frac{1}{\sigma^{-2} + \sum_{k \in \subb} \sigma_{gk}^{-2}} }.
\end{equation}
      \end{itemize}
     
     Furthermore, $\underline{q}_{l}^g < \underline{q}_{h}^g $ for both groups $g \in \{A,B\}$.

  \item [(iii)] 
  The fraction of students in group $g$ who have $\tilde q_{\subb} > q_a > \underline{q}_l^g$ and get admitted to school $J_1$ equals
     $$ D_g (\hat{a}_g (q_a)) \triangleq \pi_g \tilde\sigma_g \Phi \left( \frac{\hat{a}_g (q_a)+ b_g \mu }{\sqrt{1 + \tilde \sigma_g ^2 b_g^2}}\right) -  \pi_g \tilde\sigma_g \Phi_2 \left( \frac{\hat{a}_g (q_a)+b_g\mu}{\sqrt{1 + \tilde \sigma_g ^2 b_g^2}}, \frac{\tilde q^*_{1, \full} -\mu}{\tilde \sigma_g}; -  \frac{\tilde \sigma_g b_g}{\sqrt{1 + \tilde \sigma_g ^2 b_g^2}} \right), $$
    where $b_B $ and $\tilde \sigma_B$ are defined as in Proposition~\ref{prop.div_merit_strategic}, and 
\begin{equation*}
    \hat{a}_g (q_a) = \frac{\mu \sigma^{-2} -q_a }{\frac{\sqrt{\sum_{k \in \subb} \sigma_{gk}^{-2}}}{\sigma^{-2} + \sum_{k \in \subb} \sigma_{gk}^{-2}} \sqrt{1 + \frac{ \sum_{k \in \subb} \sigma_{gk}^{-2}}{\sigma^{-2} + \sum_{k \in \full} \sigma_{gk}^{-2}}} }.
\end{equation*}
    
  Conditional on  $\tilde q^*_{2, \subb} > \underline{q}^B_h$, school $J_1$ is more diverse than $J_2$ if and only if
  \small
  \begin{equation*}
  \begin{split}
  \frac{ D_B (\hat{\alpha}_B(\underline{q}_l^B)) }{\Phi \left ((\mu - \tilde q^*_{2,\subb} )/  \sigma \sqrt{\frac{\sum_{k \in \subb} \sigma_{Bk}^{-2}}{\sigma^{-2} + \sum_{k \in \subb} \sigma_{Bk}^{-2}}}\right) - D_B (\hat{\alpha}_B(\tilde q_{2, \subb}^*))} > \frac{C_1} {C_2}.
   \end{split}
  \end{equation*}
  \normalsize
Otherwise, school $J_1$ is more diverse than $J_2$ if and only if
\begin{equation*}
\begin{split}
\frac{ D_B (\hat{a}_B (\underline{q}_l^B)) - 
     D_B (\hat{a}_B (\underline{q}_h^B)) +  D_B (\hat{a}_B (\tilde q^*_{2, \subb})) + \tilde \sigma_B \left( \Phi\left(\frac{\hat{a}_B (\underline{q}_h^B) +b_B\mu}{\sqrt{1 + \tilde \sigma_B ^2 b_B^2}}\right) -  \Phi\left(\frac{\hat{a}_B (\tilde q^*_{2, \subb}) +b_B\mu}{\sqrt{1 + \tilde \sigma_B ^2 b_B^2}}\right) \right)}{\Phi\left ((\mu- \tilde q^*_{2, \subb} )/  \sigma \sqrt{\frac{\sum_{k \in \subb} \sigma_{Bk}^{-2}}{\sigma^{-2} + \sum_{k \in \subb} \sigma_{Bk}^{-2}}}\right)  - D_B (\hat{\alpha}_B (\underline{q}_h^B))} > \frac{C_1}{C_2}.
\end{split}
\end{equation*}

  \item [(iv)] There exist instances of the model parameters such that school $J_1$ achieves lower academic merit for group $g$ than $J_2$.
  In particular, assume that $\tilde q^*_{2, \subb} > \underline{q}_h^g$. Then, $J_1$ achieves lower academic merit for group $g$ than $J_2$ if and only if 
  \begin{equation*}
  \lambda (a_g (\tilde q^*_{1, \full}), b_g, \tilde \sigma_g, \tau_g)   <  \kappa (a_g (\tilde q^*_{2, \subb}), b_g, \tilde \sigma_g, \tau_g), 
  \end{equation*}
  where 
\begin{equation*} 
\begin{split}
    \kappa ( a_g' (q), b'_g, \tilde \sigma_g, \tau_g) \triangleq &\frac{\tilde \sigma_g^2 b'_g}{\tau_g\sqrt{1 + \tilde \sigma_g^2 (b'_g)^2}} \phi \left( \frac{a'_g(q) +{b'_g} \mu}{\sqrt{1+ \tilde \sigma_g^2 (b'_g)^2}}\right) \Phi \left ( 
\tilde q^*_{1, \full} \sqrt{1 + \tilde \sigma_g ^2 (b'_g)^2} + \frac{b'_g \tilde \sigma_g (a'_g(q) + b'_g \mu)}{\sqrt{1+\tilde \sigma_g^2 (b'_g)^2}}
\right)\\
&- \Phi\left (a'_g(q) + b'_g \mu + b'_g \tilde \sigma_g \tilde q^*_{1, \full}\right) \frac{\tilde \sigma_g^2 \phi(q)}{\tau_g} + \frac{\mu (1-\tau_g)}{ \tau_g},
\end{split}
\end{equation*}
\begin{equation*}
    a'_g (q) = \frac{\mu\sigma^{-2} - q\left( \sigma^{-2} + \sum_{k \in \subb} \sigma_{gk}^{-2}  \right) }{\sqrt{ \sum_{k \in \subb} \sigma_{gk}^{-2} }},  b'_g  = \sqrt{ \sum_{k \in \subb} \sigma_{gk}^{-2} }.
\end{equation*}
\normalsize  
  \end{itemize}
\end{proposition}

\begin{proof}{\textit{Proof of Part (i).}} The result was already proved in Lemma~\ref{lemma:twoschoools_thresholdpolicies}.
\end{proof}

\begin{proof}{\textit{Proof of Part (ii).}} 
At equilibrium, all students apply to the test-free school $J_2$. By Part (i), only students with $\tilde q_{\subb} > \tilde q^*_{2, \subb}$ get accepted. Thus, we have two separate cases:
\begin{itemize}
    \item [--] \textit{Students who get rejected by $J_2$}:  Students in group $g$  with $\tilde q_{\subb} < \tilde q^*_{2, \subb}$ decide to take the test (and apply to $J_1$) if and only if
    $$v_1 \prob (\tilde q_{\full} \geq \tilde q_{1, \full}^* \mid \tilde q_{\subb}, g, P^1_{\full}) - c_g \geq 0,$$
    thus the problem reduces to the single-school setting. By Lemma~\ref{lemma.thres.strategy}, the above condition translates to $\tilde q^*_{2, \subb} > \tilde q_{\subb} \geq \underline{q}^g_l$, where \Cref{eq.q_l_thres} follows analogously to  \Cref{eq.decision_threshold_taking_test} for $v=v_1$.

    \item [--]  \textit{Students who get accepted to $J_2$}:  Students in group $g$ with $\tilde q_{\subb} \geq \tilde q^*_{2, \subb}$ decide to take the test if and only if 
    \begin{equation}
    \label{eq.student.decision.test.fullfull}
    \begin{split}
        &v_1 \prob (\tilde q_{\full} \geq \tilde q_{1, \full}^* \mid \tilde q_{\subb}, g, P^1_{\full}) + v_2 \prob (\tilde q_{\full} <\tilde q_{1, \full}^* \mid \tilde q_{\subb}, g, P^1_{\full}) - c_g \geq v_2\\
        \Leftrightarrow & \prob (\tilde q_{\full} \geq \tilde q_{1, \full}^* \mid \tilde q_{\subb}, g, P^1_{\full}) \geq  \frac{c_g}{v_1-v_2}\\
        \Leftrightarrow & \tilde q_{\subb} \geq \underline{q}^g_h,
        \end{split}
    \end{equation}
    where $\underline{q}^g_h$ in \Cref{eq.q_h_thres} follows similarly to \Cref{eq.decision_threshold_taking_test} by replacing $v$ with $v_1-v_2$. 
\end{itemize}

The property that $\underline{q}_{l}^g < \underline{q}_{h}^g $, $g \in \{A,B\}$, follows directly from comparing \Cref{eq.q_h_thres} to \Cref{eq.q_l_thres} and using that $0<v_1 - v_2< v_1$.

Finally, note that the equilibrium described by Parts (i) and (ii) is unique. This follows using arguments similar to Lemma~\ref{lemma:reduced_equilibria_single}.
\end{proof}

\begin{proof}{\textit{Proof of Part (iii).}}
 Part (ii), together with the assumption that $\tilde q_{2, \subb}^* > \underline{q}^B_h$, implies that students in $B$ apply to $J_1$ if and only if $\tilde q_{\subb}  \geq \underline{q}^B_l$.
Thus, we can apply 
Step 2 in Part (i) of Proposition~\ref{prop.div_merit_strategic} and find that diversity at school $J_1$ equals
\begin{equation}
\label{eq.div.J1.twoschool.strategic}
\begin{split}
    \tau^1_B &= \frac{\pi_B \tilde \sigma_B }{C_1} \Phi\left(\frac{a_B (\tilde q^*_{1, \full}) +b_B\mu}{\sqrt{1 + \tilde \sigma_B ^2 b_B^2}}\right) - \frac{\pi_B \tilde \sigma_B }{C_1} \Phi_2 \left( \frac{a_B (\tilde q^*_{1, \full})+b_B\mu}{\sqrt{1 + \tilde \sigma_B ^2 b_B^2}}, \frac{\tilde q^*_{1,\full} -\mu}{\tilde \sigma_B}; -\frac{\tilde \sigma_B b_B}{\sqrt{1 + \tilde \sigma_B ^2 b_B^2}} \right)\\
    & = \frac{\pi_B \tilde \sigma_B }{C_1} \Phi\left( \frac{\hat{a}_B ( \underline{q}_l^B) +b_B\mu}{\sqrt{1 + \tilde \sigma_B ^2 b_B^2}}\right) - \frac{\pi_B \tilde \sigma_B }{C_1} \Phi_2 \left( \frac{\hat{a}_B ( \underline{q}_l^B) +b_B\mu}{\sqrt{1 + \tilde \sigma_B ^2 b_B^2}}, \frac{\tilde q^*_{1,\full} -\mu}{\tilde \sigma_B}; -\frac{\tilde \sigma_B b_B}{\sqrt{1 + \tilde \sigma_B ^2 b_B^2}} \right) \\
    & = D_B  (\hat{a}_B ( \underline{q}_l^B) )
\end{split}
\end{equation}

 Next we find the diversity level $\tau^2_B$ at $J_2$. The total mass of students from group $B$ who are eligible for acceptance at $J_2$ is 
 $\pi_B\Phi\left ((\mu - \tilde q^*_{2, \subb} )/  \sigma \sqrt{\frac{\sum_{k \in \subb} \sigma_{Bk}^{-2}}{\sigma^{-2} + \sum_{k \in \subb} \sigma_{Bk}^{-2}}}\right)$. However, only students who do not get admitted to $J_1$ actually enroll in $J_2.$
  Thus,
 
  \begin{equation*}
     \tau_B^2 = \frac{ \pi_B}{C_2} \left (\Phi\left ((\mu - \tilde q^*_{2, \subb} )/  \sigma \sqrt{\frac{\sum_{k \in \subb} \sigma_{Bk}^{-2}}{\sigma^{-2} + \sum_{k \in \subb} \sigma_{Bk}^{-2}}}\right)  \right) 
     \end{equation*}
 Requiring that $\tau^1_B > \tau_B^2$ gives us 
the condition in the statement and thus concludes the proof for the case $\tilde q^*_{2, \subb} > \underline{q}_h^B$.

For the complementary case where $\underline{q}_h^B \geq \tilde q^*_{2, \subb} > \underline{q}_l^B$, we need to take into account that students with $\tilde q_{\subb} \in [\tilde q^*_{2, \subb}, \underline{q}_h^B)$ apply and get admitted only to $J_2$. Formally, applying Equation (10,010.1) in \cite{owen1980table}, we find that
\begin{equation*}
\begin{split}
    \tau_B^1 &= \frac{\pi_B}{C_1} \left(\tilde \sigma_B  \Phi\left(\frac{\hat{a}_B (\underline{q}_l^B) +b_B\mu}{\sqrt{1 + \tilde \sigma_B ^2 b_B^2}}\right) - \tilde \sigma_B \Phi_2 \left( \frac{\hat{a}_B(\underline{q}_l^B)+b_B\mu}{\sqrt{1 + \tilde \sigma_B ^2 b_B^2}}, \frac{\tilde q^*_{1, \full} -\mu}{\tilde \sigma_B}; -\frac{\tilde \sigma_B b_B}{\sqrt{1 + \tilde \sigma_B ^2 b_B^2}} \right) \right. \\
    & + \tilde \sigma_B \Phi_2 \left( \frac{\hat{a}_B(\underline{q}_h^B)+b_B\mu}{\sqrt{1 + \tilde \sigma_B ^2 b_B^2}}, \frac{\tilde q^*_{1, \full} -\mu}{\tilde \sigma_B}; -\frac{\tilde \sigma_B b_B}{\sqrt{1 + \tilde \sigma_B ^2 b_B^2}} \right) \\
    & \left. - \tilde \sigma_B \Phi_2 \left( \frac{\hat{a}_B(\tilde q^*_{2, \subb})+b_B\mu}{\sqrt{1 + \tilde \sigma_B ^2 b_B^2}}, \frac{\tilde q^*_{1, \full} -\mu}{\tilde \sigma_B}; -\frac{\tilde \sigma_B b_B}{\sqrt{1 + \tilde \sigma_B ^2 b_B^2}} \right)\right)\\
    & = \frac{1}{C_1} \left ( D_B (\hat{a}_B (\underline{q}_l^B)) - 
     D_B (\hat{a}_B (\underline{q}_h^B)) +  D_B (\hat{a}_B (\tilde q^*_{2, \subb}))\right) \\&+ \frac{\tilde \sigma_B}{C_1} \left( \Phi\left(\frac{\hat{a}_B (\underline{q}_h^B) +b_B\mu}{\sqrt{1 + \tilde \sigma_B ^2 b_B^2}}\right) -  \Phi\left(\frac{\hat{a}_B (\tilde q^*_{2, \subb}) +b_B\mu}{\sqrt{1 + \tilde \sigma_B ^2 b_B^2}}\right) \right).
\end{split}
\end{equation*}
For school $J_2$, we have that
\begin{equation*}
\begin{split}
    \tau_B^2 = &\frac{1}{C_2} \left(\pi_B \Phi\left ((\mu- \tilde q^*_{2, \subb} )/  \sigma \sqrt{\frac{\sum_{k \in \subb} \sigma_{Bk}^{-2}}{\sigma^{-2} + \sum_{k \in \subb} \sigma_{Bk}^{-2}}}\right) - \tilde \sigma_B  \Phi\left(\frac{a_B (\underline{q}_h^B) +b_B\mu}{\sqrt{1 + \tilde \sigma_B ^2 b_B^2}}\right) \right. \\
    & \left. + \tilde \sigma_B \Phi_2 \left( \frac{a_B(\underline{q}_h^B)+b_B\mu}{\sqrt{1 + \tilde \sigma_B ^2 b_B^2}}, \frac{\tilde q^*_{1,\full} -\mu}{\tilde \sigma_B}; -\frac{\tilde \sigma_B b_B}{\sqrt{1 + \tilde \sigma_B ^2 b_B^2}} \right) \right)\\
    = & \frac{\pi_B}{C_2} \left(\Phi\left ((\mu- \tilde q^*_{2, \subb} )/  \sigma \sqrt{\frac{\sum_{k \in \subb} \sigma_{Bk}^{-2}}{\sigma^{-2} + \sum_{k \in \subb} \sigma_{Bk}^{-2}}}\right)  - D_B (\hat{\alpha}_B (\underline{q}_h^B))\right)
\end{split}
\end{equation*}
where the last two terms correspond to the students with $\tilde q_{\subb} \geq \underline{q}_h^B$ that were admitted by $J_1$.
Requiring that $\tau_B^1> \tau_B^2$ gives the result.
\end{proof}

\begin{proof}{\textit{Proof of Part (iv).}}
As in Part (iii), if $\tilde q^*_{2, \subb} > \underline{q}^g_h$, then  students in $g$ apply to $J_1$ if and only if $\tilde q_{\subb} > \underline{q}_l^g$.
However, only a fraction of them (equal to $C_1$) will get admitted. In the right panel of Figure~\ref{fig.merit.twoschools.strategic}, this mass of admitted students is depicted in yellow.
From Proposition~\ref{prop.div_merit_strategic}, it follows that the academic merit of group $g$ in the admitted class at $J_1$ is
$$  \lambda (a_g, b_g, \tilde \sigma_g, \tau_g).$$

The academic merit of the admitted class in $J_2$ equals the expected skill of the students with $\tilde q_{\subb}>\tilde q^*_{2, \subb}$ who do not get admitted to $J_1$ (this is depicted in purple in Figure~\ref{fig.merit.twoschools.strategic}). Mathematically, we can find the merit of the admitted class to $J_2$ as  the expected skill of all the students with $\tilde q_{\subb}>\tilde  q^*_{2, \subb}$ who have $\tilde q_{\full} < \tilde q^*_1$. Similarly to the proof of Part (ii) in Proposition~\ref{prop.div_merit_strategic}, we find that
\begin{equation*}
    \expec[q \mid Y_2 =1, Y_1 =0, g, P^2_{\subb}] =      \frac{1}{\tau_g} \int_{-\infty}^{\tilde q^*_{1, \full}} \tilde q
\Phi\left(\frac{ \frac{\mu \sigma^{-2} + \tilde q \sum_{k\in \subb} \sigma^{-2}_{gk} }{\sigma^{-2} + \sum_{k\in \subb} \sigma^{-2}_{gk}} -\tilde q^*_{2, \subb} }{\frac{\sqrt{\sum_{k\in \subb} \sigma^{-2}_{gk}}}{\sigma^{-2}+\sum_{k\in \subb} \sigma^{-2}_{gk}}}\right)
\phi\left(\frac{\tilde q - \mu}{\tilde \sigma_g}\right) \diff \tilde q,
\end{equation*}
which, using Equation (10,011.1) in \cite{owen1980table}, simplifies to $\kappa $ as given in the statement of the proposition.
\end{proof}

\begin{proposition}[Dropping tests with strategic students: Academic merit]
\label{thm.droptest.strategic.2school}
    Consider two schools, $J_1$ and $J_2$, both of which initially follow test-based policies $\mathbf{P}=(P_{\full}, P_{\full})$. 
    When schools optimize for academic merit only, the following statements hold: 
    \begin{itemize}
        \item [(i)]  
         Let $\tilde q^*_{1, \subb}$ be the solution to \eqref{eq:threshold_groupaware_noAA_Phi} for $P^1 = P_{\subb}$. School $J_1$ drops the test if and only if 
 \eqref{eq.condition.drop.twoschools.strategic.J1} in Theorem~\ref{thm.twoschool.strategic.equilibria} holds.
        Furthermore, conditional on $J_1$ dropping the test, then $J_2$ keeps the test if and only if 
        \eqref{eq.sub.full.J2.cond1} and \eqref{eq.sub.full.J2.cond2} in Theorem~\ref{thm.twoschool.strategic.equilibria} hold.

        \item[(ii)]  
        Let  $\tilde q^*_{2, \subb}$ be the admission threshold of $J_2$ under $\mathbf{P}=(P_{\full}, P_{\subb})$ as in \Cref{thm.twoschools}.
        School $J_2$ drops the test, while school $J_1$ keeps the test, if and only if 
    \eqref{eq.condition.drop.twoschools.strategic.J2.1}
   and  \eqref{eq.condition.drop.twoschools.strategic.J2.2}
       in Theorem~\ref{thm.twoschool.strategic.equilibria} hold.
       
       \item[(iii)] There exist functions $\underline{c}_A, \underline{c}_B: \mathbb{R}_+ \rightarrow \mathbb{R}_+$ such that neither school wants to drop the test if and only if $c_A \leq \underline{c}_A (c_B)$ and $c_B \leq \underline{c}_B (c_A)$.  
    \end{itemize}
\end{proposition}

\begin{proof}{\textit{Proof of part (i).}}
Under $\mathbf{P}=(P_{\full}, P_{\full})$,  students must take the test to be eligible to apply to schools
$J_1$ and $J_2$. A student with skill estimate $\tilde q_{\subb}$ takes the test if and only if
\begin{equation*}
        v_1 \prob ( \tilde q_{\full} \geq \tilde q^*_{1, \full} \mid \tilde q_{\subb}, g, \mathbf{P}) + v_2  \prob (\tilde q^*_{1, \full}  > \tilde q_{\full} \geq \tilde q^*_{2, \full} \mid \tilde q_{\subb}, g, \mathbf{P}) -c_g \geq 0.
\end{equation*}
Recall from \eqref{eq:q_tilde_cond_thetasubb_q} that
$\tilde q_{\full} \mid \tilde q_{\subb}, g, P_{\full} \sim \mathcal{N} ( \tilde q_{\subb}, \tilde \rho_g^2)$. Therefore, the above inequality becomes
\begin{equation}
\label{eq.testtaking.full.full}
    v_1 - c_g + (v_2 - v_1) \Phi\left( \frac{\tilde q_{1, \full}^* - \tilde q_{\subb}}{\tilde \rho_g}\right)
- v_2 \Phi \left( \frac{\tilde q_{2, \full}^* - \tilde q_{\subb}}{\tilde \rho_g}\right) \geq 0.
\end{equation}

Observe that, for each group $g$, there exists a threshold $\underline{q}_{\full, \full}^g$ such
that students from group $g$ take the test if and only $\tilde q_{\subb} \geq \underline{q}_{\full, \full}^g$. This follows from the fact that the LHS of \eqref{eq.testtaking.full.full} is strictly increasing in $\tilde q_{\subb}$ since $\Phi$ is strictly increasing and $v_2 < v_1$.
The thresholds $\underline{q}_{\full, \full}^g$, $\tilde q_{1,\full}^*$, $\tilde q_{2,\full}^*$ are the solution to the following system of equations:
\begin{equation}
\label{eq.qfullfull.g.def}
    \underline{q}_{\full, \full}^g  = q_{1, \full}^* - \tilde \rho_g \Phi^{-1} \left ( \frac{v_1 - c_g}{v_1-v_2} - \frac{v_2}{v_1 -v_2} \Phi \left( \frac{\tilde q_{2, \full}^* - \underline{q}_{\full, \full}^g}{\tilde \rho_g}\right)  \right)
\end{equation}
\begin{equation}
 C_1  = \sum_{g \in \{A,B\}}\int_{\underline{q}_{\full, \full}^g}^\infty \left (1 - \Phi \left( \frac{\tilde q_{1, \full}^* - \tilde q_{\subb} }{\tilde \rho_g}\right) \right) \phi \left (\frac{\tilde q_{\subb} - \mu }{\sigma \sqrt{\frac{\sum_{{k\in \subb}} \sigma_{gk}^{-2}  }{\sigma^{-2} +\sum_{{k\in \subb}} \sigma_{gk}^{-2}}}} \right) \diff \tilde q_{\subb}
\end{equation}
\small
\begin{equation}
\label{eq.J2.capacity.fullfull}
C_1+ C_2  = \sum_{g \in \{A,B\}} \int_{\underline{q}_{\full, \full}^g}^\infty \left (\Phi \left( \frac{\tilde q_{1, \full}^* - \tilde q_{\subb} }{\tilde \rho_g}\right) - \Phi \left( \frac{\tilde q_{2, \full}^* - \tilde q_{\subb} }{\tilde \rho_g}\right) \right) \phi \left (\frac{\tilde q_{\subb} - \mu }{\sigma \sqrt{\frac{\sum_{{k\in \subb}} \sigma_{gk}^{-2}  }{\sigma^{-2} +\sum_{{k\in \subb}} \sigma_{gk}^{-2}}}} \right) \diff \tilde q_{\subb}.
\end{equation}
\normalsize
where we used Lemma~\ref{lemma:perceived_skill_distribution} in the last two equations.

Next, we study when school $J_1$ has the incentive to drop the test. If it does, all students apply to $J_1$ since there is no test cost and $v_1 > v_2$. 
Thus, similar to Theorem~\ref{thm:threshold_char_without_aa}, school $J_1$ incurs an information loss (due to one missing feature) leading to academic merit decrease, but has the incentive to drop the test only if it now has access to higher-skilled candidates on average compared to its previous policy $P^1 = P_{\subb}$. 

The new admission threshold $\tilde q^*_{1, \subb}$ that school $J_1$ uses is the solution to \eqref{eq:threshold_groupaware_noAA_Phi}.  
The proof is by case analysis. 

\noindent\textit{Case (1):  $\tilde q^*_{1,\subb} < \underline{q}^g_{\full, \full}$, for both $g \in \{A, B\}$.} This case must be ruled out by our main assumption that $C_1, C_2$ are small enough such that both schools fill in their capacity. If case (i) were true, then school $J_1$ would not be able to fill their capacity under $\mathbf{P}=(P_{\full}, P_{\full})$, since  the LHS in \eqref{eq:threshold_groupaware_noAA_Phi} would be larger than the RHS in \eqref{eq.J2.capacity.fullfull}.

\noindent\textit{Case (2): $\tilde q^*_{1,\subb} \geq \underline{q}^g_{\full, \full}$, for both $g \in \{A, B\}$.}  
Under $P^1=P_{\subb}$, school $J_{1}$ admits only students with $\tilde q_{\subb} \ge \tilde q^{*}_{1,\subb}$. Since $\tilde q^{*}_{1,\subb} \ge \underline q^{g}_{\full,\full}$, every student who can be admitted by $J_{1}$ under $P^1=P_{\subb}$ is already contained in the set of test-taking students under $P^1=P_{\full}$. Thus,  dropping the test does not expand the pool of students from which $J_{1}$ can select its class; the relevant applicant pool is identical under $P^1=P_{\full}$ and $P^1=P_{\subb}$.

Given this fixed pool and capacity $C_{1}$, the only difference between the two policies is the information used to rank applicants. Under $P^1=P_{\full}$, school $J_{1}$ uses the more informative signal $\tilde q_{\full}$, whereas under $P^1=P_{\subb}$ it relies only on the coarser signal $\tilde q_{\subb}$. By Theorem~\ref{prop.test_group_aware_noAA}, when  capacity is fixed and there are no access barriers, $P_{\full}$ yields strictly higher expected academic merit than $P_{\subb}$.  Consequently, $J_{1}$ has no incentive to drop the test in this case, if $J_2$ also requires it.

\noindent\textit{Case (3): $\tilde q^*_{1,\subb} <\underline{q}^g_{\full, \full}$, for exactly one fixed group $g^* \in \{A, B\}$.} First, observe that this condition is equivalent to the second part of \eqref{eq.condition.drop.twoschools.strategic.J1}. Thus, school $J_1$ will expand its pool of group $g^*$ students by dropping the test. However, it might not necessarily improve its academic merit due to information loss. School $J_1$ improves academic merit by dropping the test if and only if
\begin{equation*}
\begin{split}
   & \sum_{g \in \{A,B\}} \pi_g\int_{\underline{q}_{\full, \full}^g}^\infty 
       \left( \int_{\tilde q_{1, \full}^*}^\infty
        \tilde q_{\full}  \phi \left( \frac{\tilde q_{ \full} - \tilde q_{\subb} }{\tilde \rho_g}\right) \diff \tilde q_{\full}\right) \phi \left (\frac{\tilde q_{\subb} - \mu }{\sigma \sqrt{\frac{\sum_{{k\in \subb}} \sigma_{gk}^{-2}  }{\sigma^{-2} +\sum_{{k\in \subb}} \sigma_{gk}^{-2}}}} \right) \diff \tilde q_{\subb}\\
        &<\sum_{g \in \{A,B\} }\pi_g\int_{\tilde q^*_{1, \subb}}^\infty \tilde q_{\subb} \phi \left (\frac{\tilde q_{\subb} - \mu }{\sigma \sqrt{\frac{\sum_{{k\in \subb}} \sigma_{gk}^{-2}  }{\sigma^{-2} +\sum_{{k\in \subb}} \sigma_{gk}^{-2}}}} \right) \diff \tilde q_{\subb}.
\end{split}
\end{equation*}
Since $\underline{q^*}_{\full, \full}^g$ is strictly increasing in $c_{g^*}$ (see \eqref{eq.qfullfull.g.def}), the LHS is strictly decreasing in  $c_{g^*}$ given fixed capacity $C_1$. Note that for $c_{g^*} = 0$, the LHS is larger than the RHS; for $c_{g^*} \rightarrow v_2$, the RHS becomes larger than the LHS, since no students from group $g^*$ take the test. Due to the continuity of the LHS in $c_{g^*}$, the intermediate value theorem implies that there exists a $\hat{c}_{g^*}$ such that the above inequality holds for all $c_{g^*}\geq \hat{c}_{g^*}$. 

Finally, we study when school $J_2$ has the incentive to drop the test given that school $J_1$ has dropped the test. In particular, school $J_1$ drops the test only under case (3). Under case (3), school $J_1$ admits all test-taking students from group $g^*$ under $(P_{\full}, P_{\full})$, i.e., students with $\tilde q_{\subb}  \geq \underline{q}_{\full, \full}^{g^*}$. Furthermore, under $(P_{\subb}, P_{\full})$, the pool of test-taking students decreases since no students from group $g^*$ apply to $J_2$, while for the other group $g' \neq g^*$, $\underline{q}^{g'}_{\subb, \full} > \underline{q}^{g'}_{\full, \full}$, where 
\begin{equation}
\underline{q}^{g'}_{\subb, \full} =  \tilde q_{2, \full}'-
    \tilde \rho_{g'}\Phi^{-1} \left(1 -c_{g'}/v_2\right)
\end{equation}
and $\tilde q_{2, \full}'$ is school $J_2$'s admission threshold under  $(P_{\subb}, P_{\full})$. 

It is now possible that the mass of test-taking students from $g'$ falls below $C_2$. Thus, to ensure that $J_2$ keeps the test, two conditions must hold: $M_{g'} (P_{\subb}, P_{\full}) \geq C_2$ (i.e., condition \eqref{eq.sub.full.J2.cond1}) and the expected merit of group $g'$ under $P^2=P_{\full}$ is higher than the expected merit of the two admitted groups under $P^2=P_{\subb}$, i.e., 
\begin{equation*}
\begin{split}
   &  \pi_{g'}\int_{\underline{q}_{\subb, \full}^g}^{\tilde q^*_{1, \subb}} 
       \left( \int_{\tilde q_{2, \full}^*}^\infty
        \tilde q_{\full}  \phi \left( \frac{\tilde q_{ \full} - \tilde q_{\subb} }{\tilde \rho_{g'}}\right) \diff \tilde q_{\full}\right) \phi \left (\frac{\tilde q_{\subb} - \mu }{\sigma \sqrt{\frac{\sum_{{k\in \subb}} \sigma_{gk}^{-2}  }{\sigma^{-2} +\sum_{{k\in \subb}} \sigma_{gk}^{-2}}}} \right) \diff \tilde q_{\subb}\\
        &>\sum_{g \in \{A,B\} }\pi_g\int_{\tilde q^*_{2, \subb}}^{q^*_{1, \subb}} \tilde q_{\subb} \phi \left (\frac{\tilde q_{\subb} - \mu }{\sigma \sqrt{\frac{\sum_{{k\in \subb}} \sigma_{gk}^{-2}  }{\sigma^{-2} +\sum_{{k\in \subb}} \sigma_{gk}^{-2}}}} \right) \diff \tilde q_{\subb},
\end{split}
\end{equation*}
which is equivalent to \eqref{eq.sub.full.J2.cond2}.
\end{proof}

\begin{proof}{\textit{Proof of Part (ii).}}  Fixing $P^1=P_{\full}$, school $J_2$ wants to drop the test if and only if
this expands its pool of high-skilled students. Similarly to case (3) in Part (i), this is equivalent to the first two conditions in \eqref{eq.condition.drop.twoschools.strategic.J2.1}. At the same time, school $J_1$ wants to keep the test if and only if  either it does not expand its pool by dropping the test or the academic merit of the admitted class decreases after dropping the test. The former is similar to case (2) in part (i) and occurs if and only if the first condition in \eqref{eq.condition.drop.twoschools.strategic.J2.2} holds. The latter holds if and only if the first condition in \eqref{eq.condition.drop.twoschools.strategic.J2.2} does not hold but
\begin{equation*}
\begin{split}
   & \sum_{g \in \{A,B\}} \pi_g\int_{\underline{q}_{\full, \full}^g}^\infty 
       \left( \int_{\tilde q_{1, \full}^*}^\infty
        \tilde q_{\full}  \phi \left( \frac{\tilde q_{ \full} - \tilde q_{\subb} }{\tilde \rho_g}\right) \diff \tilde q_{\full}\right) \phi \left (\frac{\tilde q_{\subb} - \mu }{\sigma \sqrt{\frac{\sum_{{k\in \subb}} \sigma_{gk}^{-2}  }{\sigma^{-2} +\sum_{{k\in \subb}} \sigma_{gk}^{-2}}}} \right) \diff \tilde q_{\subb}\\
        &<\sum_{g \in \{A,B\} }\pi_g\int_{\tilde q^*_{1, \subb}}^\infty \tilde q_{\subb} \phi \left (\frac{\tilde q_{\subb} - \mu }{\sigma \sqrt{\frac{\sum_{{k\in \subb}} \sigma_{gk}^{-2}  }{\sigma^{-2} +\sum_{{k\in \subb}} \sigma_{gk}^{-2}}}} \right) \diff \tilde q_{\subb}.
\end{split}
\end{equation*}
Similarly to case (3) in Part (i), the last inequality holds if and only if $c_g \leq \hat{c}_g''$. 
\end{proof}

\begin{proof}{\textit{Proof of Part (iii).}}
By Part (i), school $J_1$ wants to keep the test in Cases (1) and (2). It also wants to keep the test in Case (3) when \eqref{eq.condition.drop.twoschools.strategic.J1} does not hold. Equivalently, by the continuity and monotonicity of $M_g$ in $c_g, c_{g'}$, if $J_1$ wants to keep the test conditional on $P^2 = P_{\full}$,  there must exist continuous functions $\underline{c}^1_g: \mathbb{R} \rightarrow$ such that  $c_A \leq \underline{c}^1_A (c_B)$ and $c_B \leq \underline{c}^1_B (c_A)$. Similarly, conditional on $P^1 = P_{\full}$, 
and using a similar argument to Case (3) in Part (i) and the case of school $J_1$ above, there must exist
 $\underline{c}^2_g: \mathbb{R} \rightarrow$ such that school $J_2$ 
wants to keep the test if and only if  $c_A \leq \underline{c}^2_A (c_B)$ and $c_B \leq \underline{c}^2_B (c_A)$. 
 
 Finally, $(P_{\full},P_{\full})$ is an equilibrium if and only if both schools prefer to keep the test.  
This holds if and only if 
\[
c_A \le \underline{c}_A(c_B)
\quad\text{and}\quad
c_B \le \underline{c}_B(c_A),
\]
where 
$\underline{c}_A := \min\{\underline{c}^1_A, \underline{c}^2_A\}$ 
and 
$\underline{c}_B := \min\{\underline{c}^1_B, \underline{c}^2_B\}$.
\end{proof}

\thmtwoschoolsstrategicequlibria*

\begin{proof}
     The proof follows from   \Cref{thm.droptest.strategic.2school}.
\end{proof}

\begin{proposition}[Dropping tests with strategic students: Diversity]
\label{prop.strategic.twoschool.diversity}
        Consider two schools, $J_1$ and $J_2$, both of which initially follow test-based policies $\mathbf{P}=(P_{\full}, P_{\full})$. When schools optimize for diversity only, 
            school $J_1$ drops the test if and only if    \small
            \begin{equation*}
                \frac{\pi \tilde \sigma_B }{C_1} \left(\Phi\left(\frac{a_B+b_B\mu}{\sqrt{1 + \tilde \sigma_B ^2 b_B^2}}\right) -  \Phi_2 \left( \frac{a_B+b_B\mu}{\sqrt{1 + \tilde \sigma_B ^2 b_B^2}}, \frac{\tilde q^*_{1, \full} -\mu}{\tilde \sigma_B}; -\frac{\tilde \sigma_B b_B}{\sqrt{1 + \tilde \sigma_B ^2 b_B^2}} \right)\right) < 1-  \Phi \left( \Phi^{-1}(1-C_1) \sqrt{(1-\pi) \frac{\tilde{\sigma}_{A, \subb}}{ \tilde{\sigma}_{B, \subb}}+ \pi } \right),
            \end{equation*}
            \normalsize
            where $\tilde{\sigma}_{g, \subb}:= \sigma \frac{\sum_{k\in \subb} \sigma_{gk}^{-2}}{\sigma^{-2} + \sum_{k\in \subb} \sigma_{gk}^{-2}}$ and $a_B=a_B(\tilde q_{1, \full}^*)$, $b_B$, $\tilde \sigma_B$ are defined similarly to  Proposition~\ref{prop.div_merit_strategic}.

\end{proposition}

\begin{proof}
Under $P^1=P_{\full}$, similarly to  \eqref{eq.div.J1.twoschool.strategic}, the diversity level at $J_1$ equals 
\begin{equation*}
      \tau_B^1 (P_{\full}, P_{\full}) = \frac{\pi \tilde \sigma_B }{C_1} \Phi\left(\frac{a_B+b_B\mu}{\sqrt{1 + \tilde \sigma_B ^2 b_B^2}}\right) - \frac{\pi \tilde \sigma_B }{C_1} \Phi_2 \left( \frac{a_B+b_B\mu}{\sqrt{1 + \tilde \sigma_B ^2 b_B^2}}, \frac{\tilde q^*_{1, \full} -\mu}{\tilde \sigma_B}; -\frac{\tilde \sigma_B b_B}{\sqrt{1 + \tilde \sigma_B ^2 b_B^2}} \right).
\end{equation*}
If $J_1$ drops the test, then all students apply to $J_1$. Since $v_1 > v_2$, $J_1$ does not compete with $J_2$, thus under $P^1= P_{\subb}$, the diversity level at $J_1$ equals
$$\tau_B^1 (P_{\subb}, P_{\full})= 1-  \Phi \left( \Phi^{-1}(1-C) \sqrt{(1-\pi) \frac{\tilde{\sigma}_{A, \subb}}{ \tilde{\sigma}_{B, \subb}}+ \pi } \right)$$
as in the single school setting without barriers
(see the proof of Part (i) in Theorem~\ref{prop.test_group_aware_noAA}).
Requiring that $\tau_B^1 (P_{\subb}, P_{\full}) > \tau_B^1 (P_{\full}, P_{\full})$ gives the  statement.
\end{proof}

\section{Group-unaware estimation}
\label{app.group_unaware}

In the main text, we primarily consider a ``group-aware'' estimation procedure, in which the school uses students' group membership in its estimation procedure (and thus is able to plug in group-specific noise biases and variances). We now briefly discuss ``unaware'' estimation when it cannot do so. Ignoring group attributes is an oft-proposed but often problematic policy proposal to combat bias in machine learning tasks~\citep{corbett2018measure}, and so we evaluate its consequences.

Ignoring group membership complicates the skill estimation challenge. When the feature distributions differ across groups but the school cannot observe the group of a student,
the resulting estimated skill distribution is a mixture of Normal distributions. The mixture weights depend on the noise means and variances of each group $g$.
In contrast to the group-aware case, where the school manages to correct for the feature noise biases (but not variance), the biases now play an important rule in each feature's implications.

We derive this distribution below. However, we  primarily study the effects through simulation in Figure~\ref{fig:policycompare_barriers_withunaware}.

\smallskip
\noindent\textbf{Unaware estimation derivation.} Conditional on the true skill level $q$, the features are still distributed according to a group-specific Normal distribution:
\[\theta_{k} | q,g \sim N(q + \mu_{gk}, \sigma_{gk}^2) \ \ \ \ \ \ \forall k = 1 \dots K \]

But under group-unaware estimation, the school does not know or cannot use $g$,  so the posterior is now a mixture of Normal distributions. Specifically, let $f(q \mid \thetaset)$ denote the pdf of the posterior distribution, $q \mid \thetaset$; similarly, we use the notation $f(\thetaset)$ and $f(q \mid \thetaset, g)$. Thus,
\begin{align*}
	f(q | \thetaset) &= 	\sum_{g\in\{A,B\}} f(q | \thetaset,g)\prob(g | \thetaset)\\
	&= \sum_{g\in\{A,B\}} f(q | \thetaset,g) \left[\frac{f( \thetaset | g) \prob(g)}{f(\thetaset)}\right]\\
	&= \sum_{g\in\{A,B\}} w(\thetaset,g) f(q | \thetaset,g), & w(\thetaset,g) \triangleq \left[\frac{f( \thetaset | g) \prob(g)}{f(\thetaset)}\right].
\end{align*}

Then, the posterior $q | \thetaset$ is distributed as a mixture of Normal distributions, where each Normal is as in the group-aware case:
\[q | \thetaset \sim \sum_{g\in\{A,B\}} w(\thetaset,g) N\left(\tilde q(\thetaset,g), \tilde\sigma^2(\thetaset,g)\right)
\]

For the weights, we find that
\[w(\thetaset,g) \triangleq \frac{f( \thetaset | g) \prob(g)}{f(\thetaset)}=\frac{\int_{\infty}^{\infty}\Pi_kf( \theta_k | g,q)\diff F(q) \cdot  \prob(g)}{f(\thetaset)}\]
and for $K$ features,
\begin{align}
	\int_{\infty}^{\infty}\Pi_kf( \theta_k | g,q)\diff F(q) &=	\frac{
		e^{\left(-\frac
			{ \sum_{k=1}^K \left[\left(\mu + \mu_{gk} - \theta_k\right)^{2} \sigma^{-2}\sigma_{gk}^{-2}\right]
				+ \sum_{k \neq \ell} \left[\left((\mu_{\ell g} - \theta_\ell) - (\mu_{gk} - \theta_k)\right)^{2} \sigma_{\ell g}^{-2}\sigma_{gk}^{-2}\right]
			}
			{2 \, {\left(\sigma^{-2} + \sum_{k=1}^K \sigma_{gk}^{-2}\right)}}
			\right)}}
	{2 \, (1-\pi)^{K/2} \sigma\left(\Pi_k \sigma_{gk}\right)\sqrt{\sigma^{-2} + \sum_{k=1}^K \sigma_{gk}^{-2}}} \label{eqn:weightskfeatures}
\end{align}

Thus, we have
\begin{align*}
	w(\thetaset,g) &\triangleq \frac{f( \thetaset | g) \prob(g)}{f(\thetaset)}=\frac{\int_{\infty}^{\infty}\Pi_kf( \theta_k | g,q)\diff F(q) \prob(g)}{f(\thetaset)} \\
	&\propto \frac{\prob(g)
		\exp{\left(-\frac
			{ \sum_{k=1}^K \left[\left(\mu + \mu_{gk} - \theta_k\right)^{2} \sigma^{-2}\sigma_{gk}^{-2}\right]
				+ \sum_{k \neq \ell} \left[\left((\mu_{\ell g} - \theta_\ell) - (\mu_{gk} - \theta_k)\right)^{2} \sigma_{\ell g}^{-2}\sigma_{gk}^{-2}\right]
			}
			{2 \, {\left(\sigma^{-2} + \sum_{k=1}^K \sigma_{gk}^{-2}\right)}}
			\right)}}
	{\left[\Pi_k \sigma_{gk}\right]\sqrt{\sigma^{-2} + \sum_{k=1}^K \sigma_{gk}^{-2}}}
\end{align*}

\smallskip
\noindent\textbf{Derivation for equation~\eqref{eqn:weightskfeatures}}. We explicitly show the algebra for $K=1$ and $K=2$ features, and the pattern continues for $K$ features.

\noindent\textit{For  one feature $\theta_1$:}
\begin{align*}
w(\theta_1,g) &\triangleq \frac{f( \theta_1 | g) \prob(g)}{f(\theta_1)}
=\frac{\int_{\infty}^{\infty}f( \theta_1 | g,q)\diff F(q) \prob(g)}{f(\theta_1)}\\
&=\frac{\frac{1}{\sqrt{2(1-\pi)(\sigma^2 + \sigma_{g1}^2)}}\exp\left[-{\frac{{\left(\mu + \mu_{g1} - \theta_{1}\right)}^{2}}{2(\sigma^2 + \sigma_{g1}^2)}}\right]\prob(g)}
{f(\theta_1)}
= \frac{\frac{1}{\sqrt{\sigma^2 + \sigma_{g1}^2}}\exp\left[-{\frac{{\left(\mu + \mu_{g1} - \theta_{1}\right)}^{2}}{2(\sigma^2 + \sigma_{g1}^2)}}\right]\prob(g)}
{\sum_g\left[{\frac{1}{\sqrt{\sigma^2 + \sigma_{g1}^2}}\exp\left[-{\frac{{\left(\mu + \mu_{g1} - \theta_{1}\right)}^{2}}{2(\sigma^2 + \sigma_{g1}^2)}}\right]\prob(g)}\right]}.
\end{align*}

 \noindent\textit{For two features $\theta_1, \theta_2$:}
\begin{align*}
w(\thetaset,g) &\triangleq \frac{f( \thetaset | g) \prob(g)}{f(\thetaset)}=\frac{\int_{\infty}^{\infty}\Pi_kf( \theta_k | g,q)\diff F(q) \prob(g)}{f(\thetaset)},\\
\end{align*}
\begin{align*}
\int_{\infty}^{\infty}\Pi_kf( \theta_k | g,q)\diff F(q)
&=	\frac{
	e^{\left(-\frac
		{\left({\left(\mu_{g1} - \theta_1\right)} - \left(\mu_{g2} -\theta_2\right)\right)^2 \sigma_{g1}^{-2} \sigma_{g2}^{-2}
			+ {\left(\mu + \mu_{g2} - \theta_2\right)^{2}} \sigma^{-2}\sigma_{g2}^{-2}
			+ {\left(\mu + \mu_{g1}- \theta_1\right)^{2}} \sigma^{-2}\sigma_{g1}^{-2} }
		{2 \, {\left(\sigma^{-2} + \sigma_{g1}^{-2} + \sigma_{g2}^{-2}\right)}}
		\right)}}
{2 \, (1-\pi) \sigma\sigma_{g1}\sigma_{g2}\sqrt{\sigma^{-2} + \sigma_{g1}^{-2} + \sigma_{g2}^{-2}}}.
\end{align*}

\section{General distributions}
\label{app.general.distributions}

\noindent\textbf{Extended model.}
We extend the model from Section~\ref{sec:baseline_model} to non-Normal distributions. In the current setting, each candidate is characterized by a (latent) \textit{true skill} $q$ drawn from a distribution $F^0$ with support $Q = [\underline{q}, \overline{q}]$ and mean $\mu$.\footnote{Formally, we assume that there exists a probability space $(Q, \mathcal{F}, \mathbb{P})$ on which $q$ is defined.} 
We assume that $F^0$ is common for both social groups.

For each candidate, the school has access to  $K$ observable \textit{features} $\thetaset=(\theta_k)_{k=1}^K$. Throughout this section, we thus assume that the school uses policy $P_{\full}$ and omit it from the notation.  

Conditional on the true skill level $q$ and group $g$, feature $\theta_k$ is independently drawn from a distribution $F^k_{q, g}$. Let $\Theta_k= [\underline{\theta}_k, \overline{\theta}_k]$  be the support of each feature $\theta_k$.
We assume that the distributions $F^0$, $\{F^k_{q,g}\}_{k=1}^K$ are common knowledge.
Without loss of generality and for the sake of simplicity, we  further assume that  $F^0$, $\{F^k_{q,g}\}_{k=1}^K$ are continuous (although being measurable would suffice).

At an aggregate level per group $g$, the information structure $\left ( \times_{k=1}^K \Theta_k, F^0, \{F^k_{q,g}\}_{k=1}^K\right)$ induces a \textit{skill estimate} distribution, $\hat{F}_g$, for candidates in group $g$, i.e., $\tilde q \mid g \sim \hat{F}_g$, where $\tilde q(\thetaset, g) \triangleq \expec [q \mid \thetaset, g]$ as in the main model. We also let $\hat{F} = (1-\pi) \hat{F}_A + \pi \hat{F}_B$.

\smallskip\noindent\textbf{Preliminaries.} We will need the following technical terms and properties.

\begin{definition}[\textbf{\protect\cite{blackwell1953equivalent}}]
$\{\Pi_q, q \in Q\}$ is \emph{sufficient} for $\{\Pi'_q, q \in Q\}$ if there exists a transformation $T(x, dy)$ such that for all $q \in Q$, $\Pi'_q (\cdot) = \int_X T(x, \cdot) \Pi_q (dx)$.
\end{definition}

\begin{lemma}[\textbf{\protect\cite{eckwert2004economic, zhang2009comparison}}]
\label{thm:zhang}
The following statements are equivalent:
\begin{itemize}
    \item $\{\Pi_q, q \in Q\}$ is sufficient for $\{\Pi'_q, q \in Q\}$;
    \item The distribution of posteriors $ \left \langle \Pi'_q \right\rangle$ second-order stochastically dominates $\left \langle \Pi_q \right\rangle$.
\end{itemize}
\end{lemma}

\begin{lemma}[\textbf{\protect\cite{gentzkow2016rothschild}}]
\label{lemma:MPS_distr_posteriors_to_posterior_mean}
If the distribution of posteriors $\langle \Pi_q  \rangle$ is a mean-preserving spread of $\langle \Pi'_q \rangle$, then the posterior mean distribution under $\langle \Pi_q \rangle$ is a mean-preserving spread of the posterior mean distribution under $\langle \Pi'_q \rangle$.
\end{lemma}

\begin{lemma}
\label{lemma:MPS_SSD}
Let $X$, $Y$ be two random variable with equal means $\expec[X]=\expec[Y]$, support $[\underline{q}, \overline{q}]$, and CDFs $F$ and $G$, respectively. Then, 
the following are equivalent:
\begin{itemize}
    \item[(i)] $Y \succ_{\textrm{SSD}} X$;
    \item[(ii)] $X$ is a mean-preserving spread of $Y$;
    \item[(iii)]  $\int_{\underline{q}}^{\overline{q}} u(y)\diff G(y) \geq \int_{\underline{q}}^{\overline{q}} u(x)\diff F(x)$ for every weakly increasing concave function $u: [\underline{q}, \overline{q}] \rightarrow \mathbb{R}$.
\end{itemize}
\end{lemma}

\begin{lemma}
\label{lemma:crossing_MPS}
Let $X$ and $Y$ be two random variables with support $[\underline{q}, \overline{q}]$ and CDFs $F$ and $G$, such that $Y$ is a mean preserving of $X$. Then,
$F$ crosses $G$ exactly once at a point $q_+ \in [\underline{q}, \overline{q}]$. For $q \in [\underline{q}, q_+)$, $F(q) > G(q)$ whereas for $q \in (q_+, \underline{q}]$,  $F(q) < G(q)$.
\end{lemma}

\smallskip
\noindent\textbf{Generalizing Proposition~\ref{prop:group_aware_noAA}.}
We are finally ready to prove a generalized version of Proposition~\ref{prop:group_aware_noAA}.

\begin{proposition}
Suppose that $(\tilde q \mid A) \prec_{\textrm{SSD}} (\tilde q \mid B) $ with crossing point $q_+$. Consider a school that uses admissions policy $P_{\full}$. Then, $(\tilde q \mid A) \prec_{\textrm{SSD}} (\tilde q \mid B) $ is equivalent to each of the following conditions.
\begin{itemize}
    \item[(i)] \emph{Diversity}: Group $B$ is under-represented if and only if $C< 1-\hat{F}(q_+) $;
    \item[(ii)] \emph{Academic merit}: For any capacity $C$, the policy achieves worse academic merit for admitted students from group $B$.
\end{itemize}
Furthermore, suppose that $\{\times_{k=1}^K F^k_{q,A}, q \in Q\}$ is sufficient for $\{\times_{k=1}^K F^k_{q,B}, q\in Q\}$.  Equivalently,
\begin{itemize}
   \item[(iii)] \emph{Individual fairness}: there exists a threshold $\hat{q}$ such that $I(q; P_S) >0$ if and only if $q>\hat{q}$.
\end{itemize}
\end{proposition}

\begin{proof}
We prove each part separately.

\noindent\textit{Proof of part (i).} 
Let $\tilde q^*$ denote the optimal acceptance threshold as given by the following equation:
\begin{equation*}
    (1-\pi) \hat{F}_A(\tilde q^*) + \pi \hat{F}_B(\tilde q^*)  = 1-C \Longleftrightarrow \hat{F} (\tilde q^*) = 1-C.
\end{equation*}
Therefore, for $C< 1-\hat{F}(q_+)$, it holds that $\tilde q^*> q_+$, and vice versa.
Thus, part (i) follows directly from Lemma~\ref{lemma:crossing_MPS}.

~\\\noindent\textit{Proof of part (ii).} 
Part (ii) follows from the equivalence between (i) and (iii) in  Lemma~\ref{lemma:MPS_SSD} where we consider $u(x)$ to be the linear function $u(x)=x$. 

~\\\noindent\textit{Proof of part (iii).} By Lemma~\ref{thm:zhang}, sufficiency equivalently guarantees that the posterior distribution $\{\times_{k=1}^K F^k_{q,A}, q \in Q\}$ second-order stochastically dominates $\{\times_{k=1}^K F^k_{q,B}, q\in Q\}$. By Lemma~\ref{lemma:crossing_MPS}, this immediately translates to the following property:
\begin{equation*}
    \prob[Y=1 \mid q, A] = \prob[\tilde q \geq \tilde q^* \mid q, A] >     \prob[Y=1 \mid q, B] = \prob[ q \geq \tilde q^* \mid q, B] 
\end{equation*}
if and only $q > \tilde q^*$,
where $\tilde q^*$ is the optimal acceptance threshold corresponding to some capacity $C$.
\end{proof}

Note that an analog of the above proposition can also be obtained for any subset of features $S$.

\section{Affirmative action}
\label{sec:affirmative_action}
Schools often have an additional lever in their admissions policies: whether or not to use affirmative action. The term \textit{affirmative action} refers to admissions policies that partially base decisions on applicants’ membership in social groups with legally protected characteristics (e.g., race, ethnicity, or gender), to promote equal opportunity as well as the educational benefits of diversity~\citep{alon2015race}.

We define affirmative action as a constraint on the fraction of students from each group. As a result, the admissions policy may  use different admission thresholds for different groups.
This approach is common in the literature \citep{fang2011theories} and a proxy of the practices adopted by universities. 
However, due to the recent lawsuit against Harvard~\citep{harvard2019} and the Supreme Court decision in 2023~\citep{saul_2023},  %
	the legal framework around such affirmative action is restrictive.
	Explicit, predetermined \textit{racial} quotas  are generally illegal, as is (newly) broad consideration of race separate from individuals' contexts; conversely, University of Texas admits students using a high school-based quota system \citep{ut_news_2019}.
    
    From a theoretical standpoint, the class of affirmative-action policies is interesting because it generates a Pareto frontier between the academic merit and diversity objectives. A fully Bayesian school,  using group information when forming skill estimates but then accepting students with the highest skill estimates regardless of group, would maximize academic merit. To instead maximize some weighted combination of academic merit and diversity, an optimal school (with no legal constraints) would be fully Bayesian \textit{within} each group, ranking students within each group according to their expected true skill and then accepting the top students \textit{from each group} to achieve some desired balance between academic merit and diversity objectives. Different weights would correspond to different fractions of students from each group, tracing out a Pareto curve.

Next, we study outcomes when schools can decide both whether to require standardized testing \textit{and} whether to use affirmative action.

\medskip
\noindent\textbf{Affirmative action under a fixed testing policy.} As a stylized model of affirmative action, we extend the main setup of Section~\ref{sec:baseline_model} by introducing a constraint on the diversity level $\tau(P)$ achieved by a policy $P$, i.e., %
 consider admissions policies of the form $P^\tau_{S}$, where $\tau \in (\tau(P_{S}), \pi]$ is the target diversity level set by the school.
Thus, the school still optimizes for academic merit but under the additional constraint that a fraction $\tau$ of admitted students must belong to group $B$.
To do so, the common admission decision threshold  is now replaced
by two group-dependent thresholds, $\tilde q^*_{A,S}$ and $\tilde q^*_{B,S}$.\footnote{{In Proposition~\ref{prop:group_aware_AA}, the assumptions that $\gamma_A \geq \frac{2(1-\tau)C}{1-\pi}$ and $\gamma_B \geq \frac{2\tau C}{\pi}$ ensure that, even in the presence of barriers, the admission to the school is over-demanded (in the sense that the school cannot admit all applicants) and selective (meaning that the admission thresholds satisfy $\tilde q^*_{g,S} \geq \mu$).}}
Note that $\tau(P^\tau_S) = \tau$, thus
under affirmative action, diversity improves by definition, and group fairness holds when the target diversity level is set to $\tau = \pi$.\footnote{Proposition \ref{prop:group_aware_AA} focuses only on diversity levels $\tau \in (\tau (P_{S}), \pi]$. The lower bound is reasonable since $\tau (P_{S})$ is the diversity level achieved by a school  optimizing solely for academic merit (Theorem~\ref{prop:group_aware_noAA}). The upper bound achieves group fairness. Note that higher levels $\tau>\pi$ could have also been considered with similar results; however, higher values of $\tau$ may be infeasible for certain values of $C$ and $(1-\pi)$ therefore are omitted.
} 
Affirmative action can be utilized on top of test-free or test-based policies. Whereas the testing policy determines the amount of information available in the estimation process, the affirmative action changes the selection process given information.

 We find that although affirmative action increases diversity, it does not change the information that schools have on students, and as a result the school still cannot identify high-skilled students in group $B$ as well as it can identify group $A$ students. 
We show that with unequal precision, affirmative action  improves the individual fairness gap but does not eliminate it, as disparities in the identification of the highest-skilled students remain. It further increases the gap in academic merit across social groups. Affirmative action alone cannot address the fundamental issue caused by variance in the features. As a result, we consider this decision as orthogonal.

\begin{restatable}[Affirmative action with a fixed testing policy]{proposition}{propgroupawareAA}
\label{prop:group_aware_AA}
Fix %
the target diversity level $\tau(P_{S})<\tau \leq \pi$  and assume unequal precisions.
{Let also $\gamma_B \leq \gamma_A \leq 1$ such that $\gamma_A \geq \frac{2 (1-\tau) C }{1-\pi}$, $\gamma_B \geq \frac{2 \tau C }{\pi}$.}
Then, 
\begin{itemize}

    \item[(i)] \emph{Individual fairness}: In comparison to $P_{S}$, the {individual fairness gap}  improves, i.e., 
    $ I(q ; P^\tau_{S}) < I(q ; P_{S})$ for all $q$.
    However, group A students still have higher probability of admission than same-skilled  group B students, i.e., $I(q; P^\tau_{S}) >0$,
	if and only if
 \begin{equation*}
    q  >  \frac{\left( 
    \frac{{\sum_{k\in S} \sigma_{Ak}^{-2}} + \sigma^{-2}}{\sqrt{\sum_{k\in S} \sigma_{Ak}^{-2}}}
    \right) \tilde q^*_{A,S} - \left( 
    \frac{{\sum_{k\in S} \sigma_{Bk}^{-2}} + \sigma^{-2}}{\sqrt{\sum_{k\in S} \sigma_{Bk}^{-2}}}
    \right) \tilde q^*_{B,S}}{\sqrt{\sum_{k\in S} \sigma_{Ak}^{-2}}-\sqrt{\sum_{k\in S} \sigma_{Bk}^{-2}}} + \frac{\mu \sigma^{-2}}{ \sqrt{\sum_{k\in S} \sigma_{Ak}^{-2}} \sqrt{\sum_{k\in S} \sigma_{Bk}^{-2}}}.
\end{equation*}

     Finally, there exist parameters such that 
     $ I(q ; P^\tau_{S}) < 0< I(q ; P_{S})$ for some $q$.

      \item [(ii)] %
   \emph{Academic merit}: 
   Policy $P^\tau_{S}$ always achieves worse {academic merit} for admitted  group B students than for group A students.
    Furthermore, in comparison to $P_{S}$, the academic merit of admitted students decreases for group B, while it increases for group A.

\end{itemize}
\end{restatable}

\begin{proof}{\textit{Proof of Part (i).}}
With affirmative action, the common threshold $\tilde q^*_{S}$ in \Cref{def:threshold_groupaware_noAA} is replaced
by two group-dependent thresholds, $\tilde q^*_{A,{S}}$ and $\tilde q^*_{B,{S}}$:
\begin{equation}
\label{eq:thresholds_AA}
    \begin{split}
        (1-\pi) \gamma_A (1- F_{\tilde q \mid A, P_{S}}(\tilde q^*_{A,{S}})) = (1-\tau) C, \,\,\,\,
        \pi \gamma_B (1-F_{\tilde q \mid B, P_{S}}(\tilde q^*_{B,{S}}))  = \tau C.
    \end{split}
\end{equation}
Note further that the distribution $F_{\tilde q \mid g, P_{S}} \equiv F_{\tilde q \mid g, P^\tau_{S}}$, $g\in \{A,B\}$, remains unchanged under both admissions policies $P^\tau_{S}$ and $P_{S}$, as both share the same (group-aware) estimation policy and  feature set $S$.

First, observe that  \Cref{eq:thresholds_AA} gives us
\begin{equation}
\label{eq:two_thresholds_tau}
\begin{split}
       \tilde q^*_{A,{S}} = F_{\tilde q \mid A, P_{S}}^{-1} \left(1- \frac{1-\tau}{(1-\pi)\gamma_A}C \right), \,\,\,
        \tilde q^*_{B,{S}} = F_{\tilde q \mid B, P_{S}}^{-1} \left(1- \frac{\tau}{\pi \gamma_B}C \right).
\end{split}
\end{equation}
Since $\tau>\tau(P_{S})$ and $\gamma_B \leq \gamma_A \leq 1$,
it follows that $\tilde q^*_{B,{S}} < \tilde{q}^*_{S} < \tilde q^*_{A,{S}}$.
Due to our assumptions that $\gamma_A \geq \frac{2(1-\tau)C}{1-\pi}$ and $\gamma_B \geq \frac{2\tau C}{\pi}$, we also get that $\mu<\tilde q^*_{B,{S}} < \tilde{q}^*_{S} < \tilde q^*_{A,{S}}$.

For the first statement of part (i), observe that, due to $\tilde q^*_{A,S} > \tilde q^*_{S}$ and  $\tilde q^*_{B,S} < \tilde q^*_{S}$ for all $\tau (P_{S})< \tau\leq \pi$, 
$\prob[\tilde q \geq \tilde q^*_{A,S} \mid q, A, P^\tau_{S}] < \prob[\tilde q \geq \tilde q^*_{S} \mid q, A, P_{S}],$ and ${\prob[\tilde q \geq \tilde q^*_{B,S} \mid q, B, P^\tau_{S}]> \prob[\tilde q \geq \tilde q^*_{S} \mid q, B, P_{S}],}$
since the distribution of $\tilde q \mid q, P$ remains the same under both $P \in \{P_{S}, P_{S}^\tau\}$. Consequently, $I(q ; P^\tau_{S}) <I(q ; P_{S})$.

For the proof of the second statement in Part (i), we apply the argument used in Proposition~\ref{prop:appendix_baseline}, Part (ii). Thus, we get that
$I(q; P_{S}^\tau) >0$ if and only if
\small
\begin{equation*}
        \frac{ \tilde q^*_{A,S}\sigma^{-2} + \tilde q^*_{A,S}\sum_{k \in S}\sigma_{Ak}^{-2}  - \mu \sigma^{-2} - q \sum_{k\in S} \sigma_{Ak}^{-2}}{ \sqrt{\sum_{k\in S} \sigma_{Ak}^{-2}}} < \frac{ \tilde q^*_{B,S}\sigma^{-2} +\tilde q^*_{B,S} \sum_{k\in S} \sigma_{Bk}^{-2}  - \mu \sigma^{-2} - q \sum_{k\in S} \sigma_{Bk}^{-2}}{ \sqrt{\sum_{k\in S} \sigma_{Bk}^{-2}}},
\end{equation*}
\normalsize
which is equivalent to
\small
 \begin{equation*}
 \begin{split}
  q  >  \frac{\left( 
    \frac{{\sum_{k\in S} \sigma_{Ak}^{-2}} + \sigma^{-2}}{\sqrt{\sum_{k\in S} \sigma_{Ak}^{-2}}}
    \right) \tilde q^*_{A,S} - \left( 
    \frac{{\sum_{k\in S} \sigma_{Bk}^{-2}} + \sigma^{-2}}{\sqrt{\sum_{k\in S} \sigma_{Bk}^{-2}}}
    \right) \tilde q^*_{B,S}}{\sqrt{\sum_{k\in S} \sigma_{Ak}^{-2}}-\sqrt{\sum_{k\in S} \sigma_{Bk}^{-2}}} + \frac{\mu \sigma^{-2}}{ \sqrt{\sum_{k\in S} \sigma_{Ak}^{-2}} \sqrt{\sum_{k\in S} \sigma_{Bk}^{-2}}}.
 \end{split}
\end{equation*}
\normalsize

Finally, we prove the third statement in Part (ii).  Consider an instance of the model parameters where
\begin{equation}
\label{eq:instance_condition_derivative}
    \sqrt{\sum_{k\in S} \sigma_{Bk}^{-2}}  \sqrt{\sum_{k\in S} \sigma_{Ak}^{-2}} > \sigma^{-2},
\end{equation}
and under $P^\tau_{S}$, the condition in Part (ii) in Proposition \ref{prop:appendix_baseline}, holds with equality for some $\hat{q},$ 
i.e.,
$$(\tilde q^*_{A,S}-\hat{q}) \sqrt{\sum_{k\in S} \sigma_{Bk}^{-2}}  \sqrt{\sum_{k\in S} \sigma_{Ak}^{-2}} = \sigma^{-2}(\tilde q^*_{A,S} -\mu).$$
Therefore, $\prob[\tilde q > \tilde q^*_{A,S} \mid \hat{q}, A] = \prob[\tilde q > \tilde q^*_{A,S} \mid \hat{q}, B].$
Since $\tilde q^*_{B,S}<\tilde q^*_{A,S}$, it further holds that ${\prob[\tilde q > \tilde q^*_{B,S} \mid \hat{q}, B]> \prob[\tilde q > \tilde q^*_{A,S} \mid \hat{q}, B]}$. Thus, $I(\hat{q}; P^\tau_{S}) <0$.

However, for  $q = \hat{q}$, we also have that
    $$(\tilde q^*_{S}-\hat{q}) \sqrt{\sum_{k\in S} \sigma_{Bk}^{-2}}  \sqrt{\sum_{k\in S} \sigma_{Ak}^{-2}} <\sigma^{-2}(\tilde q^*_S -\mu).$$
To see why, observe that given the condition in \Cref{eq:instance_condition_derivative}, the function $$g(\tilde q) =  (\tilde q -\hat{q}) \sqrt{\sum_{k\in S} \sigma_{Bk}^{-2}}  \sqrt{\sum_{k\in S} \sigma_{Ak}^{-2}} - \sigma^{-2}(\tilde q -\mu)$$ is increasing in $\tilde q$ since $$\frac{\diff g(\tilde q)}{\diff \tilde q} = \sqrt{\sum_{k\in S} \sigma_{Bk}^{-2}}  \sqrt{\sum_{k\in S} \sigma_{Ak}^{-2}} - \sigma^{-2} >0.$$
Consequently, for $\tilde q^*_S < \tilde q^*_{A,S}$,  $g(\tilde q^*_S) < g(\tilde q^*_{A,S}) =0$.
Part (ii) in Proposition \ref{prop:appendix_baseline}  further guarantees that $I(\hat{q}; P_{S}) >0 $ for this particular instance of model parameters.
Consequently,  we have constructed an instance of model parameters such that $I(\hat{q}; P_{S}) >0 > I(\hat{q}; P^\tau_{S})$ for some $\hat{q}$. Thus, such an instance exists. 
\end{proof}

\smallskip
\begin{proof}{\textit{Proof of Part (ii)}.}
We  use an argument similar to part (iii) in Proposition~\ref{prop:appendix_baseline} (note that this part holds for any common threshold greater than $\mu$ and not only $\tilde q^*_{S}$). Similarly to \Cref{eq:expected_academic_merit}, we derive that for both $g \in \{A, B\}$,
  $   {\expec[ q \mid \tilde q \geq \tilde q_{g,{S}}^*, g, P^\tau_{S}] = \expec[ \tilde q \mid \tilde q \geq \tilde q_{g,{S}}^*, g, P^\tau_{S}].}$
By the same part (iii) in Proposition~\ref{prop:appendix_baseline}, replacing $\tilde q^*_{S}$ with threshold $\tilde q_{A,{S}}^*>\mu$ implies that
$ \expec[\tilde q \mid \tilde q \geq \tilde q_{A,{S}}^*, A, P^\tau_{S}] > \expec[ \tilde q \mid \tilde q \geq \tilde q_{A,{S}}^*, B, P^\tau_{S}].$
Next, we have that
\begin{equation*}
\begin{split}
    \expec[\tilde q \mid Y=1, B, P^\tau_{S}] =& \expec[\tilde q \mid \tilde q \geq \tilde{q}^*_{B,{S}}, B,  P^\tau_{S}] \\
    =& \frac{1}{1- F_{\tilde q \mid B, P_{S}}(\tilde q^*_{B,{S}})} \int_{\tilde q^*_{B,S}}^{\infty} \tilde q d F_{\tilde q \mid B, P_{S}} (\tilde q)\\
     =& \frac{1}{1- F_{\tilde q \mid B, P_{S}}(\tilde q^*_{B,{S}})} \left( \int_{\tilde q^*_{B,S}}^{\tilde q^*_{A,{S}}} \tilde q d F_{\tilde q \mid B, P_{S}} (\tilde q) +  \int_{\tilde q^*_{A,S}}^{\infty} \tilde q d F_{\tilde q \mid B, P_{S}} (\tilde q) \right)  \\
     =&  \frac{F_{\tilde q \mid B, P_{S}}(\tilde q^*_{A,S}) -F_{\tilde q \mid B, P_{S}}(\tilde q^*_{B,S})}{1- F_{\tilde q \mid B, P_{S}}(\tilde q^*_{B,S})}  \expec[\tilde q \mid \tilde{q}^*_{A,S} >\tilde q \geq \tilde{q}^*_{B,S}, B,  P^\tau_{a,K}]\\
    &{}+ \frac{1 -F_{\tilde q \mid B, P_{S}}(\tilde q^*_{A,S})}{1- F_{\tilde q \mid B, P_{S}}(\tilde q^*_{B,S})}  \expec[\tilde q \mid \tilde q \geq \tilde{q}^*_{A,S}, B,  P^\tau_{a,K}].
\end{split}
\end{equation*}

The fact that $\expec[\tilde q \mid  \tilde{q}^*_{A,S}> \tilde q \geq\tilde q^*_{B,S}  , B, P^{\tau}_{S}] < \expec[\tilde q \mid \tilde q \geq \tilde{q}^*_{A,S}, B,P^{\tau}_{S}]$, together with    
the inequalities above,
finally imply that
$$\expec[ q \mid  Y=1, B, P^\tau_{S}] = \expec[\tilde q \mid \tilde q \geq \tilde{q}^*_{B,S}, B, P^\tau_{S}] < \expec[\tilde q \mid \tilde q \geq \tilde{q}^*_{A,S}, A, P^\tau_{S}] =\expec[ q \mid Y=1, A, P^\tau_{S}].$$

Regarding the second statement of part (ii), recall that the distributions $F_{\tilde q \mid g, P_{S}}$ and $F_{\tilde q \mid g, P^\tau_{S}}$ are identical. Since $\tilde q^*_{B,S} < \tilde{q}^*_{S} < \tilde q^*_{A,S}$, it follows that the conditional expectations satisfy
$$ \expec[ q \mid Y=1, A, P^\tau_{S}] = \expec[\tilde q \mid \tilde q \geq \tilde{q}^*_{A,S}, A, P^\tau_{S}] > \expec[\tilde q \mid \tilde q \geq \tilde{q}^*_{S}, A, P_{S}]= \expec[ q \mid Y=1, A, P_{S}],$$
$$\expec[ q \mid Y=1, B, P^\tau_{S}] = \expec[\tilde q \mid \tilde q \geq \tilde{q}^*_{B,S}, B, P^\tau_{S}]<\expec[\tilde q \mid \tilde q \geq \tilde{q}^*_{S}, B, P_{S}]=\expec[ q \mid Y=1, B, P_{S}].$$
Thus, the academic merit of admitted students increases for group $A$ while it decreases for group $B$.
\end{proof}

\medskip
\noindent\textbf{Dropping the test under affirmative action.} We now study how test-free and test-based policies with affirmative actions compare in a setting with unequal barriers $\gamma_g$ to test access. Recall that Theorem~\ref{thm:threshold_char_without_aa}   (without affirmative action) shows that, conditional on the information environment, if there are substantial barriers to test access, removing the test requirement improves academic merit.
The following theorem establishes the same result for a school using affirmative action. Recall that  the function $\textsc{HR}$ denote the hazard rate of the Normal distribution $\Phi$,  
    $\textsc{HR}(z) = \frac{\phi(z)}{1-\Phi(z)}$. 
\begin{restatable}[Dropping tests under affirmative action with barriers]{proposition}{thmbudgettestsvspoolsize}
\label{thm.budget_tests_vs_poolsize}
Fix group $g \in \{A, B\}$, variances $\sigma^2_{gk}$, and target diversity level $\tau$. Let $\tau_A \triangleq 1-\tau$ and $\tau_B \triangleq \tau$. Dropping the test score requirement improves the academic merit of admitted students from group $g$, i.e., $\E [q \mid Y=1, g, P^\tau_{\full}] < \E [q \mid Y=1, g, P^\tau_{\subb}]$, if and only if $\gamma_g \leq \hat{\gamma}_g$, where
\begin{equation}
\label{eq:hat_gamma_threshold}
    \hat{\gamma}_g=\frac{\tau_g C}{1-\Phi\left( \textsc{HR}^{-1} \left (\frac{\sqrt{\frac{\sum_{k \in \subb} \sigma_{gk}^{-2}  }{\sigma^{-2} +\sum_{k \in \subb} \sigma_{gk}^{-2}}}}
    {\sqrt{\frac{\sum_{k \in \full} \sigma_{gk}^{-2}  }{\sigma^{-2} +\sum_{k \in \full} \sigma_{gk}^{-2}}} } \textsc{HR} (
    \Phi^{-1}(1-\frac{\tau_gC}{\pi_g}))\right)\right)}.
\end{equation}
Fixing all other parameters,  the threshold $\hat{\gamma}_g$  increases as test variance $\sigma_{gK}$ for group $g$ increases.
\end{restatable}

\begin{proof}
Let $\tilde w^*_{g,{\full}}$ be the group-dependent threshold in a policy with barriers and  affirmative action. Define
$$t_g = \frac{\tilde w^*_{g,{\full}} - \mu}{\sigma\sqrt{ \frac{\sum_{k\in \full} \sigma_{gk}^{-2}  }{\sigma^{-2} +\sum_{k\in \full} \sigma_{gk}^{-2}}}}, \,\,\, t'_g = \frac{\tilde q^*_{g,\subb} - \mu}{\sigma \sqrt{\frac{\sum_{k\in \subb} \sigma_{gk}^{-2}  }{\sigma^{-2} +\sum_{k\in \subb} \sigma_{gk}^{-2}}}}.$$
For such a policy with admission thresholds $\tilde w^*_{g,{\full}}$, $g\in\{A,B\}$, Lemma~\ref{lemma:conditional_expectation_normal} implies that the expected skill level of admitted students in group $g$ equals
$$\expec[q \mid Y=1, g, P^\tau_{\full}]= \mu +  \sigma \sqrt{\frac{\sum_{k\in \full} \sigma_{gk}^{-2}  }{\sigma^{-2} +\sum_{k\in \full} \sigma_{gk}^{-2}}} \cdot \frac{\phi(t_g)}{1-\Phi(t_g)}.$$
Similarly, for a policy using affirmative action but no tests, and admission thresholds $\tilde q^*_{g,\subb}$, we get that
$$\expec[ q \mid Y=1, g, P^\tau_{\subb}]=  \mu +  \sigma \sqrt{\frac{\sum_{k\in \subb} \sigma_{gk}^{-2}  }{\sigma^{-2} +\sum_{k\in \subb} \sigma_{gk}^{-2}}} \cdot \frac{\phi(t'_g)}{1-\Phi(t'_g)}.$$

To compute the threshold $\hat{\gamma}_g$, we require that
$\expec[ q \mid Y=1, g, P^\tau_{\subb}]=\expec[q \mid Y=1, g, P^\tau_{\full}].$
Based on the above equations, this condition is equivalent to
\begin{equation*}
\begin{split}
    \sqrt{\frac{\sum_{k \in \subb} \sigma_{gk}^{-2}  }{\sigma^{-2} +\sum_{k \in \subb} \sigma_{gk}^{-2}}} \textrm{HR} \left(\frac{\tilde q^*_{g,\subb} - \mu}{ \sigma\sqrt{\frac{\sum_{k \in \subb} \sigma_{gk}^{-2}  }{\sigma^{-2} +\sum_{k \in \subb} \sigma_{gk}^{-2}}}}\right) &=
     \sqrt{\frac{\sum_{k\in \full} \sigma_{gk}^{-2}  }{\sigma^{-2} +\sum_{k\in \full} \sigma_{gk}^{-2}}}  \textrm{HR} \left(  \frac{\tilde w^*_{g,\full} - \mu}{\sigma\sqrt{ \frac{\sum_{ k \in \full} \sigma_{gk}^{-2}  }{\sigma^{-2} +\sum_{ k \in \full} \sigma_{gk}^{-2}}}}\right).
\end{split}
\end{equation*}
Letting $\tau_B=\tau$, $\tau_A=1-\tau$ and using  \Cref{eq:two_thresholds_tau} 
to compute the thresholds $\tilde w^*_{g,\full}$, $\tilde q^*_{g,\subb}$, we get that
\begin{equation*}
\begin{split}
    &\sqrt{\frac{\sum_{k \in \subb} \sigma_{gk}^{-2}  }{\sigma^{-2} +\sum_{k \in \subb} \sigma_{gk}^{-2}}} \textrm{HR} \left(
    \Phi^{-1}\left(1-\frac{\tau_gC}{\pi_g}\right)\right)
    =
     \sqrt{\frac{\sum_{k\in \full} \sigma_{gk}^{-2}  }{\sigma^{-2} +\sum_{k\in \full} \sigma_{gk}^{-2}}}  \textrm{HR} \left(  \Phi^{-1}\left(1-\frac{\tau_g C}{\pi_g \hat{\gamma}_g}\right) \right)
\end{split}
\end{equation*}
Thus, solving for $\hat{\gamma}_g$, we finally get  \Cref{eq:hat_gamma_threshold}. Note that the expected skill level of admitted students in the test-based policy is given -- due to Lemma~\ref{lemma:conditional_expectation_normal} -- by 
$$\mu + \sigma\sqrt{\frac{\sum_{k\in \full} \sigma_{gk}^{-2}  }{\sigma^{-2} +\sum_{k\in \full} \sigma_{gk}^{-2}}}  \textrm{HR} \left(  \Phi^{-1}\left(1-\frac{\tau_g C}{\pi_g\hat{\gamma}_g}\right) \right).$$
By Lemma~\ref{lemma:HR}, it follows that $\textrm{HR}$ is increasing. However, $\Phi^{-1}\left(1-\frac{\tau_g C}{\pi_g\hat{\gamma}_g}\right)$ is decreasing in $\hat{\gamma}_g$.   Therefore, the academic merit of  $g$ must be decreasing in $\hat{\gamma}_g$. Thus, dropping the test increases academic merit for   $g$ if and only if $\gamma_g \leq \hat{\gamma}_g$.

Finally, we prove the second claim. As $\sigma_{gK}$ increases, $\sqrt{\frac{\sum_{k\in \full} \sigma_{gk}^{-2}  }{\sigma^{-2} +\sum_{k\in \full} \sigma_{gk}^{-2}}}$ decreases. Thus, the quantity
$$\frac{\sqrt{\frac{\sum_{k \in \subb} \sigma_{gk}^{-2}  }{\sigma^{-2} +\sum_{k \in \subb} \sigma_{gk}^{-2}}}}
    {\sqrt{\frac{\sum_{k\in \full} \sigma_{gk}^{-2}  }{\sigma^{-2} +\sum_{k\in \full} \sigma_{gk}^{-2}}} } \textsc{HR} \left(
    \Phi^{-1}\left(1-\frac{\tau_g C}{\pi_g}\right)\right)$$
increases. By Lemma~\ref{lemma:HR}, the hazard rare ($\mathrm{HR}$) is increasing so its inverse $\mathrm{HR}^{-1}$ is also increasing. Since the CDF $\Phi$ is increasing, their composition $\Phi(\textrm{HR}^{-1}(\cdot))$ must be also increasing, which in turn implies that the denominator in \Cref{eq:hat_gamma_threshold} is decreasing in $\sigma_{gK}$. Consequently, $\hat{\gamma}_g$ increases as $\sigma_{gK}$ increases.
\end{proof}

Observe that the threshold $\hat{\gamma}_g$ now depends only the characteristics of group $g$ and $\tau$, in contrast to Theorem \ref{thm:threshold_char_without_aa}, where the threshold depends on characteristics of both groups. The result further holds regardless of the economic inequality $\gamma_A - \gamma_B$ between the two groups; under affirmative action with a fixed diversity level, the school conducts the selection process for the two groups separately.
Finally, as expected, if the test has a higher variance for a certain group, then it is more beneficial for that group to drop the test.

\medskip 

\begin{figure}[t]
	\begin{subfigure}[b]{0.48\textwidth}
		\centering
		\includegraphics[width=\linewidth]{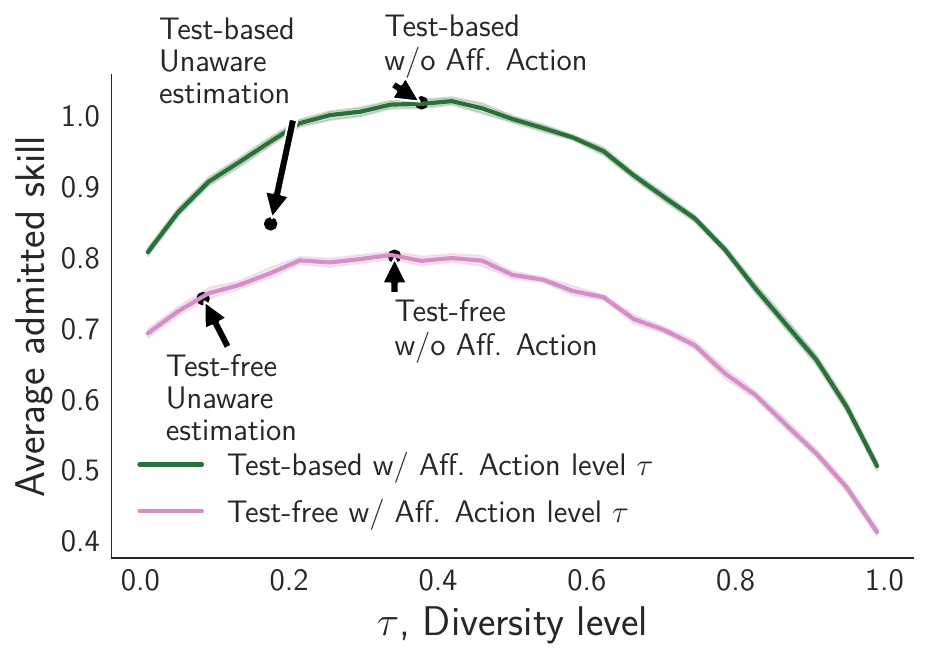}
		\caption{\small Average admitted skill vs. diversity level \normalsize}
		\label{fig:paretosimulations_barriers_withunaware}
	\end{subfigure}
	\hfill
	\begin{subfigure}[b]{0.48\textwidth}
		\centering
		\includegraphics[width=\linewidth]{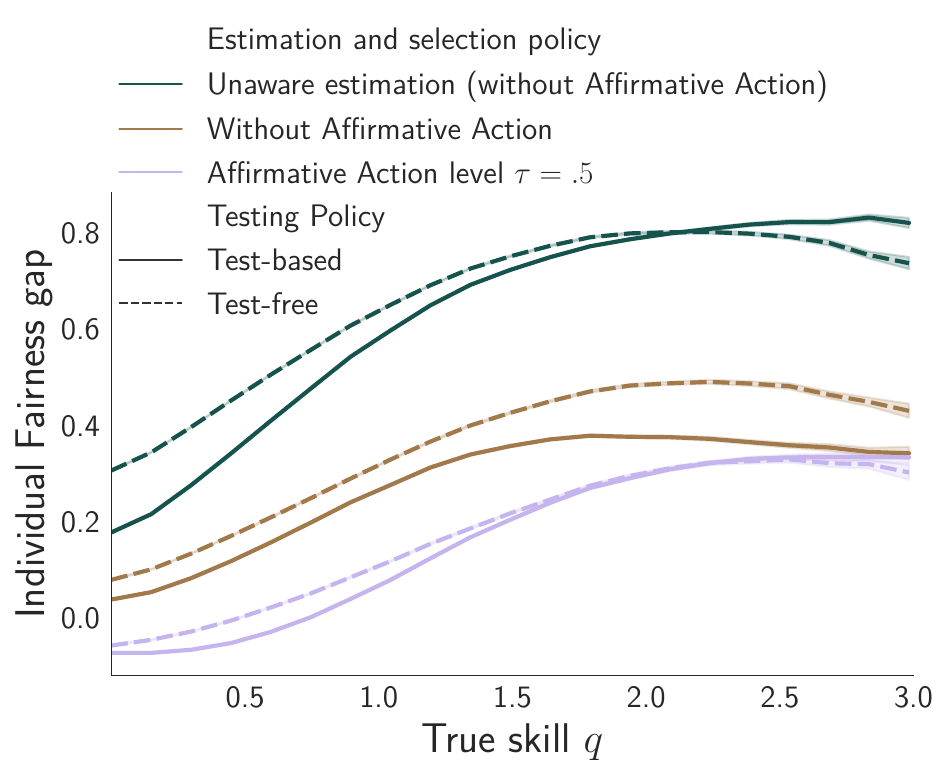}
		\caption{\small Individual fairness gap \normalsize}
		\label{fig:IFsimulations_barriers_withunaware}
	\end{subfigure}

	\caption{Performance of various policies, in simulation in a setting where features are more informative for group $A$, and with testing barriers for group $B$.  Group-unaware policies, analyzed in \Cref{app.group_unaware}, are those in which the school does not use group information in its Bayesian skill estimation. Affirmative action, analyzed in \Cref{sec:affirmative_action}, is defined as fixing the diversity level $\tau$ as a constraint. Figure (a) then shows results for the full range of diversity levels. Affirmative action in general improves both diversity and individual fairness, while dropping the test score has an ambiguous impact. Group-unaware policies generally perform the worst on all metrics. More generally, dropping the test has equivocal effects, depending on the objective (diversity level, academic merit, individual fairness gap) and other policies. The full parameter set can be found in Electronic Companion \ref{appsec:simparams}.
 }
	\label{fig:policycompare_barriers_withunaware}
\end{figure}

\noindent\textbf{Comparing the policies in simulation.} Figure~\ref{fig:policycompare_barriers_withunaware} compares, for one parameter setting, our policies: with and without testing, and with and without \textit{affirmative action} (where a fixed fraction $\tau$ of the admitted class is group B; see Section \ref{sec:affirmative_action}). In Figure~\ref{fig:paretosimulations_barriers_withunaware}, the Pareto curves trace the trade-off between diversity and academic merit, for each testing policy. In this scenario, constraining each group's admitted class to be proportional to its group size (affirmative action at level $\tau = \pi = \frac12$) does not substantially affect academic merit, while improving both group and individual fairness substantially. Furthermore, dropping tests has an equivocal effect: it worsens diversity levels and academic merit, as well as the individual fairness gap in the case without affirmative action. However, it (slightly) improves the individual fairness gap with affirmative action. 

{Figure~\ref{fig:policycompare_barriers_withunaware} also includes \textit{group-unaware estimation} policies, that ignore the social group that a student belongs to; in this case, estimating student skill levels requires calculating the posterior from a mixture of Normal distributions.  Ignoring group attributes is an oft-proposed but often problematic policy proposal to combat bias~\citep{corbett2018measure}.
Perhaps unsurprisingly, group-unaware estimation policies perform most poorly. It worsens both the average academic merit of the admitted class and the diversity level, compared to the policy with group-aware estimation. It also leads to large individual fairness gaps, especially for high-skilled students. More details can be found in Electronic Companion~\ref{app.group_unaware}.}

\end{document}